\documentclass[11pt]{article}

\usepackage[english]{babel}

\usepackage[letterpaper,top=2cm,bottom=2cm,left=3cm,right=3cm,marginparwidth=1.75cm]{geometry}

\usepackage{amsthm,amssymb,amsmath}  
\usepackage{mathtools}
\usepackage[T1]{fontenc}
\usepackage{graphicx}
\usepackage[colorlinks=true, allcolors=blue]{hyperref}
\usepackage[capitalize]{cleveref}
\usepackage{xspace}
\usepackage{thmtools}
\usepackage{thm-restate}
\usepackage{enumitem}
\usepackage{url}
\usepackage{algorithm}
\usepackage[noend]{algpseudocode}
\usepackage{comment}
\usepackage{caption}
\usepackage{authblk}
\usepackage{array}
\usepackage{colortbl}
\usepackage{xcolor}
\usepackage{soul}
\usepackage{multicol}
\usepackage{float}
\usepackage[utf8]{inputenc}
\usepackage{subcaption}
\usepackage[short,nocomma]{optidef}
\definecolor{mygray}{rgb}{0.9,0.9,0.9}
\sethlcolor{mygray}
\definecolor{myblue}{rgb}{0.7,0.8,1.0}
\definecolor{mygreen}{rgb}{0.8,1.0,0.8} 
\definecolor{myorange}{rgb}{1.0,0.8,0.6} 

\definecolor{darkgray}{rgb}{0.7,0.7,0.7} 
\definecolor{darkblue}{rgb}{0.5,0.6,0.8} 
\definecolor{darkgreen}{rgb}{0.5,0.9,0.5}
\definecolor{darkorange}{rgb}{0.9,0.5,0.4}

\newtheorem{problem}{Problem}

\DeclareMathOperator{\SA}{SA}
\DeclareMathOperator{\LCP}{LCP}
\DeclareMathOperator{\LCA}{LCA}
\DeclareMathOperator{\LCE}{LCE}
\DeclareMathOperator{\ISA}{ISA}

\DeclareMathOperator{\ST}{ST}
\DeclareMathOperator{\PT}{PT}

\DeclareMathOperator{\hp}{hp}
\DeclareMathOperator{\sd}{sd}
\DeclareMathOperator{\parent}{parent}
\DeclareMathOperator{\preorder}{pre}
\DeclareMathOperator{\rLeaf}{rLeaf}
\DeclareMathOperator{\str}{str}
\def\dd{\mathinner{.\,.}}
\usepackage{graphicx}   
\usepackage{subcaption} 
\usepackage{xspace}
\usepackage{wrapfig}

\newcommand{\Oh}{\mathcal{O}}
\newcommand{\Ohtilde}{\widetilde{\Oh}}

\usepackage{cite}
\graphicspath{{images}}

\newtheorem{theorem}{Theorem}

\newtheorem{example}{Example}

\usepackage{tikz}
\usetikzlibrary{calc,shapes,shapes.misc,calc,arrows,positioning,trees,decorations.pathmorphing,decorations.markings, decorations.pathreplacing, patterns, calligraphy, backgrounds}

\def\dd{\mathinner{.\,.}}
\newcommand{\cO}{\mathcal{O}}
\newcommand{\CPM}{\textsc{CPM}\xspace}

\newcommand{\CPRI}{\textsc{CPRI}\xspace}
\newcommand{\CPRIS}{\textsc{SA}\xspace}

\newcommand{\CPC}{\textsc{CPC}\xspace}
\newcommand{\CI}{\textsc{CI}\xspace}

\newcommand{\B}{\mathbf{B}\xspace}
\newcommand{\M}{\mathbf{M}\xspace}

\newcommand{\IM}{\textsc{IM}\xspace}
\newcommand{\EM}{\textsc{EM}\xspace}
\newcommand{\RI}{\textsc{RI}\xspace}

\newcommand{\tA}{\ensuremath{\mathtt{A}}}
\newcommand{\tC}{\ensuremath{\mathtt{C}}}
\newcommand{\tG}{\ensuremath{\mathtt{G}}}
\newcommand{\tT}{\ensuremath{\mathtt{T}}}

\newcommand{\chr}{\textsc{CHR}\xspace}
\newcommand{\sars}{\textsc{SARS}\xspace}

\title{Contextual Pattern Mining and Counting}
\author[1]{Ling Li}
\author[2]{Daniel Gibney}
\author[3]{Sharma V. Thankachan}
\author[4,5]{Solon P.\ Pissis}
\author[1]{\\Grigorios Loukides}
\affil[1]{King's College London, London, UK}
\affil[2]{University of Texas at Dallas, Texas, USA}
\affil[3]{North Carolina State University, North Carolina, USA}
\affil[4]{CWI, Amsterdam, The Netherlands}
\affil[5]{Vrije Universiteit, Amsterdam, The Netherlands}

\date{\vspace{-.5cm}}

\begin{document}

\maketitle

\thispagestyle{empty}

\begin{abstract}
Given a string $P$ of length $m$, a longer string $T$ of length $n>m$, and two integers $l\geq 0$ and $r\geq 0$, the \emph{context} of $P$ in $T$ is \emph{the set} of all string pairs $(L,R)$, with $|L|=l$ and $|R|=r$, such that the string $LPR$ occurs in $T$. We introduce two problems related to the notion of context: (1) the \emph{Contextual Pattern Mining} (\CPM) problem, which given $T$, $(m,l,r)$, and an integer $\tau>0$, asks for outputting the context of each substring $P$ of length $m$ of $T$, provided that the size of the context of $P$ is at least $\tau$; and (2) the \emph{Contextual Pattern Counting} (\CPC) problem, which asks for preprocessing $T$ so that the size of the context of a given query string $P$ of length $m$ can be found efficiently. Both problems have direct applications in text mining and bioinformatics and are challenging to solve for realistically long strings. 

For \CPM, we propose a linear-work algorithm that either uses only internal memory, or a bounded amount of internal memory and external memory, which allows much larger datasets to be handled. For \CPC, we propose an $\Ohtilde(n)$-space index that can be constructed in $\Ohtilde(n)$ time and answers queries in $\Oh(m)+\Ohtilde(1)$ time. We further improve the practical performance of the \CPC index by optimizations that exploit the LZ77 factorization of $T$ and an upper bound on the query length. Using billion-letter datasets from different domains, we show that the external memory version of our \CPM algorithm can deal with very large datasets using a small amount of internal memory while its runtime is comparable to that of the internal memory version. Interestingly, we also show that our optimized index for \CPC outperforms an approach based on the state of the art for the reporting version of \CPC [Navarro, SPIRE 2020] in terms of query time, index size, construction time, and construction space, often by more than an order of magnitude. 
\end{abstract}

\clearpage
\setcounter{page}{1}

\section{Introduction}\label{sec:intro}

Strings (sequences of letters over some alphabet) encode data arising from different sources such as DNA sequences~\cite{NGS}, natural language text~\cite{DBLP:books/lib/JurafskyM09}, and log files (e.g., sequences of visited webpages in a webserver log~\cite{DBLP:journals/pvldb/Boncz0L20} or visited locations in a location-based application~\cite{gis}). Analyzing such strings is key in applications including DNA sequence analysis~\cite{NGS}, document analytics~\cite{simonjea}, network analytics~\cite{ndsi16}, and location-based service provision~\cite{locis}. At the heart of these applications, there are two fundamental string-processing problems: (1) the frequent pattern mining problem~\cite{fisher}; and (2) the text indexing problem~\cite{DBLP:books/daglib/0020103}. 
Given a string $T$ and an integer $\tau>0$, frequent pattern mining asks for each string $P$ that occurs as a substring of $T$ at least $\tau$ times. Text indexing asks to preprocess a string $T$ into a compact data structure that supports efficient pattern matching queries; e.g., \emph{report} the set of starting positions of $P$ in $T$; or \emph{count} the size of the latter set. In what follows, we refer to $T$ as the \emph{text} and to $P$ as the \emph{pattern}. 

In this work, we shift the focus of these problems from $P$ to the string $L$  of length $l$ and the string $R$ of length $r$, which occurs immediately to the left and to the right of $P$ in $T$, respectively. The \emph{set} of such string pairs $(L, R)$ is referred to as the \emph{context} of $P$ in $T$ and denoted by $\mathcal{C}_{T}(P,l,r)$. 

We introduce the following two problems:

\begin{problem}[Contextual Pattern Mining (\CPM)]
    Given a text $T$ of length $n = |T|$ and a tuple of integers $(\tau, m, l, r)$, where $\tau, m > 0$ and $l, r \geq 0$, output 
    $\mathcal{C}_T(P,l,r)$, for each substring $P$ of $T$ with length $m$ such that $|\mathcal{C}_T(P,l,r)|\geq \tau$. 
\end{problem}

\begin{example} \label{EX:CPC}
For text $T$ below and $(\tau, m, l, r)=(3, 2, 2,1)$, the \CPM solution is 
$\mathcal{C}_T(\texttt{AA},2,1)=\{(\texttt{CT},\texttt{G}), (\texttt{AG}, \texttt{G}), (\texttt{AG},\texttt{T}), (\texttt{TG}, \texttt{C})\}$. Each occurrence of $P=\texttt{AA}$ and its context is highlighted by a different color. We have $|\mathcal{C}_T(\texttt{AA},2,1)|\geq \tau=3$, while $|\texttt{AA}|=m=2$, $|\texttt{CT}|=|\texttt{AG}|=|\texttt{TG}|=l=2$, and $|\texttt{G}|=|\texttt{T}|=|\texttt{C}|=r=1$.    
\begin{center}
\setlength{\tabcolsep}{0.05cm} 
\scalebox{1}{
\begin{tikzpicture}

    \node[anchor=north west] (table) at (0, 0) {
        \begin{tabular}{c|*{15}{>{\centering\arraybackslash}p{0.025\textwidth}}}
        $i$ & $1$ & $2$ & $3$ & $4$ & $5$ & $6$ & $7$ & $8$ & $9$ & $10$ & $11$ & $12$ & $13$ & $14$ & $15$ \\ \hline
        $T$ & \tC & \tT & \cellcolor{mygray}\tA & \cellcolor{mygray}\tA & \tG & \cellcolor{myorange}\tA & \cellcolor{myorange}\tA & \tG & \cellcolor{mygreen}\tA & \cellcolor{mygreen}\tA & \tT & \tG & \cellcolor{myblue}\tA & \cellcolor{myblue}\tA & \tC
        \end{tabular}
    };

\draw[line width=1.2pt, darkgray] ([yshift=-3pt]table.south west) ++(0.5cm, 0.2cm) -- ++(1.0cm, 0); 
\draw[line width=1.2pt, darkgray] ([yshift=-3pt]table.south west) ++(2.5cm, 0.2cm) -- ++(0.50cm, 0); 
\draw[line width=1.2pt, myorange] ([yshift=-3pt]table.south west) ++(2.0cm, 0.13cm) -- ++(1.0cm, 0); 
\draw[line width=1.2pt, myorange] ([yshift=-3pt]table.south west) ++(3.95cm, 0.13cm) -- ++(0.5cm, 0); 
\draw[line width=1.2pt, darkgreen] ([yshift=-3pt]table.south west) ++(3.45cm, 0.20cm) -- ++(1.cm, 0); 
\draw[line width=1.2pt, darkgreen] ([yshift=-3pt]table.south west) ++(5.45cm, 0.20cm) -- ++(0.5cm, 0);

\draw[line width=1.2pt, darkblue] ([yshift=-3pt]table.south west) ++(5.45cm, 0.13cm) -- ++(1cm, 0); 
\draw[line width=1.2pt, darkblue] ([yshift=-3pt]table.south west) ++(7.4cm, 0.13cm) -- ++(0.5cm, 0);
\end{tikzpicture}
}
\end{center}

\end{example}

\begin{problem}[Contextual Pattern Counting (\CPC)]
A text $T$ of length $n = |T|$ is given for preprocessing. Given a query consisting of a tuple $(P, l, r)$, where $P$ is a pattern of length $m$ and $l, r \geq 0$ are integers, output
the size of $~\mathcal{C}_T(P,l,r)$. 
\end{problem}

\begin{example}[cont'd from Example~\ref{EX:CPC}] After preprocessing $T$, the query \( (P = \texttt{AA}, l = 2, r = 1) \) in \CPC outputs $4$. 
\end{example}

\subparagraph{Motivation.}~An important application of \CPM comes from bioinformatics. Specifically, the pair of strings $(L,R)$ in the context of a pattern corresponds to
\emph{flanking sequences}. The flanking sequences are highly relevant to understanding the role of genes~\cite{matlock2021flanker} and DNA shape fluctuations~\cite{li2024predicting}, and their lengths $l$, $r$ are relatively short, e.g., $9$, $12$, or $15$~\cite{plosgen}. Thus, \CPM outputs the set of flanking sequences for each pattern of a given length $m$ that has at least $\tau$ distinct contexts; $m$ and $\tau$ are to filter out patterns that are too short or have a very small number of distinct contexts, if needed. Similarly, \CPC outputs the size of the set of flanking sequences when one has a specific pattern $P$ of interest (e.g., a CpG site in a certain type of DNA region called CpG island~\cite{nucleicacidres}, which helps to study the evolutionary history of mammalian genomes~\cite{nucleicacidres}). 

Another application of \CPM is in text mining, where the text represents a document and its letters represent the words in the document. Mining the set of contexts of a word or phrase is important for explaining terms~\cite{kostic2023mapping}. For example, the meaning of the word $P=\texttt{bank}$ is different if  $L=\texttt{south}$ and $R=\texttt{of~Thames}$ compared to when $L=\texttt{State Savings}$ and $R=\texttt{of Ukraine}$; in the first case $P$ refers to the bank of a river, while in the second case to a financial institution. In addition, text mining systems (e.g.,~\cite{voyant}) compute the size of the set of contexts of a user-specified pattern, and this is precisely the objective of the \CPC problem. 

\begin{table}[t]
\caption{Contexts for the user sessions $S_1$ and $S_2$ in Example~\ref{exp:log}.}  \label{tab:BrowsingBehaviors}
\centering
\begin{tabular}{|c|c|}
\hline
         &  The context of $P=\texttt{/login}$                                  \\ \hline
$S_1$ & \begin{tabular}[c]{@{}c@{}}\{(\texttt{/search}, \texttt{/error}), (\texttt{/error}, \texttt{/help}), $\ldots$, (\texttt{/contact}, \texttt{/forum}),\\(\texttt{/forum}, \texttt{/admin}), (\texttt{/admin}, \texttt{/contact})\}\end{tabular} \\ \hline
$S_2$  &    \{(\texttt{/home}, \texttt{/dashboard}), (\texttt{/cart}, \texttt{/order})\} \\ \hline
\end{tabular}
\end{table}

\CPM and \CPC are also useful in log analysis. An application of  \CPM in this domain is 
\emph{structured event template} identification (i.e., identifying key information associated with key system events such as ``logged'' or ``received'') from a system log~\cite{he2021survey}. 
For example, a substring \texttt{logged} may appear together with many distinct 
contexts such as $L=\texttt{User1}$ and $R=\texttt{in}$, or $L=\texttt{User2}$ and $R=\texttt{out}$, etc. \CPM helps to create structured event templates, by mining the context  for any pattern that has context size at least $\tau$. An  application of \CPC is \emph{web log anomaly detection}, where user sessions with large  contexts for specific patterns correspond to a potential security threat called \emph{web application probing}~\cite{asselin2016anomaly}; see Example~\ref{exp:log}, which presents a concrete scenario. 

\begin{example}\label{exp:log}
A web server log dataset~\cite{zaker2019online} is comprised of 1,288 strings, each modeling a user session comprised of actions such as \texttt{/order}, \texttt{/login}, \texttt{/profile}, etc. By applying \CPC on each string with $(P,l,r)=(\texttt{/login}, 1,1)$, we found that the result of \CPC in $14$ user sessions was at least $11$. These sessions had very large contexts, indicating web application probing. For instance, the context for such a session, denoted by $S_1$, and for a normal session, denoted by $S_2$, is depicted in Table~\ref{tab:BrowsingBehaviors}.
\end{example}

\subparagraph{Challenges.}~\CPC and \CPM are challenging to solve efficiently on realistically large strings. In typical bioinformatics datasets, for example, the length $n$ of $T$ is often in the order of millions or even billions and the size of 
$\mathcal{C}_T(P,l,r)$ for a typical pattern $P$ is in the order of several thousands~\cite{navarro2020contextual}. This makes straightforward solutions to \CPC, which: (1)  construct an index on $T$ to efficiently locate all the occurrences of a pattern $P$ in $T$; (2) find the pair $(L,R)$ associated to each occurrence; and (3) construct from thereon the context for $P$, too costly~\cite{navarro2020contextual}. For \CPM, such solutions are also very costly, since they need to consider each length-$m$ substring of $T$ and process the (many) occurrences of the substring independently.  

To address \CPC more efficiently than the straightforward approach above, one can employ an index that  \emph{reports} $\mathcal{C}_T(P,l,r)$ for a given pattern $P$ of length $m$~\cite{navarro2020contextual} and then \emph{count} the size of $\mathcal{C}_T(P,l,r)$. We  refer to this approach as \CPRI (for Context Pattern Reporting Index); see Section~\ref{sec:related} for details. 
However, \CPRI is impractical because its query time bound is linear in $|\mathcal{C}_T(P,l,r)|$, which is often very large as mentioned above.  
\CPRI can also be used for addressing \CPM: every length-$m$ substring $P$ of $T$ is used as a query and $P$ becomes part of the output if $|\mathcal{C}_T(P,l,r)|\geq \tau$. This, however, is also impractical: it takes $\cO(mn)$ time in total (as there are $\cO(n)$ substrings of length $m$ in $T$) and $\cO(n)$ space, and $n$ is typically in the order of millions or billions in practice. 

\subparagraph{Contributions and Paper Organization.}~In addition to introducing the \CPM and \CPC problems, our work makes the following contributions.

{\bf 1.}~Motivated by the fact that it is impractical to solve \CPM for realistically long strings via \CPRI, we propose a different linear-work approach and an algorithm that has two versions to realize it: the first uses exclusively internal memory and runs in $\cO(n)$ time and space; the second uses also external memory\footnote{The objective in the external-memory model is to process very large datasets (that do not fit into the main memory) efficiently by fetching them from and accessing them in the hard disk in blocks.} and takes $\cO(\textsf{sort}(n))$ I/Os (i.e., the I/Os required to sort $n$ elements~\cite{DBLP:journals/fttcs/Vitter06}). The benefit of the external memory  version is that one can specify the maximum amount of internal memory to be used. This allows processing very large datasets on commodity hardware. See Section~\ref{sec:external}. 

\begin{example}
Our external memory version was applied to a genomic dataset \cite{chr} of $n=16\cdot 10^9$ letters in total with parameters $(\tau, m, l, r)=(1000, 6, 9,9)$ and the main memory was limited to $32$ GBs. It required about $5.4$ hours to terminate. On the other hand, neither \CPRI nor the internal memory version terminated, as they needed more than $512$ GBs of memory: the total amount of memory on our system.  
\end{example}

{\bf 2.}~We design an index for \CPC  answering queries in 
$\Oh(m+(\log n/ \log\log n)^3)$ time and 
occupying 
$\Oh(n\cdot(\log n)^3/(\log\log n)^2)$ 
space.  
This index is constructible in $\Ohtilde(n)$ time and space.
The benefit of our index is that the query time bound does \emph{not} depend on $|\mathcal{C}_T(P,l,r)|$. This is useful because in the aforementioned applications of \CPC the context size $|\mathcal{C}_T(P,l,r)|$ is very large. 
We also propose optimizations to the index which exploit 
the LZ77 factorization~\cite{DBLP:journals/tit/ZivL77} of $T$ 
and an upper bound $B$ on query length: $ l + m + r \leq B$. These enable us to replace $n$ in the bounds of our index by $\min(B\cdot z, n)$, where $z$ is the number of phrases in the LZ77 factorization of $T$. See Section~\ref{sec:Counting Indexes}. As $z$ is often much smaller than $n$ ~(see Table~\ref{tab:data}), our index offers orders of magnitude faster query times than 
\CPRI~\cite{navarro2020contextual},  
while it occupies less memory, and it is faster and more space-efficient to construct.

\begin{example} 
For  the first $n=8\cdot 10^{9}$ letters of the  genomic dataset of~\cite{chr}, our (optimized) index was constructed in $3821$ seconds, using $127$ GBs of memory, and occupies $49$ GBs. On the contrary, \CPRI~\cite{navarro2020contextual} was constructed in $9737$ seconds, using $500$ GBs of memory, and occupies $425$ GBs. Remarkably, the average query time of our index (over all length-$6$ substrings of the dataset as $P$, and $l=r$ in $[3,15]$) is at least $2$ and up to $3$ orders of magnitude faster than \CPRI (see Fig.~\ref{fig:intro}). 
\end{example}

\begin{figure}[ht]
    \centering
    \includegraphics[width=0.5\columnwidth]{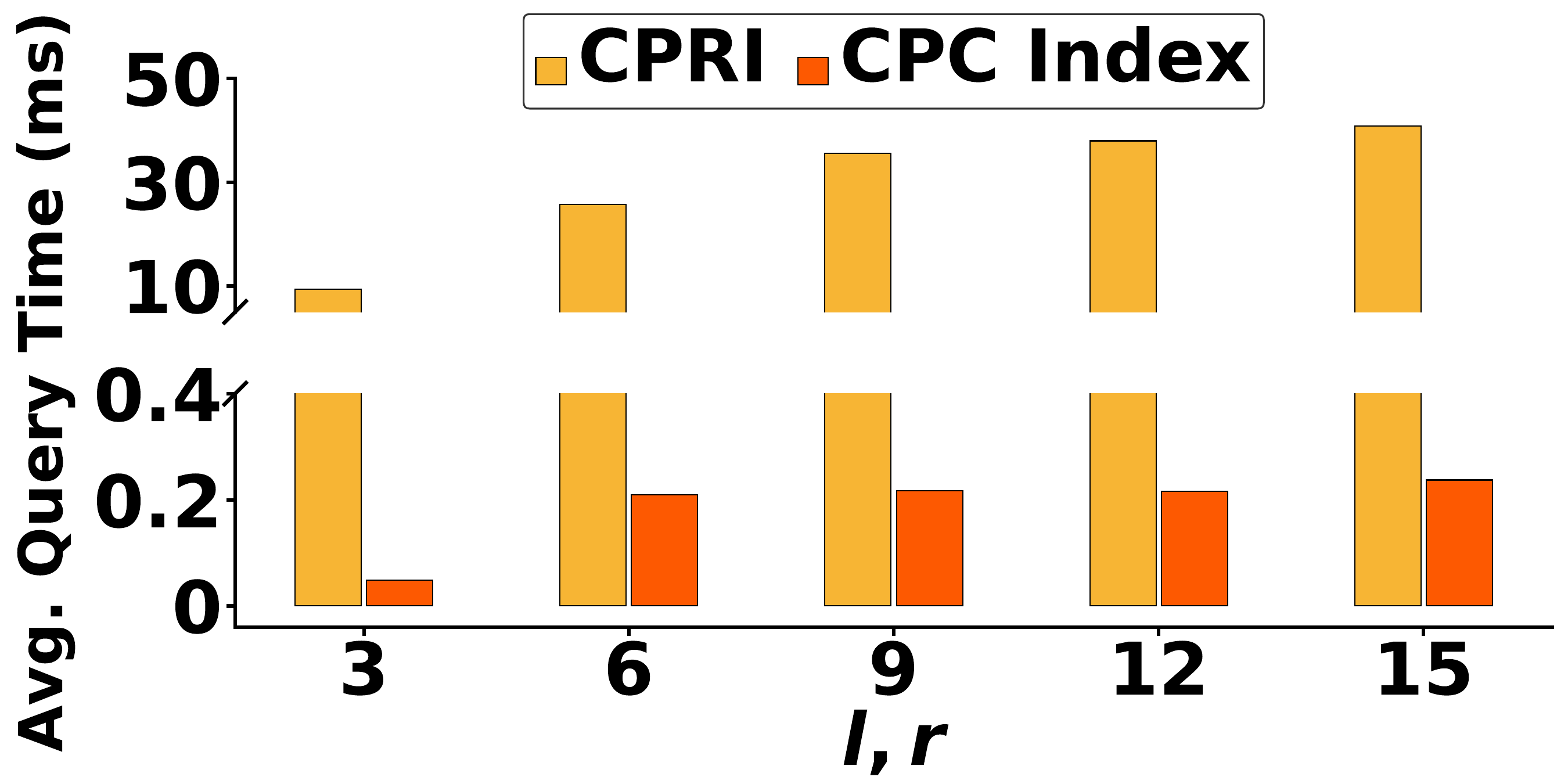}
    \caption{Average query time vs. $l=r$ for  all length-$6$ substrings.}
    \label{fig:intro}
\end{figure}

{\bf 3.}~We conducted experiments on $5$ real-world datasets with sizes up to $16$ GBs. The results show that the external memory version of our \CPM algorithm can deal with much larger datasets than the internal memory one, while requiring comparable time, which is very encouraging. 
Notably, we also show that our (optimized) index for \CPC substantially outperforms \CPRI in terms of all four measures of efficiency~\cite{pvldb23}: average query time; index size; construction time; and construction space. For example, on prefixes of lengths $1$, $2$, $4$, and $8$ billion letters of the genomic dataset of~\cite{chr}, our index was constructed on average 
$52\%$ faster and using less than $75\%$ space compared to \CPRI, while  its size was $74\%$ smaller. Also, for a query workload asking for all length-$9$ substrings of the dataset  and $l=r=9$, our index answered queries $45$ times faster on average. See Section~\ref{sec:experiments}.

We organize the rest of the paper as follows. Section~\ref{sec:preliminaries} provides the necessary notation, definitions, and tools. 
Section~\ref{sec:related} reviews related work, and Section~\ref{sec:conclusion} concludes the paper. 

\section{Preliminaries}\label{sec:preliminaries}

In this section, we discuss some basic concepts.

\subparagraph{Strings.}~An \emph{alphabet} $\Sigma$ is a finite set of elements called \emph{letters}.
We consider throughout that $\Sigma$ is an integer alphabet.  
For a string $T[1\dd n] \in \Sigma^n$, we denote its length $n$ by $|T|$ and its $i$th letter by $T[i]$. A \emph{substring} of $T$ starting at position $i$ and ending at position $j$ of $T$ is denoted by $T[i\dd j]$. A substring $S$ of $T$ may have multiple occurrences in $T$.
We thus characterize an \emph{occurrence} of $S$ in $T$ by its \emph{starting position} $i\in[1,n]$; i.e.,
$S=T[i\dd i+|S|-1]$.
A substring of the form $T[i\dd n]$ is a \emph{suffix} of $T$ and a \emph{proper suffix} if $i > 1$. 
A substring of the form $T[1\dd i]$ is a \emph{prefix} of $T$ and a \emph{proper prefix} if $i < n$. 
The \emph{reverse} of $T$ is denoted by $T^R$. For example, for string   $T=\mathtt{CTAAG}$, the reverse is $T^R=\mathtt{GAATC}$. 
The \emph{concatenation} of strings $X$, $Y$ is denoted by $X\cdot Y$ (or by $XY$).
For convenience, we assume that $T$ always ends with a terminating letter $T[n] = \$$ such that $\$$ occurs only at position $n$ of $T$ and is the lexicographically smallest letter. 

\subparagraph{Compact Tries, Suffix Trees, and Suffix Arrays.}~A \emph{compact trie} is a trie in which each maximal branchless path is replaced with a single edge whose label is a string equal to the concatenation of that path's edge labels. 
For a node $v$ in compact trie $\mathcal{T}$, $\str(v)$ is the concatenation of edge labels on the root-to-$v$ path. We define the \emph{string depth} of a node $v$ as $\sd(v) = |\str(v)|$. The \emph{locus} of a pattern $P$ is the node $v$ in $\mathcal{T}$ with the  smallest string depth such that $P$ is a prefix of $\str(v)$. We also denote the \emph{parent} of a node $v$ by $\parent(v)$ and the \emph{lowest common ancestor}  of two nodes $u$ and $v$ by $\LCA(u,v)$.

\begin{example}
    The tree in Fig.~\ref{fig:ST_PT} (left) is a compact trie. The  edge labeled $\texttt{na}$ replaces two edges labeled $\texttt{n}$ and $\texttt{a}$ in the (non-compacted) trie for the same string $\texttt{banana\$}$. For node $u_2$, $\str(u_2)=\texttt{ana}$, which is the concatenation of the edge labels $\texttt{a}$ and $\texttt{na}$ on the path from root to $u_2$. The string depth of node $u_2$ is $\sd(u_2)=|\texttt{ana}|=3$. The locus of pattern $P=\texttt{ana}$ is $u_2$, since $u_2$ is the node with the smallest string depth such that $P$ is a prefix of $\str(u_2)=\texttt{ana}$. The parent of node $u_2$ is $\parent(u_2)=u_1$. The lowest common ancestor of nodes $\ell_2$ and $u_2$ is $\LCA(\ell_2,u_2)=u_1$.
\end{example}

\begin{figure}[t]
    \centering
    \includegraphics[width=.7\linewidth]{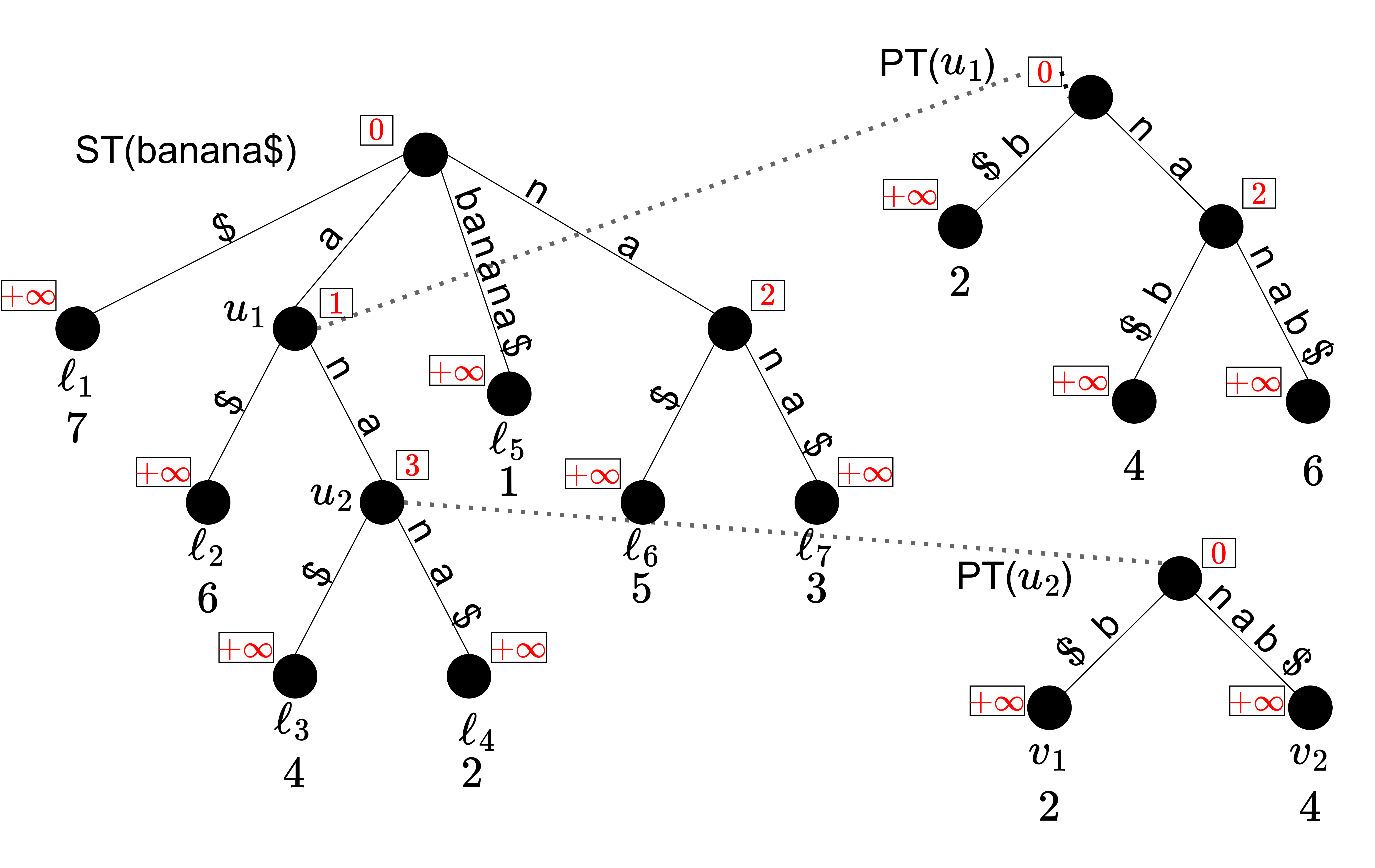}
    \caption{(Left) The suffix tree for the string $\texttt{banana\$}$. (Right) The prefix trees $\PT(u_1)$ and $\PT(u_2)$. The string depths in all trees are indicated in boxes and each leaf is labeled with the corresponding suffix array value.}
    \label{fig:ST_PT}
\end{figure}

The \emph{suffix tree} of a string $T[1\dd n]$, denoted by $\ST(T)$, is the  compact trie of all suffixes of $T$ with its leaves, $\ell_1, \hdots,\ell_n$, ordered in ascending lexicographical rank of the corresponding suffixes~\cite{DBLP:conf/focs/Weiner73}. We associate  each leaf with an id, $1, \ldots, n$, so that the $i$th leaf from left to right has $id$ equal to $i$, and set  its string depth to $+\infty$. See Fig.~\ref{fig:ST_PT} (left) for an example of the suffix tree for string $T=\texttt{banana\$}$. The suffix tree can be constructed in $\cO(n)$ time for integer alphabets~\cite{DBLP:conf/focs/Farach97}. 

The \emph{suffix array} $\SA[1\dd n]$ of a string $T$ is a permutation of $\{1,\ldots,n\}$ such that $T[\SA[i]\dd n]$ is the $i$th smallest suffix when ordered  lexicographically~\cite{DBLP:journals/siamcomp/ManberM93}.  The \emph{inverse suffix array} $\ISA[1\dd n]$ is a permutation such that $\ISA[\SA[i]] = i$. 
In Fig.~\ref{fig:ST_PT} (left), the numbers under the leaves induce the $\SA[1\dd 7]=[7,6,4,2,1,5,3]$. Also, $\ISA[1\dd 7]=[5,4,7,3,6,2,1]$. 

We use $\LCE(i,j)$ to denote the length of the \emph{longest common prefix} (LCP) of the suffixes $T[i\dd n]$ and $T[j\dd n]$. Note that by definition, $\LCE(i,j) = \sd(\LCA(\ell_{\ISA[i]}, \ell_{\ISA[j]}))$, and that this quantity can be computed in $\cO(1)$ time after an $\cO(n)$-time preprocessing of the suffix tree~\cite{DBLP:journals/siamcomp/HarelT84}. 

\begin{example}
For $T=\texttt{banana\$}$ 
in Fig.~\ref{fig:ST_PT},  
$\LCE(3,5)=2$ because the longest common prefix of  $T[3\dd 7]=\texttt{nana\$}$ and 
$T[5\dd 7]=\texttt{na\$}$ is $\texttt{na}$ which has length $2$. Note that  
$\ISA[3]=7$, $\ISA[5]=6$, and $\LCA(\ell_7, \ell_6)$ is the lowest common ancestor of $\ell_7$ and $\ell_6$, which has string depth $2$. Thus, $\LCE(3,5)=\sd(\LCA(\ell_7, \ell_6))=2$.   
\end{example} 

The $\LCP[1\dd n]$ array of a string $T$ stores the length of the LCP of lexicographically adjacent suffixes~\cite{DBLP:journals/siamcomp/ManberM93}. For $j>1$, $\LCP[j]$ stores the length of the LCP between the suffixes starting at $\SA[j-1]$ and $\SA[j]$, and $\LCP[1]=0$. Given $\SA$, the $\LCP$ array of $T$ can be computed in $\cO(n)$ time~\cite{DBLP:conf/cpm/KasaiLAAP01}. 

\begin{example}
The $\LCP$ array of $T=\texttt{banana\$}$ is $[0,0,1,3,0,0,2]$. 
For example, the length of the LCP between the suffix $\texttt{a\$}$  corresponding to $\SA[2]=6$ and the suffix $\$$  corresponding to $\SA[1]=7$ is $0$, while the length of the LCP  between the suffix corresponding to $\SA[2]$ and the suffix $\texttt{ana\$}$ corresponding to 
$\SA[3]$ is $1$. 
\end{example}

\subparagraph{Orthogonal Range Counting.}~We preprocess a multiset $\mathcal{P}$ of $N$ $d$-dimensional points.\footnote{Multiplicities greater than one can be handled via rank space mapping~\cite{DBLP:journals/tcs/DurocherEMT15}.} A query then consists of $d$ intervals defining a $d$-dimensional rectangle $[a_1,b_1] \times\dots \times[a_d,b_d]$ and the output (query answer) is the number of points in the submultiset $\mathcal{P} \cap [a_1,b_1] \times\dots \times[a_d,b_d]$ (i.e., how many points in $\mathcal{P}$ are contained in the rectangle including multiplicities). 

\begin{example}
    Let $\mathcal{P}=\{(1,1), (1,3),(2,2),(4,1),(4,2),(5,3)\}$. Consider a query that consists of $2$ intervals, one $[1,3]$ (for the first dimension) and another $[1,3]$ (for the second dimension); see Fig.~\ref{fig:grid}. This query is represented by the blue rectangle and represents all points $(x,y)$ with $x\in[1,3]$ and $y\in[1,3]$. The output of the query is $3$ because the size of the submultiset $\mathcal{P}\cap [1,3]\times [1,3]$ is $3$: the rectangle encloses the points $(1,1)$, $(1,3)$, and $(2,2)$ (blue points in Fig.~\ref{fig:grid}).

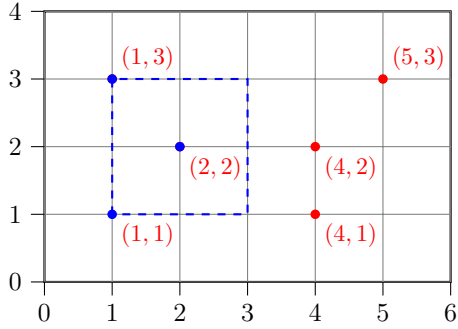
\begin{figure}
    \centering
\begin{tikzpicture}[scale=0.9, transform shape]

\draw[thick] (0,0) rectangle (6,4);

\draw[step=1cm,gray,thin] (0,0) grid (6,4);

\foreach \x in {0,1,2,3,4,5,6} {
    \draw (\x,0) -- (\x,-0.2) node[below] {\x};
}

\foreach \y in {0,1,2,3,4} {
    \draw (0,\y) -- (-0.2,\y) node[left] {\y};
}

\fill[red] (1,1) circle (1pt) node[below right] {\small $(1,1)$};
\fill[red] (2,2) circle (2pt) node[below right] {\small $(2,2)$};
\fill[red] (1,3) circle (2pt) node[above right] {\small $(1,3)$};
\fill[red] (4,1) circle (2pt) node[below right] {\small $(4,1)$};
\fill[red] (4,2) circle (2pt) node[below right] {\small $(4,2)$};
\fill[red] (5,3) circle (2pt) node[above right] {\small $(5,3)$};

\draw[blue, thick, dashed] (1,1) rectangle (3,3);

\fill[blue] (2,2) circle (2pt);
\fill[blue] (1,1) circle (2pt);
\fill[blue] (1,3) circle (2pt);

\end{tikzpicture}
    \caption{Example of orthogonal range counting.}
    \label{fig:grid}
\end{figure}
\end{example}

Our theoretical results use a data structure by J{\'{a}}J{\'{a}} et al.~\cite{DBLP:conf/isaac/JaJaMS04}.
This data structure occupies $\Oh(N(\log N/\log\log N)^{d-2})$ space, it can be constructed in $\Ohtilde(N)$ space and time, and it answers queries in $\Oh((\log N / \log \log N)^{d-1})$ time. 
As this data structure is complex and (as far as we know) has never been implemented, for our  implementations, we evaluated practical orthogonal range counting data structures, such as Range trees~\cite{bentley1978decomposable,DBLP:conf/focs/Lueker78}, KD trees~\cite{bentley1975multidimensional}, and  $R^*$-trees~\cite{beckmann1990r}; see Section~\ref{sec:experiments}.

\begin{figure}[ht]
    \centering
    \scalebox{1}
    {
    \begin{tabular}{c||c|c|ccc|cc|cccccccccc}
    $i$&$1$&$2$&$3$&$4$&$5$&$6$&$7$&$8$&$9$&$10$&$11$&$12$&$13$&$14$ &$15$\\ \hline
    $T$ & \tA & \tT &  \tA & \tT &\tA &  \tA &  \tA & \tT &  \tA &  \tA & \tA &  \tT & \tA &  \tA & \tA
    \end{tabular}
    }
    \begin{tikzpicture}[remember picture, overlay]
        \draw[->, thick, black] 
        (-9.2, -0.4) to[out=220, in=320] (-10.4, -0.4);
        \draw[->, thick, black] 
        (-7.3, -0.4) to[out=220, in=320] (-7.9, -0.4);
        \draw[->, thick, black] 
        (-6, -0.4) to[out=210, in=330] (-8.5, -0.4);
    \end{tikzpicture}

    \caption{An example LZ77 factorization of $T$ with $n=|T|=15$ and $z = 5$. Phrase boundaries are at positions $s_1=1$, $s_2=2$, $s_3=3$, $s_4=6$, and $s_5 = 8$. Starting positions of previous occurrences of phrases are indicated with an arrow.}
    \label{fig:LZ77}
\end{figure}

\subparagraph{LZ77.}~LZ77~\cite{DBLP:journals/tit/ZivL77} is a popular form of repetition-based compression and is the basis of the \texttt{zip} compression format~\cite{DBLP:journals/rfc/rfc1951} and several existing compressed text indexes~\cite{DBLP:conf/latin/GagieGKNP14,DBLP:conf/cpm/KreftN11}. The LZ77 factorization partitions a given string $T[1\dd n]$ into \emph{phrases} such that a phrase starting at position $s_i$ is either the earliest occurrence of the letter $T[s_i]$ in $T$ or a maximum length substring $T[s_i\dd j]$ such that there exist $i' < s_i$ and $j' <j$ where $T[i'\dd j'] = T[s_i\dd j]$; see Fig.~\ref{fig:LZ77}. For example, a phrase boundary at position $3$ is a maximum length substring $T[3\dd 5]$, as we have $T[1\dd 3]=T[3\dd 5]=\tA\tT\tA$. 
 We use $z$ to denote the number of LZ77 phrases in $T$ and $s_1,\dots,s_z$ to denote their starting positions in $T$ from left to right. 
The LZ77 factorization of $T$ can be computed in $\Oh(n)$ time~\cite{DBLP:journals/jacm/StorerS82}.

\section{Algorithm for the \CPM Problem}\label{sec:external}

We present the two versions of our algorithm for \CPM. 

\emph{External Memory (EM) Version.}~This version works in the standard external memory (EM) model~\cite{DBLP:journals/fttcs/Vitter06}. In this model, the RAM has size $\M$ and the disk has block size $\B$. There are two common operations for processing a sequence of elements: scanning and sorting. The complexity of EM algorithms is typically measured in I/O operations. For a sequence of $n$ elements, the I/O complexities are $\textsf{scan}(n)=n/\B$, for scanning it, and $\textsf{sort}(n)=n/\B \log_{\M/\B} n/\B$ for sorting it. The I/O complexity of the best $\SA$ and $\LCP$ array construction algorithms~\cite{DBLP:conf/icalp/KarkkainenS03,DBLP:journals/jea/Bingmann0O16} for a string $T$ of length $n$ is $\cO(\textsf{sort}(n))$.

\subparagraph{High-Level Idea.}~In Phase $1$, our algorithm  builds auxiliary data structures over $T$. In Phases $2$ and $3$, it groups together every occurrence of a string $P\cdot R$ of total length $m+r$ (i.e., the length-$m$ substring and the length-$r$ substring right after it), so that there is a group for each such distinct $P\cdot R$. In Phase $4$, it  groups together every occurrence of a string $L$ of length $l$ (i.e., the length-$l$ substring right before $P$), so that there is a group for each distinct $L$ of length $l$. In Phases $5$ and $6$, it combines  the groups, so that, for each substring $L\cdot P \cdot R$ of $T$ of total length $l+m+r$, there is a group representing all occurrences  of $L\cdot P \cdot R$, and then it outputs the  distinct pair(s) of strings $(L,R)$ for each group with at least $\tau$ occurrences.

\subparagraph{Phase 1.}~We construct the $\SA$ and $\LCP$ array for $T$ 
    and the same arrays for $T^R$, which we denote by $\SA^R$ and $\LCP^R$.
    
\subparagraph{Phase 2.}~We partition $[1,n]$ into a family $\mathcal{I}$ of maximal intervals $I_1,I_2,\ldots, I_{|\mathcal{I}|}$ such that for each interval $I\in \mathcal{I}$ and all $i,j\in I$: $T[\SA[i] \dd \SA[i]+m-1]=T[\SA[j] \dd \SA[j]+m-1].$  
    This means that the suffixes of $T$ starting at positions $\SA[i]$ and $\SA[j]$ share a prefix of length $m$. 
    We achieve this by
    scanning the $\LCP$ array and creating a new interval when $\LCP[i] < m$. 
    We also maintain the following association:
    For each $i \in [1, n]$, $\textsf{int-id}(i)$ maps $i$ to its interval id; e.g., $\textsf{int-id}(13)=5$ means that $13\in I_5$.
    This function can be implemented in EM within the same scan of the $\LCP$ array.

\subparagraph{Phase 3.}~We further partition each interval $I\in \mathcal{I}$ 
    into a family $\mathcal{S}(I)$ of subintervals $S_1,S_2,\ldots,S_{|\mathcal{S}(I)|}$ 
    such that for each subinterval $S\in\mathcal{S}(I)$ and all $i,j\in S$: $[\SA[i] \dd\SA[i]+m+r-1]=T[\SA[j] \dd \SA[j]+m+r-1].$ 
    This means that the suffixes of $T$ starting at positions $\SA[i]$ and $\SA[j]$ of $T$ share a prefix of length $m+r$.
    Again we achieve this by scanning the $\LCP$ array
    and mapping each $i$ to its subinterval id $\textsf{sint-id}(i)$.
    As a result, every occurrence $i$ of string $P\cdot R$, for two strings $P$, with $|P|=m$ and $R$, with $|R|=r$, will be mapped to the same pair $(\textsf{int-id}(i), \textsf{sint-id}(i))$. 
    With this information, for each $i \in [1, n]$, we create a $4$-tuple $(\SA[i], i, \textsf{int-id}(i), \textsf{sint-id}(i))$. 

\subparagraph{Phase 4.}~ We partition $[1, n]$ into a family $\mathcal{I}^R$ of maximal intervals $I^R_1,I^R_2,\ldots, I^R_{|\mathcal{I}^R|}$ via scanning $\LCP^R$. The suffixes of $T^R$ within each interval $I^R$ share a prefix of length $l$. 
    For each $i \in [1, n]$, $\textsf{rint-id}(i)$ maps $i$ to its interval id.
    For each $i \in [1, n]$, we create a $2$-tuple $(n - \SA^R[i], \textsf{rint-id}(i))$. 
    
\subparagraph{Phase 5.}~We observe that when $\SA[i] = n - \SA^R[j]$, for some $i,j\in[1,n]$, both encode the starting position of a string $P$ in $T$. Thus, we sort all $4$-tuples $(\SA[i], i, \textsf{int-id}(i), \textsf{sint-id}(i))$ by their first element in EM 
    and also sort separately all $2$-tuples $(n - \SA^R[j], \textsf{rint-id}(j))$ by their first element in EM. We then scan the sorted tuples and merge each pair of a $4$-tuple $(\SA[i], i, \textsf{int-id}(i), \textsf{sint-id}(i))$ and a $2$-tuple $(n - \SA^R[j], \textsf{rint-id}(j))$ with $\SA[i] = n - \SA^R[j]$ into a $5$-tuple $(\SA[i], i, \textsf{int-id}(i), \textsf{sint-id}(i), \textsf{rint-id}(j))$. 
    
\subparagraph{Phase 6.}~We sort the $5$-tuples $(\SA[i],\allowbreak i, \allowbreak\textsf{int-id}(i),\allowbreak \textsf{sint-id}(i), \allowbreak\textsf{rint-id}(j))$ in EM based on the \emph{second element}. 
    In particular, this sorting step will sort the tuples according to $i$, which is the index of $\SA$,
    thus sorting the tuples according to the $\SA$ order.
    We do this in order to be able to cluster all tuples with the same values of $\textsf{int-id}(i)$ and \textsf{sint-id}(i) together. Within each cluster corresponding to the same $\textsf{int-id}(i)$ and $\textsf{sint-id}(i)$, we conduct an additional sort based on $\textsf{rint-id}(j)$ and then the tuples with the same value of $\textsf{rint-id}(j)$ will also be clustered together. Next, we scan through the sorted $\textsf{rint-id}(j)$ within each cluster. Whenever, a distinct $\textsf{rint-id}(j)$ is encountered, we increase a counter $c_I$, where $I = \textsf{int-id}(i)$. For every cluster corresponding to the same $I$, the counter $c_I$ accumulates the results from all distinct $\textsf{sint-id}(i)$ values in that cluster.
    The total count $c_I$ per cluster represents the number of distinct strings $L\cdot P\cdot R$ of each substring $P$ that occurs  in $T$, where $|L| = l$, $|P| = m$, and $|R| = r$.   
    Whenever a distinct $\textsf{int-id}(i)$ is encountered compared to the previous one (indicating a transition to the next distinct $P$), we check if $c_I \geq \tau$. If it holds, we output the $c_I$ distinct $(L,R)$ string pairs. Otherwise, we ignore this $P$ and proceed to the next one.
    
The $\SA$ and $\LCP$ array can be constructed in $\cO(\textsf{sort}(n))$ I/Os.
All steps of the algorithm above either scan or sort $n$ elements in EM.
We thus arrive at the following result.

\begin{theorem}\label{the:EM}
\CPM can be solved using $\cO(\textsf{sort}(n))$ I/Os.
\end{theorem}

Since $\textsf{sort}(n)=n/\B \log_{\M/\B} n/\B$,
by Theorem~\ref{the:EM}, the only parameters that affect our algorithm are $n$, $\B$, and $\M$. 

\emph{Internal Memory (IM) Version.}~The above algorithm can be implemented fully in internal memory (IM). After constructing the $\SA$ and $\LCP$ array in IM in $\cO(n)$ time and space,
we simply follow the steps of the above algorithm by
replacing the external-memory sorting routine with standard radix sort
and the external-memory scanning routine with standard loops (iterations).
Each invocation of radix sort on $n$ tuples takes $\cO(n)$ time and space because
the integers in the tuples are from the range $[1,n]$. Note that this IM algorithm avoids the range data structures used by \CPRI~\cite{navarro2020contextual} and thus is very fast in practice (see Section~\ref{sec:experiments}). 
We obtain the following result.

\begin{theorem}\label{the:IM}
\CPM can be solved in $\cO(n)$ time and space.
\end{theorem}

\section{Index for the \CPC Problem} \label{sec:Counting Indexes}

We first present a warm-up, intuitive approach, based on suffix trees and orthogonal range counting data structures. Unfortunately, this approach will only lead to a $\Ohtilde(n^2)$-space indexing solution. We will then show a more complex, but space-efficient, approach that has the same query time. This will yield our index for \CPC, which we will then optimize. 

\subsection{A Simple Approach}
\label{sec:simple}

\subparagraph{High-Level Idea.}~We use the suffix tree $\ST(T)$ to search $T$.
Upon a query pattern $P$ of length $m$, we find the node representing $P$. 
This node induces a subtree in $\ST(T)$, where each node $u$ at string depth $m+r$
represents a string $P\cdot R$, with $|R|=r$. For each such node $u$, we need to count the number of distinct strings $L$, such that $L\cdot P\cdot R$ occurs in $T$. To do this, for each node $u$, we use a prefix tree $\PT(u)$, which indexes the reversed prefixes of $T$ ending right before $P\cdot R$. We thus need to count the number of nodes $v$ in $\PT(u)$ at string depth $l=|L|$. To count efficiently, we use \emph{preprocessing}: we construct points in a five-dimensional space, for all pairs $(u,v)$, and perform range counting using the appropriate length constraints.

\subparagraph{Construction.}~We construct the suffix tree $\ST(T)$ and the corresponding suffix array $\SA[1\dd n]$. We compute the preorder rank for each node $u$ of $\ST(T)$ (i.e., the rank of $u$ in a preorder traversal of $T$), denoted by $\preorder(u)$. The root of $\ST(T)$ has preorder rank $1$. As mentioned in Section~\ref{sec:preliminaries}, the string depth of each leaf of $\ST(T)$ is set to $+\infty$. 
This is to correctly handle distinct right contexts at the end of string $T$.

For each node $u$ in $\ST(T)$, we construct the \emph{prefix tree} $\PT(u)$ as the compact trie of the reverse of all prefixes in the set $\{ \$T[1\dd \SA[i]-1] \mid i \in [a,b] \}$, where $[a,b]$ is the range of the ids of leaves under $u$ in $\ST(T)$. 
 A letter $\$$ is prepended to the reverse prefix to ensure that each reverse prefix corresponds to a leaf in the prefix tree. 
 Note that, going up from any leaf $\ell$ of $\PT(u)$ to the root (ignoring the $\$$), and then from the root of $\ST(T)$ to the corresponding leaf of $\ell$ in $\ST(T)$ spells the entire $T$.
 To correctly handle left contexts at the start of $T$, we define the string depth of the leaves of $\PT(u)$ as $+\infty$.
See Fig.~\ref{fig:ST_PT} for an example; in $\PT(u_1)$, the branch labeled $\texttt{nanab}$ corresponds to the reverse of the prefix $T[1\dd \SA[2]-1]]=T[1\dd 5]=\texttt{banan}$.

Let $u_P$ be the locus of pattern $P$ in $\ST(T)$, $u$ be a node in $\ST(T)$, and $v$ be a node in $\PT(u)$. Observe that $|\mathcal{C}_T(P,l,r)|$ is equal to the number of pairs $(u,v)$ of nodes for which the following three constraints hold:
\begin{enumerate}
\item $u$ is a node in the subtree of $u_P$;
\item $\sd(\parent(u)) < m + r \leq \sd(u)$; 
\item $\sd(\parent(v)) < l \leq \sd(v)$.
\end{enumerate}

Let $L\cdot P\cdot R$ a string occurring in $T$, with $|L|=l$ and $|R|=r$.
Constraint $1$ ensures that the string $R$ appears after $P$. Constraint $2$ ensures that the string $P\cdot R$ is of length $m+r$ or is at the end of the string $T$ (for the latter case the leaf nodes have $\sd$ equal to $+\infty$). Constraint  $3$ ensures that $L$ has length $l$ or is at the beginning of the string $T$.

However, checking these constraints for each pair $(u,v)$ independently is too costly. Therefore, we preprocess each pair $(u,v)$ where $u$ is a node in $\ST(T)$ and $v$ is a node in $\PT(u)$, by creating the five-dimensional point: 
\[
(\preorder(u), \sd(\parent(u))+1, \sd(u), \sd(\parent(v))+1, \sd(v)).
\]
We construct a range counting structure over these points. 

\subparagraph{Querying.}~To answer a query $(P,l,r)$, we first find the locus $u_P$ of $P$ in $\ST(T)$, and then answer the range counting query: 
\[
[\preorder(u_P), \preorder(\rLeaf(u_P))] \times [0,m+r] \times [m+r,+\infty] \times [0,l] \times [l,+\infty],
\]
where $\rLeaf(u_P)$ denotes the rightmost leaf in the subtree of $u_P$. The first dimension of the query ensures that Constraint 1 holds. This follows from the definition of preorder. Constraints 2 and 3 are clearly enforced by the remaining dimensions. 

\begin{example}
For $T = \texttt{banana\$}$, we construct $\ST(T)$ and a prefix tree for every node in $\ST(T)$; some of these prefix trees are shown in Fig.~\ref{fig:ST_PT}. A range counting structure is constructed over all resulting points. 
 To answer the \CPC query $(P = \texttt{a}, l = 1, r = 2)$, we first find its locus $u_P = u_1$ in $\ST(T)$. Note that $\preorder(u_1) = 3$ and $\preorder(\rLeaf(u_1)) = 7$. 
 We make the query 
 $
 [3,7]\times [0, 3] \times [3,+\infty] \times [0,1]\times [1,+\infty].
 $
 The returned value is $|\mathcal{C}_T(P,l,r)|=3$. To see why, note that Constraint $1$ selects the points corresponding to nodes $u_1$, $u_2$, $\ell_2$, $\ell_3$, and $\ell_4$; i.e., those in the subtree of $u_P = u_1$. 
  Among these, Constraint 2 restricts the selected nodes to $u_2$ and $\ell_2$, as they satisfy $\sd(\parent(u_2)) = 1 < m+r = |P|+r=3 \leq \sd(u_2) = 3$ and $\sd(\parent(\ell_2)) = 1 <m+r 
 \leq \sd(\ell_2) = +\infty$. In $\PT(u_2)$, the nodes $v_1$ and $v_2$ satisfy $\sd(\parent(v_1)) = 0 < l=1 \leq \sd(v_1) = +\infty$ and $\sd(parent(v_2))=0 < l=1 \leq \sd(v_2) = +\infty$ hold. Hence, the points
\resizebox{\linewidth}{!}{%
\begin{minipage}{\linewidth}
\begin{align*}
& (\preorder(u_2), \sd(\parent(u_2))+1, \sd(u_2), \sd(\parent(v_1))+1, \sd(v_1)) \\
& = (5,2,3,1, +\infty) \\
& (\preorder(u_2), \sd(\parent(u_2))+1, \sd(u_2), \sd(\parent(v_2))+1, \sd(v_2)) \\
& = (5,2,3,1, +\infty)
\end{align*}
\end{minipage}
}
\noindent are counted (i.e., one point with multiplicity $2$).   
For leaf $\ell_2$, the prefix tree $\PT(\ell_2)$ (not shown) has a single leaf. For this leaf, Constraint 3 is also satisfied and its corresponding point is counted.  Therefore, a total count of $3$ is returned by the range counting data structure.
\end{example}

\subparagraph{Complexities.}~As the locus of $P$ in $\ST(T)$ can be found in $\Oh(m)$ time and a five-dimensional range query can be answered in $\Oh((\log  n / \log\log n)^4)$ time~\cite{DBLP:conf/isaac/JaJaMS04}, we obtain an $\Oh(m + (\log  n / \log\log n)^4)$ query time complexity. The space complexity of this simple approach is quadratic in $n$ due to the sum of prefix tree sizes being quadratic in the worst case and a point being created for every node in a prefix tree. 

\subsection{The \CPC Index}
\label{sec:CPC_index}

\subparagraph{High-Level Idea.}~Instead of constructing points for every pair $(u,v)$, where $u$ is a node in $\ST(T)$ and $v$ a node in $\PT(u)$, we apply a well-known decomposition technique on the suffix tree, which partitions its nodes into \emph{heavy} and \emph{light}, and treat each node type separately. This allows us to bound the size of the prefix trees constructed by $\widetilde{\cO}(n)$ resulting in an index of $\Ohtilde(n)$ size with the same query time as in the simple approach.  

\subparagraph{Heavy-Path Decomposition.}~Our improved approach makes use of a \emph{heavy-path decomposition}~\cite{DBLP:journals/jcss/SleatorT83} of the suffix tree $\ST(T)$. We define the \emph{subtree size} of node $u$ as the number of leaves in the subtree rooted at $u$. A heavy-path decomposition categorizes the nodes in $\ST(T)$ into two types: \emph{light} and \emph{heavy}. The root $r$ of $\ST(T)$ is light. We then take the child $u$ of $r$ with a largest subtree size (ties broken arbitrarily) and call $u$ heavy. We continue from $u$, constructing a path of heavy nodes rooted at $r$ until the heavy node has no child. We classify all of the nodes adjacent to the heavy path as light and recursively apply the same procedure to each of these light nodes.
We use $\hp(u_l)$ to denote the heavy path rooted at light node $u_l$ and ending with a leaf of $\ST(T)$. See Fig.~\ref{fig:phi_values} (left) for an example.

A key property of heavy-path decomposition is that the number of light nodes on any root-to-leaf path is at most $\log n$~\cite{DBLP:journals/jcss/SleatorT83}. Since each leaf is in at most $\log n$ light-node subtrees, the sum of subtree sizes of all light nodes is $\Oh(n \log n)$. Applying the property, as we describe next, will enable us to reduce the space required from quadratic to $\Ohtilde(n)$.

\subparagraph{Construction.}~We first construct $\ST(T)$ and then perform a heavy-path decomposition to $\ST(T)$. For a light node $u_l$ in $\ST(T)$, we define the prefix tree $\PT(u_l)$ in the same way as in Section \ref{sec:simple}.
Observe that by the stated property of heavy-path decomposition, the sum of sizes of $\PT(u_l)$ for all light nodes $u_l$ is $\Oh(n \log n)$. This follows from the fact that the number of nodes is $\PT(u_l)$ is proportional to the number of leaves in the subtree rooted at $u_l$ in $\ST(T)$.

For every light node $u_l$ and every node $v$ in $\PT(u_l)$, we associate a value, denoted by $\phi(v)$. Intuitively, $\phi(v)$ captures the maximum string depth at which any suffix succeeding $\str(v)^R$ ``departs'' the heavy path $\hp(u_l)$. First, suppose $v$ is a leaf node in $\PT(u_l)$ and its corresponding prefix is $T[1\dd \SA[i] -1]$. Let $\ell_i$ be the leaf in the suffix tree corresponding to $T[\SA[i]\dd n]$ and let $\ell_{u_l}$ be the leaf in $\hp(u_l)$. We make $\phi(v) = \sd(\LCA(\ell_i, \ell_{u_l}))$; i.e., the string depth of the lowest ancestor of $\ell_i$ that is on the heavy path rooted at $u_l$. Next, for any internal node $v$ in $\PT(u_l)$, we make $\phi(v)$ equal to the maximum $\phi(\cdot)$ values of all leaves in the subtree of $v$. See Fig.~\ref{fig:phi_values} for an example; $u_l$ is the root, $v$ in $\PT$ is $v_1$ corresponding to $\ell_2$ in $\ST(T)$ and $\ell_{u_l}=\ell_4$, so $\phi(v_1)=\sd(\LCA(\ell_2,\ell_4))=\sd(u_1)=1$. Indeed the suffix $\texttt{a}$ departs the heavy path $\hp(u_l)$ at maximum string depth $1$. 

\begin{figure}
    \centering
    \includegraphics[width=\linewidth]{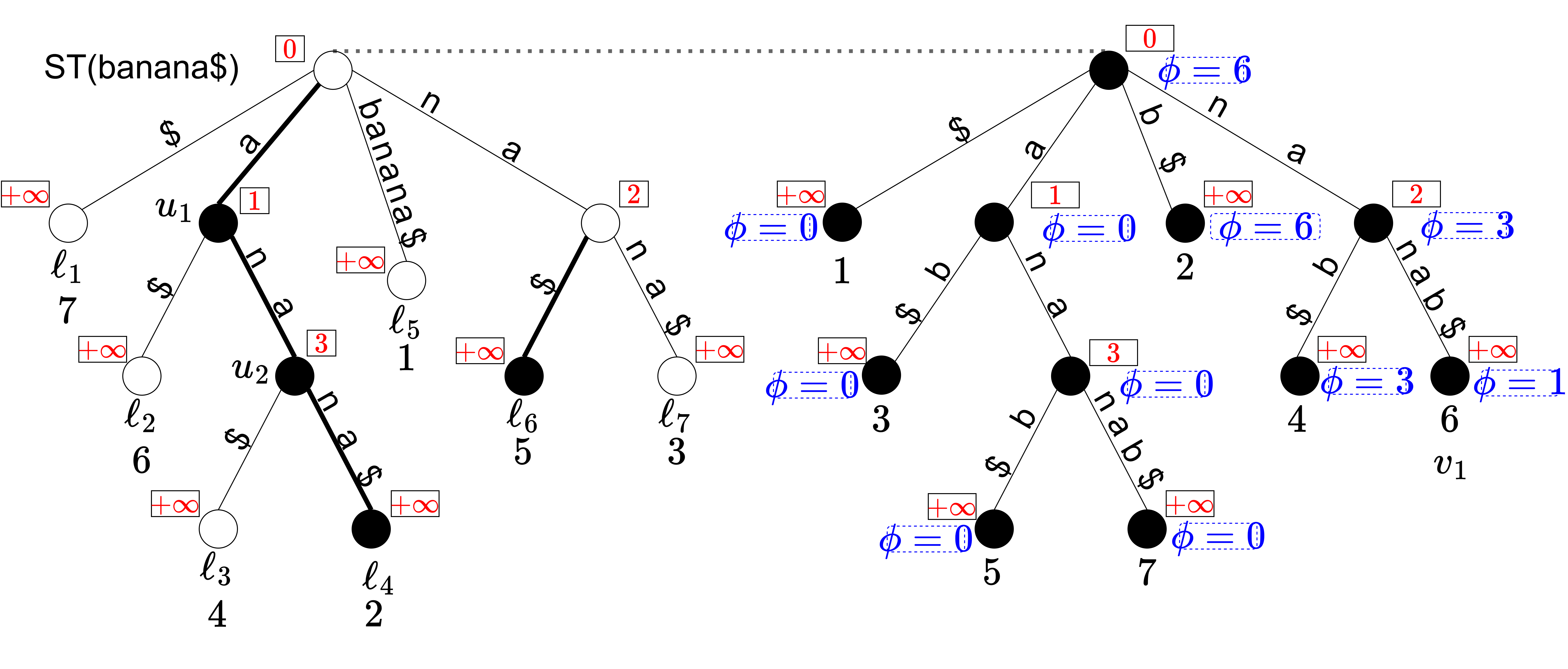}
    \caption{(Left) The suffix tree of $\texttt{banana\$}$ and its heavy-path decomposition. Light nodes are shown in white and heavy nodes in black. Heavy paths are shown with bold edges. (Right)  The $\PT$ tree for the root (a light node). String depths are shown in red and $\phi$ values in blue.}
    \label{fig:phi_values}
\end{figure}

Let $u_P$ be the locus of $P$ in $\ST(T)$. We consider two cases.

\emph{Case 1: $u_P$ is a light node.} 
We will consider nodes $u_l$ that are light and are in the subtree of $u_P$.
By the decomposition, every node has exactly one heavy child.
We will treat $u_l$ separately, based on whether $u_l$ departs or enters a heavy path.

The first subcase is easy. We count the occurrences corresponding to nodes $v$ in $\PT(u_l)$ of light nodes $u_l$ where
\begin{enumerate}
\item $u_l$ is a node in the subtree of $u_P$;
\item $\sd(\parent(u_l)) < m+r \leq \sd(u_l)$;
\item $\sd(\parent(v)) < l \leq \sd(v)$.
\end{enumerate}
These are the same conditions described in Section \ref{sec:simple}, but applied only to light nodes. Intuitively, these conditions account for occurrences that depart a heavy path at string depth $m+r$; see the nodes in red in Fig.~\ref{fig:context_counting_combined} (left), for an example.

We also need to count the occurrences corresponding to nodes $v$ in $\PT(u_l)$ of light nodes $u_l$ where
\begin{enumerate}
\item $u_l$ is a node in the subtree of $u_P$;
\item $\sd(u_l) < m+r \leq \phi(v)$;
\item $\sd(\parent(v)) < l \leq \sd(v)$.
\end{enumerate}
Intuitively, these conditions account for occurrences that ``enter'' a heavy path and remain on it until string depth $m+r$; consider the occurrences that follow a heavy path rooted at green nodes in Fig.~\ref{fig:context_counting_combined} (left) until string depth $m + r$. 

\emph{Case 2: $u_P$ is a heavy node.}~Besides any occurrences that we need to find as in Case 1, we need to account for some other occurrences. Let $u_l$ be the lowest light ancestor of $u_P$. We also need to count points associated with nodes $v$ in $\PT(u_l)$ where
\begin{enumerate}
\item $\sd(\parent(v)) < l \leq \sd(v)$; and (2) $m+r \leq \phi(v)$.  
\end{enumerate}
Constraint 1 is as before. 
For Constraint 2, note that $\hp(u_l)$ is defined because $u_l$ is a light node. 
By the decomposition, every (light) node has exactly one heavy child. 
Thus, $u_P$ lies (by construction) on $\hp(u_l)$. An example of the light node $u_l$ we are concerned with is shown in blue in Fig.~\ref{fig:context_counting_combined} (right). 
Recall that $\phi(v)$ gives the string depth of the longest suffix succeeding $\str(v)^R$ 
that also departs $\hp(u_l)$. Thus $m+r \leq \phi(v)$ considers only
the (remaining) strings $P\cdot R$, for which $|R|=r$.

\begin{figure}[t]
    \centering
    \includegraphics[width=.7\linewidth]{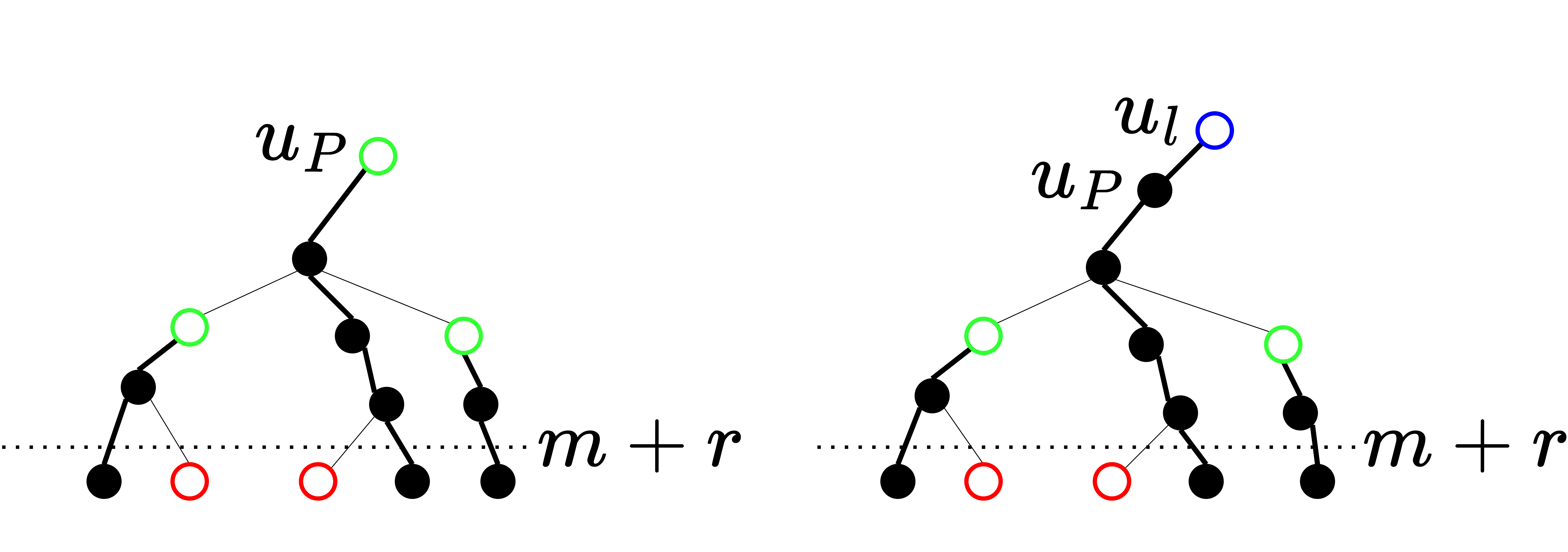}
    \caption{(Left) Occurrences for the  light nodes in red are counted with a query to $\mathcal{D}_1$. 
    Occurrences for the light nodes in green are counted with a query to $\mathcal{D}_2$. (Right) Occurrences for the heavy nodes are counted with a query to $\mathcal{D}(u_l)$.} 
    \label{fig:context_counting_combined}
\end{figure}

We now describe our data structures. For each light node $u_l$ in $\ST(T)$ and node $v$ in $\PT(u_l)$, we create 3 types of points:

\begin{enumerate} [left=20pt]
    \item[type 1:] A five-dimensional point is defined as: {\small
    \begin{equation*}
    \hspace*{-\leftmargini}
    (\preorder(u_l),\sd(\parent(u_l))+1, \sd(u_l), \sd(\parent(v))+1, \sd(v))
    \end{equation*}
    }

    \item[type 2:] Another five-dimensional point is defined as:{\small
    \begin{equation*}
    \hspace*{-\leftmargini}
    (\preorder(u_l), \sd(u_l)+1, \phi(v), \sd(\parent(v))+1, \sd(v))
    \end{equation*}
    }

    \item[type 3:] A three-dimensional point is defined as:{\small
    \begin{equation*}
    \hspace*{-\leftmargini}
    (\sd(\parent(v))+1, \sd(v), \phi(v))
    \end{equation*}
    }
\end{enumerate}

We maintain a range counting structure $\mathcal{D}_1$ over the collection of all type 1 points and a range counting structure 
$\mathcal{D}_2$ 
over the collection of all type 2 points. 
For each light node $u_l$ in $\ST(T)$, we maintain a separate range counting structure $\mathcal{D}(u_l)$ over the type 3 points created for every $v$ in
$\PT(u_l)$.

\subparagraph{Querying.}~We first find the locus of $P$ in $\ST(T)$. 
Let node $u_P$ be this locus. Our querying procedure considers two cases:

\emph{Case 1: $u_P$ is a light node.} We make the following query to $\mathcal{D}_1$:
$
[\preorder(u_P), \preorder(\rLeaf(u_P))] \times [-\infty,m+r] \times [m+r, +\infty] \times [-\infty, l]  \times [l, +\infty]. 
$
We also make the identical query to $\mathcal{D}_2$.
These correspond to the two sets of constraints observed in Case 1.
As the answer to \CPC, we return the sum of the counts returned by $\mathcal{D}_1$ and $\mathcal{D}_2$.

\emph{Case 2: $u_P$ is a heavy node.} Let $u_l$ be the lowest light ancestor of $u_P$. We make the queries described above to $\mathcal{D}_1$ and $\mathcal{D}_2$ and also the following query to $\mathcal{D}(u_l)$:
$
[-\infty,l] \times [l,+\infty] \times [m+r,+\infty].
$
As the answer to \CPC, we  return the sum of the counts returned by $\mathcal{D}_1$, $\mathcal{D}_2$, and $\mathcal{D}(u_l)$. A concrete example is provided in Example~\ref{exp:CPC}.

\begin{example} \label{exp:CPC}
 For $T =\texttt{banana\$}$, 
we construct the suffix tree  $\ST(T)$ in Fig.~\ref{fig:phi_values} (left) and perform a heavy-path decomposition to it. We then build the data structures $\mathcal{D}_1$, $\mathcal{D}_2$, and $\mathcal{D}(u_l)$ for each light node $u_l$ in Fig.~\ref{fig:phi_values}. 
 To answer a \CPC query $(P = \texttt{a}, l = 1, r = 2)$, we first identify $u_P = u_1$ representing $P$, which is a heavy node. Thus, we are in Case $2$.
 Next, we make the following queries:
\begin{enumerate}
    \item Q$1$: {\small $[3,7]\times [-\infty, 3] \times [3, +\infty] \times [-\infty, 1] \times [1, +\infty]$} to $\mathcal{D}_1$.
    
    \item Q$2$: {\small $[3,7]\times [-\infty, 3] \times [3, +\infty] \times [-\infty, 1] \times [1, +\infty]$} to $\mathcal{D}_2$.
    
    \item Q$3$: {\small $[-\infty, 1]\times [1, +\infty] \times [3, +\infty]$} to $\mathcal{D}(u_l)$, where $u_l$ is the lowest light ancestor of $u_P$, i.e., the root of $\ST(T)$.
\end{enumerate}

Q$1$ returns $1$, Q$2$ returns $0$, and Q$3$ returns $2$. Finally, we sum up these results to output $|\mathcal{C}_T(P,l,r)|=3$. 
\end{example}

Our result for \CPC is Theorem \ref{thm:CPC_index}.

\begin{theorem}[\CPC Index]
\label{thm:CPC_index}
For any text $T[1\dd n]$, there exists an index occupying $
\Oh(n\cdot(\log n)^3/(\log\log n)^2)$ space that answers any \CPC query in $\Oh(m+(\log n/ \log\log n)^3)$ time. The index can be constructed
in $\Ohtilde(n)$ time and space.
\end{theorem}

\begin{proof}
   The total number of points in $\mathcal{D}_1$, $\mathcal{D}_2$, and $\mathcal{D}(u_l)$ over all light nodes $u_l$ is upper-bounded by the sum of $\PT(u_l)$ sizes over all light nodes, which is $\cO(n\log n)$. Furthermore, since the number of distinct non-$\pm \infty$ boundary values used for five-dimensional queries is four, applying the techniques of Munro et al.~\cite{DBLP:conf/cccg/MunroNT15}, the five-dimensional queries can be reduced to four-dimensional queries. Finally, we apply the structure of J{\'{a}}J{\'{a}} et al.~\cite{DBLP:conf/isaac/JaJaMS04} for our counting data structures. For $N$ four-dimensional points this takes $\Oh(N(\log N/ \log\log N)^2)$ space and has query time $\Oh((\log N/\log\log N)^3)$ over our $N=\Oh(n\log n)$ points. We also observe here that all of the construction steps can be performed in $\Ohtilde(n)$ time and space. This is because the suffix tree construction and the heavy path decomposition can be done in $\Oh(n)$ time~\cite{DBLP:conf/focs/Farach97,DBLP:journals/jcss/SleatorT83}. Following this, a given prefix tree can be constructed in time proportional to its size. All points for a given prefix tree can then be extracted in time proportional to size of the prefix tree. The orthogonal range counting structures can then be constructed over the $\Ohtilde(n)$ points in $\Ohtilde(n)$ time.
\end{proof}

\subsection{Optimizations to the \CPC Index}\label{sec:cpc:optimizations}

A parameter $B$ that upper-bounds the total query length $l+m+r$ 
enables useful optimizations. It can be set to a small constant based on domain knowledge~\cite{plosgen}, and its impact on efficiency can be seen in Theorem~\ref{thm:BCPC_index} and Section~\ref{sec:experiments}.  

\subparagraph{LZ77 Factorization and Primary Occurrences.}~Let $s_1 < s_2 <$  $\hdots$ $< s_z$ be the starting positions of the phrases in the LZ77 factorization of $T$.
Let $S$ be a substring of $T$. We say $T[i\dd j]$ is a \emph{primary} occurrence of $S$ if $T[i\dd j] = S$ and $s_h \in [i,j]$ for some $h \in [1,z]$. A key property of the LZ77 factorization is that every substring of $T$ has a primary occurrence~\cite{karkkainen1996lempel}. Our index uses this property in conjunction with the bound $B$ to consider only the relevant substrings around phrase boundaries. In particular, every occurrence of string $L\cdot P\cdot R$ in the \CPC problem  must have a primary occurrence intersecting a phrase boundary, and this primary occurrence is \emph{all that is required}  to count 
the distinct occurrences of $L\cdot P\cdot R$. 

\subparagraph{Modified String.}~As a first step, given $T$ and its LZ77 factorization with phrases starting at $s_1 < s_2 <$  $\hdots$ $< s_z$, we construct a new string $T'$ as follows: 
$T'$ is initially the empty string. For $i \in [1, n]$, if $\min_{1\leq j \leq z}|s_j -i| < B$, append $T[i]$ to $T'$. If $i$ is \emph{the first} index in a run of indices satisfying $\min_{1\leq j \leq z} |s_j -i| \geq B$, we append a letter $\#\notin \Sigma$ to $T'$. Otherwise, we do nothing to $T'$.  See Fig.~\ref{fig:T_prime} for an example. The string $T'$ can be constructed in $\Oh(n)$ time by maintaining a pointer to the currently closest phrase start while iterating over $i \in [1,n]$. We also have $|T'| = \Oh(\min(B\cdot z, n))$, since $T'$ contains at most $\mathcal{O}(B)$ letters (either appended from $T$ or a $\#$) for each of the $z$ phase starts, but no more than $n$ letters.

\begin{figure}
    \centering
    \scalebox{0.8}
    {
    \begin{tabular}{c||c|c|cc|ccccccc|cccccc}
    $i$&$1$&$2$&$3$&$4$&$5$&$6$&$7$&$8$&$9$&$10$&$11$&$12$&$13$&$14$ &$15$\\ \hline
    $T$ & \tT &\tA &  \tA &  \tA & \tT &  \tA &  \tA & \tA &  \tT & \tA &  \tA & \tT & \tA & \tA & \tA\\
    $T'$ & \tT & \tA & \tA & \tA & \tT & \tA & \tA & \multicolumn{2}{c}{$\#$} & \tA & \tA & \tT & \tA & \tA  & \multicolumn{1}{c}{$\#$} \\
    \end{tabular}
    }

    \caption{String $T'$ is constructed from the LZ77 factorization of string $T$ with $B = 3$. Phrases start at $1,2,3,5,12$. For position $8$ in $T$, we append a $\#$ in $T'$ and for position $9$ we do nothing. For position $15$ in $T$, we append a $\#$ in $T'$. Thus, $|T'|=14$.}
    \label{fig:T_prime}
\end{figure}

\subparagraph{Truncated Suffix Tree.}~It does not quite suffice to directly construct our \CPC index over the string $T'$, as this would allow for spurious occurrences; i.e., occurrences containing a $\#$ in either $L$ or $R$. To solve this issue, we first perform an additional preprocessing step on $\ST(T')$, in which we  truncate any paths from the root containing a $\#$  just prior to the first $\#$. (If an edge starts with a $\#$, we remove that edge.) Let us denote this truncated suffix tree by $\widetilde{\ST}(T')$.  A heavy-path decomposition is then performed on $\widetilde{\ST}(T')$.

\subparagraph{Truncated Prefix Tree.}~For each light node $u_l$ in $\widetilde{\ST}(T')$, we construct a prefix tree $\PT(u_l)$. We use the untruncated $\ST(T')$ for getting the correct set of prefixes for the leaves in the subtree of $u_l$. Note, due to truncation, a leaf in $\widetilde{\ST}(T')$ created through truncation may have multiple corresponding leaves  in the untruncated $\ST(T')$. 

We now define the $\phi(\cdot)$ values. Suppose $v$ is a leaf node in $\PT(u_l)$ and its corresponding prefix is $T[1\dd \SA[i] -1]$. Let $\ell_i$ be the leaf in $\widetilde{\ST}(T')$ corresponding to $T[\SA[i]\dd n]$ and $\ell_{u_l}$ the leaf in $\hp(u_l)$. Let $d$ be the string depth at which the first $\#$ occurs on the root-to-$\ell_i$ path in $\ST(T')$ (inclusive), or $+\infty$ if no such $\#$ occurs. We define $\phi(v) = \min(d-1,\sd(\LCA(\ell_i, \ell_{u_l})))$. As in Section \ref{sec:CPC_index}, for any internal node $v$ in $\PT(u_l)$, we make $\phi(v)$ equal to the maximum $\phi(\cdot)$ value of the leaves in the subtree of $v$. 
For every light node $u_l$ in $\widetilde{\ST}(T')$, we truncate the prefix tree $\PT(u_l)$ by truncating any path from its root containing a $\#$  just prior to its first $\#$ (any edge starting with $\#$ is removed).

The data structures from Section \ref{sec:CPC_index} can now be constructed over $\widetilde{\ST}(T')$ and the collection of truncated prefix trees. Since truncating the suffix and prefix trees only decreases the overall index size, we directly obtain Theorem~\ref{thm:BCPC_index}.

\begin{theorem}[Optimized \CPC Index]
\label{thm:BCPC_index}
For any text $T[1\dd n]$ with an LZ77 factorization of $z$ phrases, a parameter $B\in[1,n]$, and $k :=  \min(B\cdot z, n)$, there exists an index occupying $
\Oh(k \cdot(\log k)^3/(\log\log k)^2)$ space that answer any \CPC query, with $l+m+r\leq B$, in $\Oh(m+(\log k/ \log\log k)^3)$ time. The index can be constructed in $\Oh(n)+\Ohtilde(k)$ time and space.
\end{theorem}

\section{Related Work}\label{sec:related}

Our work is the first to propose and address the \CPM problem. However, there is a large  literature on mining other types of patterns from strings. For example,~\cite{takeaki} mines maximal frequent substrings,~\cite{icdm2005,icdm2008} frequent substrings,~\cite{icdm2007,maxcfp} approximate frequent substrings, and~\cite{aaai20} statistically significant maximal substrings. 

Currently, there are no indexes designed specifically for the \CPC problem we introduced here. However, 
there are many indexes~\cite{DBLP:conf/focs/Weiner73,DBLP:journals/siamcomp/ManberM93,DBLP:conf/focs/Farach97,DBLP:journals/jacm/KarkkainenSB06,DBLP:journals/algorithmica/0001KL15,DBLP:conf/esa/LoukidesP21,DBLP:journals/csur/Navarro21a,sharma9,rindex,sh,talgChristiansen} for the text indexing problem (see Section~\ref{sec:intro}). There are also  indexes whose size is related to some good measure of compressibility for highly repetitive string collections~\cite{DBLP:journals/csur/Navarro21a}. Examples are indexes built on measures like the size of the Lempel-Ziv factorization~\cite{sharma9}, of  
the Run-Length-Encoded Burrows–Wheeler Transform (RL-BWT)~\cite{rindex}, or of other compressed forms~\cite{sh,talgChristiansen}. 

Closer to our work are three indexes that report $\mathcal{C}_T(P,l,r)$, but not its size.  
The first index~\cite{navarro2020contextual} uses a suffix array and other 
auxiliary data structures; it has $\cO(|P|+|\mathcal{C}_T(P,l,r)|)$ query time and uses $\cO(n)$ space. It is used in \CPRI (see Section~\ref{sec:intro}), which our approach outperforms in terms of all measures of efficiency~\cite{pvldb23} (see Section~\ref{sec:experiments}). 
The second index~\cite{navarro2020contextual} has 
$\cO(|P|\log\log n + |\mathcal{C}_T(P,l,r)|\log n)$ query time and uses $\cO(\bar{r}\log (n/\bar{r}))$ space, where $\bar{r}$ is the maximum number of runs in the BWT of $T$ and of its reverse. The third index~\cite{abedinDCC23} has $\cO(|P|+|\mathcal{C}_T(P,l,r)|\log\ell\log(n/e))$ query time, where $\ell=\max(l,r)$ and $e$ is the number of equal-letter runs in the BWT of $T$, and uses $\cO(e\log(n/e))$ space. The query time for the third index  is slightly slower than that for the second. The last two indexes are 
of theoretical interest, as they need some functionality of  
the $r$-index~\cite{rindex} that (to our knowledge) is not available in any implementation. 
We therefore compare our approach to a simpler  approach utilizing the $r$-index. 
Our experiments show that the time required to construct just the $r$-index  
exceeds by far that of our index. 

\section{Experimental Evaluation}\label{sec:experiments}

\emph{Datasets.}~We used $5$ publicly available, large-scale datasets from different domains. The datasets have different
length $n$, alphabet size $|\Sigma|$, and number $z$ of LZ77 phases (see Table~\ref{tab:data} for details). The largest dataset, \chr, has length $n=16\cdot 10^9$ letters. 
WIKI is a collection of different versions of Wikipedia articles about Albert Einstein up to November 10, 2006, released as part of the Pizza \& Chili corpus \cite{pizzachili}.
BST and SDSL represent all commits from two GitHub repositories (Boost \cite{boost} and SDSL \cite{sdsl}), retrieved from \cite{Getgit}. \sars is a genomic (SARS-CoV-2) dataset  \cite{NCBI}.
\chr is a dataset from the $1,000$ Genomes Project~\cite{GenomesProject} 
representing chromosome 19 sequences from $1,000$ human haplotypes \cite{chr}. We have extracted the prefixes of length $16,000,000$ of all $1,000$ sequences to form the dataset used in our experiments.

\begin{table}[ht]
\centering
\caption{Datasets characteristics} \label{tab:data}
\begin{tabular}{|c|c|c|c|c|}
\hline
{\bf Datasets} &  {\bf Domain} & $n$ &  $|\Sigma|$ & $z$\\ \hline \hline
WIKI \cite{pizzachili} &biographical article& 474,214,798 & 36& 266,159 \\ \hline
BST \cite{boost} &github repository& 3,186,957,848  & 68& 25,745 \\ \hline
SDSL \cite{sdsl} & github repository& 5,508,023,144  & 67 &321,791 \\ \hline
 
    \sars \cite{NCBI}     & bioinformatics & 10,000,000,000 & 4 & 2,060,594\\ \hline
    \chr  \cite{chr}    & bioinformatics & 16,000,000,000 & 4 & 2,314,193\\ \hline
    
\end{tabular}
\end{table}

\begin{figure*}[t]
\centering
\begin{subfigure}{0.19\textwidth}
    \includegraphics[width=1.03\linewidth]{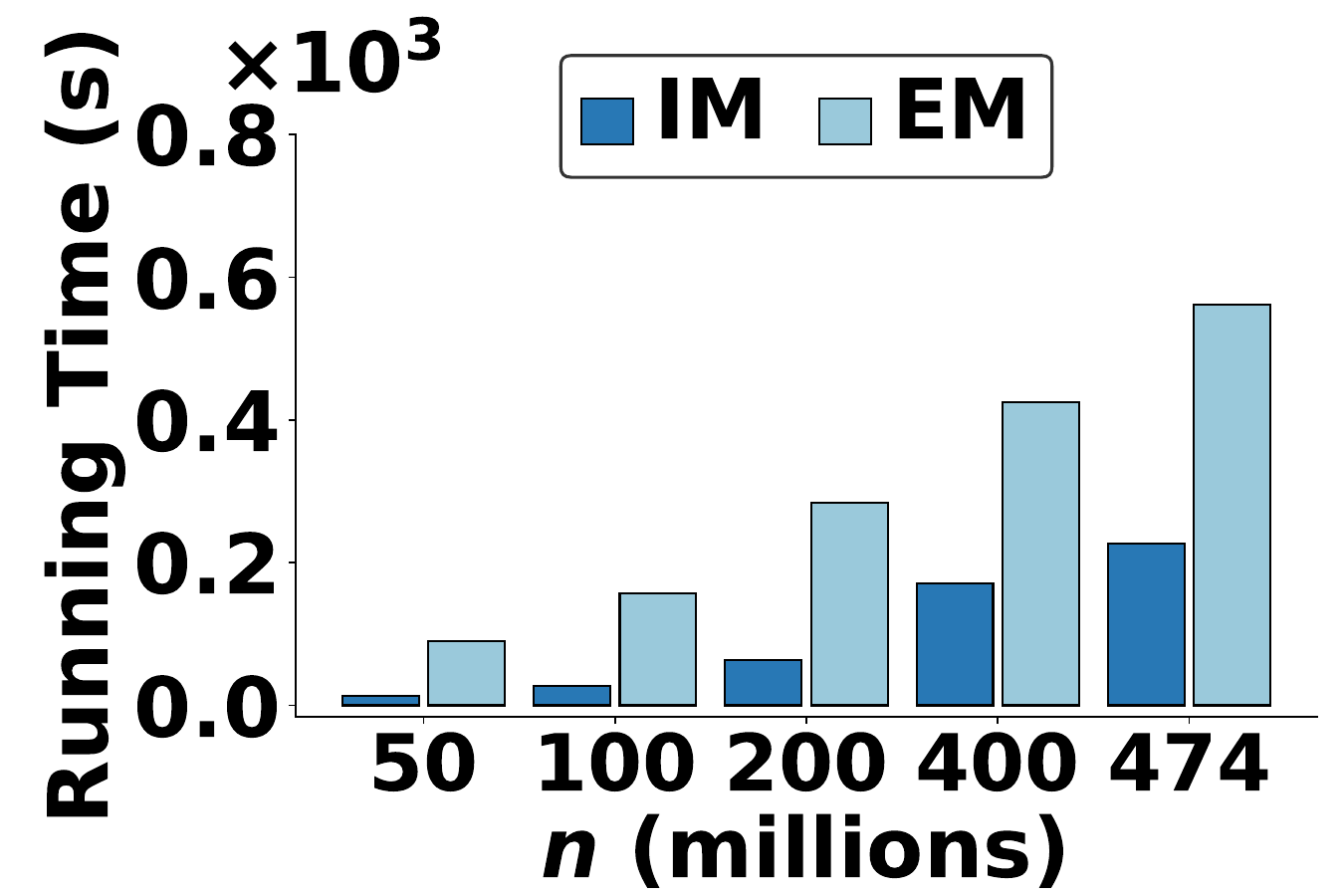}
    \caption{WIKI}
    \label{fig:EM_WIKI_time_varying_n}
\end{subfigure}
\hfill
\begin{subfigure}{0.19\textwidth}
    \includegraphics[width=1.03\linewidth]{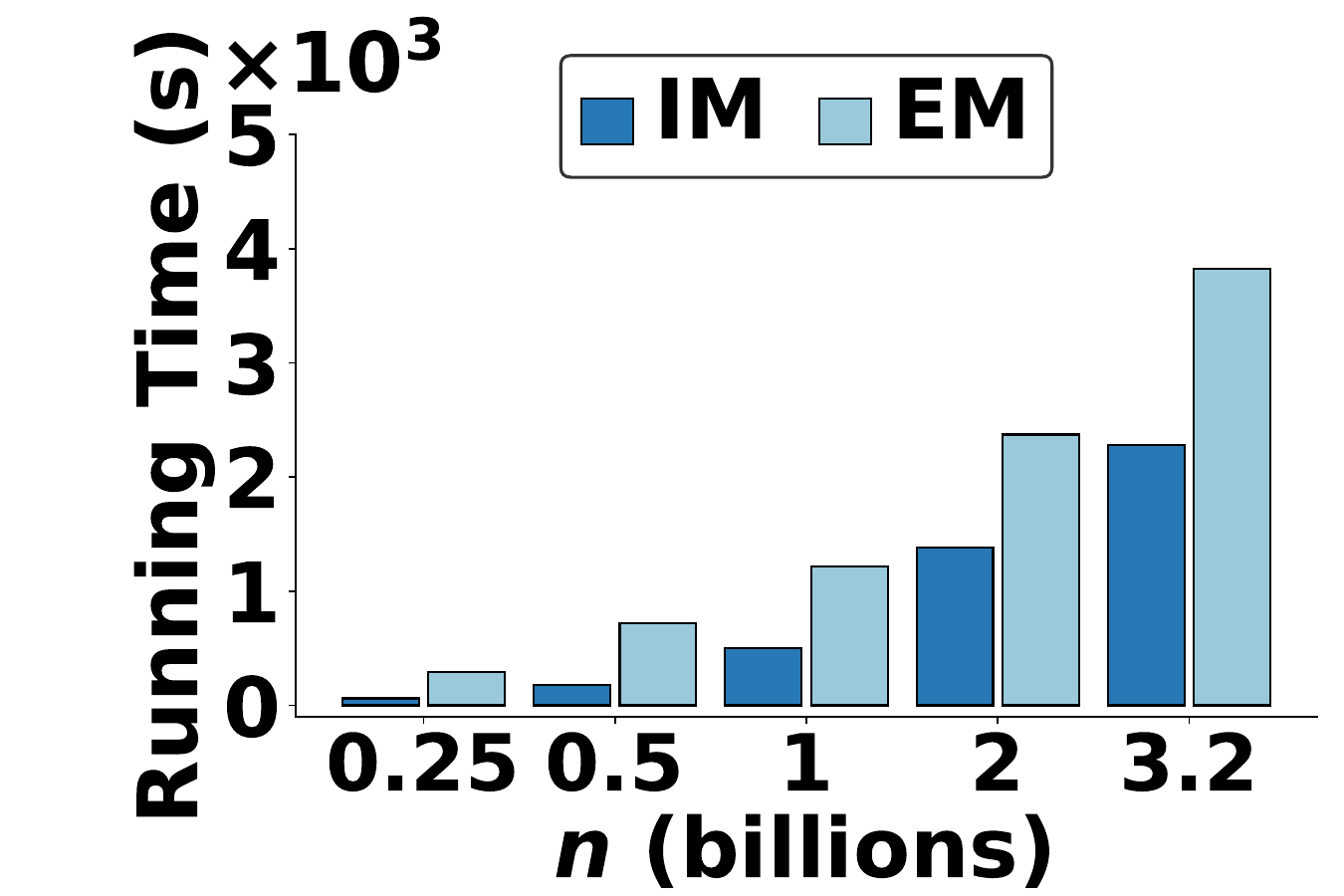}
    \caption{BST}
\end{subfigure}
\hfill
\begin{subfigure}{0.19\textwidth}
    \includegraphics[width=1.03\linewidth]{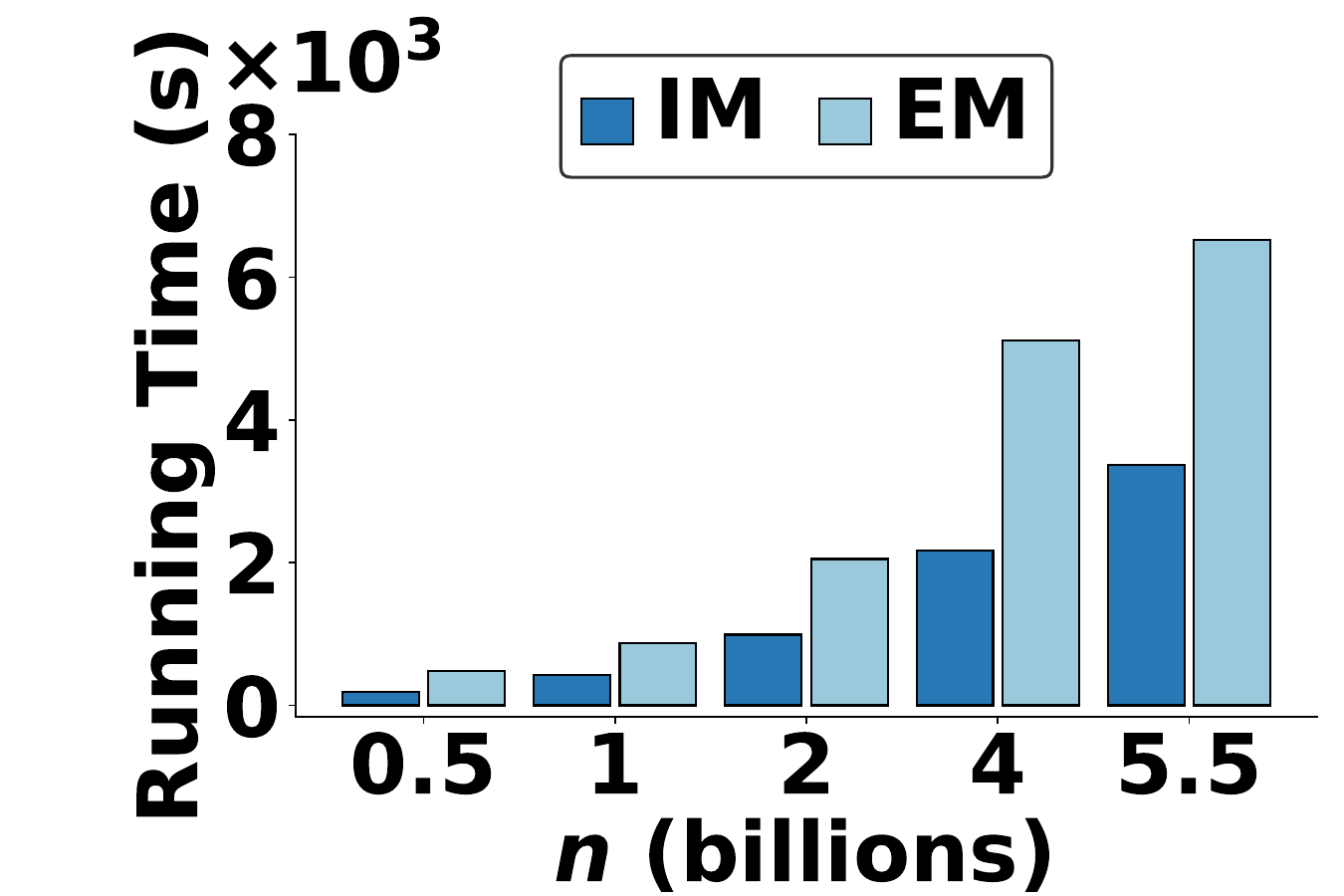}
    \caption{SDSL}
\end{subfigure}
\hfill
\begin{subfigure}{0.19\textwidth}
    \includegraphics[width=1.03\linewidth]{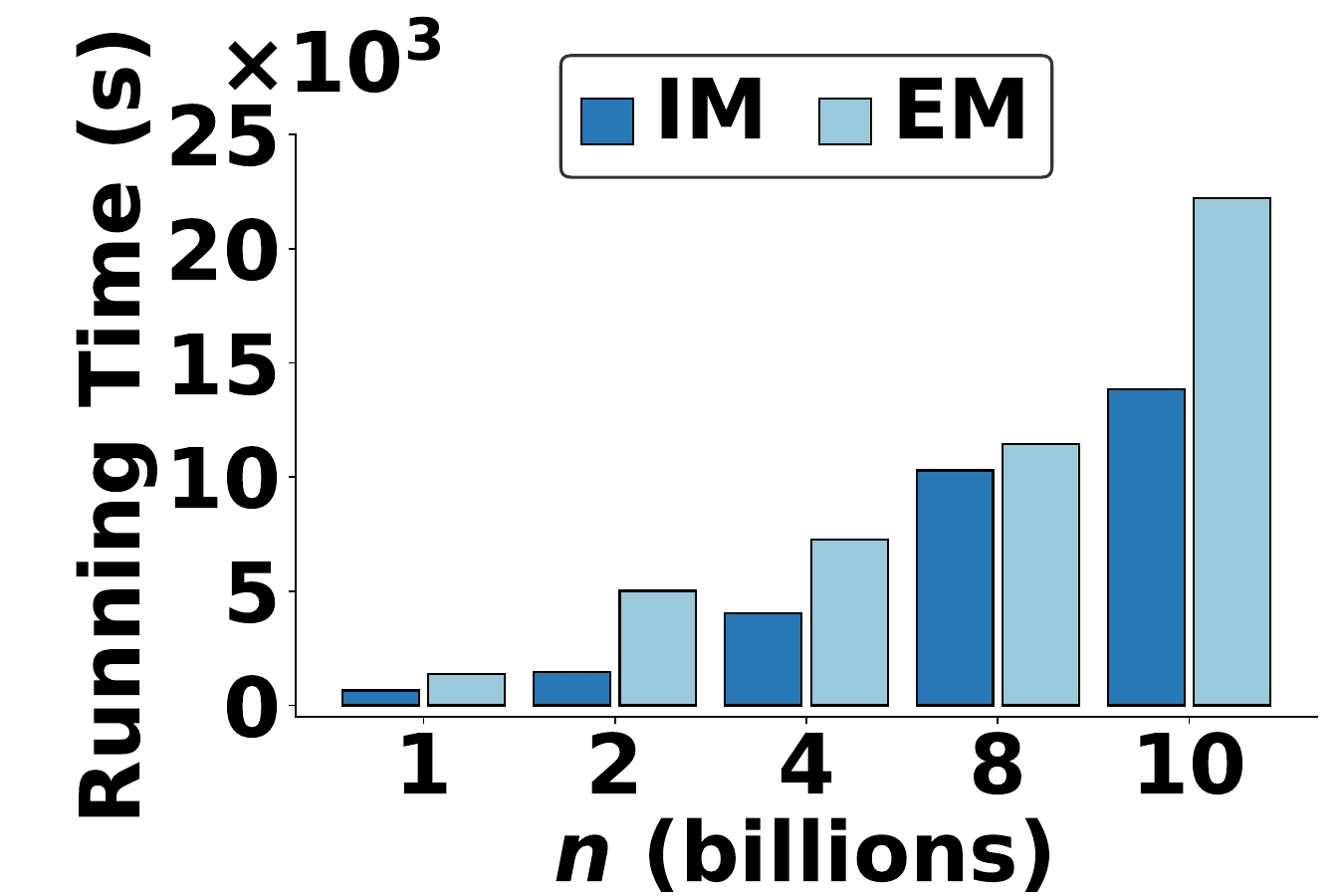}
    \caption{\sars}
    \label{fig:EM_sars_time_varying_n}
\end{subfigure}
\hfill
\begin{subfigure}{0.19\textwidth}
    \includegraphics[width=1.03\linewidth]{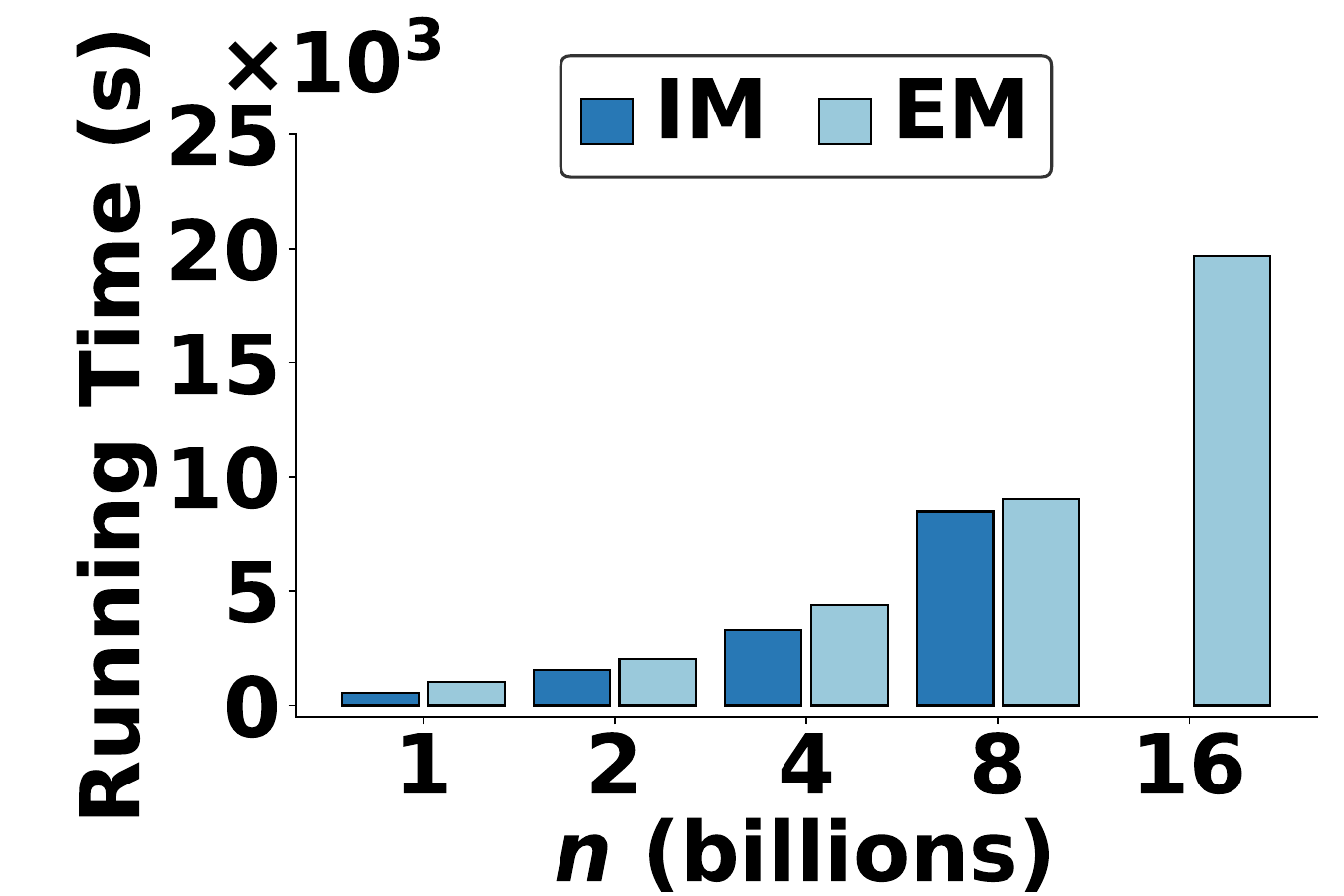}
    \caption{\chr}
    \label{fig:EM_chr_time_varying_n}
\end{subfigure}

\begin{subfigure}{0.19\textwidth}
    \includegraphics[width=1.03\linewidth]{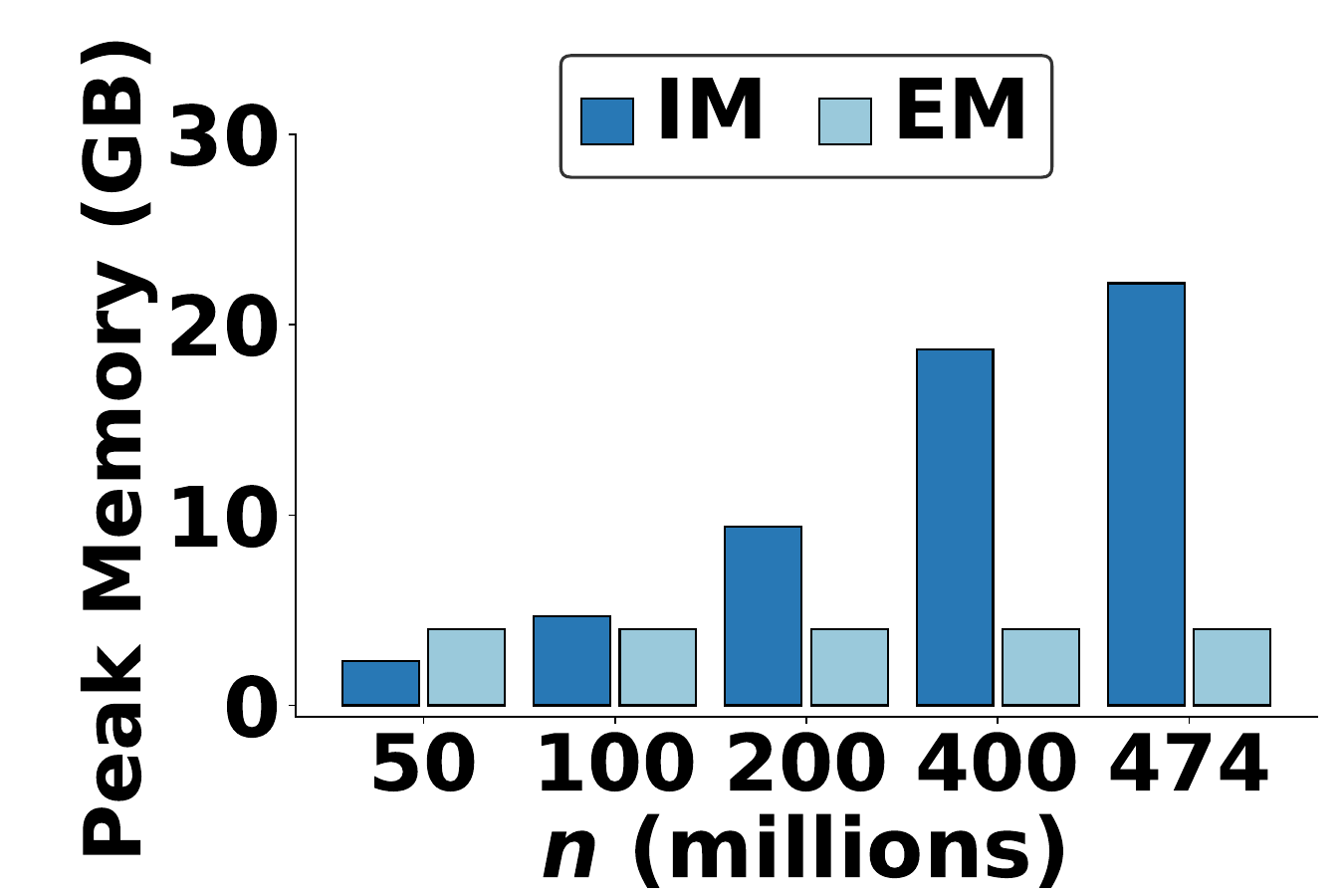}
    \caption{WIKI}
    \label{fig:EM_WIKI_Ram_varying_n}
\end{subfigure}
\hfill
\begin{subfigure}{0.19\textwidth}
    \includegraphics[width=1.03\linewidth]{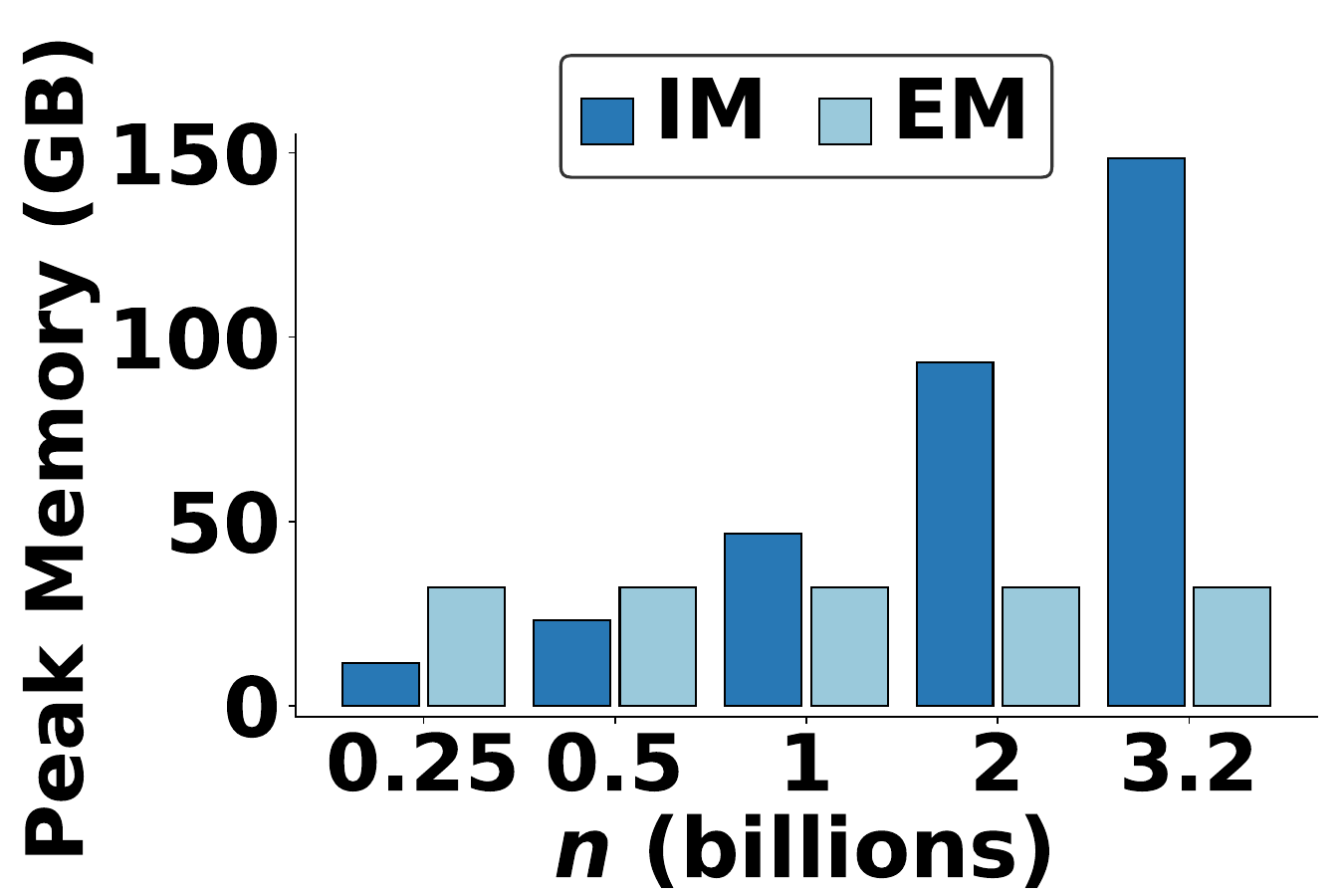}
    \caption{BST}
\end{subfigure}
\hfill
\begin{subfigure}{0.19\textwidth}
    \includegraphics[width=1.03\linewidth]{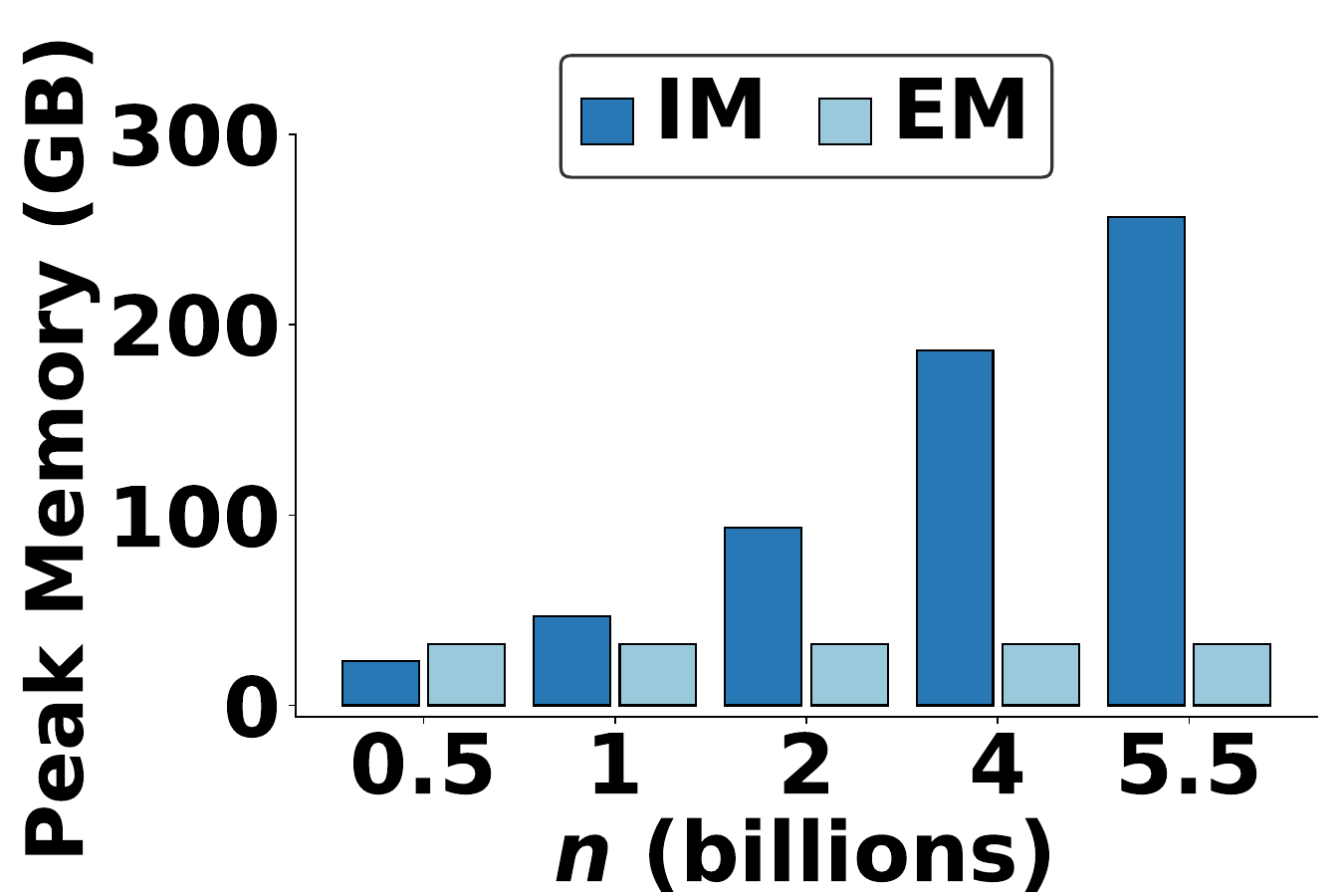}
    \caption{SDSL}
\end{subfigure}
\hfill
\begin{subfigure}{0.19\textwidth}
    \includegraphics[width=1.03\linewidth]{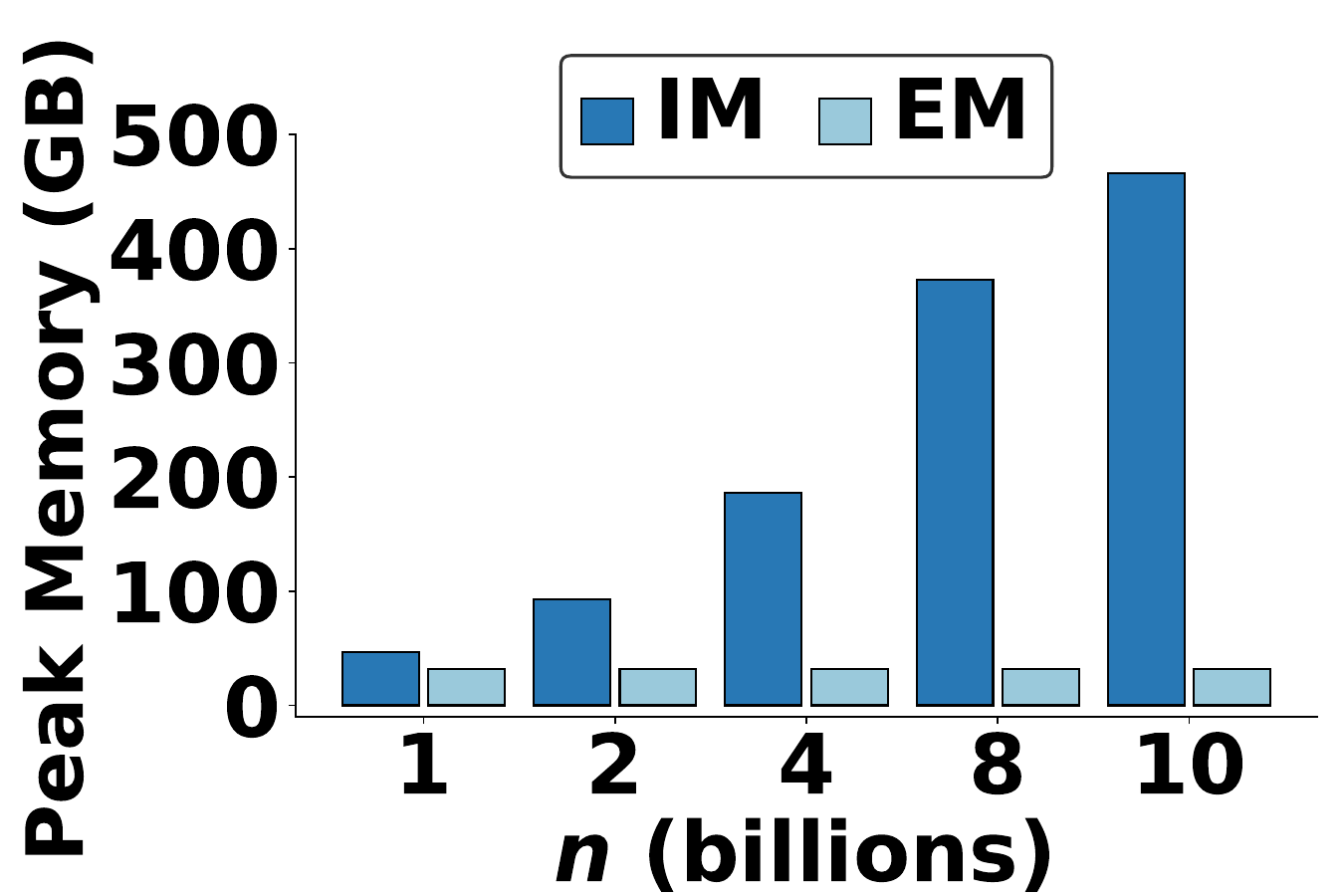}
    \caption{\sars}
    \label{fig:EM_sars_Ram_varying_n}
\end{subfigure}
\hfill
\begin{subfigure}{0.19\textwidth}
    \includegraphics[width=1.03\linewidth]{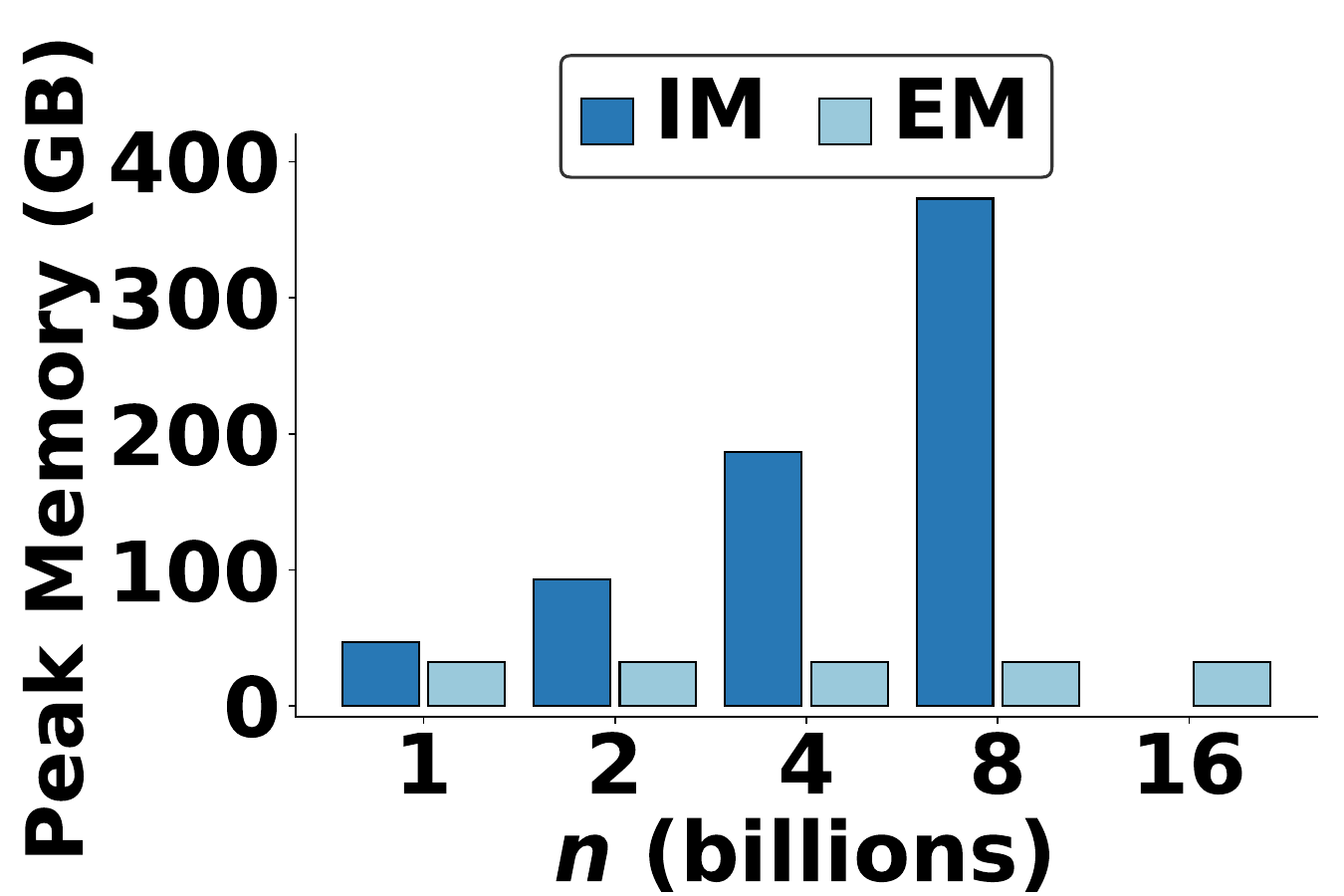}
    \caption{\chr}
    \label{fig:EM_chr_Ram_varying_n}
\end{subfigure}

\begin{subfigure}{0.19\textwidth}
    \includegraphics[width=1.03\linewidth]{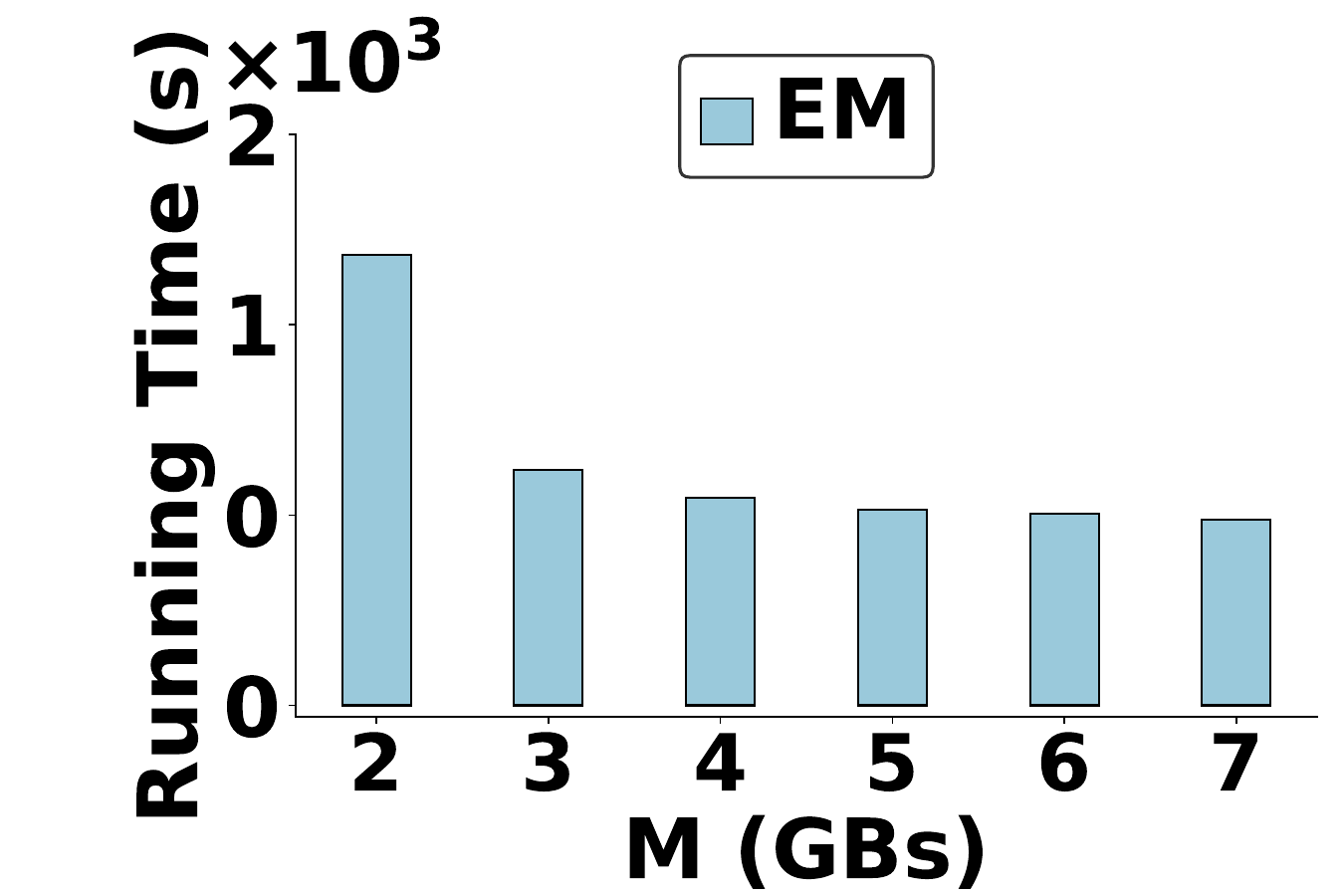}
    \caption{WIKI}
    \label{fig:EM_WIKI_varyingRam}
\end{subfigure}
\hfill
\begin{subfigure}{0.19\textwidth}
    \includegraphics[width=1.03\linewidth]{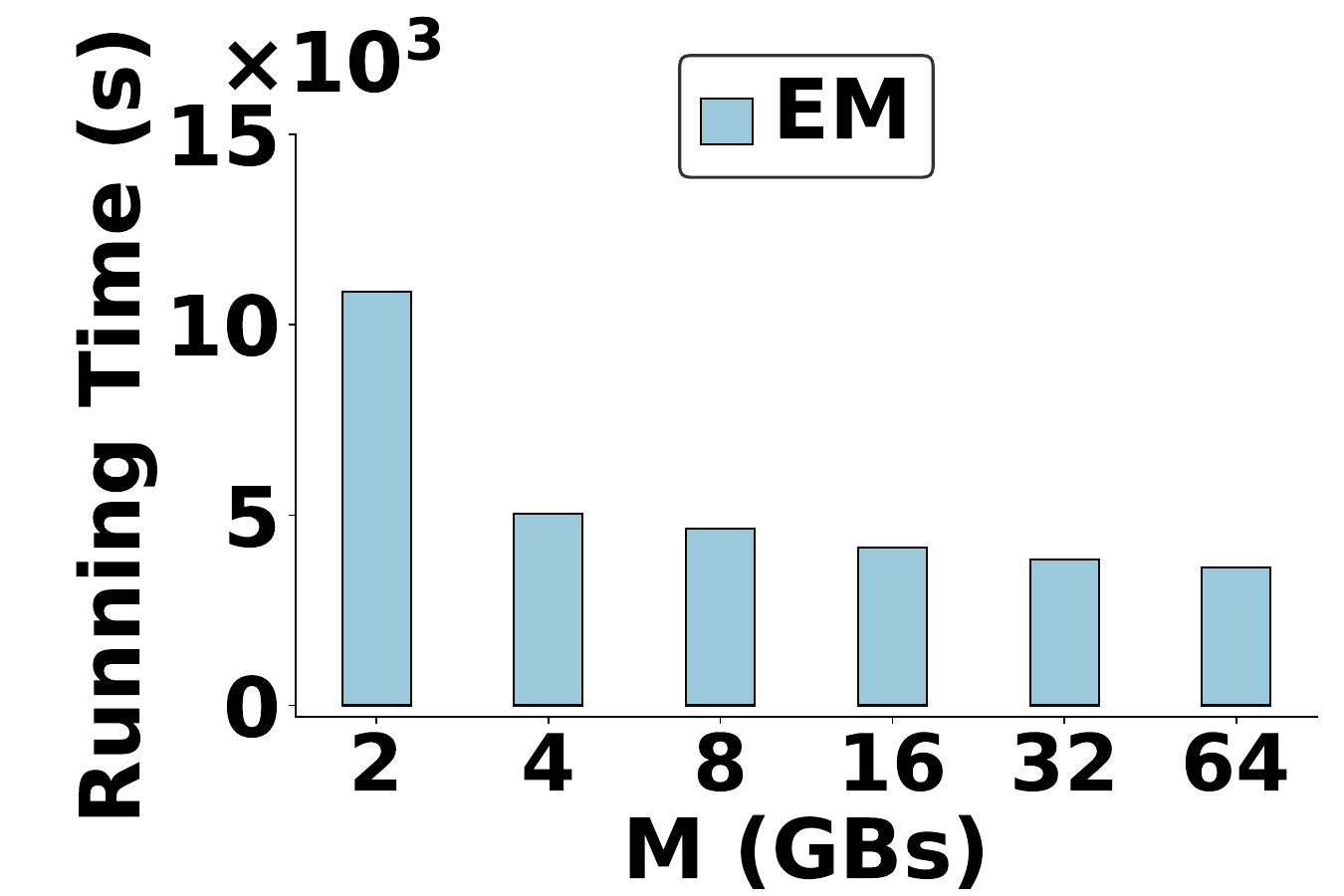}
    \caption{BST}
\end{subfigure}
\hfill
\begin{subfigure}{0.19\textwidth}
    \includegraphics[width=1.03\linewidth]{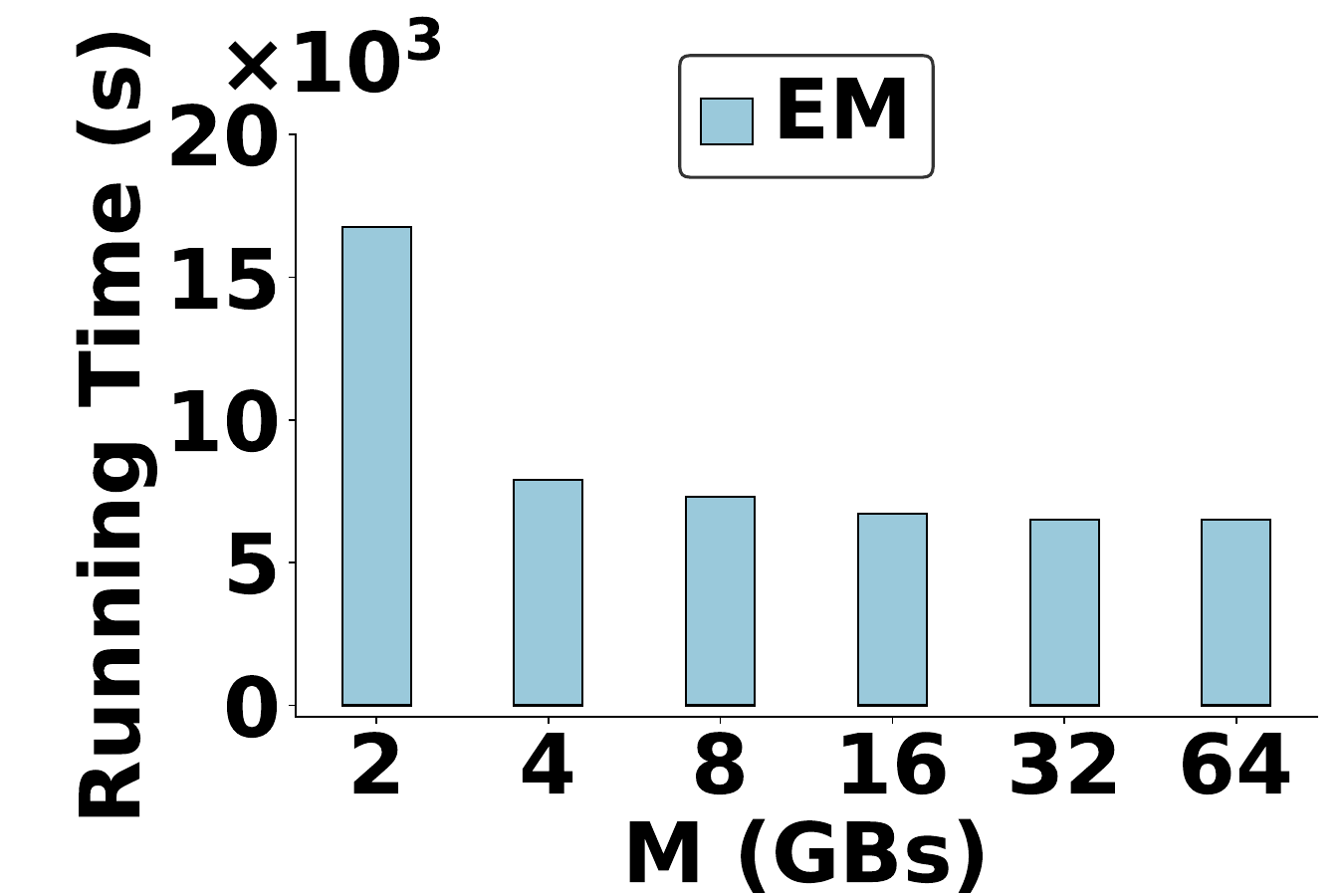}
    \caption{SDSL}
\end{subfigure}
\hfill
\begin{subfigure}{0.19\textwidth}
    \includegraphics[width=1.03\linewidth]{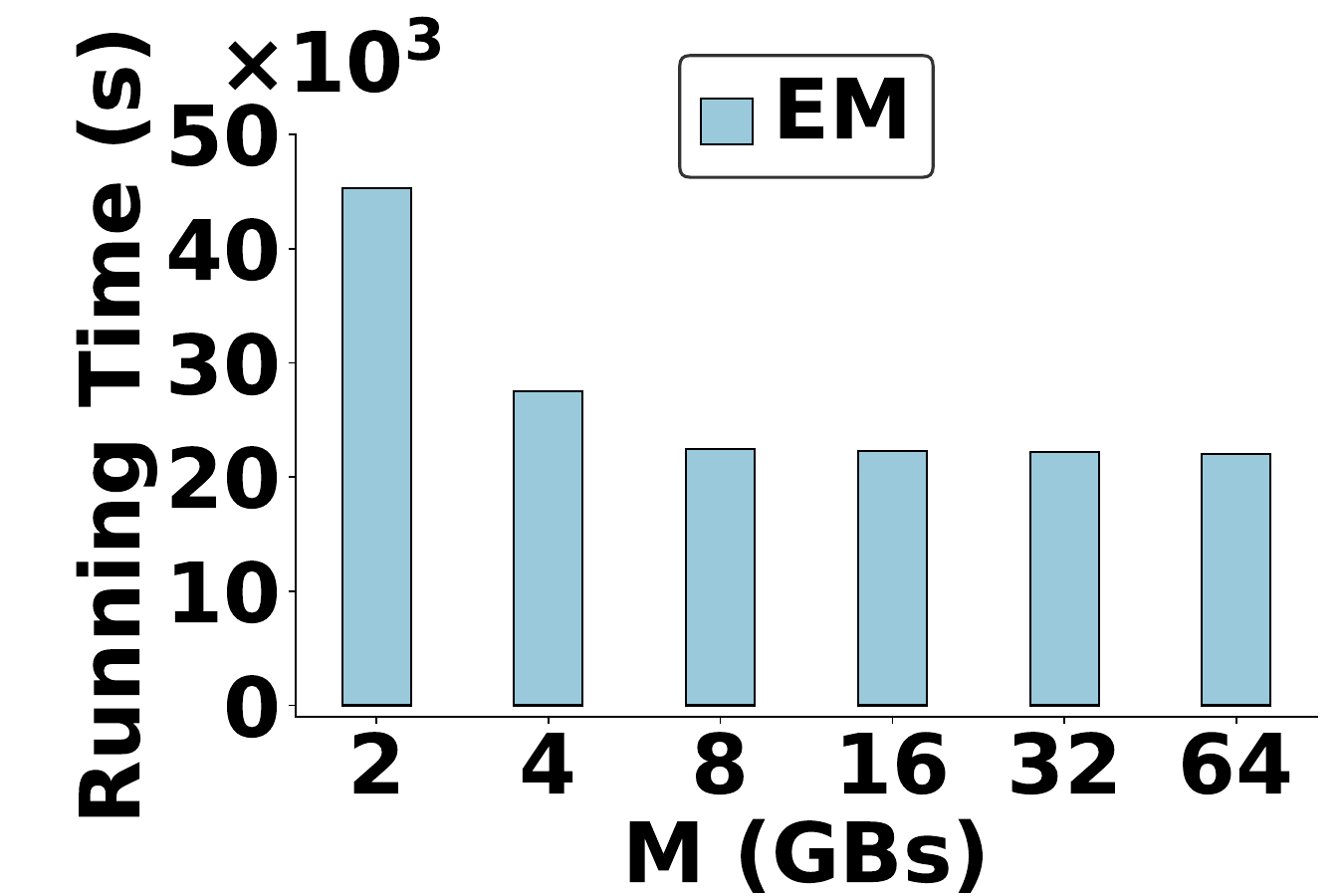}
    \caption{\sars}
    \label{fig:EM_sars_varyingRam}
\end{subfigure}
\hfill
\begin{subfigure}{0.19\textwidth}
    \includegraphics[width=1.03\linewidth]{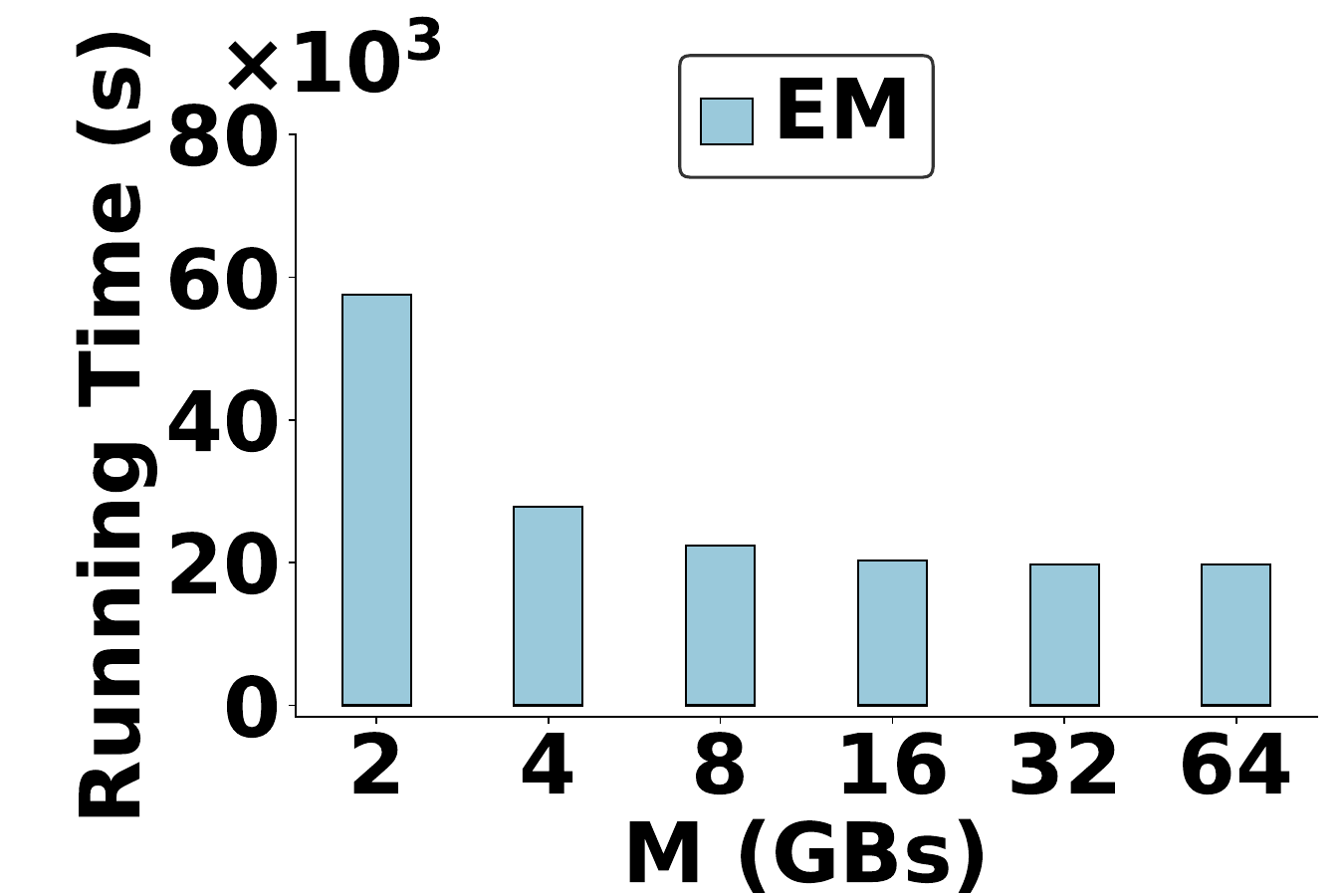}
    \caption{\chr}
    \label{fig:EM_chr_varyingRam}
\end{subfigure}
\caption{(a--e) Running time vs.\ $n$; (f--j) peak memory vs.\ $n$; (k--o) running time vs.\ $\mathbf{M}$ across all datasets. Missing bars indicate that a method needed more than $512$ GBs of memory.}
\end{figure*}

\emph{Environment.}~All experiments were conducted on an AMD EPYC 7282 CPU with $512$ GBs RAM and $1T$ disk space.  
All methods were implemented in \texttt{C++}. 
Our code and datasets are available at \url{https://github.com/LINGLI97/CPM}.

\subsection{Contextual Pattern Mining (\CPM) Algorithm} 

\subparagraph{Methods.}~We compare our external memory version (\EM)  to the internal memory one (\IM).  We evaluated the two approaches in the ``Challenges'' subsection of Section~\ref{sec:intro}, but we do not present their results as they were much slower.  

\noindent{\bf Measures.}~We measured runtime via the \texttt{chrono C++} library and peak memory usage via \texttt{/usr/bin/time -v}.  

\subparagraph{Parameters.}~We fixed $\mathbf{B}=2$MB and $(\tau,m,l,r)=(1000,9,9,9)$; recall from Section~\ref{sec:external} that $\tau$, $m$, $l$, and $r$ do not affect the runtime or space of either version of our algorithm. By default, $\mathbf{M}=32$GB, except for WIKI, where $\mathbf{M}=4$GB due to the relatively small size of the dataset. 

\subparagraph{Results.}~Figs.~\ref{fig:EM_WIKI_time_varying_n} to \ref{fig:EM_chr_time_varying_n} show the runtime of \IM and \EM as a function of the length $n$ of the input string (prefix of a dataset). \IM is on average $2.6$ times faster than \EM, as the latter accesses external memory (i.e., it performs I/O operations, which are much slower than accessing RAM). 

Figs.~\ref{fig:EM_WIKI_Ram_varying_n} to \ref{fig:EM_chr_Ram_varying_n} show the peak memory usage for the experiments of Figs.~\ref{fig:EM_WIKI_time_varying_n} to  \ref{fig:EM_chr_time_varying_n}. The memory needed by \IM increases linearly with $n$, as expected by its space complexity, and for \EM it was at most $\mathbf{M}$ GBs as expected by its design. For the entire \chr dataset ($n=16\cdot 10^9$), \IM needed more than $512$ GBs of memory (the full available amount) so it did not terminate, while \EM only $32$ GBs. 
Thus, \EM used \emph{$4.3$ times less memory} and was \emph{approximately $2.6$ times slower} on average over all $5$ datasets when both \IM and \EM terminated. 

Figs.~\ref{fig:EM_WIKI_varyingRam} to  \ref{fig:EM_chr_varyingRam} show that \EM takes less time as $\textbf{M}$ increases because it performs fewer I/Os (see Section~\ref{sec:external}).  
However, after some point its  runtime stays the same, as the minimum number of I/Os that must be performed determine the runtime.

\subsection{Contextual Pattern Counting (\CPC) Index} 

\subparagraph{Methods.}~We compared our optimized  \CPC index (referred to as \CI) to: (1) The \CPRI approach which uses the suffix array (referred to as \CPRIS)~\cite{navarro2020contextual}. This is the only practical index for reporting $\mathcal{C}_T(P,l,r)$ (see Section~\ref{sec:related}). (2) The ``straightforward  approach'' (see ``Challenges'' in Section~\ref{sec:intro}) based on the $r$-index~\cite{rindex}, a state-of-the-art index exploiting another form of compressibility than LZ77 factorization (see Section~\ref{sec:related}). We refer to this approach as \RI. The \CI without our optimizations (referred to as CI-) in Section~\ref{sec:CPC_index} was not competitive (we included it only in subsection ``Ablation study'' below); and  the index in Section~\ref{sec:simple} needed much more space (omitted).  

\subparagraph{Measures.}~We used all four relevant measures of efficiency~\cite{pvldb23}: (1) query time; (2) construction time; (3) index size; and (4) construction space. For (1) and (2), we used the \texttt{chrono C++} library. For (3), we used the \texttt{mallinfo2 C++} function. For (4), we recorded the maximum resident set size  
using the \texttt{/usr/bin/time -v} command. 

\subparagraph{Query Workloads.}~We considered query workloads of type 
$\mathcal{W}_{m,l,r}$. Each such workload contains as queries $(P,l,r)$ all substrings of length $m=|P|$ of a dataset and $l=r$, where $m,l,r\in \{3,6,9,12,15\}$. These $l$, $m$, and $r$ values were chosen to stress-test the approaches, as they result in very large contexts. We set $B = l + m + r$ by default. 

\subparagraph{Comparison of Range Counting Data Structures.}~We evaluated \CI based on three popular range counting structures: Range trees~\cite{bentley1978decomposable,DBLP:conf/focs/Lueker78}, KD trees~\cite{bentley1975multidimensional}, and $R^*$-trees~\cite{beckmann1990r}; see Fig.~\ref{fig:BIwith_trees}. The R$^*$-tree based implementation has the fastest query time and construction time and only slightly larger index size and construction space compared to the KD tree based implementation. The Range tree based implementation was the worst. Thus, in the remaining of the paper, we will only report results using the R$^*$-tree based implementation. 

\begin{figure}[t]
\centering
\begin{subfigure}[t]{0.24\columnwidth}
    \centering
    \includegraphics[width=1\linewidth]{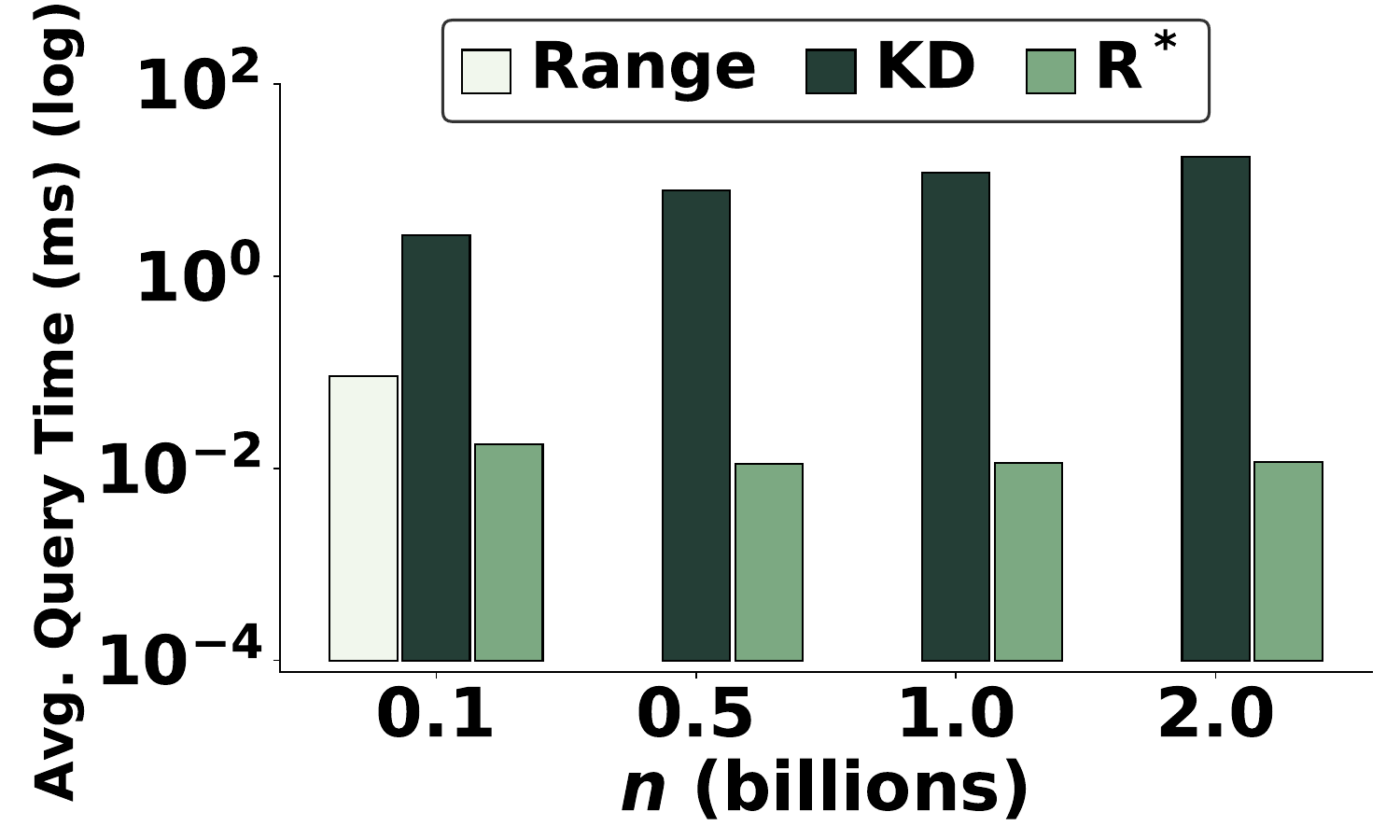}
    \caption{SARS}\label{fig:ds:a}
\end{subfigure}
\begin{subfigure}[t]{0.24\columnwidth}
    \centering
    \includegraphics[width=1\linewidth]{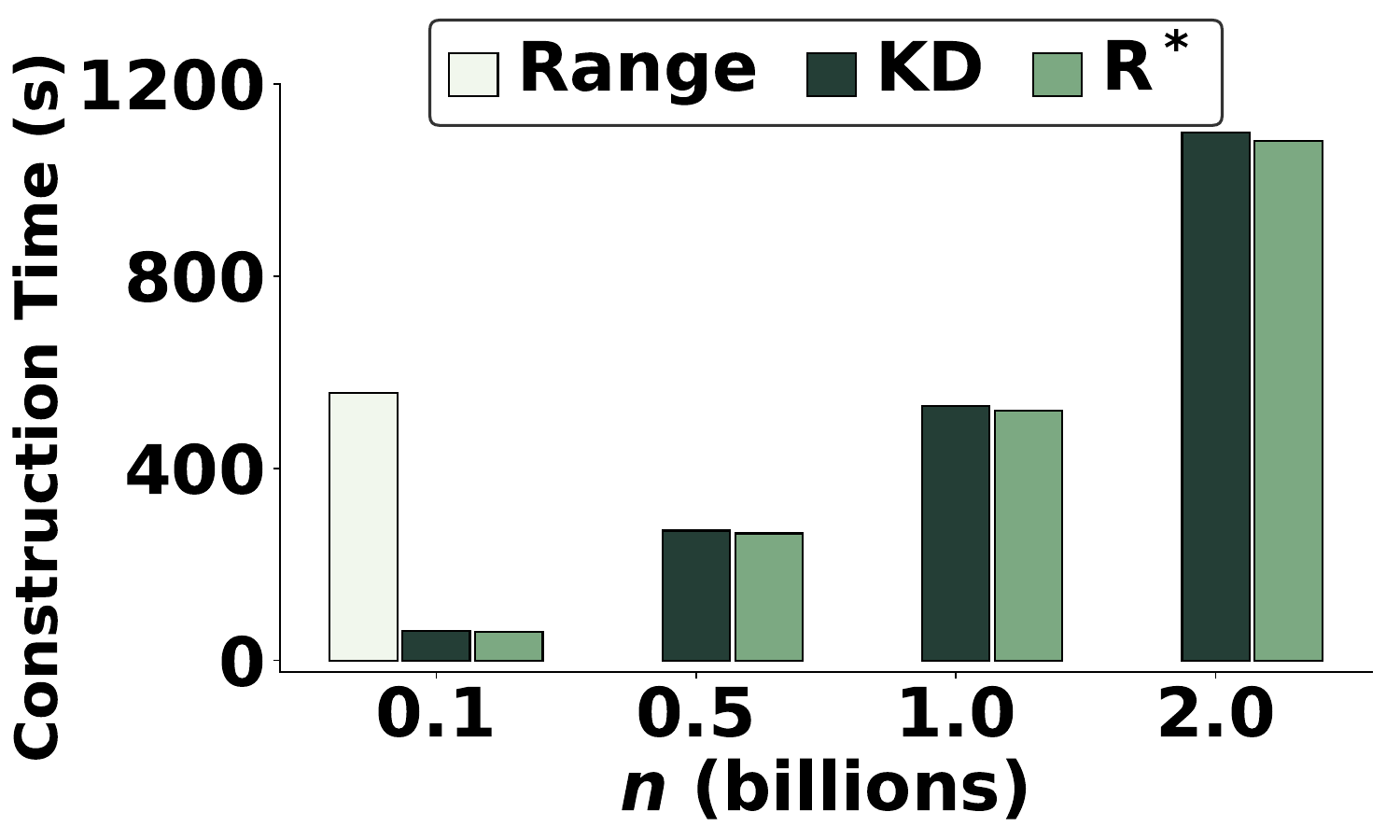}
    \caption{SARS}\label{fig:ds:b}
\end{subfigure}
\begin{subfigure}[t]{0.24\columnwidth}
    \centering
    \includegraphics[width=1\linewidth]{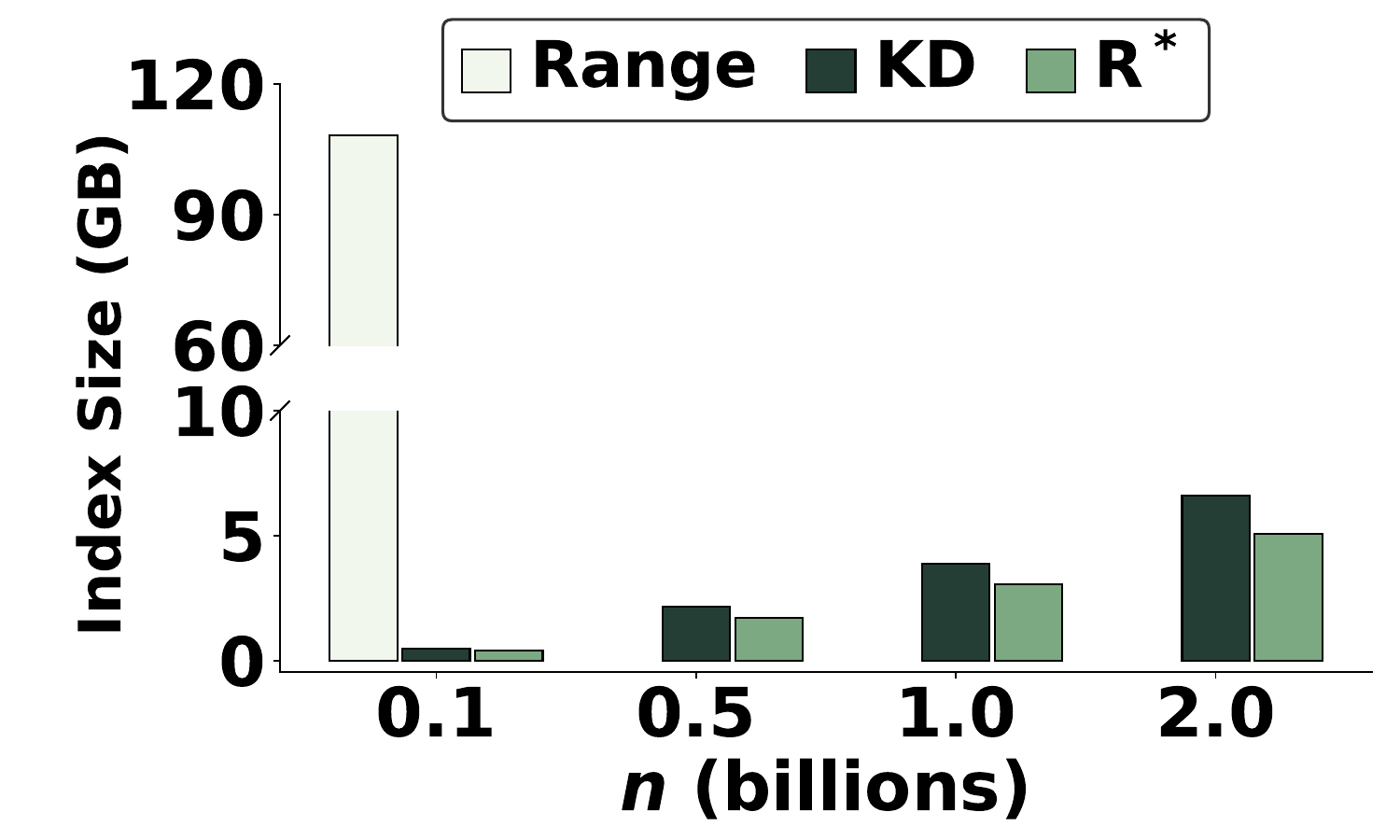}
    \caption{SARS}\label{fig:ds:c}
\end{subfigure}
\begin{subfigure}[t]{0.24\columnwidth}
    \centering
    \includegraphics[width=1\linewidth]{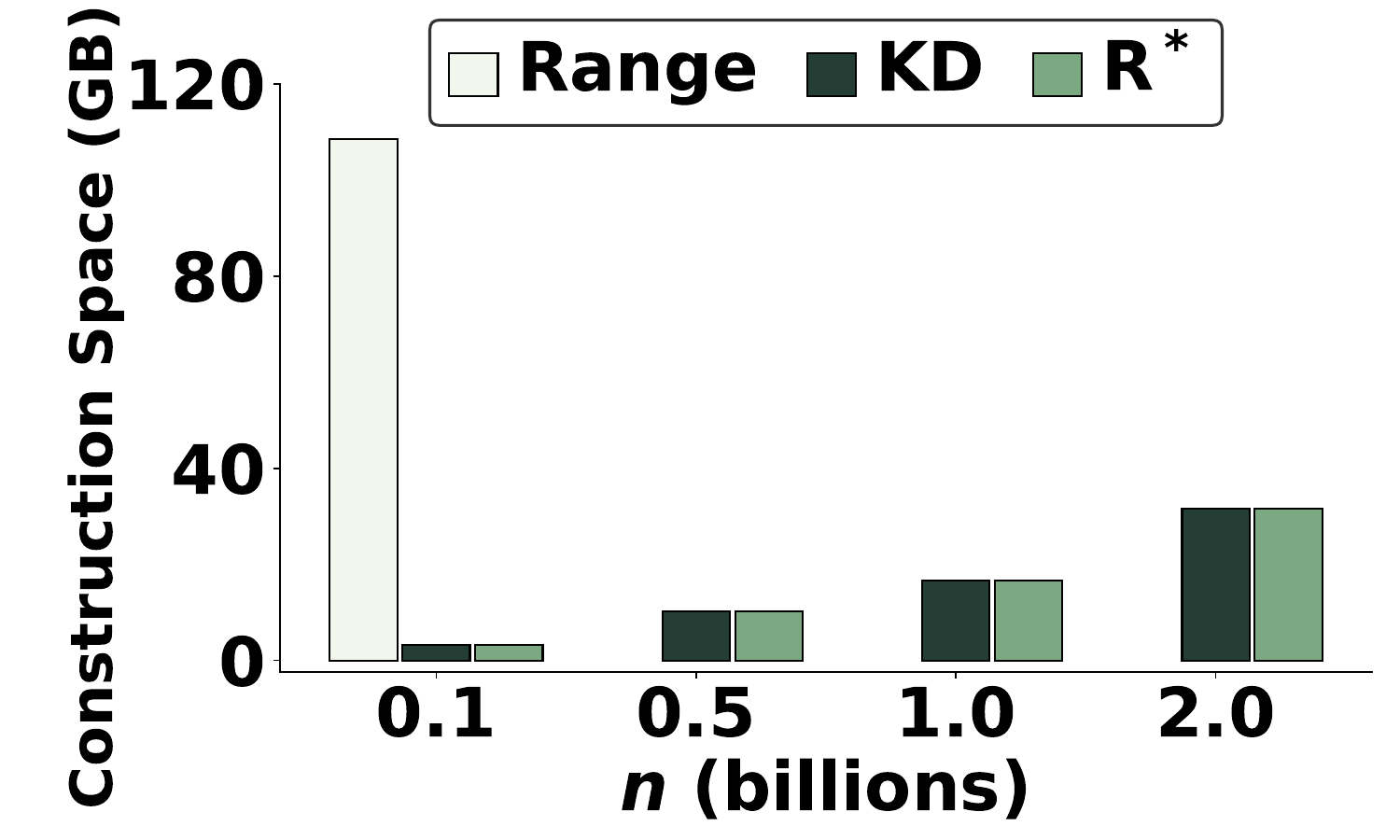}
    \caption{SARS}\label{fig:ds:d}
\end{subfigure}

\caption{\CI with Range trees, KD trees, or R$^*$ trees for workload $\mathcal{W}_{9,9,9}$: (a) average query time, (b) construction time, (c) index size, and (d) construction space. Missing bars indicate that a method needed more than $512$ GBs of memory.}
\label{fig:BIwith_trees}

\end{figure}

In the following, we present the results of the competitors and our \CI index on all $5$ datasets. Note that \CPRIS did not terminate when  applied to the entire \chr dataset, as it needed more than the $512$ GBs of memory that was available.

\subparagraph{Average Query Time.}~Figs.~\ref{fig:querytime_wiki_n} to \ref{fig:querytime_chr_n} show the average query time for the queries in the workload $\mathcal{W}_{9,9,9}$, using prefixes of the $5$ datasets with varying length $n$. \CI is on average $52$ times and up to  
$172$ times faster than \CPRIS and \emph{more than three orders of magnitude faster} than \RI. 
The query time for \CPRIS and \RI increases significantly with $n$, unlike that of \CI, since their query time bounds depend on $|\mathcal{C}_T(P,lr)|$, which increases with $n$.  Figs.~\ref{fig:querytime_wiki_xy} to  \ref{fig:querytime_chr_xy} show the average query time for the workloads $\mathcal{W}_{9,l,r}$ with $l=r$ and $l,r \in \{3,6,9,12,15\}$ and Figs.~\ref{fig:querytime_wiki_m} to  \ref{fig:querytime_chr_m} for the workloads $\mathcal{W}_{m,9,9}$ with $m \in \{3,6,9,12,15\}$.  The query times of all indexes increase as $l$ and/or $r$ gets larger, or as $m$ gets smaller. For \RI and \CPRIS, this is because $|\mathcal{C}_T(P,l,r)|$ increases. For \CI, this is because $B$ increases (see Theorem~\ref{thm:BCPC_index}):  as $l+m+r \leq B$, when $l$, $r$ and/or $m$ increases, $B$ increases as well.  

\begin{figure*}[t]
    \centering

    \begin{subfigure}{0.188\linewidth}
        \includegraphics[width=1.05\linewidth]{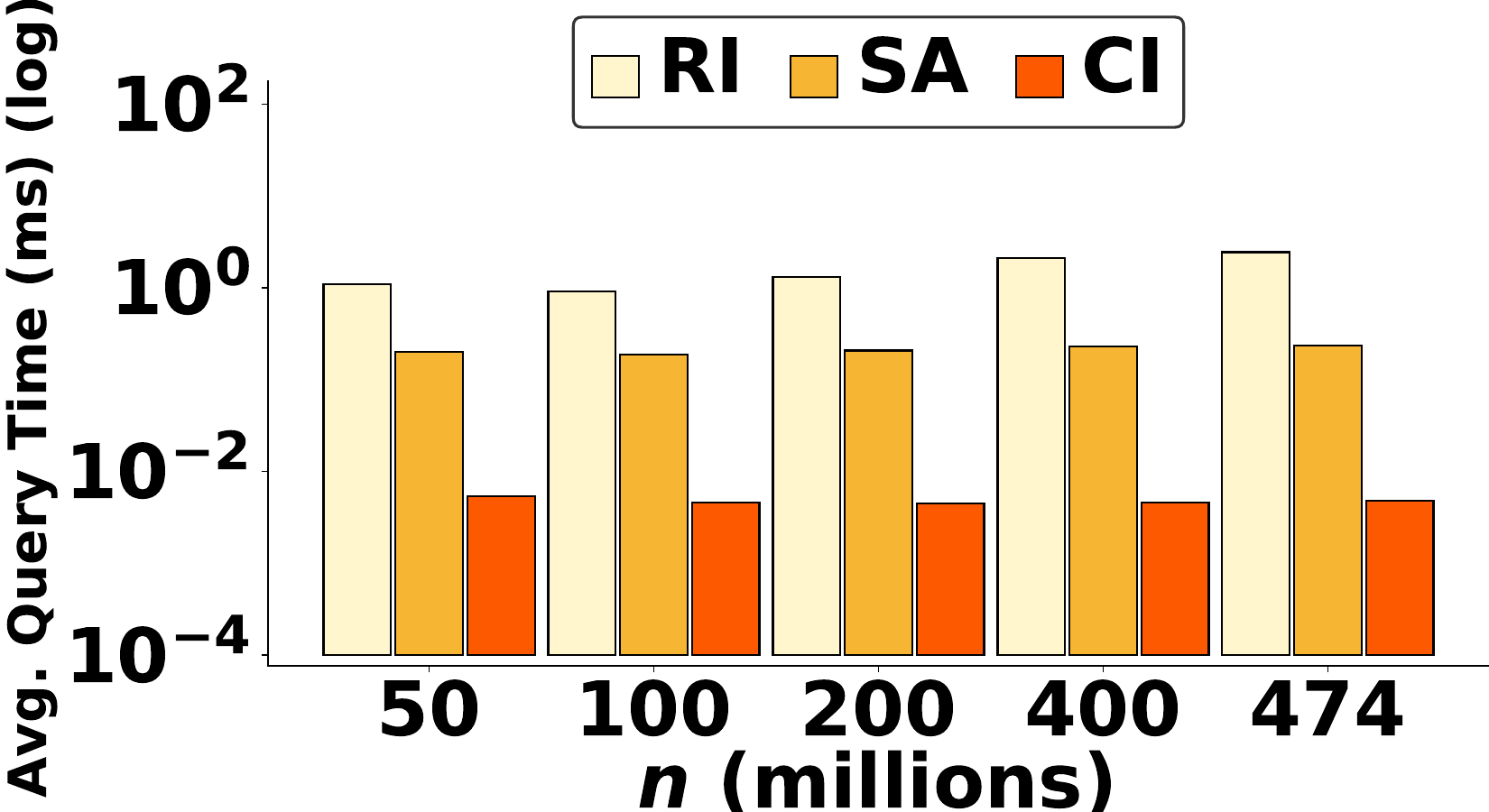}
        \caption{WIKI}
            \label{fig:querytime_wiki_n}
    \end{subfigure}%
\hfill
    \begin{subfigure}{0.188\linewidth}
        \includegraphics[width=1.05\linewidth]{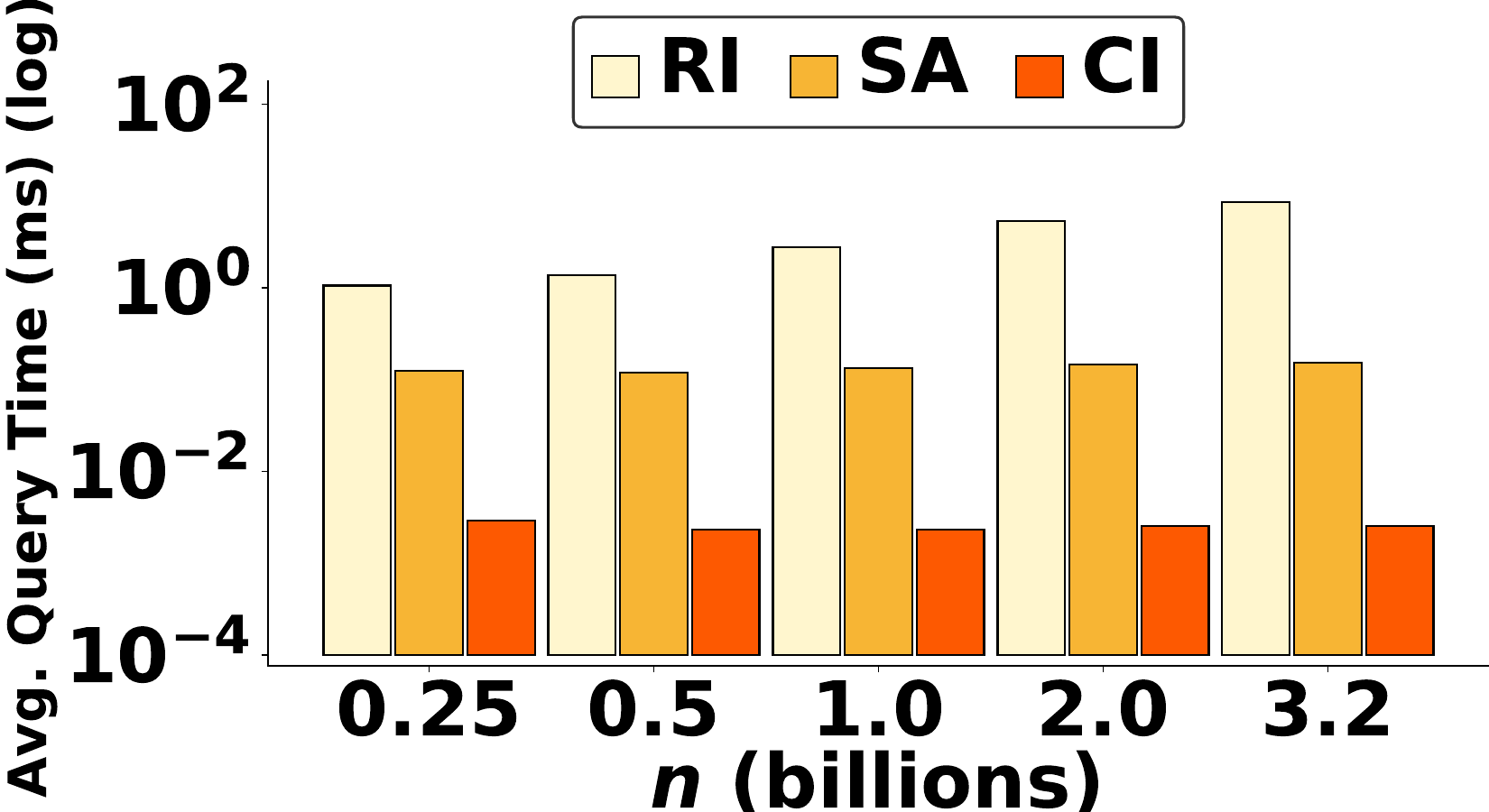}
        \caption{BST}
    \end{subfigure}%
\hfill
    \begin{subfigure}{0.188\linewidth}
        \includegraphics[width=1.05\linewidth]{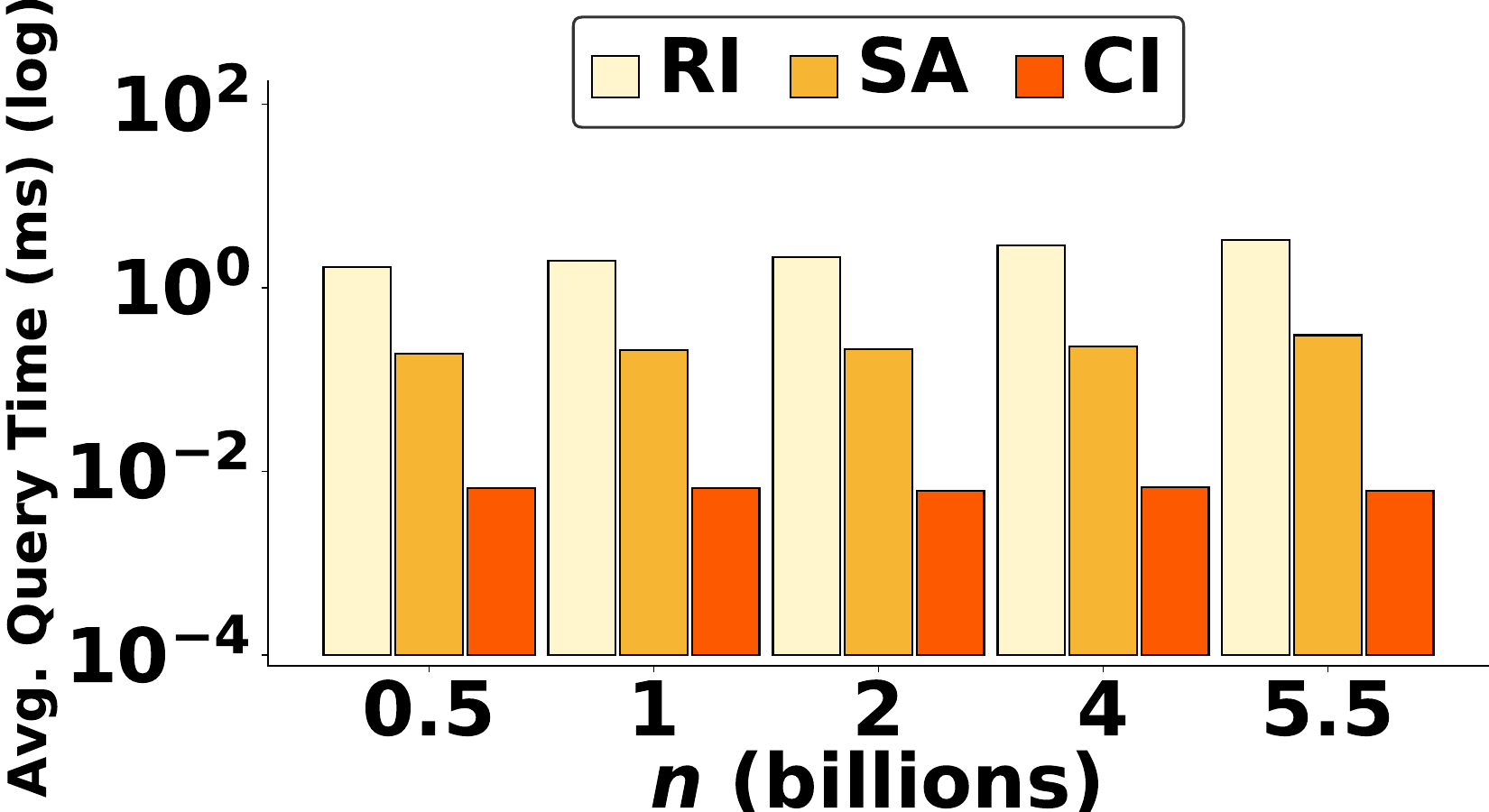}
        \caption{SDSL}
    \end{subfigure}%
\hfill
    \begin{subfigure}{0.188\linewidth}
        \includegraphics[width=1.05\linewidth]{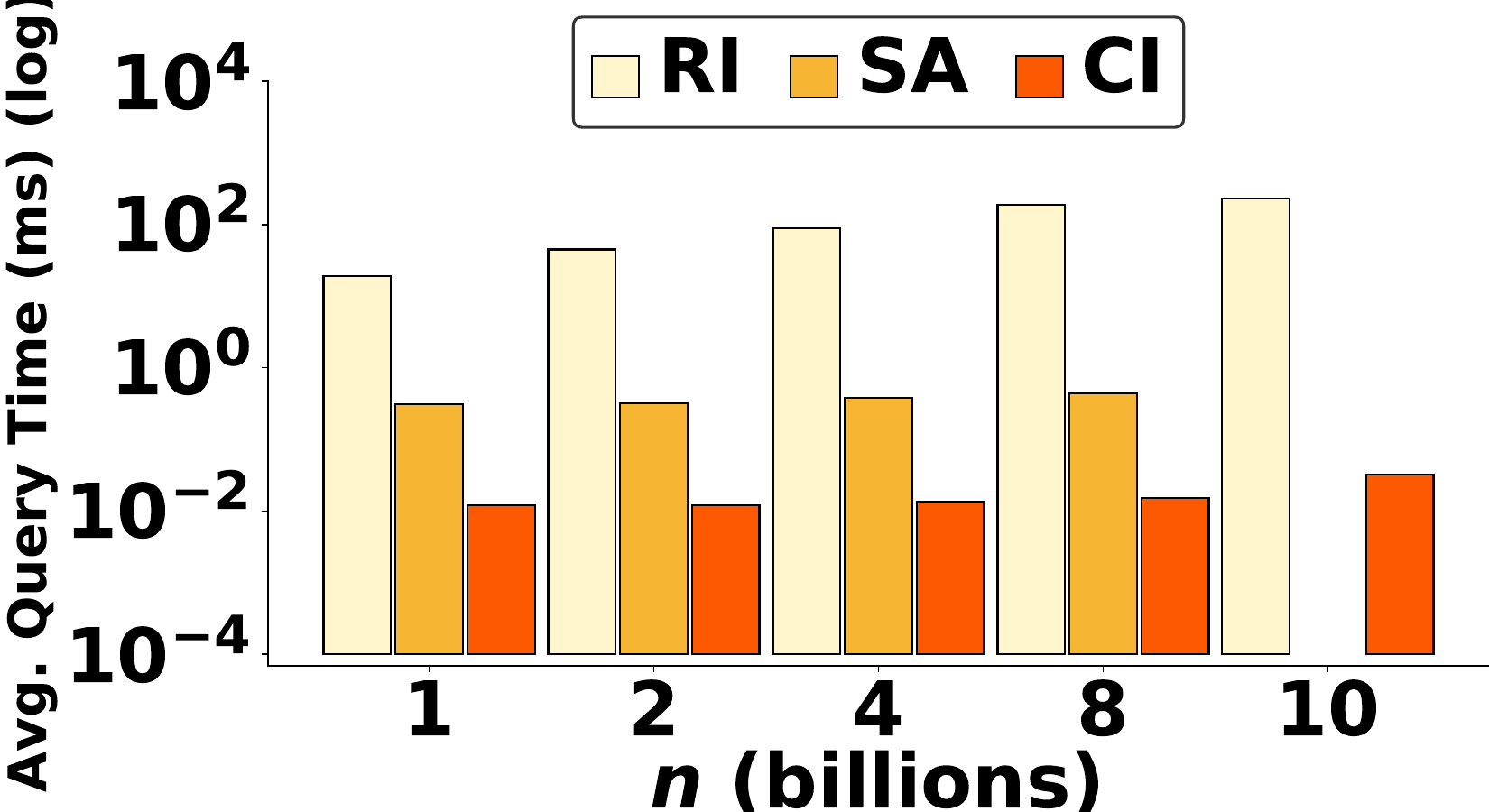}
        \caption{\sars}
    \end{subfigure}%
\hfill
    \begin{subfigure}{0.188\linewidth}
        \includegraphics[width=1.05\linewidth]{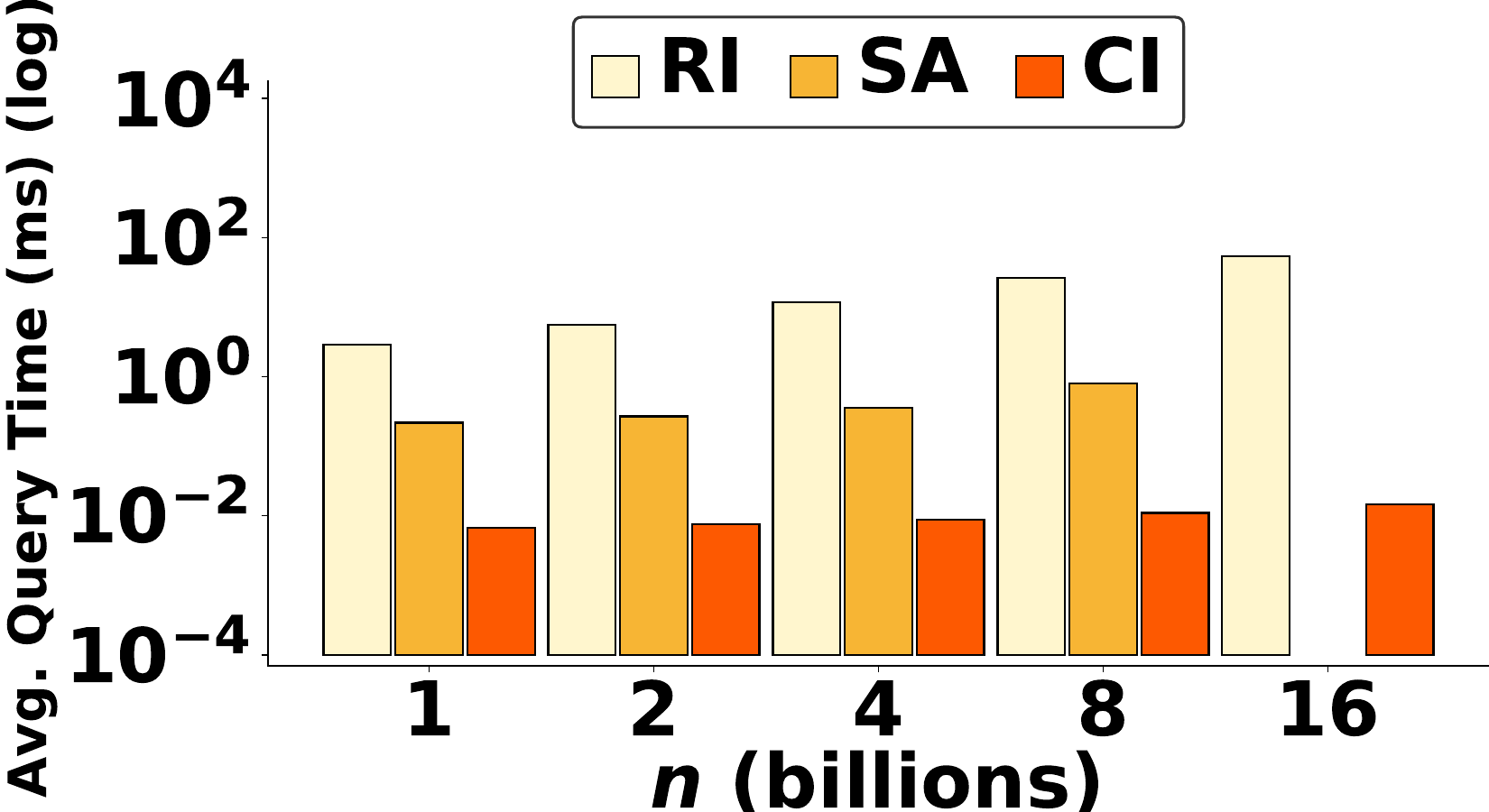}
        \caption{\chr}
        \label{fig:querytime_chr_n}
    \end{subfigure}%

    \begin{subfigure}{0.188\linewidth}
        \includegraphics[width=1.05\linewidth]{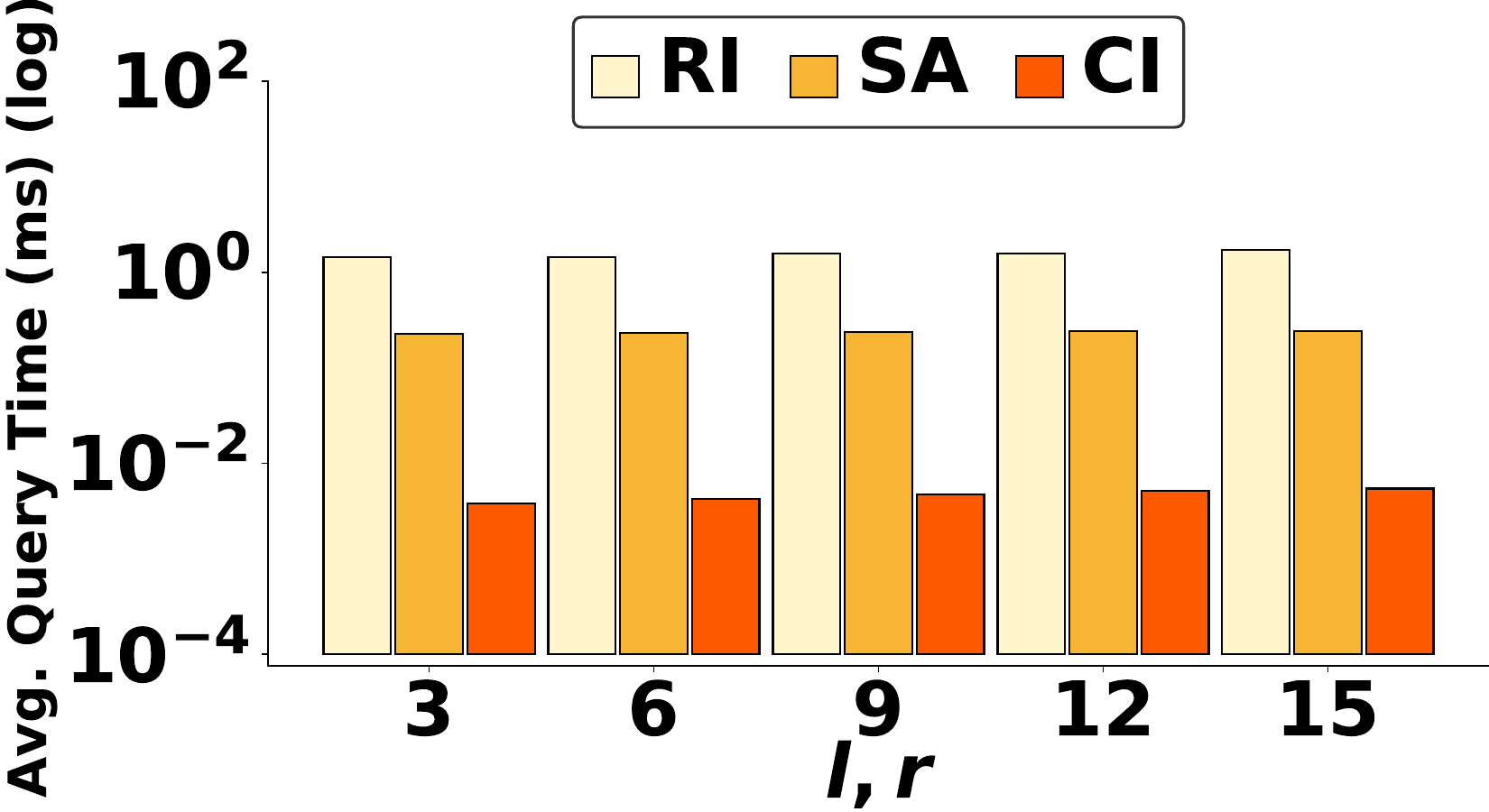}
        \caption{WIKI}
        \label{fig:querytime_wiki_xy}
    \end{subfigure}%
\hfill
    \begin{subfigure}{0.188\linewidth}
        \includegraphics[width=1.05\linewidth]{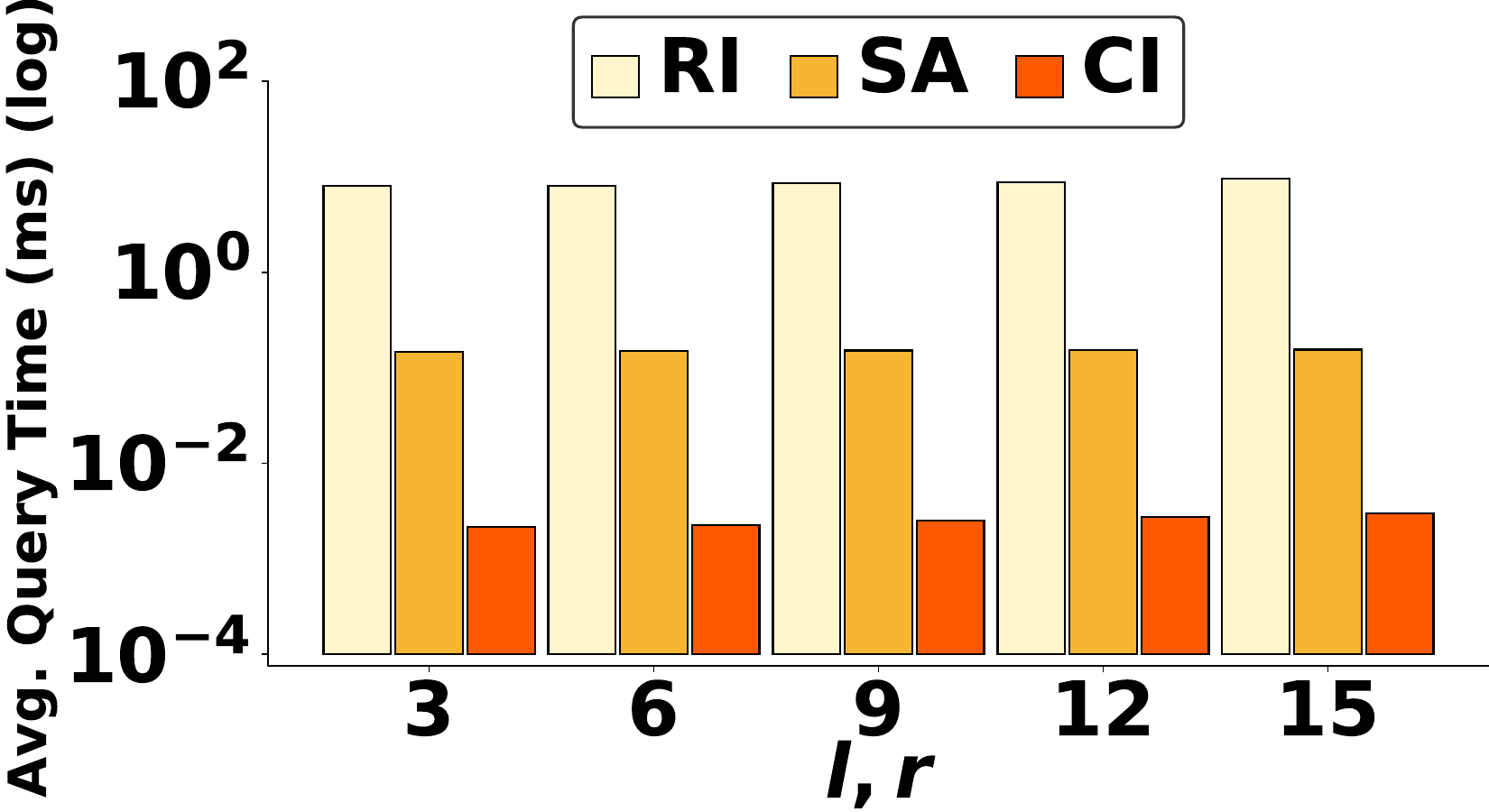}
        \caption{BST}
    \end{subfigure}%
\hfill
    \begin{subfigure}{0.188\linewidth}
        \includegraphics[width=1.05\linewidth]{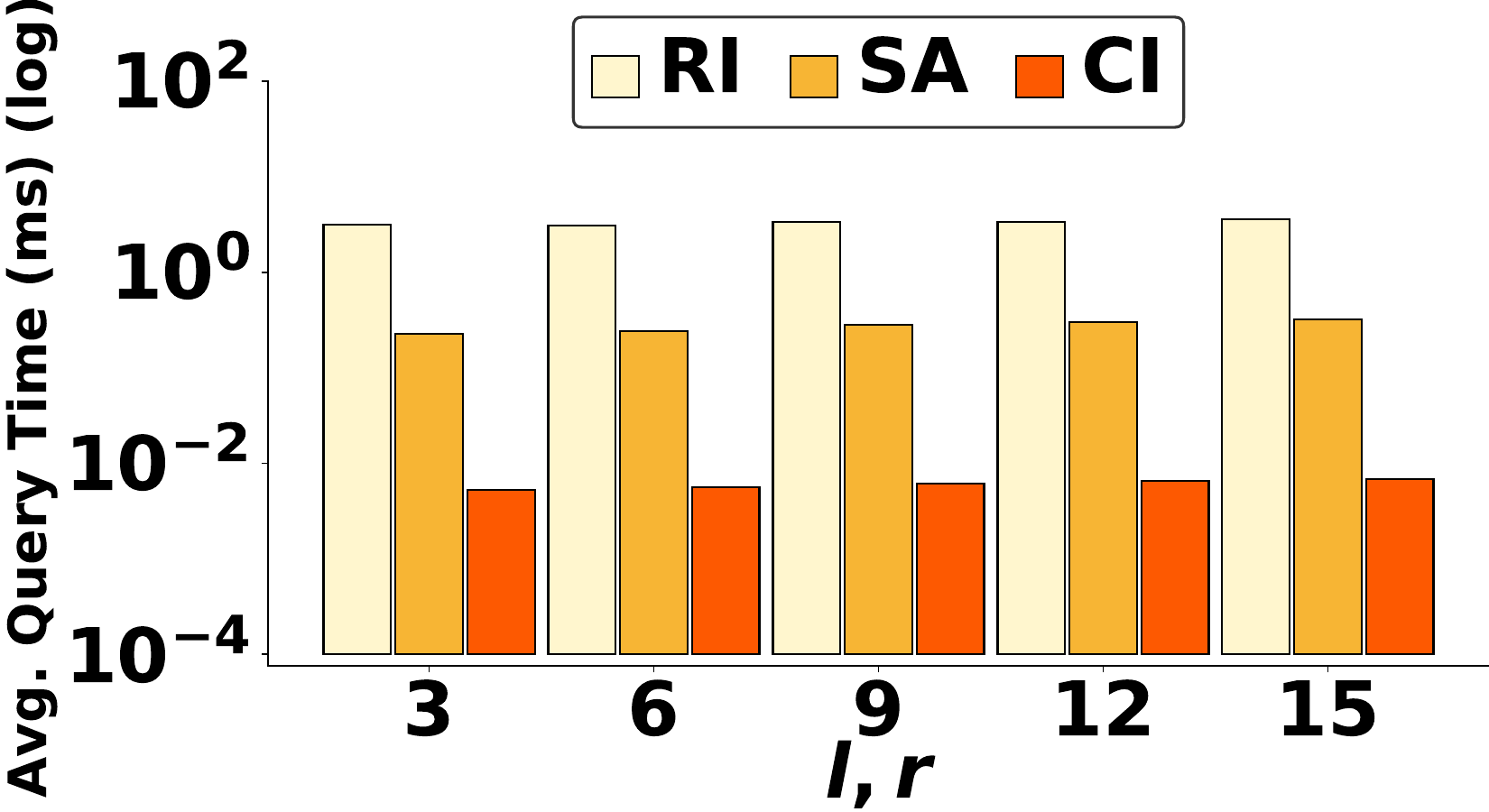}
        \caption{SDSL}
    \end{subfigure}%
\hfill
     \begin{subfigure}{0.188\linewidth}
        \includegraphics[width=1.05\linewidth]{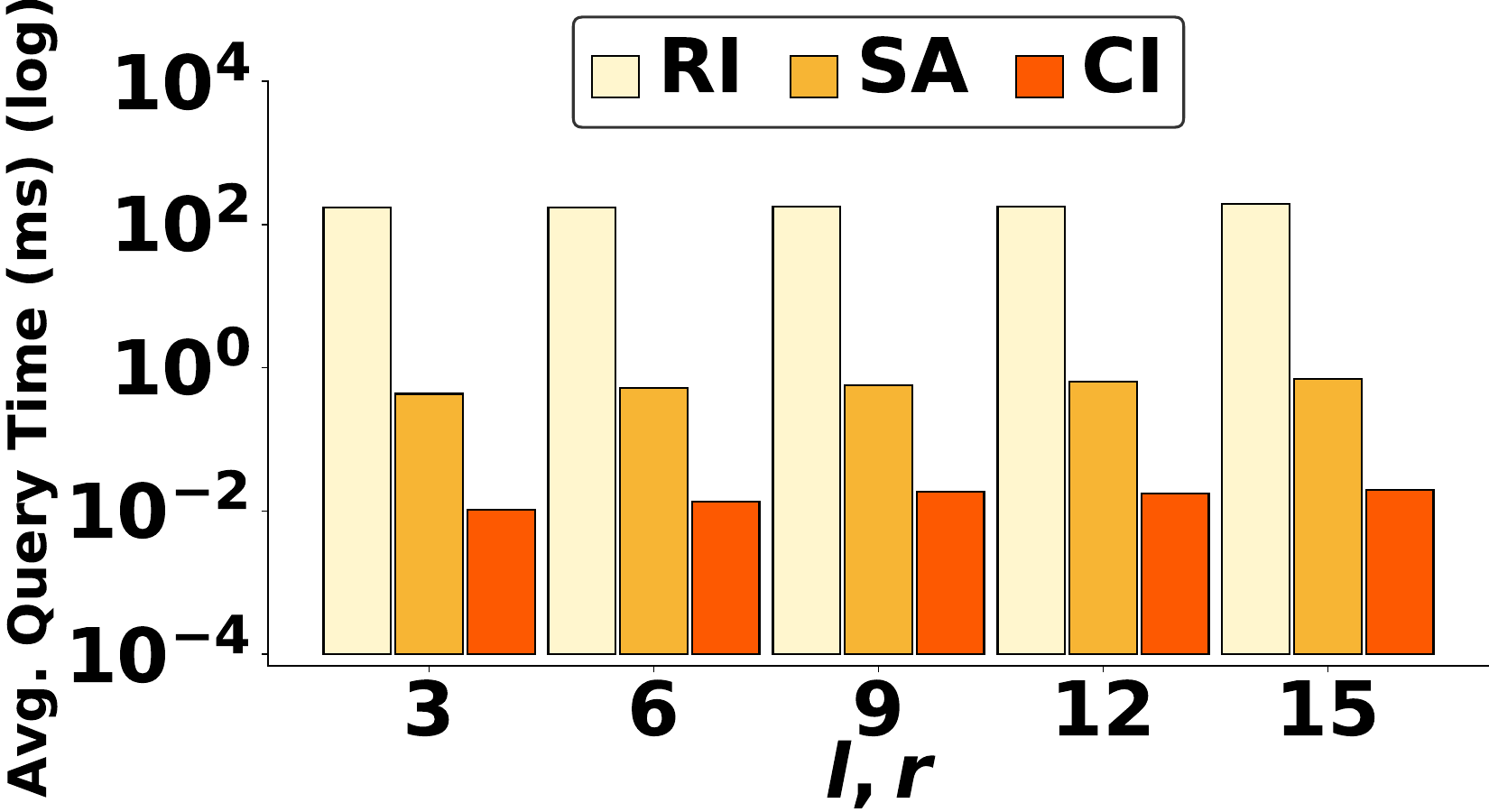}
        \caption{\sars}
    \end{subfigure}%
\hfill
    \begin{subfigure}{0.188\linewidth}
        \includegraphics[width=1.05\linewidth]{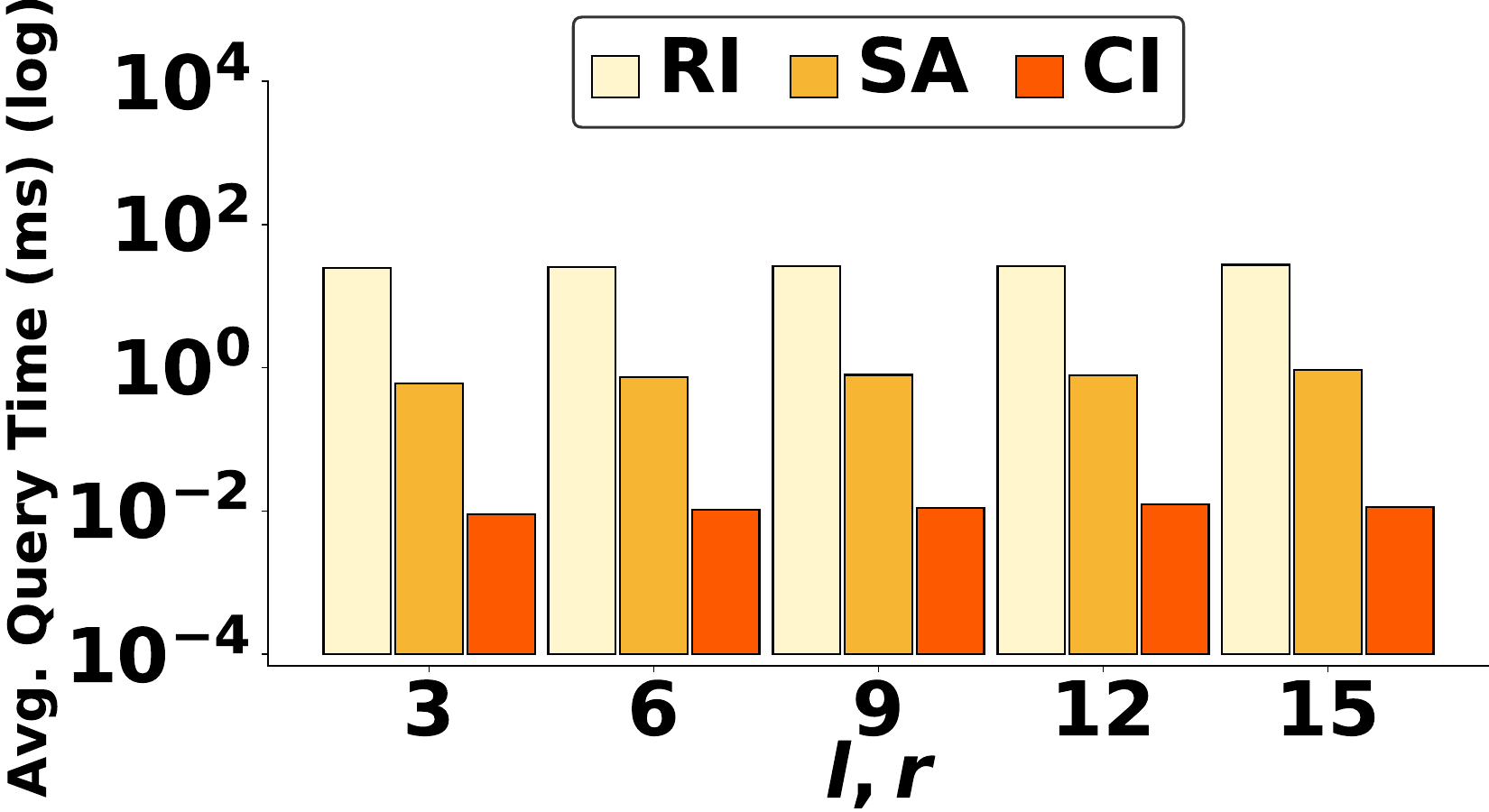}
        \caption{\chr}
        \label{fig:querytime_chr_xy}
    \end{subfigure}%

    \begin{subfigure}{0.188\linewidth}
        \includegraphics[width=1.05\linewidth]{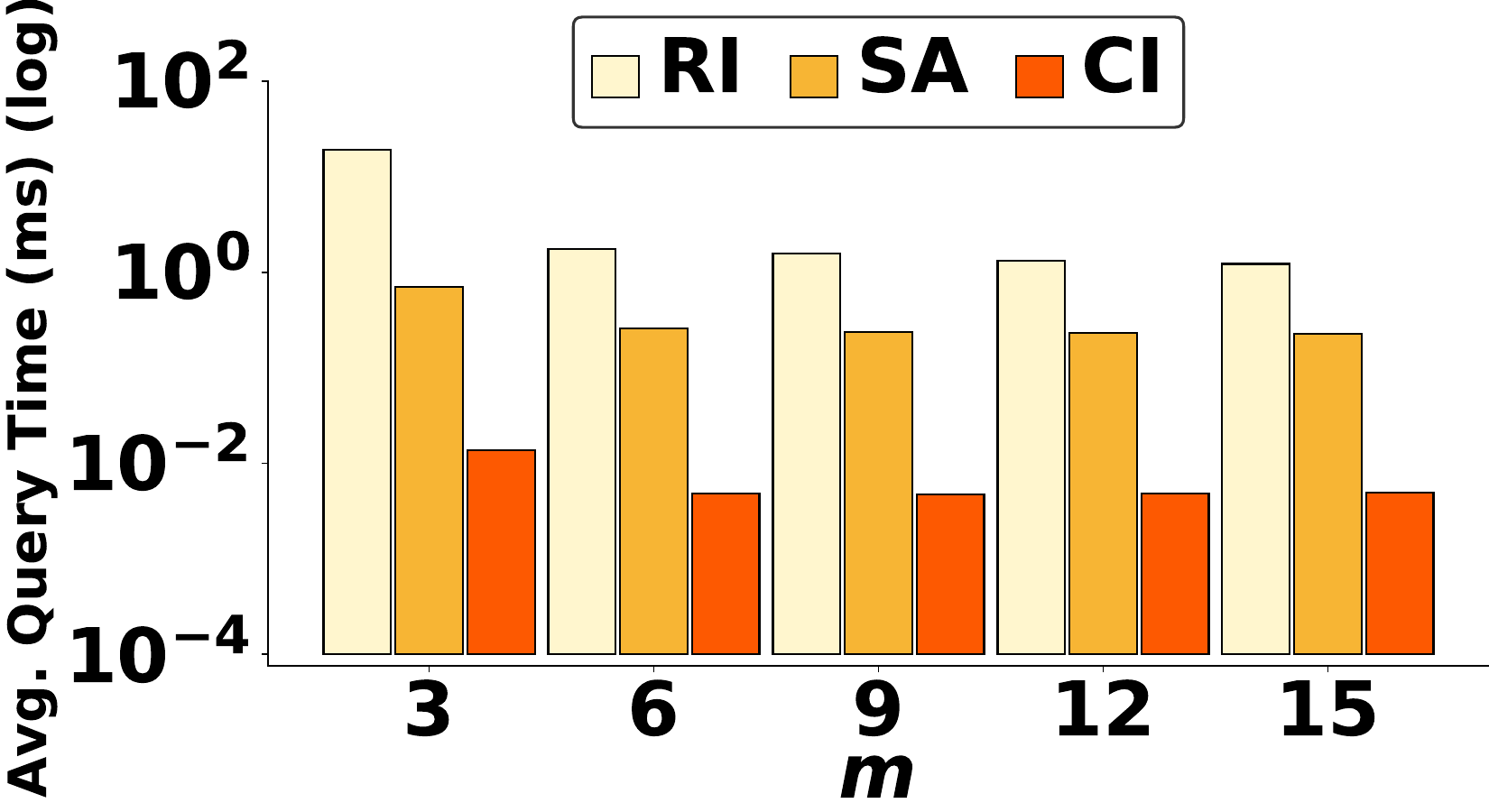}
        \caption{WIKI}
    \label{fig:querytime_wiki_m}
    \end{subfigure}%
\hfill
    \begin{subfigure}{0.188\linewidth}
        \includegraphics[width=1.05\linewidth]{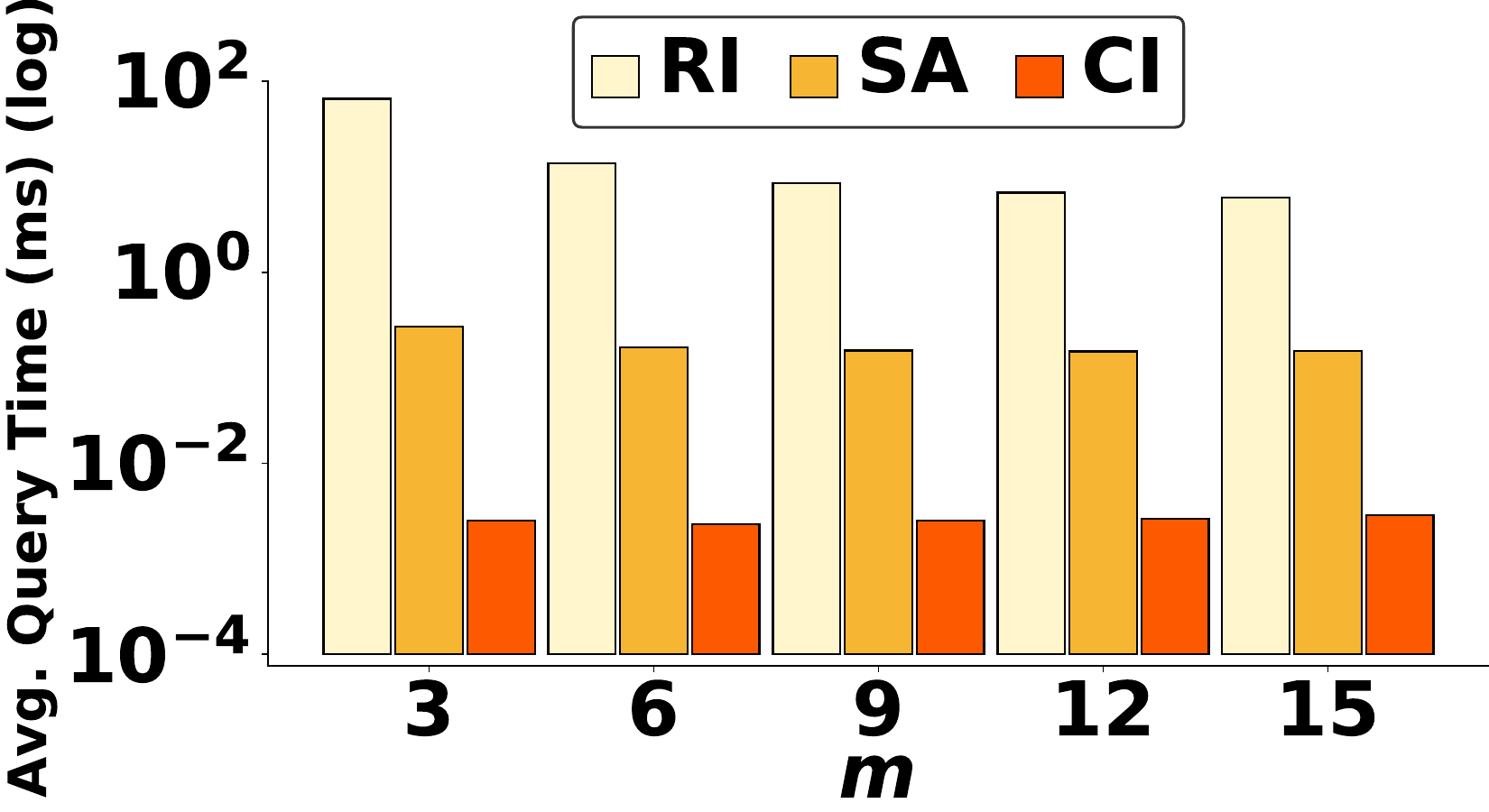}
        \caption{BST}
    \end{subfigure}%
\hfill
    \begin{subfigure}{0.188\linewidth}
        \includegraphics[width=1.05\linewidth]{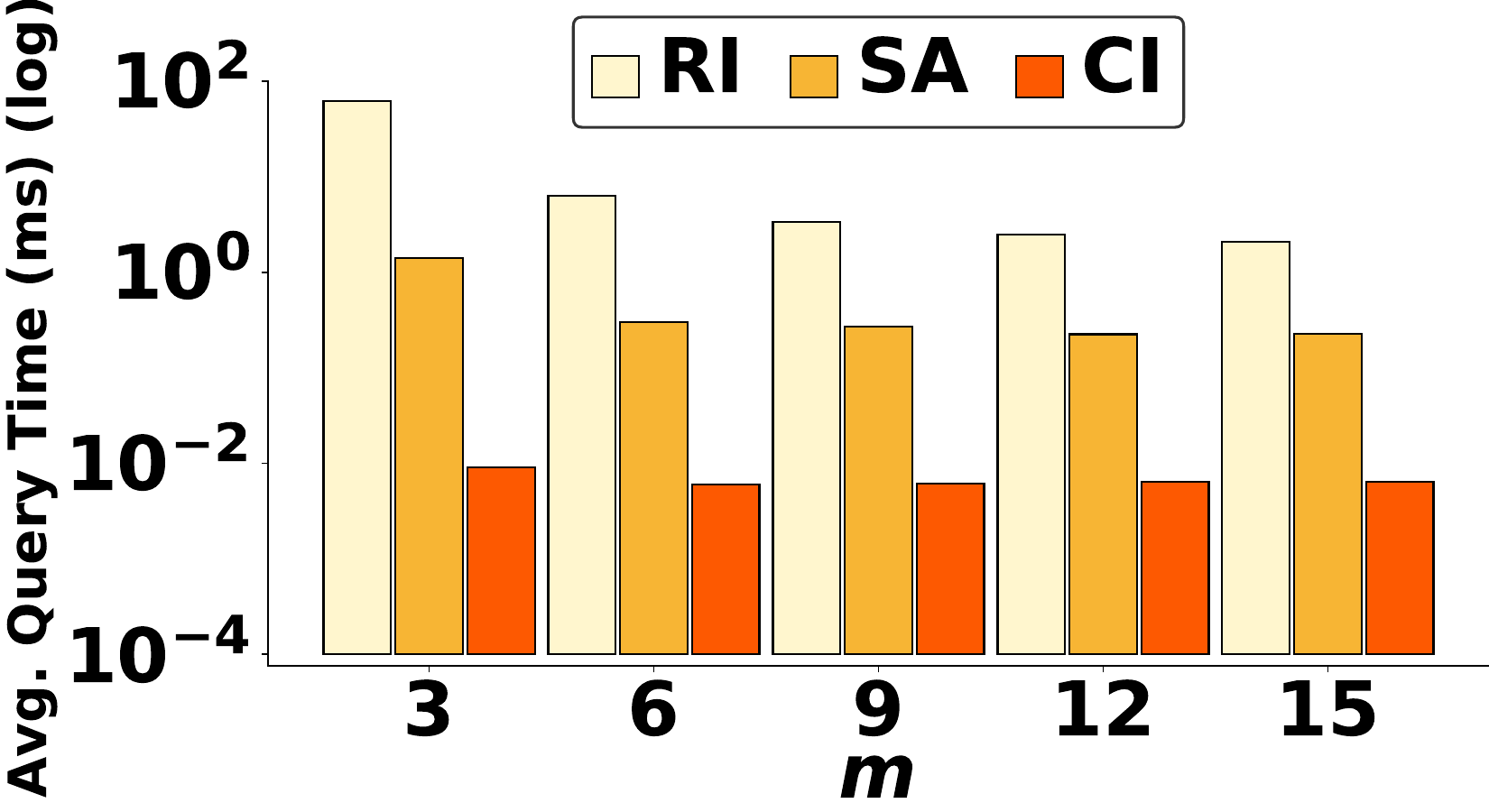}
        \caption{SDSL}
    \end{subfigure}%
\hfill
    \begin{subfigure}{0.188\linewidth}
        \includegraphics[width=1.05\linewidth]{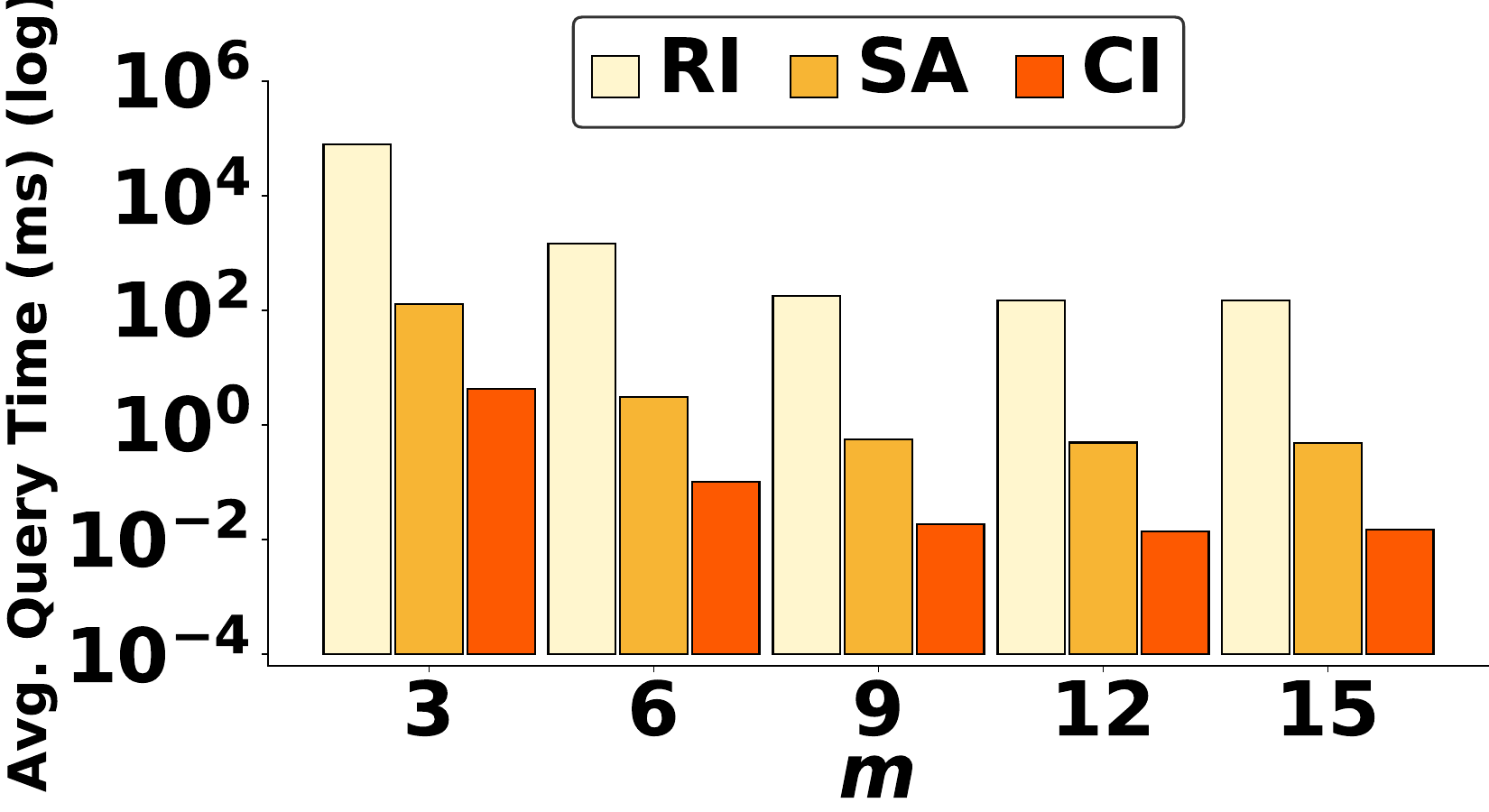}
        \caption{\sars}
    \end{subfigure}%
\hfill
        \begin{subfigure}{0.188\linewidth}
        \includegraphics[width=1.05\linewidth]{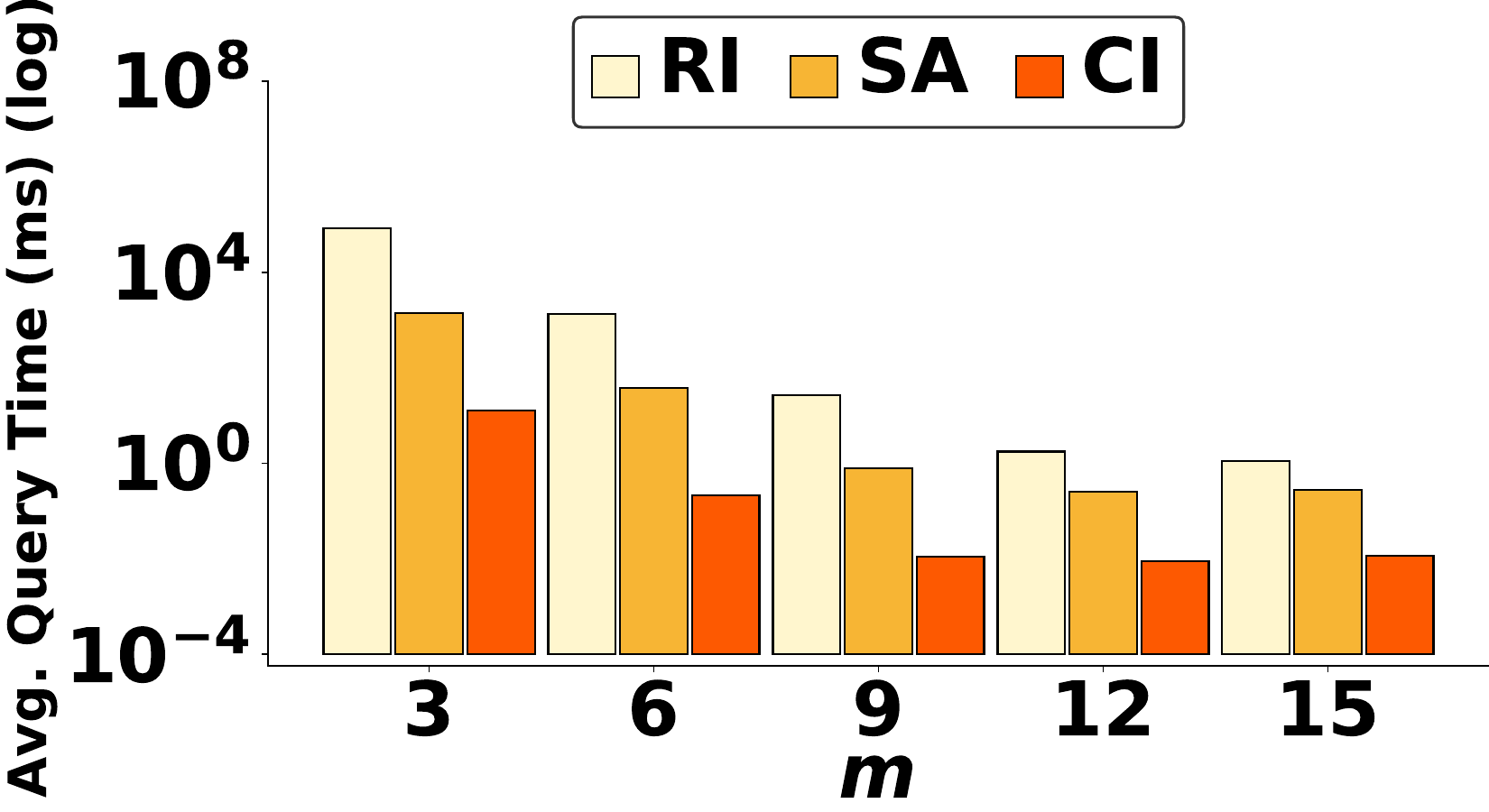}
        \caption{\chr}
          \label{fig:querytime_chr_m}
    \end{subfigure}%

    \caption{Average query time across all datasets: (a--e) vs.\ $n$, (f--j) vs.\ $l,r$, and (k--o) vs.\ $m$.}
    \label{fig:querytime}
\end{figure*}

\subparagraph{Construction Time.}~Figs.~\ref{fig:construction_time_WIKI_n} to \ref{fig:construction_time_chr_n} show the construction time of all methods for the experiments of Figs.~\ref{fig:querytime_wiki_n} to  \ref{fig:querytime_chr_n}. All times increase linearly with $n$, as expected. Figs.~\ref{fig:construction_time_WIKI_xy} to  \ref{fig:construction_time_CHR_xy} (respectively, Figs.~\ref{fig:construction_time_WIKI_m} to  \ref{fig:construction_time_CHR_m}) show the construction time for the experiment of Figs. \ref{fig:querytime_wiki_xy} to  \ref{fig:querytime_chr_xy} (respectively, Figs. \ref{fig:querytime_wiki_m} to  \ref{fig:querytime_chr_m}). 
The construction time of \CI increases slightly with $l$, $r$, or $m$, as \CI is constructed in $\Oh(n)+\Ohtilde(k)$ time, where $k =  \min(B\cdot z, n)$ and $l+m+r \leq B$. The construction times of \RI and \CPRIS are not affected, as expected by their construction time complexities. 
The construction time of \CI is $38\%$ ($46\%$) of that of \RI (\CPRIS) on average, over all experiments of Fig.~\ref{fig:construction_time}.  

\begin{figure*}[t]
    \centering
    \begin{subfigure}{0.188\linewidth}
        \includegraphics[width=1.05\linewidth]{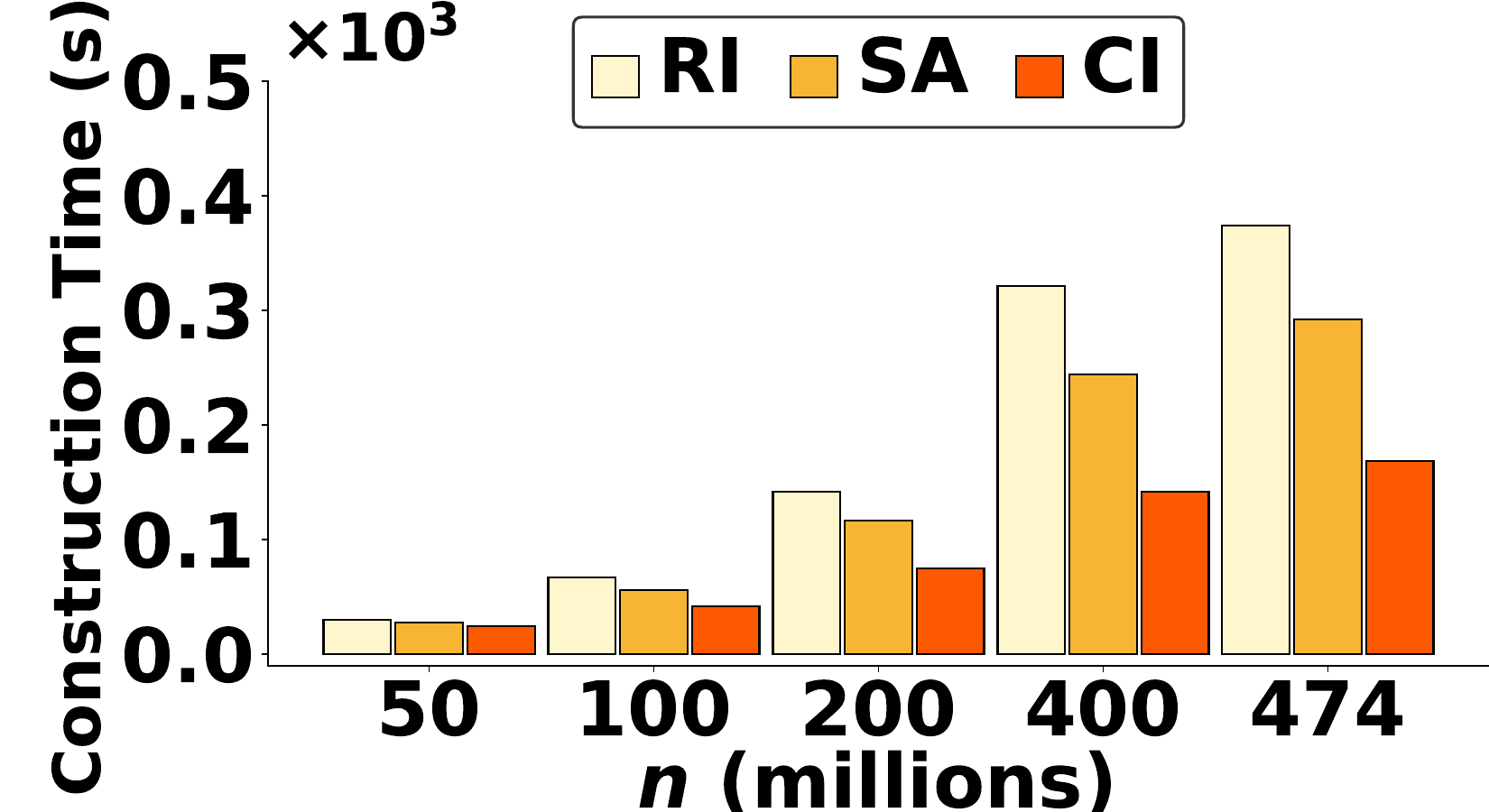}
        \caption{WIKI}
    \label{fig:construction_time_WIKI_n}
    \end{subfigure}%
 \hfill
    \begin{subfigure}{0.188\linewidth}
        \includegraphics[width=1.05\linewidth]{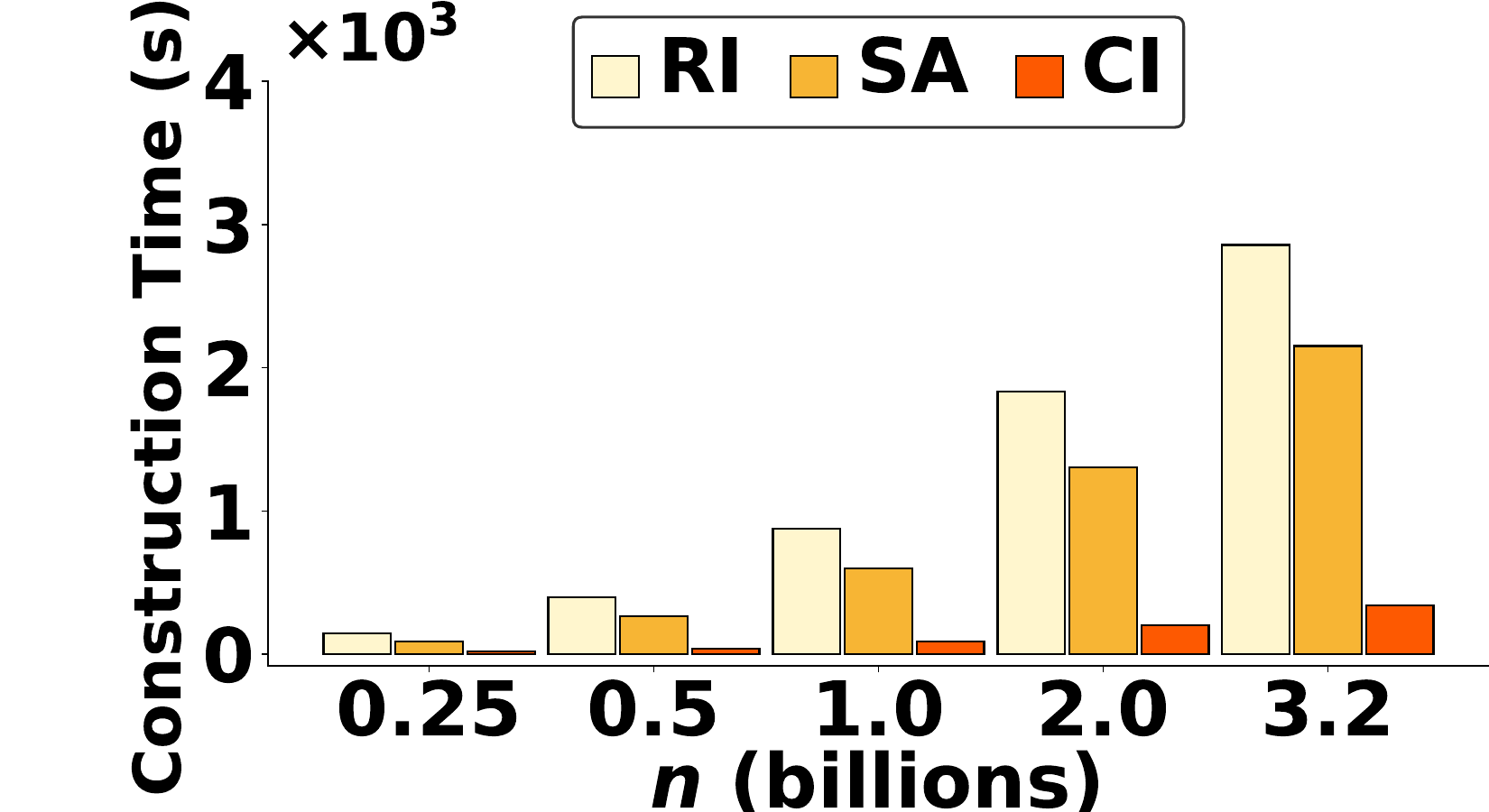}
        \caption{BST}
    \end{subfigure}%
 \hfill
    \begin{subfigure}{0.188\linewidth}
        \includegraphics[width=1.05\linewidth]{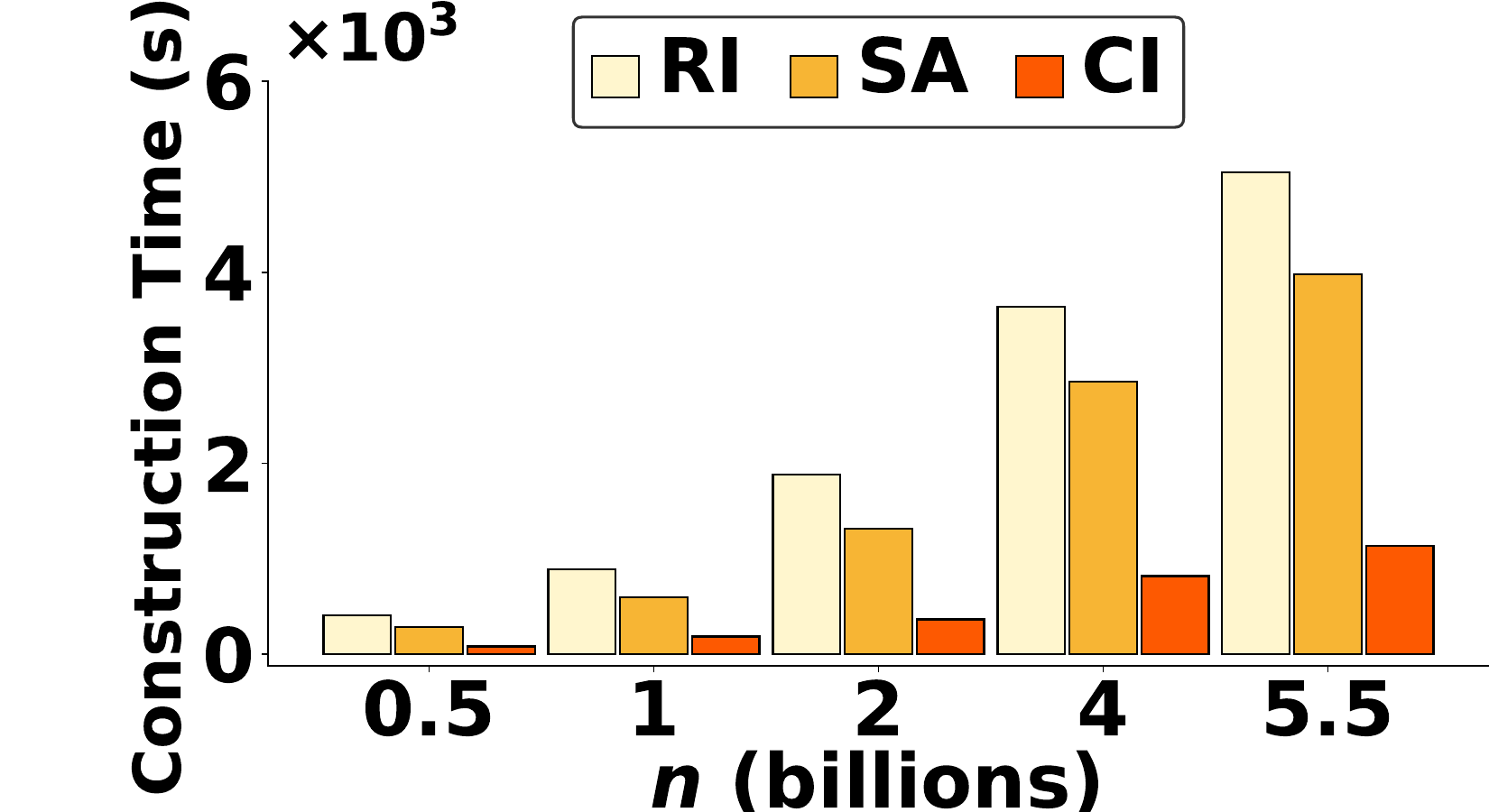}
        \caption{SDSL}
    \end{subfigure}%
\hfill
        \begin{subfigure}{0.188\linewidth}
        \includegraphics[width=1.05\linewidth]{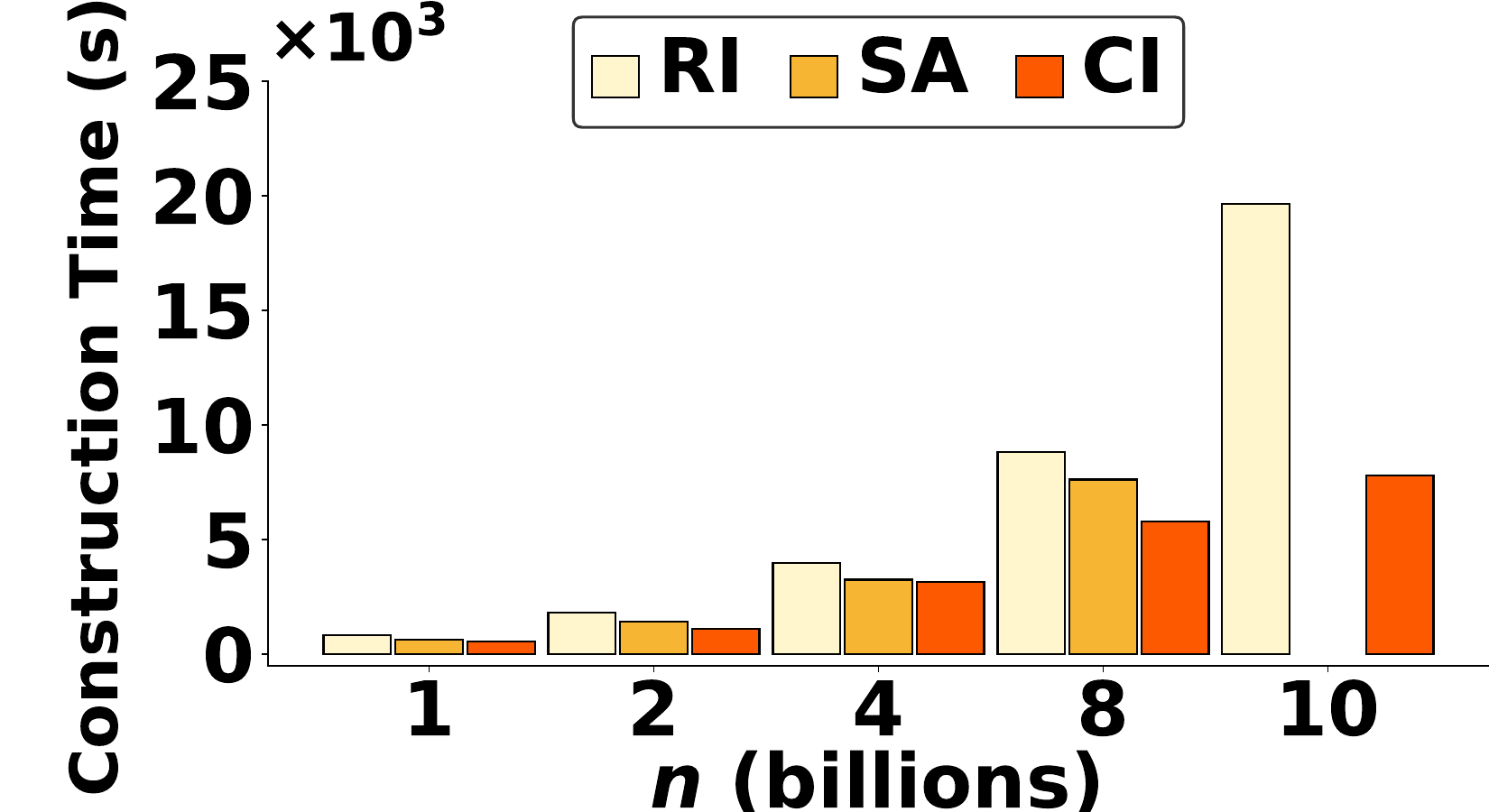}
        \caption{\sars}
    \end{subfigure}%
\hfill
    \begin{subfigure}{0.188\linewidth}
        \includegraphics[width=1.05\linewidth]{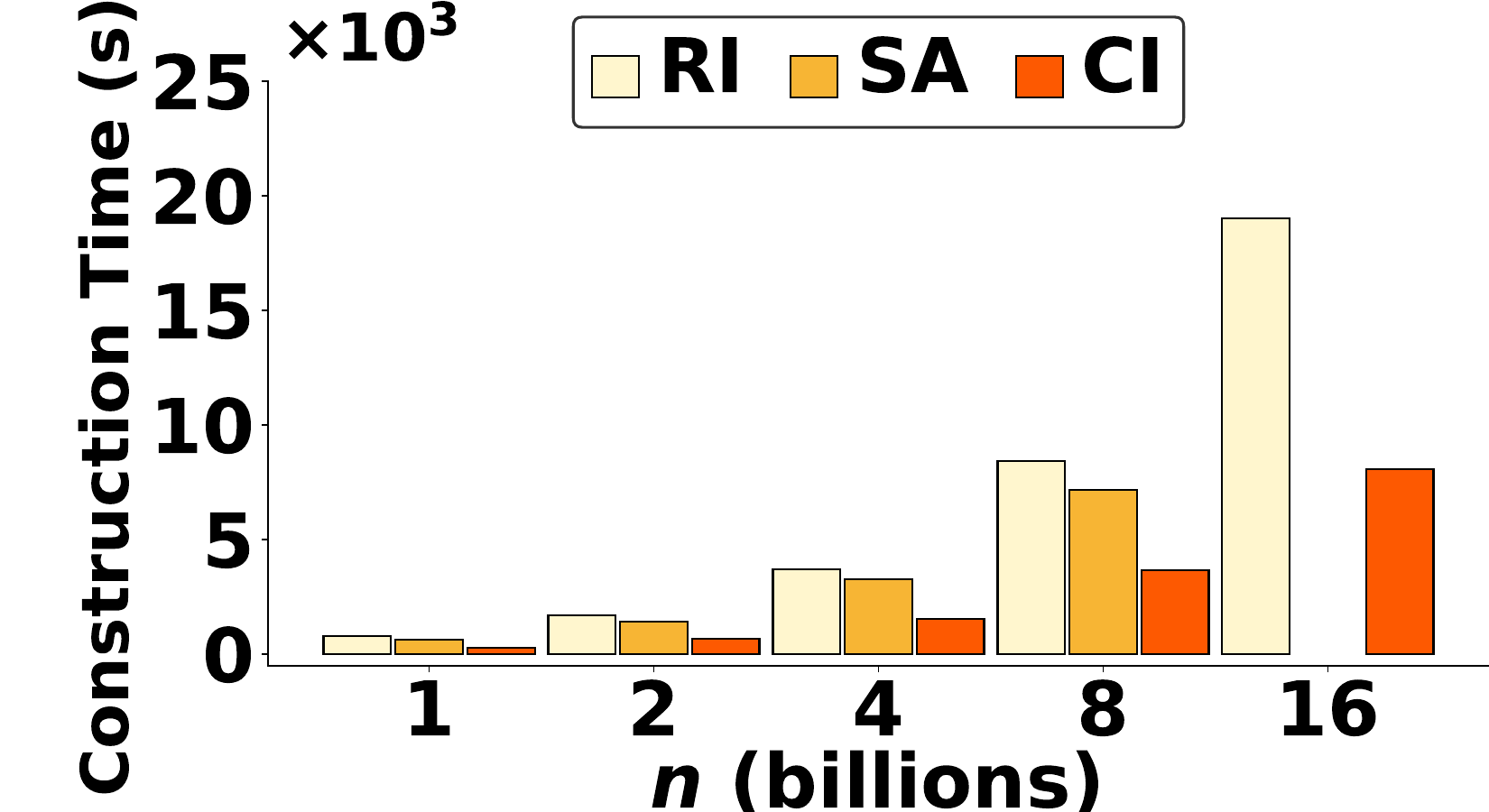}
        \caption{\chr}
    \label{fig:construction_time_chr_n}
    \end{subfigure}%

    \begin{subfigure}{0.188\linewidth}
        \includegraphics[width=1.05\linewidth]{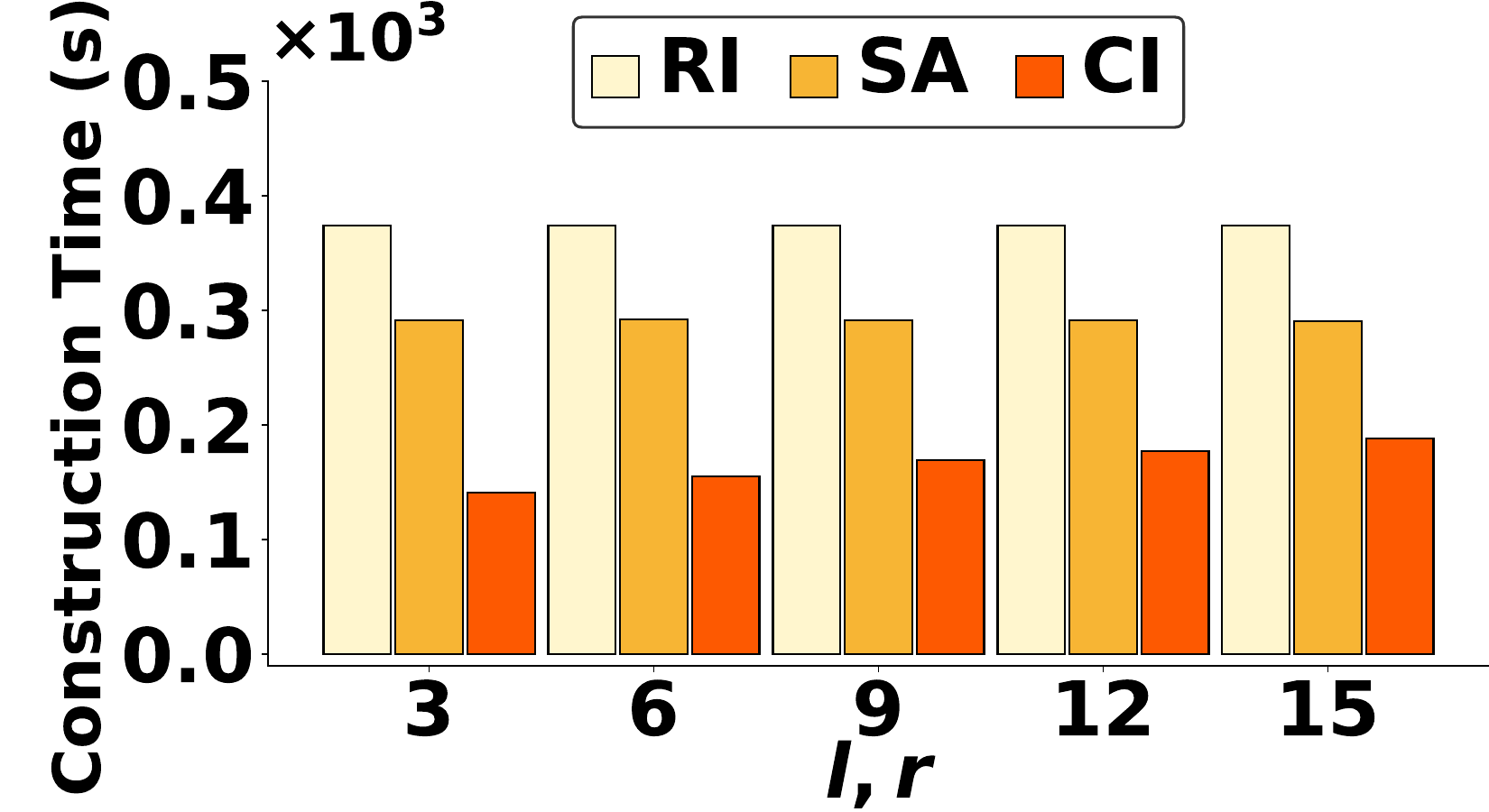}
        \caption{WIKI}
        \label{fig:construction_time_WIKI_xy}
    \end{subfigure}%
\hfill
    \begin{subfigure}{0.188\linewidth}
        \includegraphics[width=1.05\linewidth]{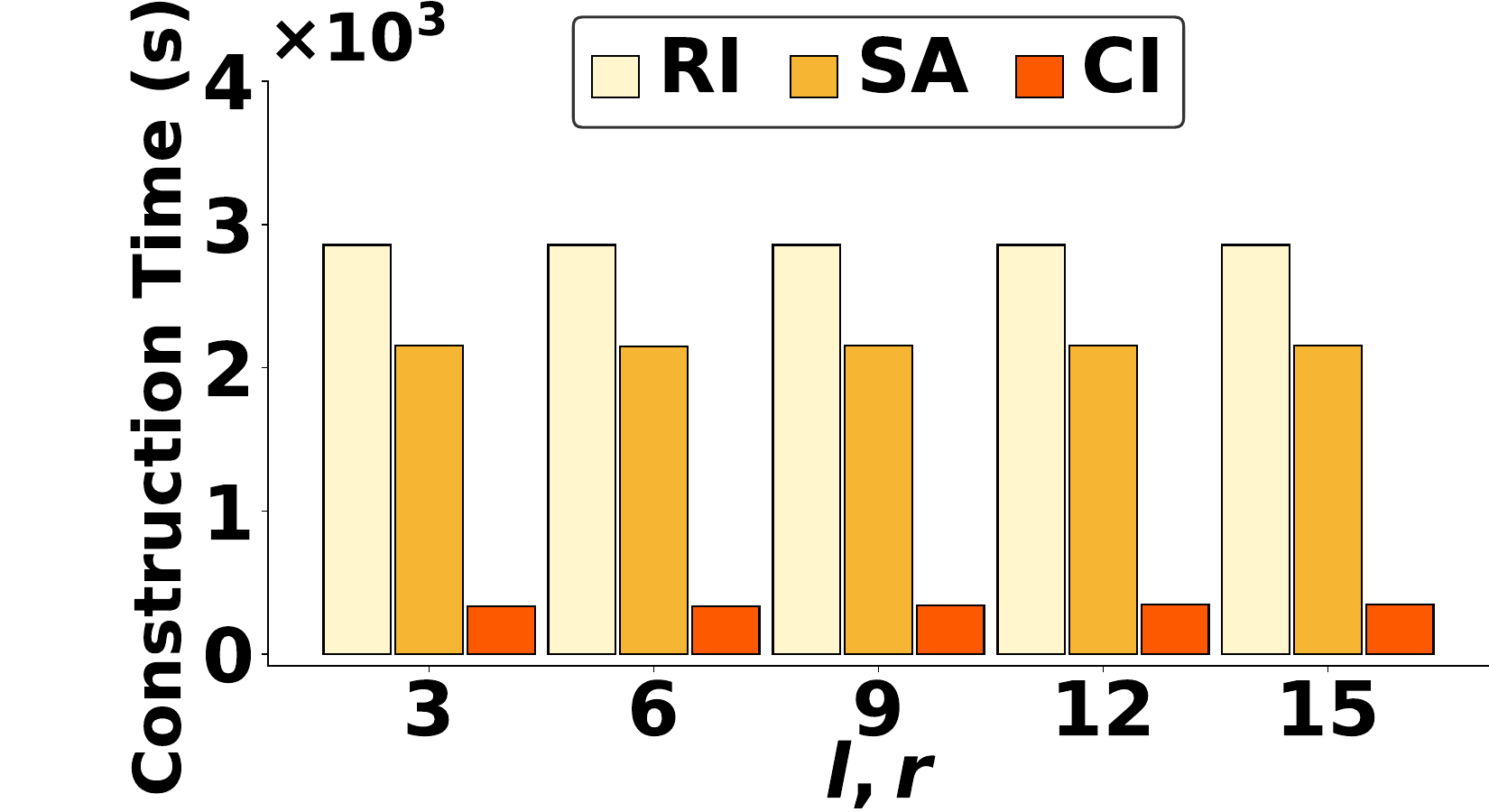}
        \caption{BST}
    \end{subfigure}%
\hfill
    \begin{subfigure}{0.188\linewidth}
        \includegraphics[width=1.05\linewidth]{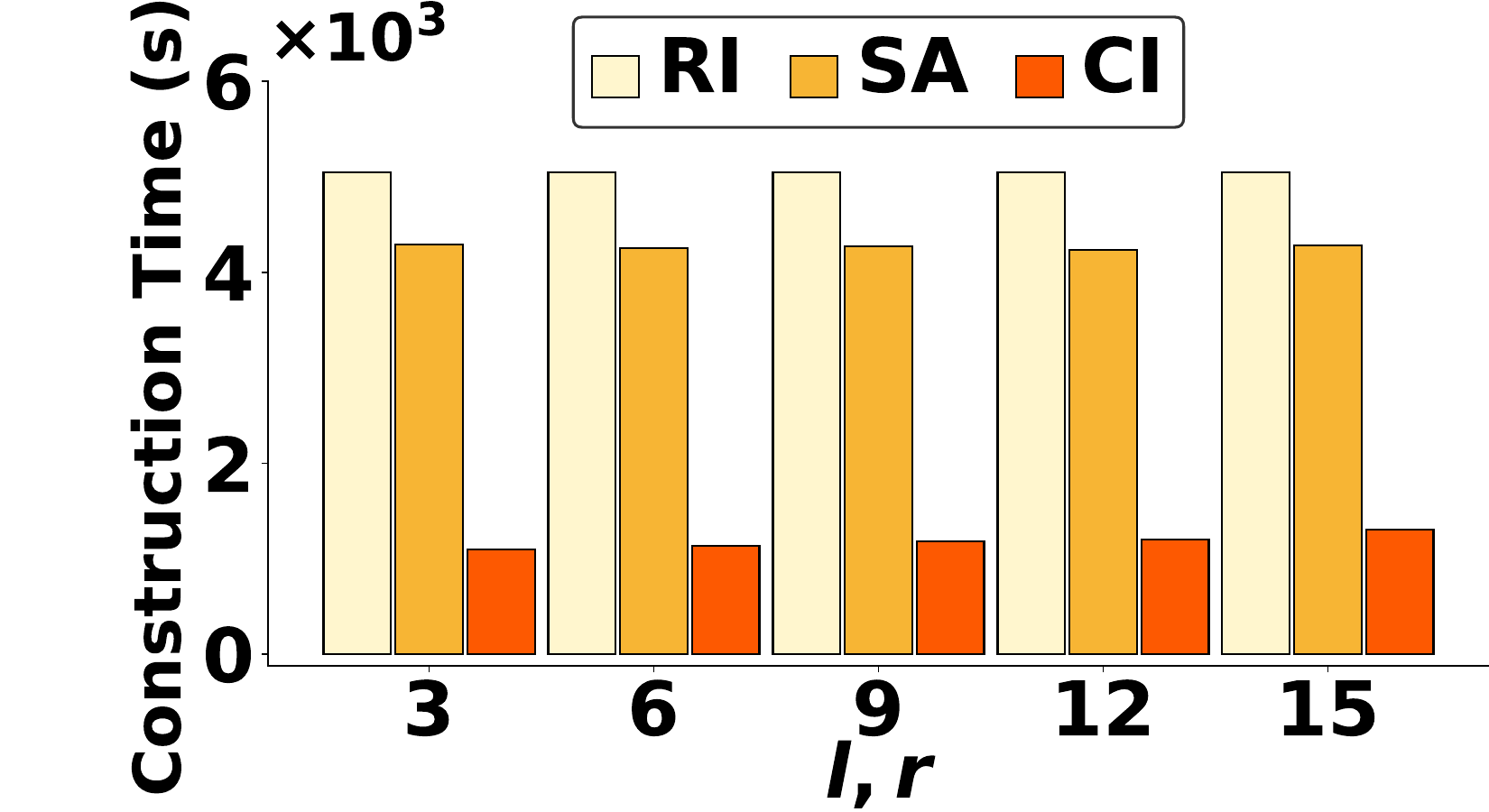}
        \caption{SDSL}
    \end{subfigure}%
\hfill
        \begin{subfigure}{0.188\linewidth}
        \includegraphics[width=1.05\linewidth]{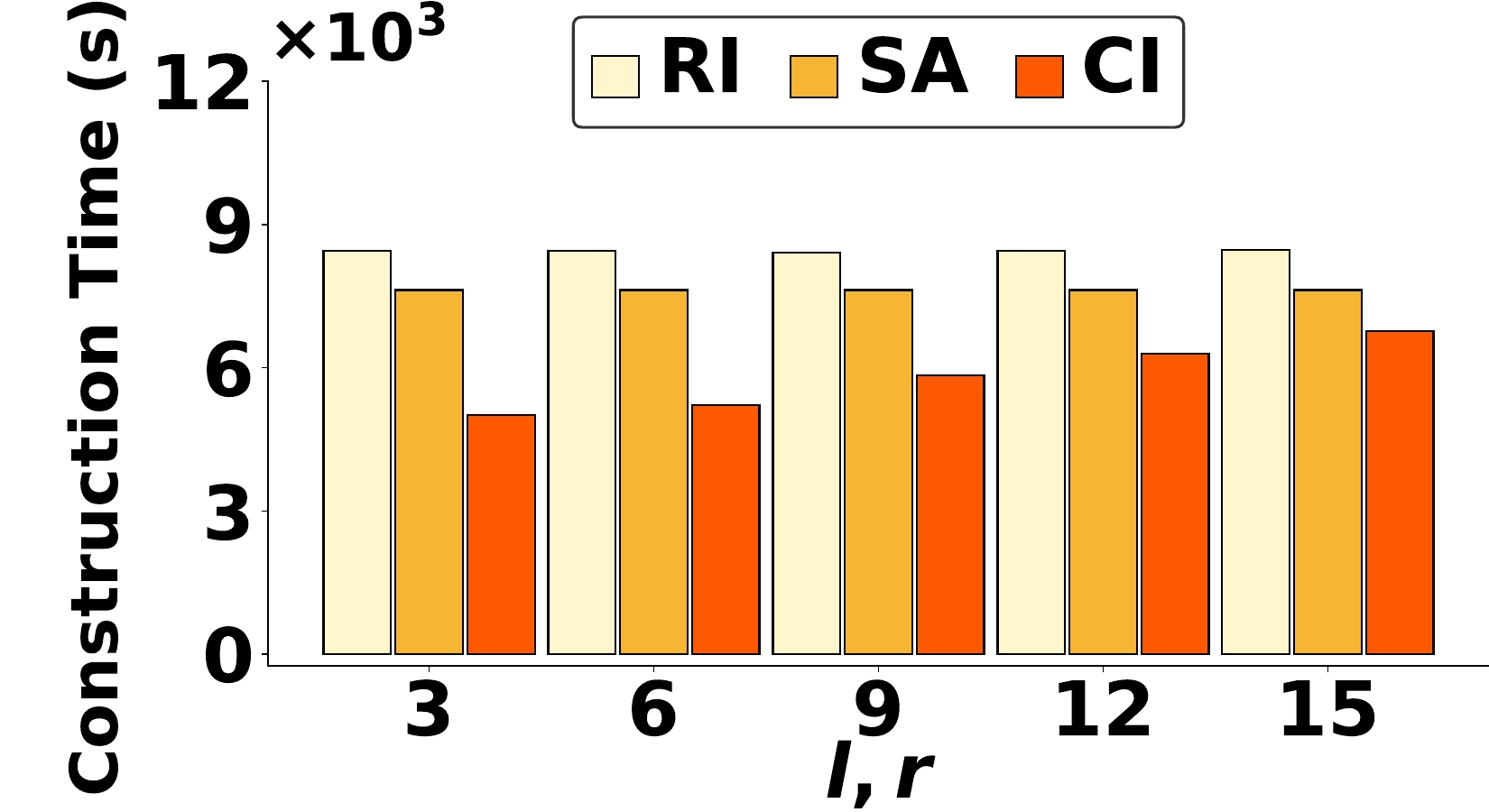}
        \caption{\sars}
    \end{subfigure}%
\hfill
    \begin{subfigure}{0.188\linewidth}
        \includegraphics[width=1.05\linewidth]{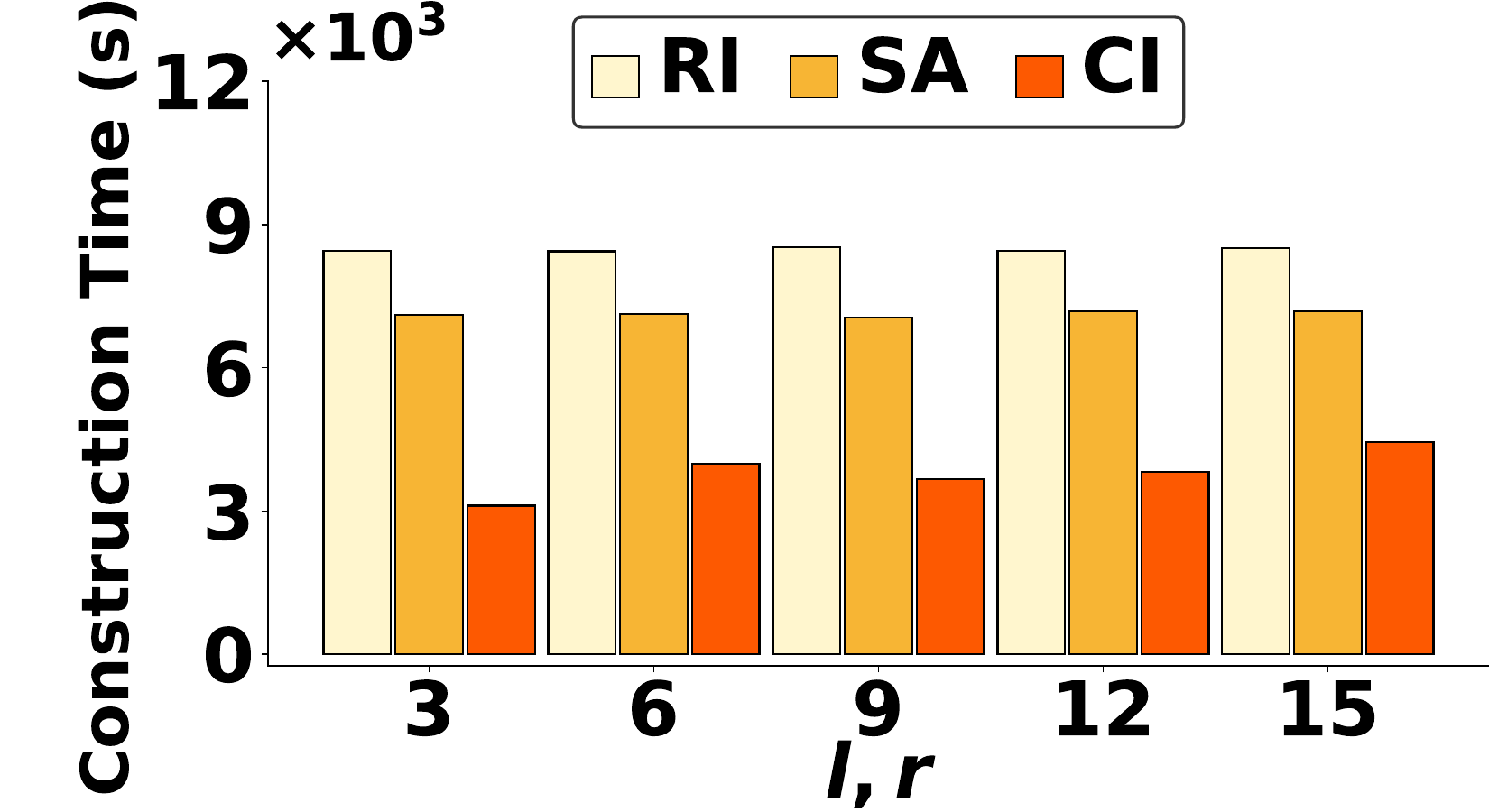}
        \caption{\chr}
    \label{fig:construction_time_CHR_xy}
    \end{subfigure}%

    \begin{subfigure}{0.188\linewidth}
        \includegraphics[width=1.05\linewidth]{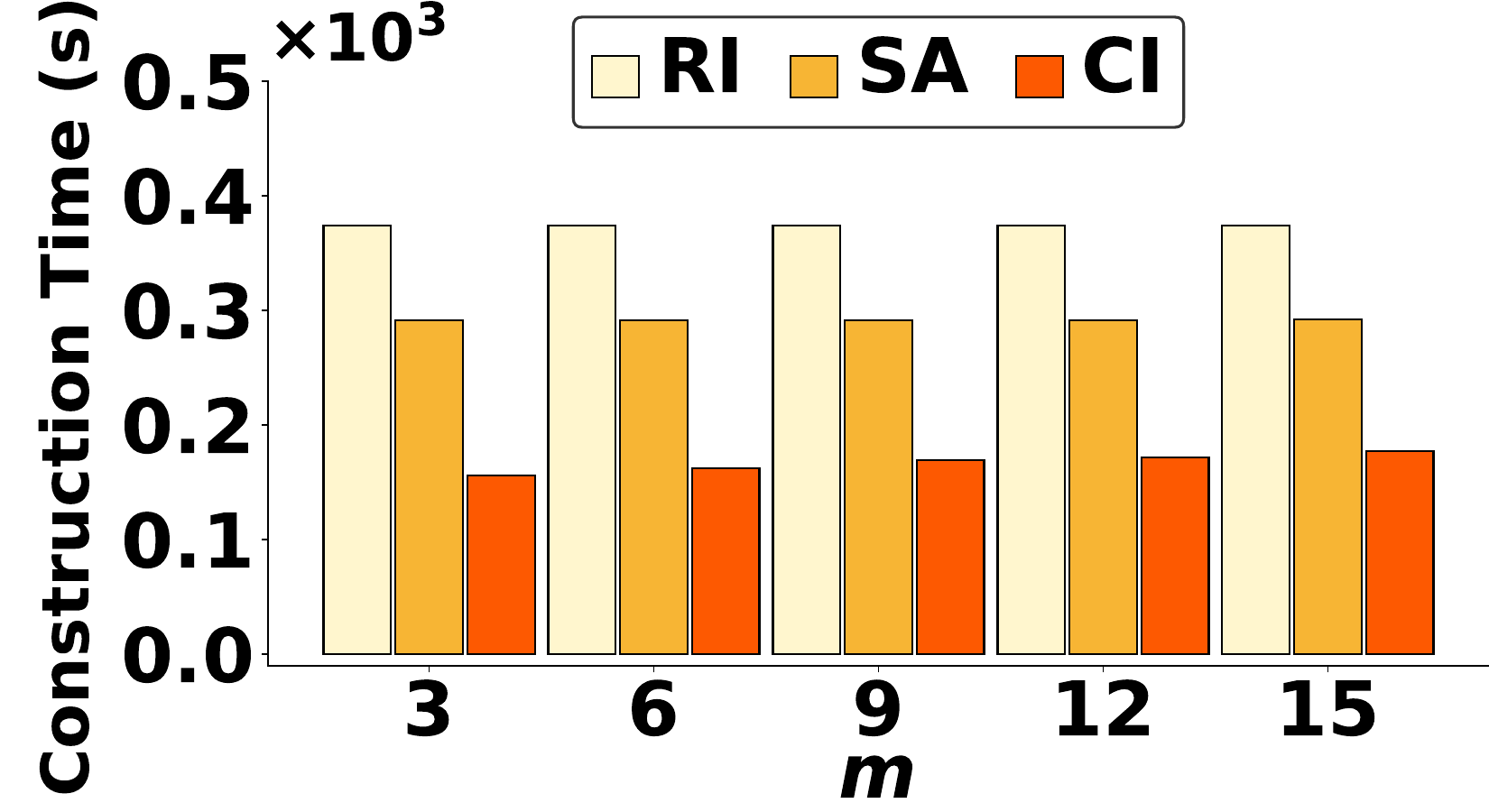}
        \caption{WIKI}
        \label{fig:construction_time_WIKI_m}
    \end{subfigure}%
\hfill
    \begin{subfigure}{0.188\linewidth}
        \includegraphics[width=1.05\linewidth]{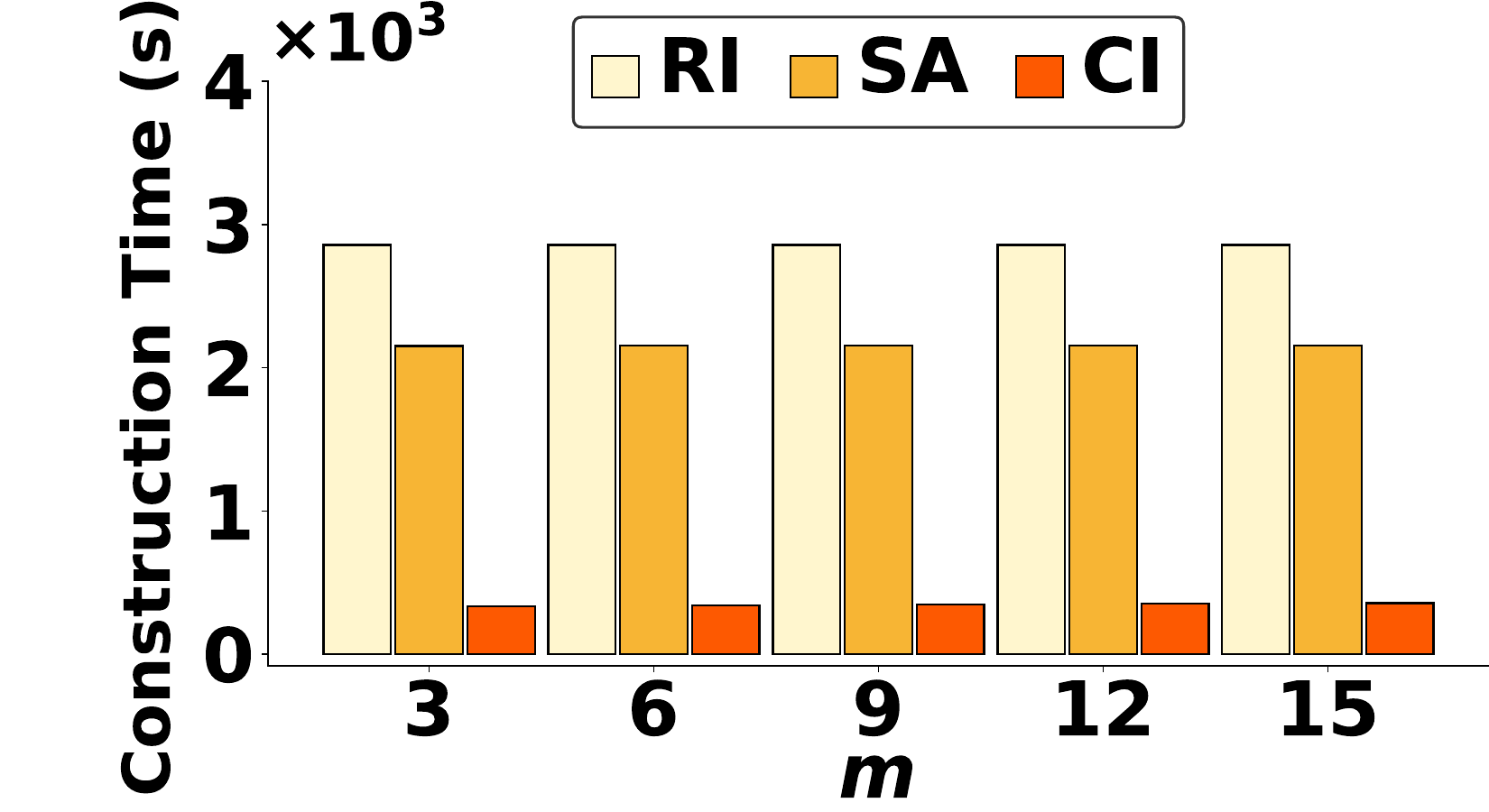}
        \caption{BST}
    \end{subfigure}%
\hfill
    \begin{subfigure}{0.188\linewidth}
        \includegraphics[width=1.05\linewidth]{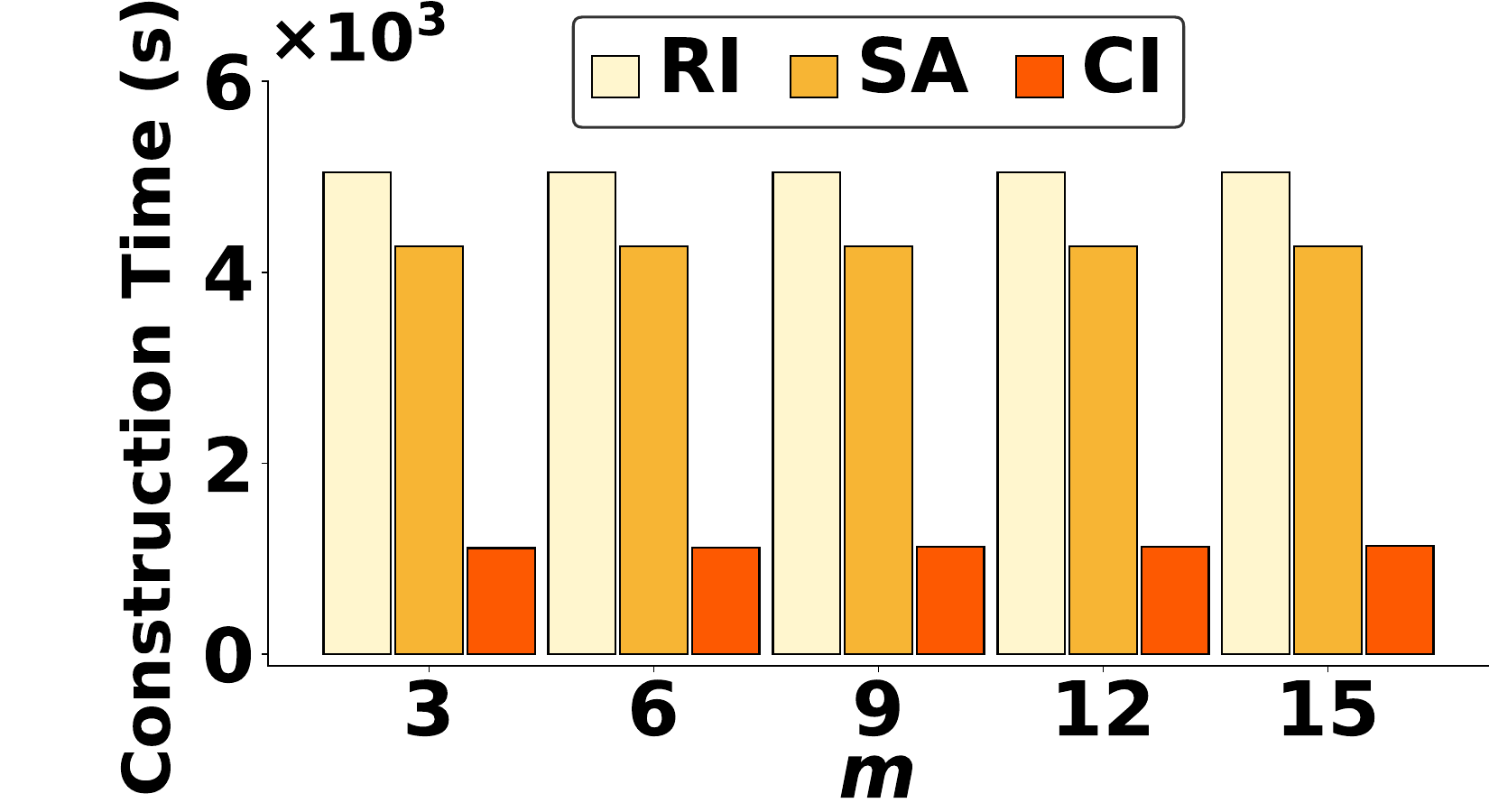}
        \caption{SDSL}
    \end{subfigure}%
        \hfill
    \begin{subfigure}{0.188\linewidth}
        \includegraphics[width=1.05\linewidth]{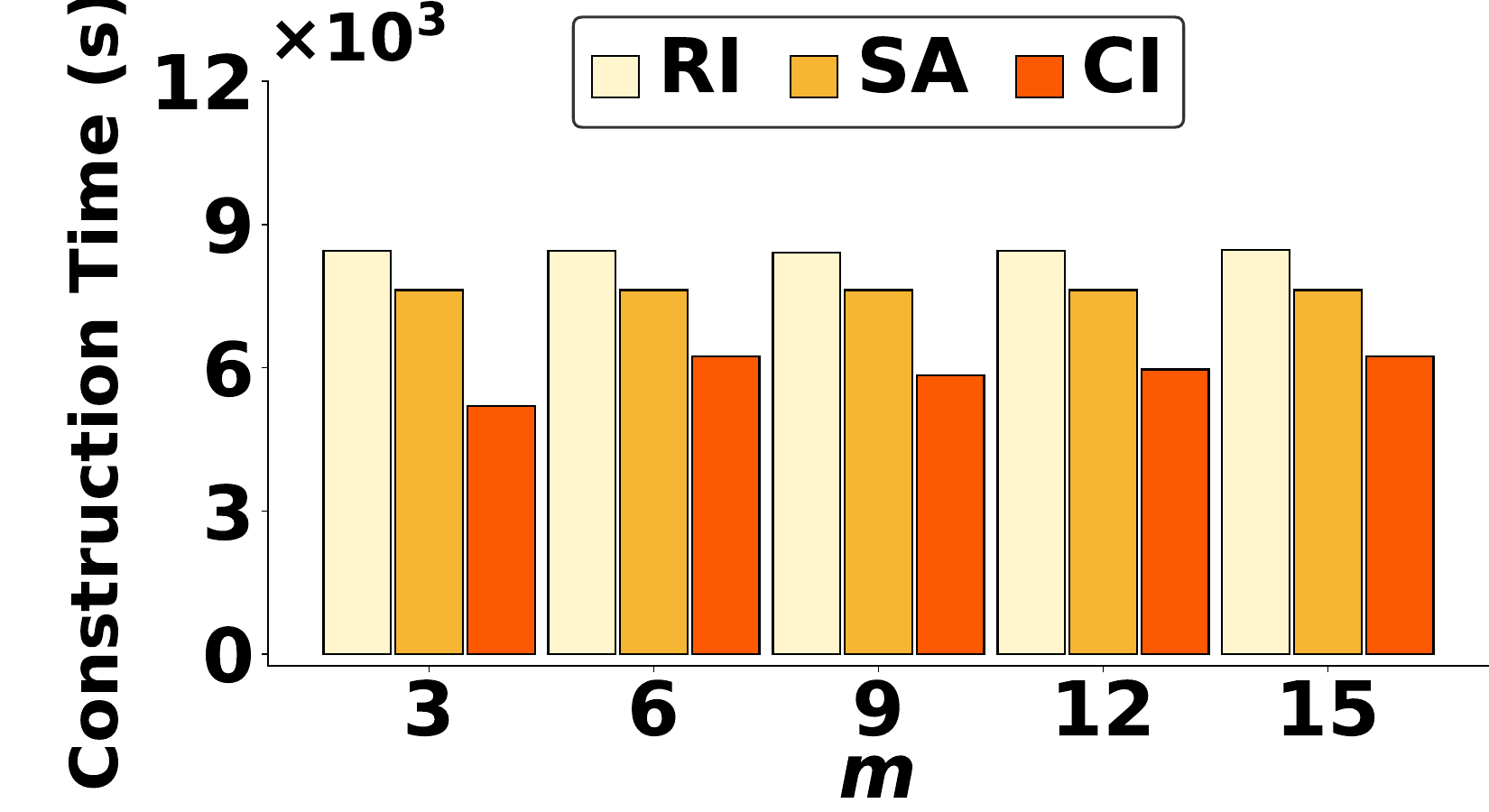}
        \caption{\sars}
    \end{subfigure}%
\hfill
        \begin{subfigure}{0.188\linewidth}
        \includegraphics[width=1.05\linewidth]{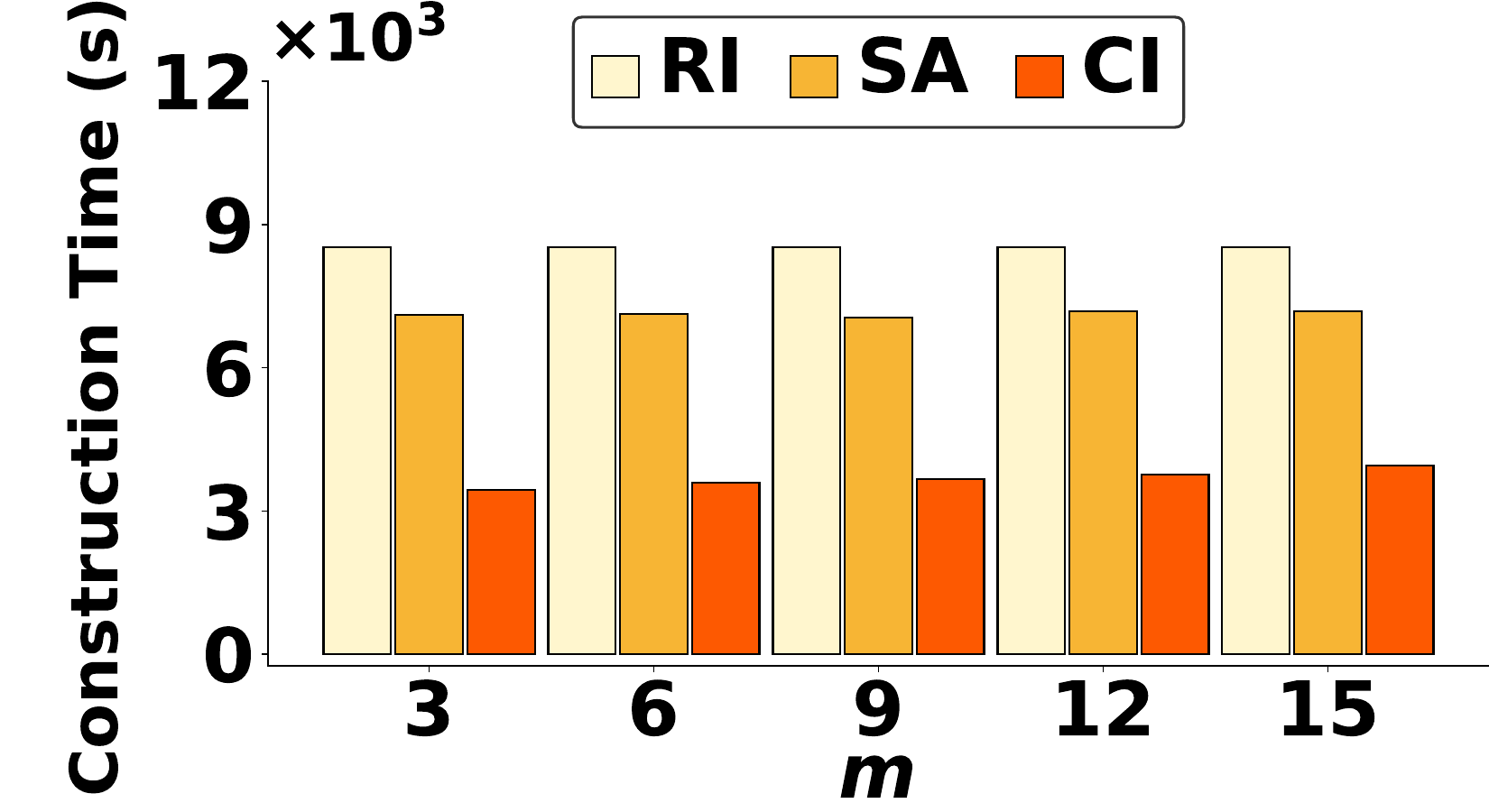}
        \caption{\chr}
        \label{fig:construction_time_CHR_m}
    \end{subfigure}%
    \caption{Construction time across all datasets: (a--e) vs.\ $n$, (f--j) vs.\ $l,r$, and (k--o) vs.\ $m$.}
    \label{fig:construction_time}
\end{figure*}

\subparagraph{Index Size.}~Figs. \ref{fig:index_space_WIKI_n} to  \ref{fig:index_space_CHR_m} show the index size for the experiments of Figs.~\ref{fig:querytime_wiki_n}  to \ref{fig:querytime_chr_m}.  
The index size of \CI increases as $n$, $l$, $r$, or $m$ increases, as the space complexity of \CI is $\Oh(k \cdot(\log k)^3/(\log\log k)^2)$, where $k=\min(B \cdot z,n)$ and $l+m+r\leq B$. The index sizes of \RI and \CPRIS increase linearly with $n$, as expected by their space complexities. The index size of \CI is only $9\%$ of that of \CPRIS on average. It is on average $5$ times larger than that of \RI, but \RI has prohibitively expensive query time, as shown above. 

\begin{figure*}[t]
    \centering
    \begin{subfigure}{0.188\linewidth}
        \includegraphics[width=1.05\linewidth]{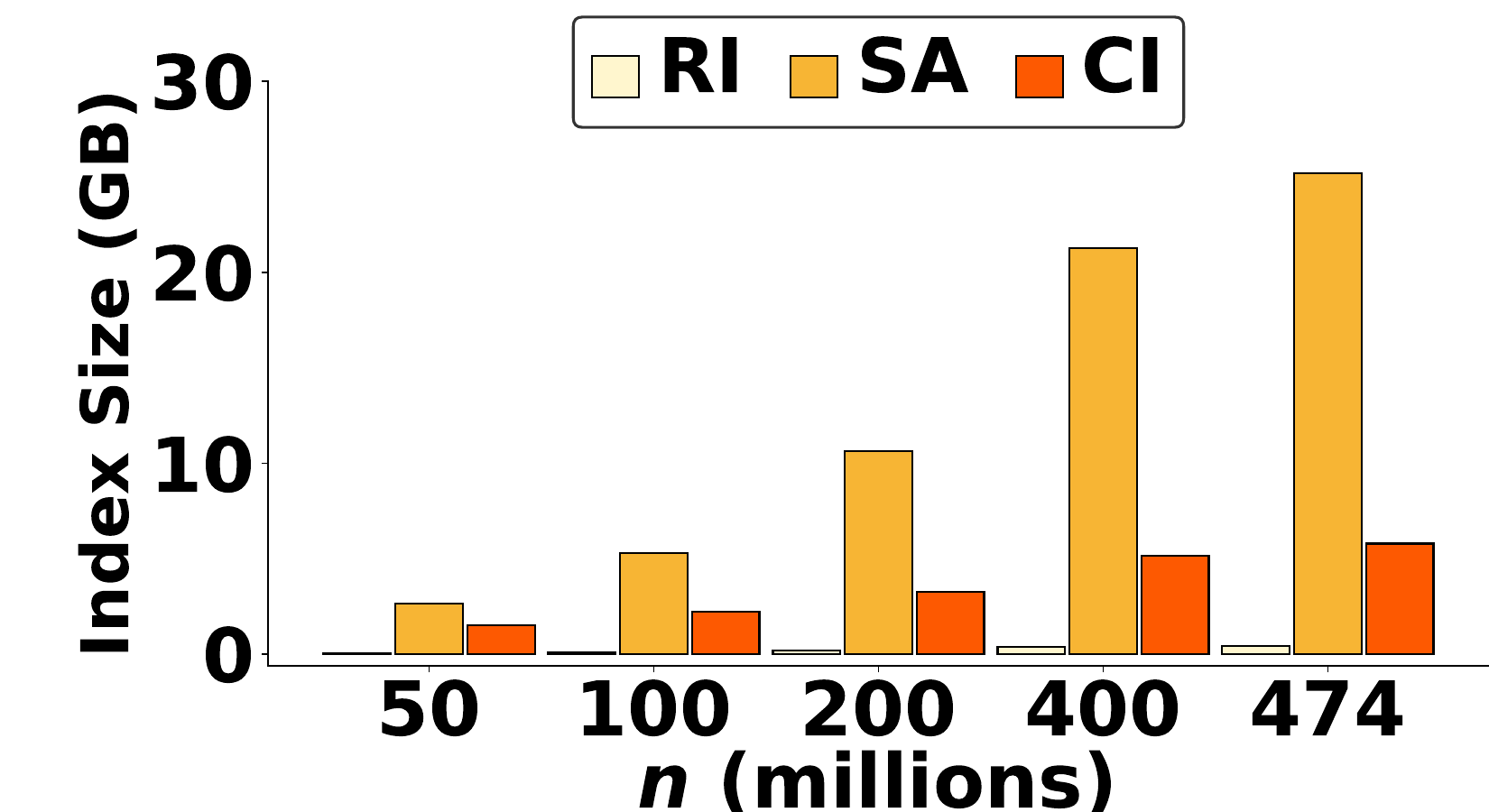}
        \caption{WIKI}
        \label{fig:index_space_WIKI_n}
    \end{subfigure}%
\hfill
    \begin{subfigure}{0.188\linewidth}
        \includegraphics[width=1.05\linewidth]{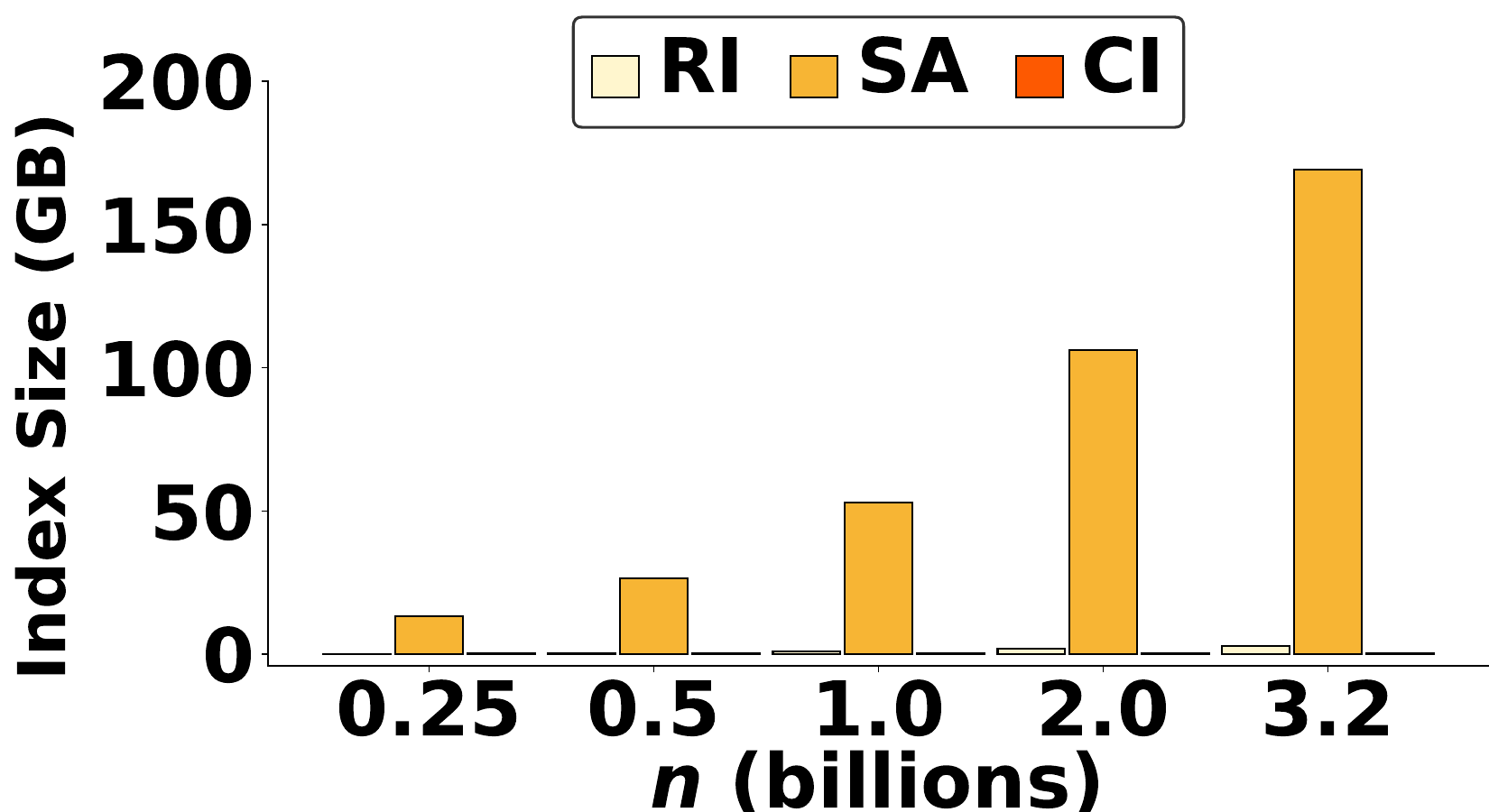}
        \caption{BST}
    \end{subfigure}%
\hfill
    \begin{subfigure}{0.188\linewidth}
        \includegraphics[width=1.05\linewidth]{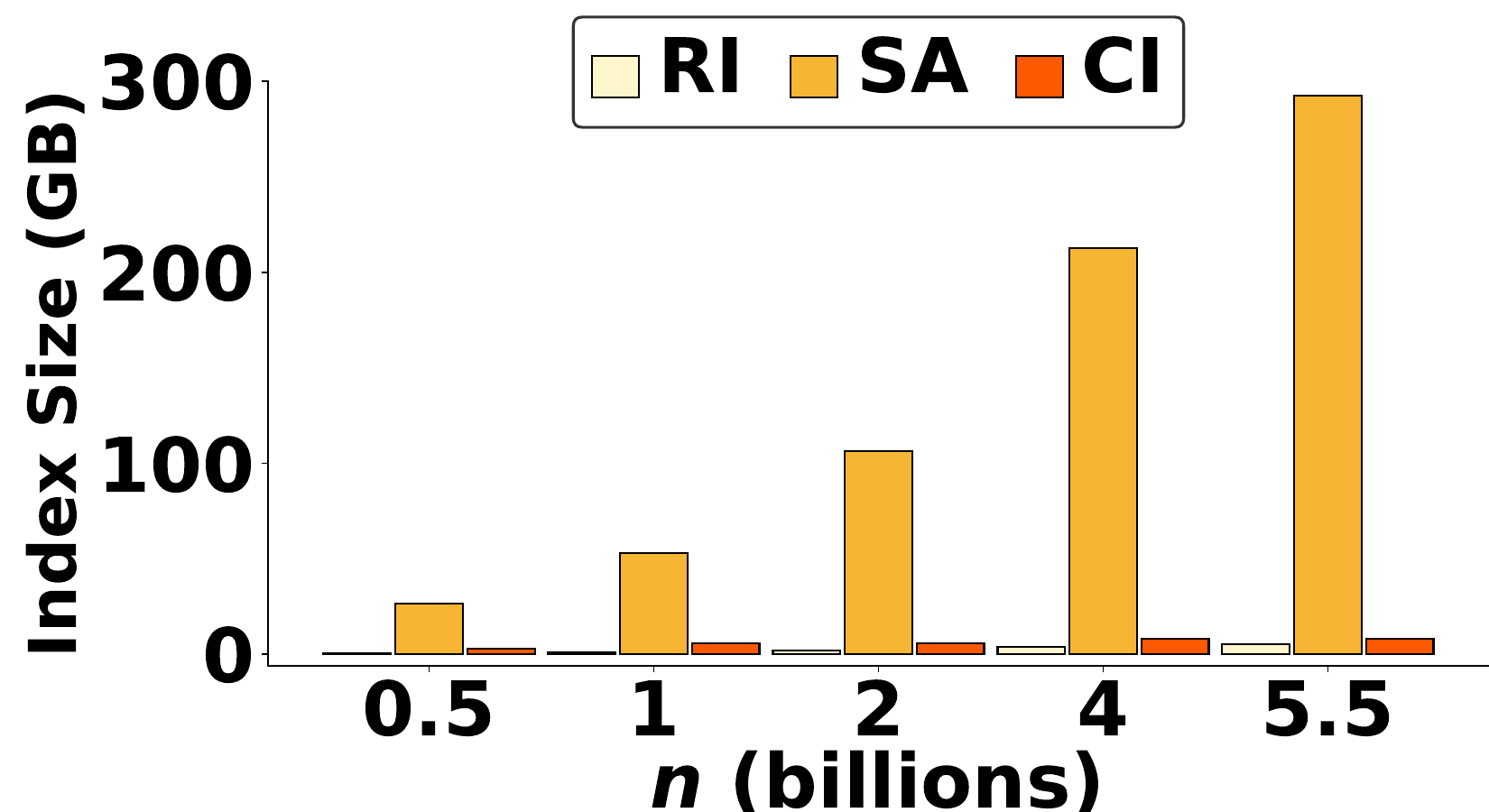}
        \caption{SDSL}
    \end{subfigure}%
  \hfill
    \begin{subfigure}{0.188\linewidth}
        \includegraphics[width=1.05\linewidth]{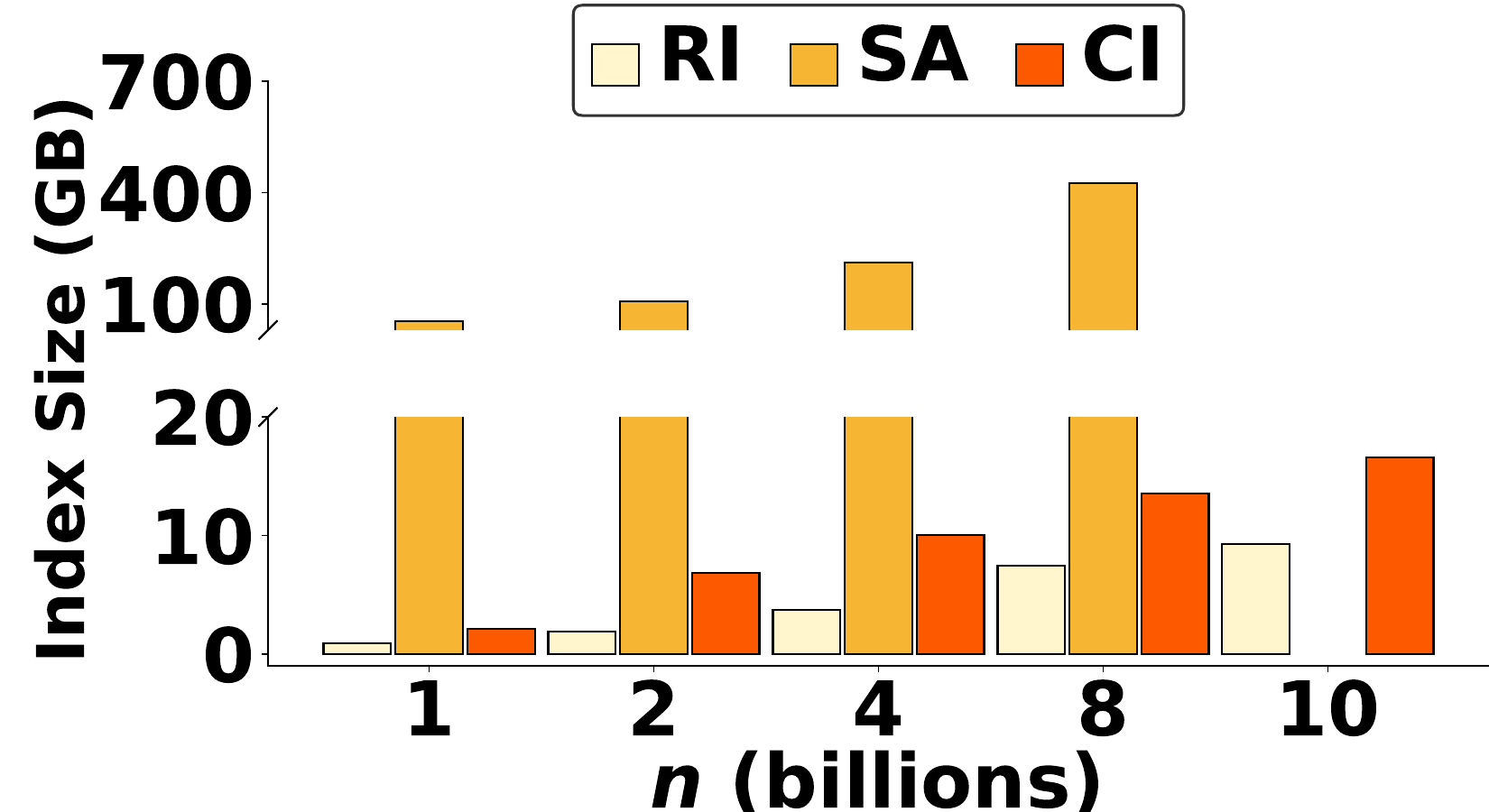}
        \caption{\sars}
    \end{subfigure}%
\hfill
    \begin{subfigure}{0.188\linewidth}
        \includegraphics[width=1.05\linewidth]{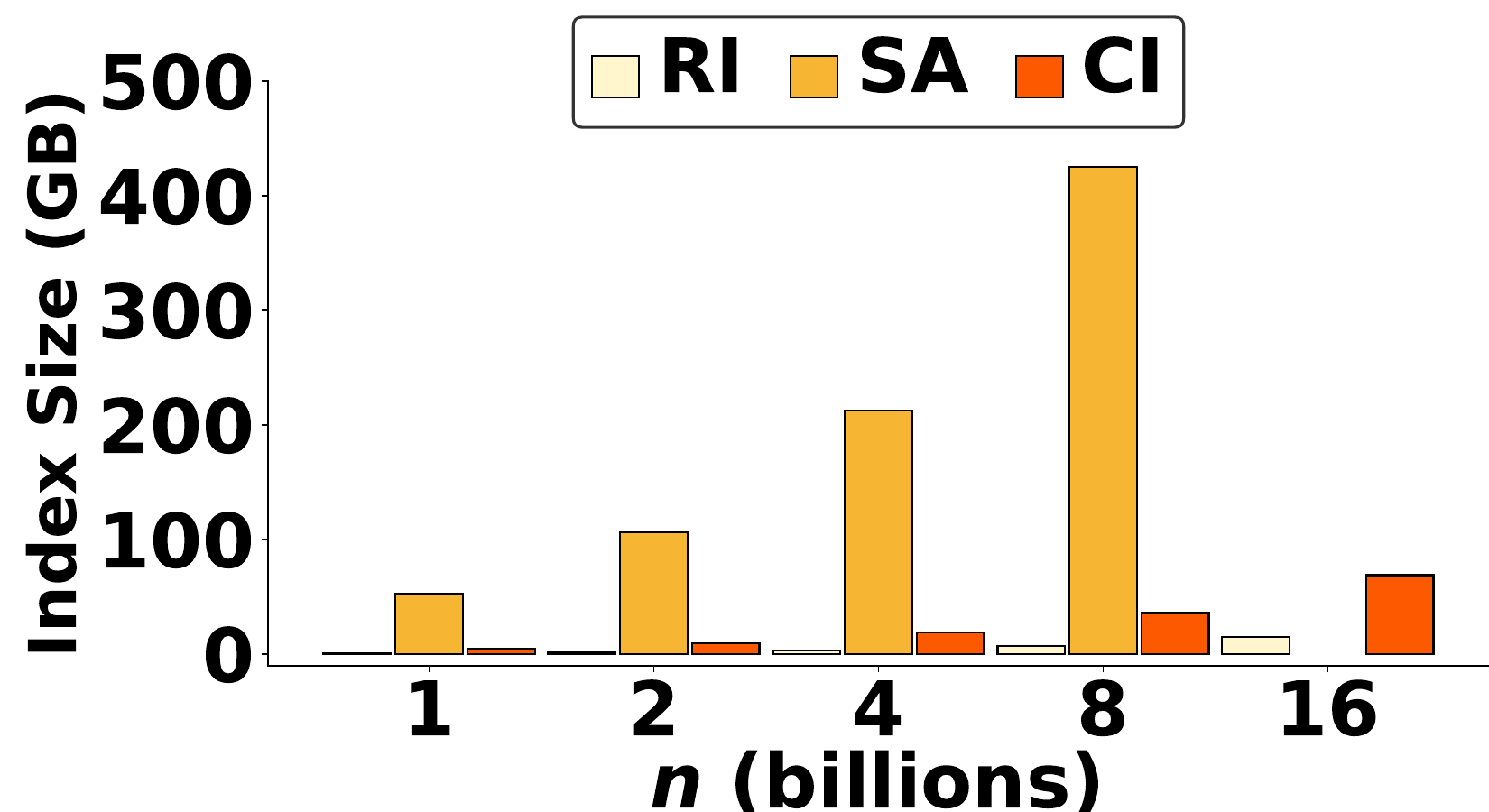}
        \caption{\chr}
        \label{fig:index_space_CHR_n}
    \end{subfigure}%

    \begin{subfigure}{0.188\linewidth}
        \includegraphics[width=1.05\linewidth]{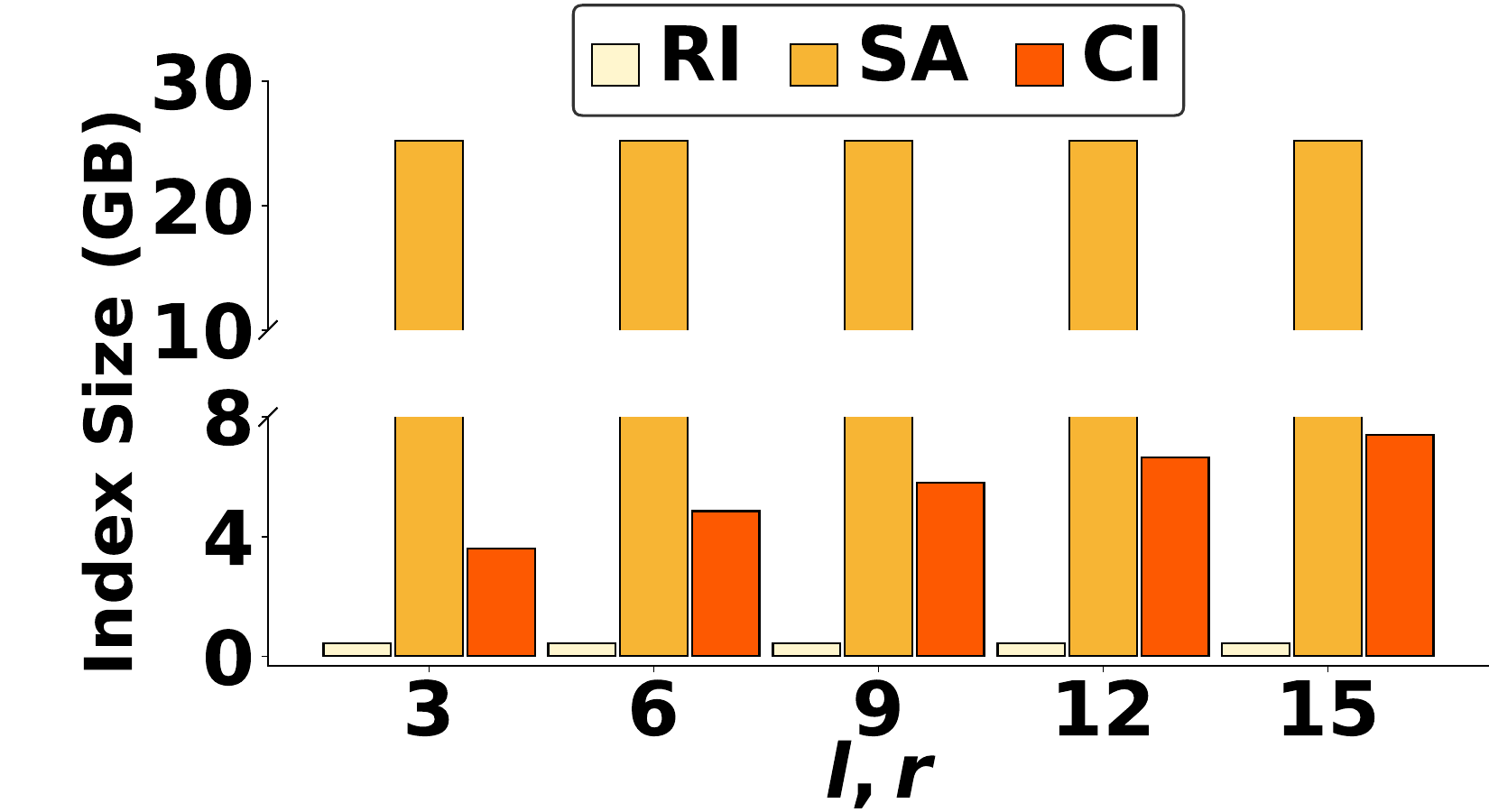}
        \caption{WIKI}
        \label{fig:index_space_WIKI_xy}
    \end{subfigure}%
 \hfill
    \begin{subfigure}{0.188\linewidth}
        \includegraphics[width=1.05\linewidth]{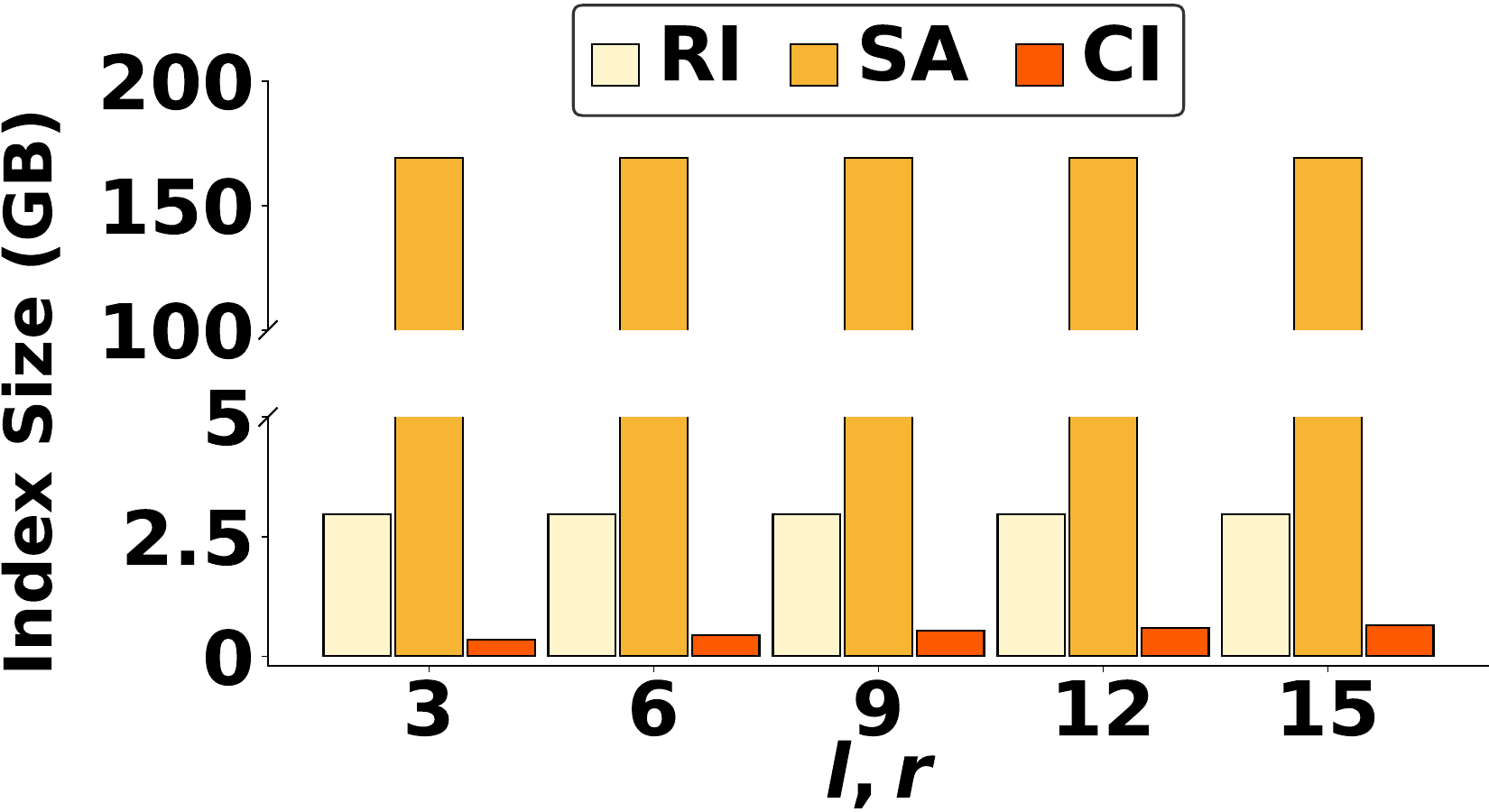}
        \caption{BST}
    \end{subfigure}%
\hfill
    \begin{subfigure}{0.188\linewidth}
        \includegraphics[width=1.05\linewidth]{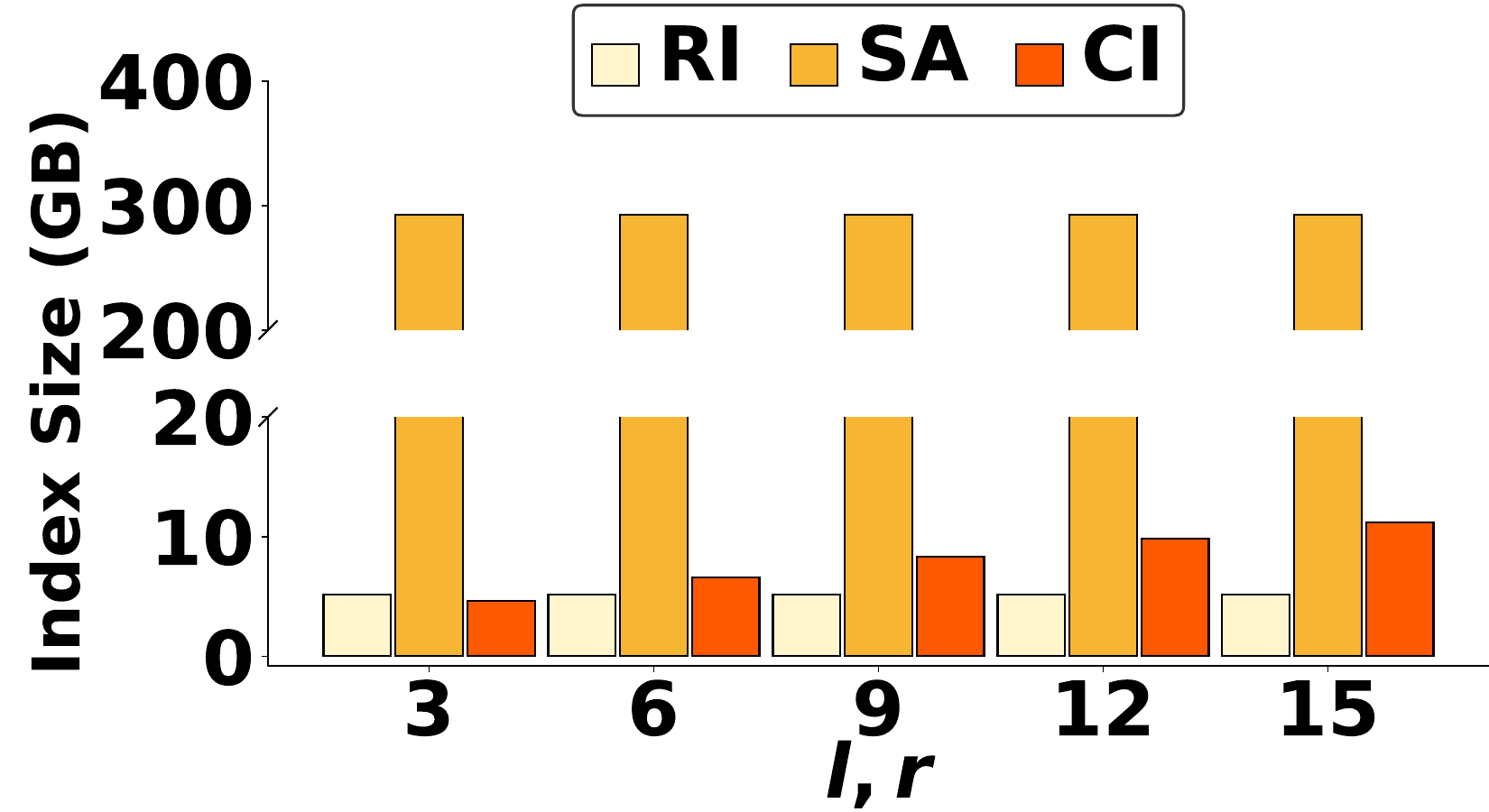}
        \caption{SDSL}
    \end{subfigure}%
\hfill
    \begin{subfigure}{0.188\linewidth}
        \includegraphics[width=1.05\linewidth]{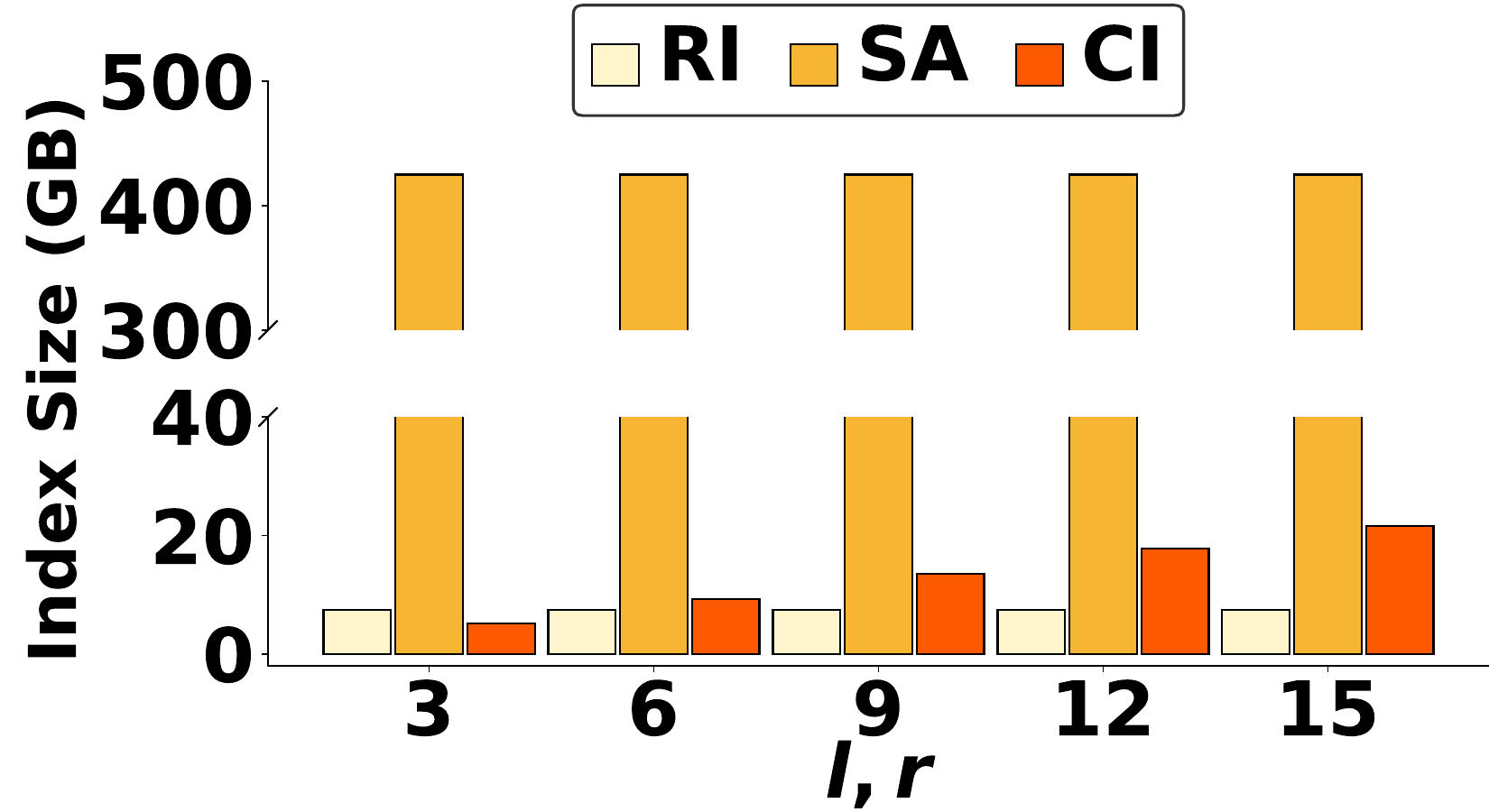}
        \caption{\sars}
    \end{subfigure}%
\hfill
    \begin{subfigure}{0.188\linewidth}
        \includegraphics[width=1.05\linewidth]{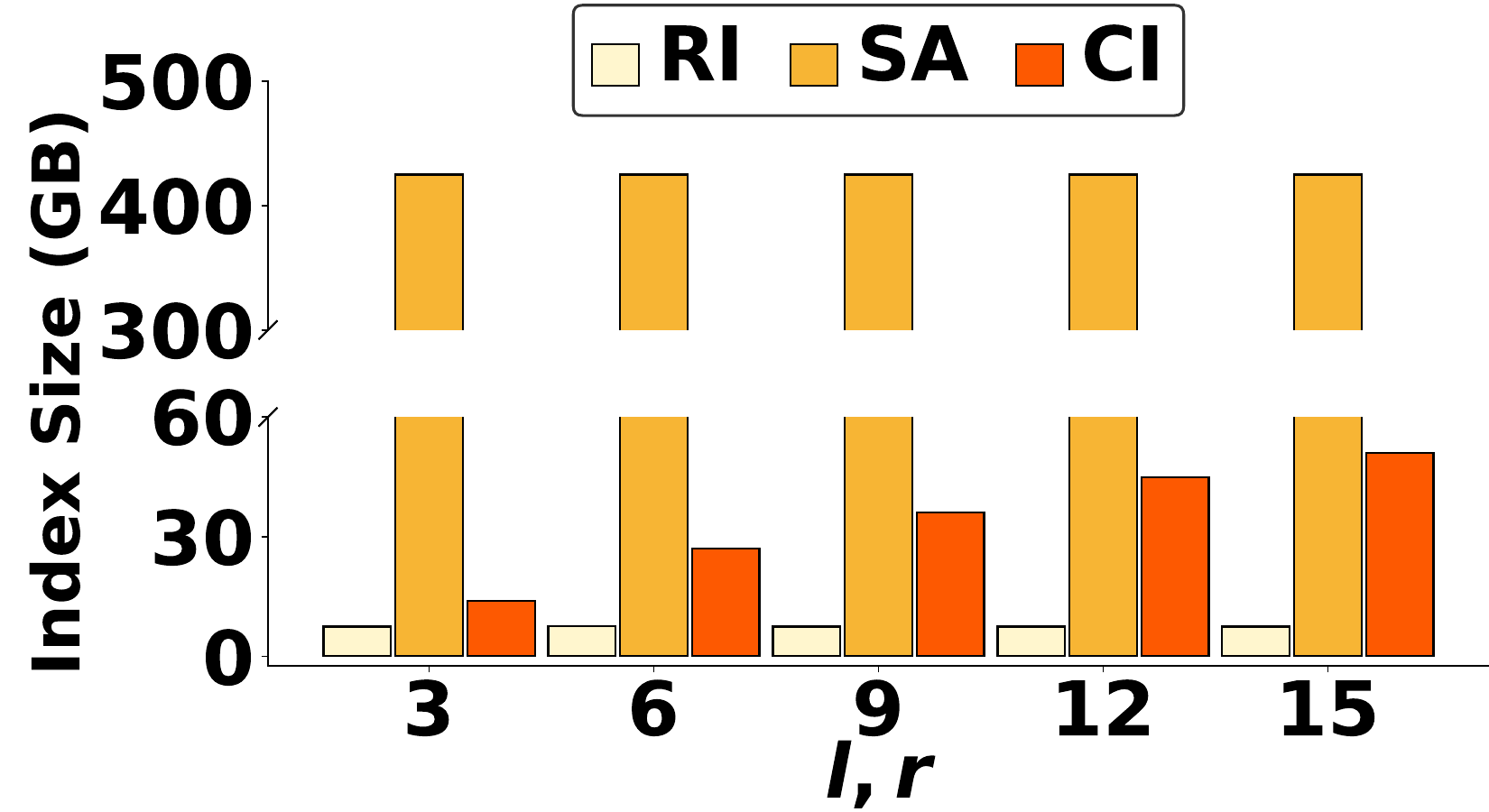}
        \caption{\chr}
        \label{fig:index_space_CHR_xy}
    \end{subfigure}%

    \begin{subfigure}{0.188\linewidth}
        \includegraphics[width=1.05\linewidth]{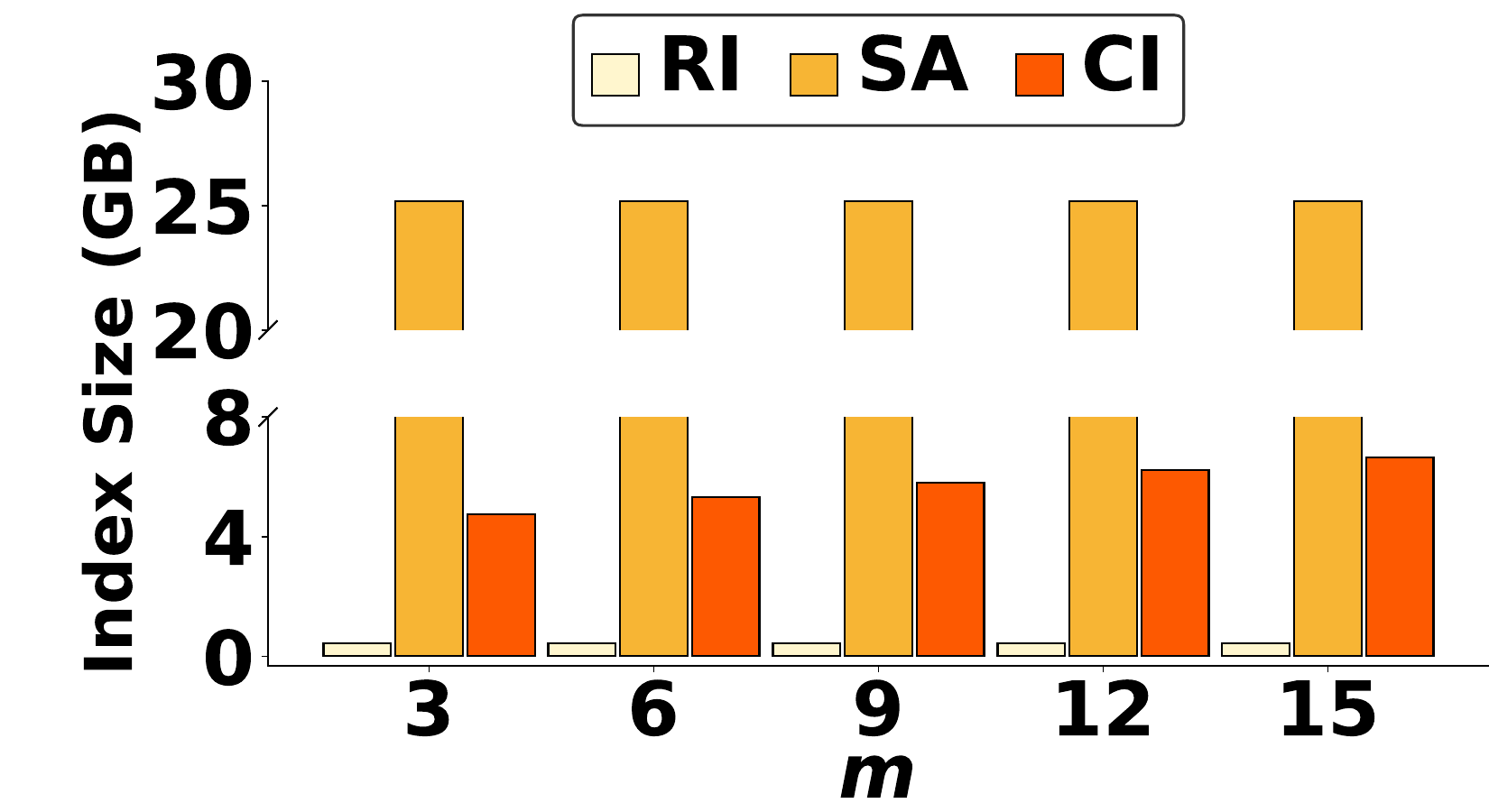}
        \caption{WIKI}
        \label{fig:index_space_WIKI_m}
    \end{subfigure}%
\hfill
    \begin{subfigure}{0.188\linewidth}
        \includegraphics[width=1.05\linewidth]{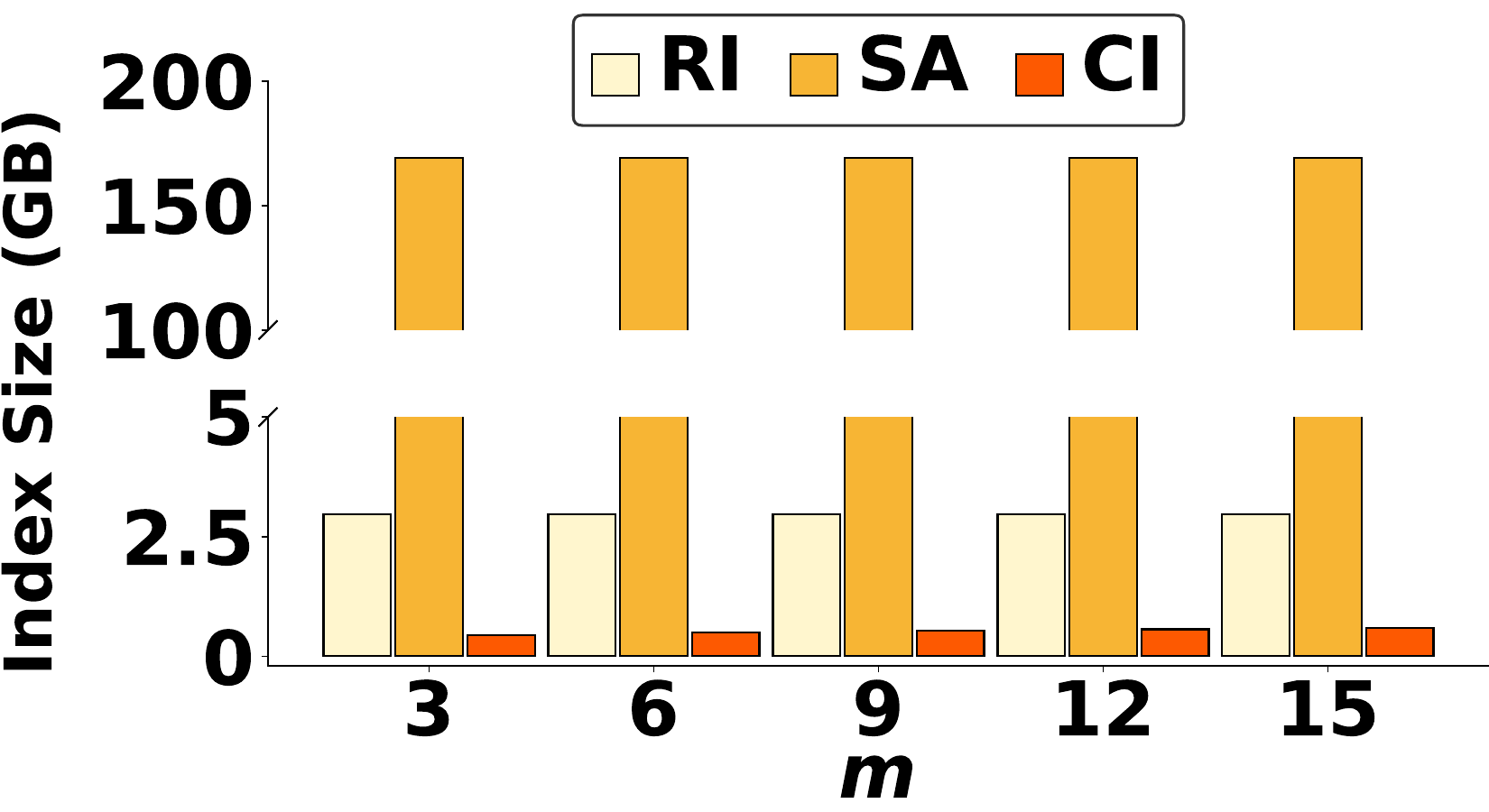}
        \caption{BST}
    \end{subfigure}%
    \hfill
    \begin{subfigure}{0.188\linewidth}
        \includegraphics[width=1.05\linewidth]{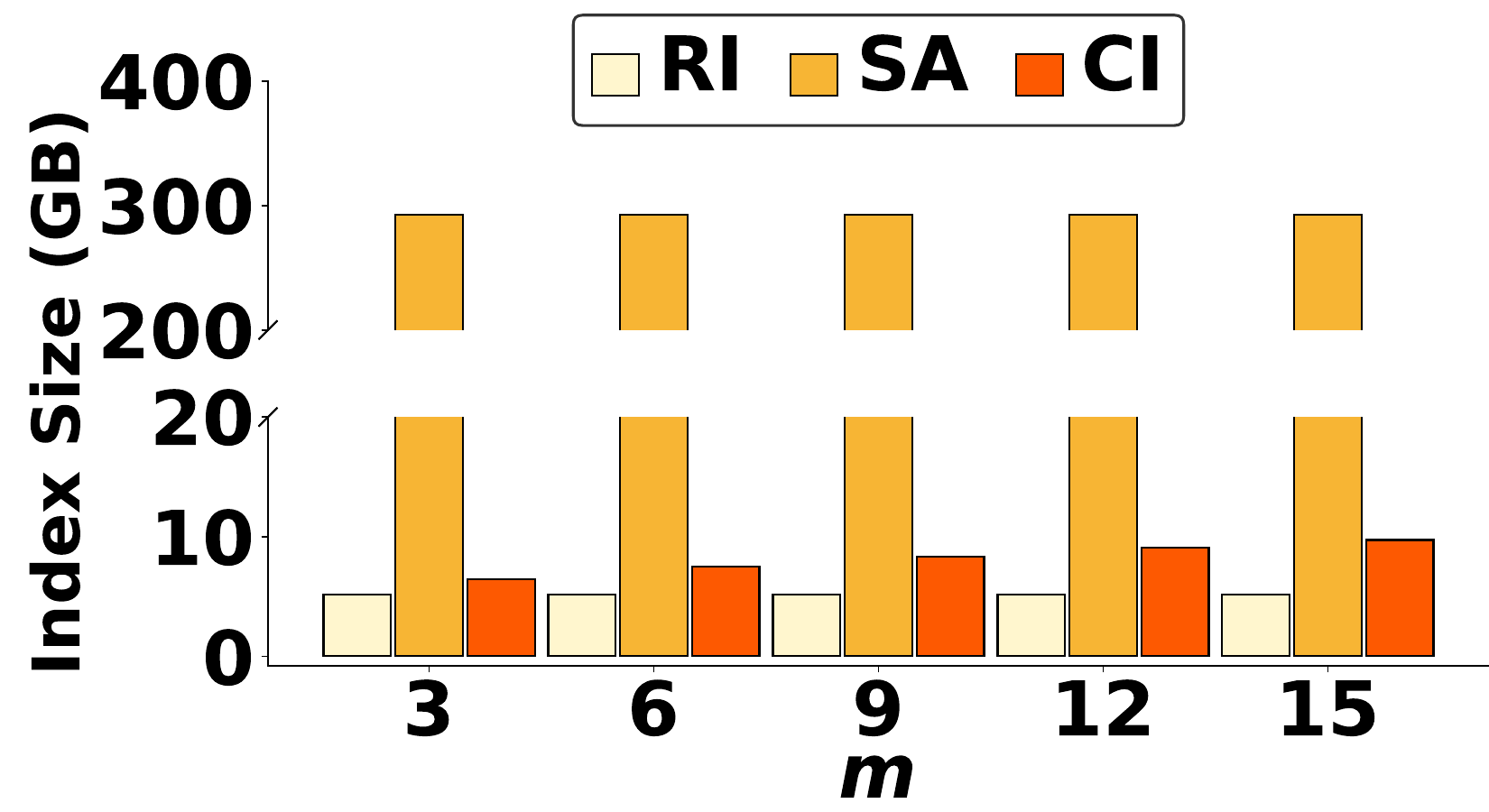}
        \caption{SDSL}
    \end{subfigure}%
\hfill
        \begin{subfigure}{0.188\linewidth}
        \includegraphics[width=1.05\linewidth]{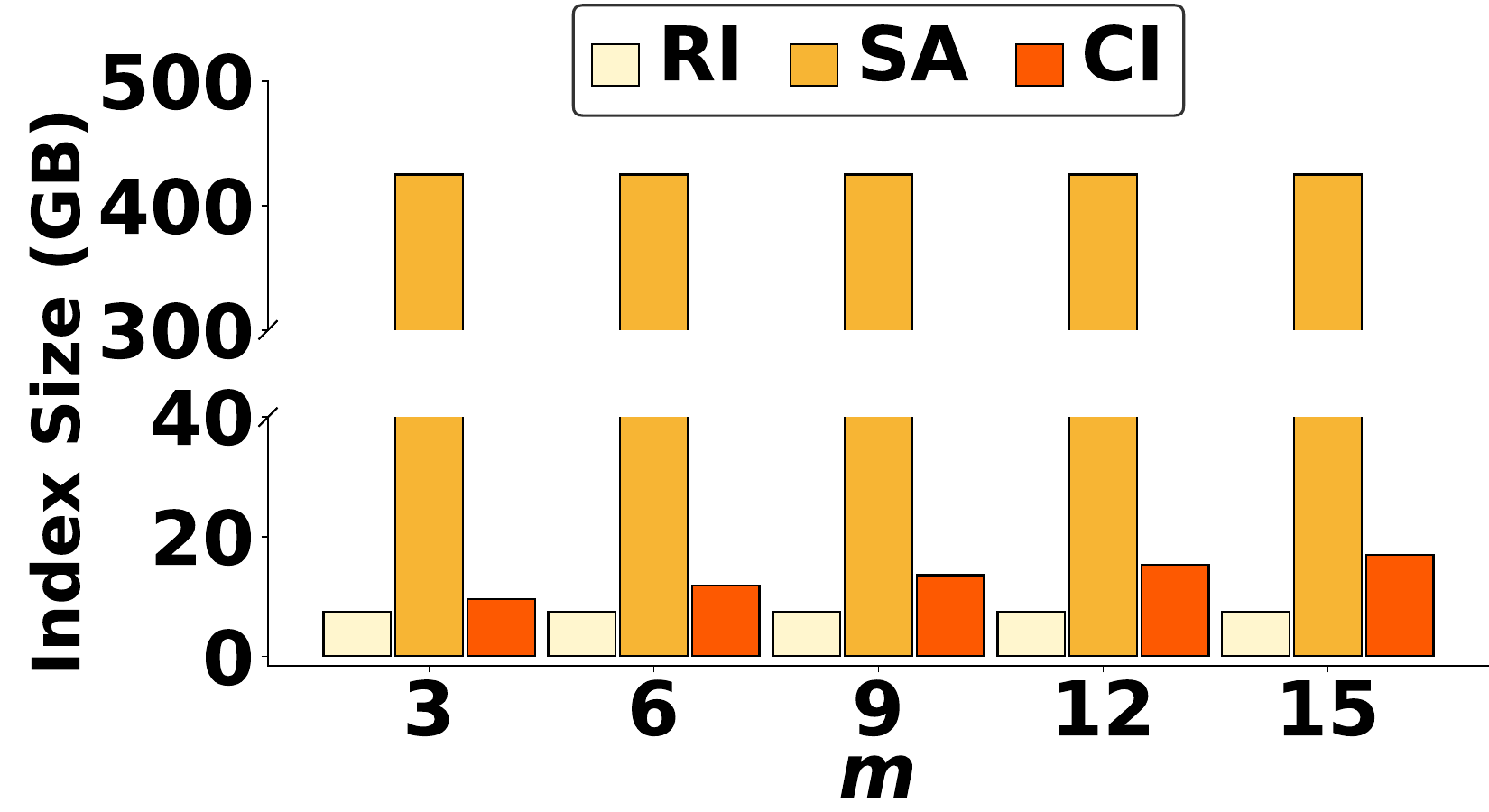}
        \caption{\sars}
    \end{subfigure}%
\hfill
    \begin{subfigure}{0.188\linewidth}
        \includegraphics[width=1.05\linewidth]{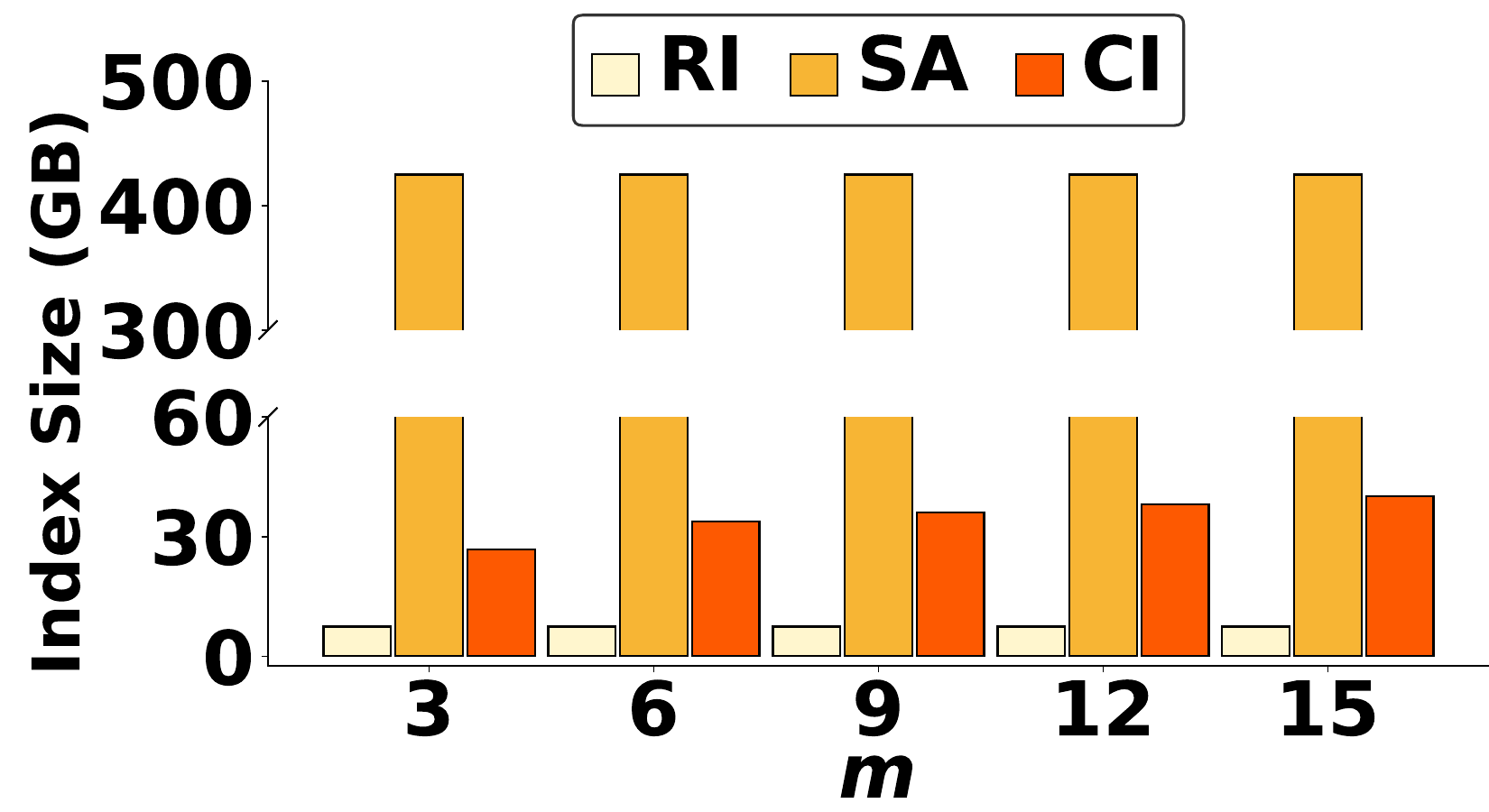}
        \caption{\chr}
        \label{fig:index_space_CHR_m}
    \end{subfigure}%

    \caption{Index size across all datasets: (a--e) vs.\ $n$, (f--j) vs.\ $l,r$, and (k--o) vs.\ $m$.}
    \label{fig:Index_space}
\end{figure*}

\subparagraph{Construction Space.}~Figs.~\ref{fig:construction_space_WIKI_n} to \ref{fig:construction_space_CHR_m} show the construction space for the experiments of Figs.~\ref{fig:querytime_wiki_n} to \ref{fig:querytime_chr_m}, respectively. The results for SA and \RI are analogous to those for index size, as  construction space is heavily determined by index size. In contrast, the construction space of \CI increases linearly with $n$ and is not affected substantially by $l$, $r$, or $m$. This is because  
the peak memory usage occurred during the LZ77 factorization phase, which takes $\cO{(n)}$ space~\cite{DBLP:journals/jacm/StorerS82}. 
The construction space of \CI is only $27\%$ of that of \CPRIS on average. It is on average $4.4$ times larger than that of \RI, which nevertheless has prohibitively expensive query time, as shown above. 

\begin{figure*}[t]
    \centering
    \begin{subfigure}{0.188\linewidth}
        \includegraphics[width=1.05\linewidth]{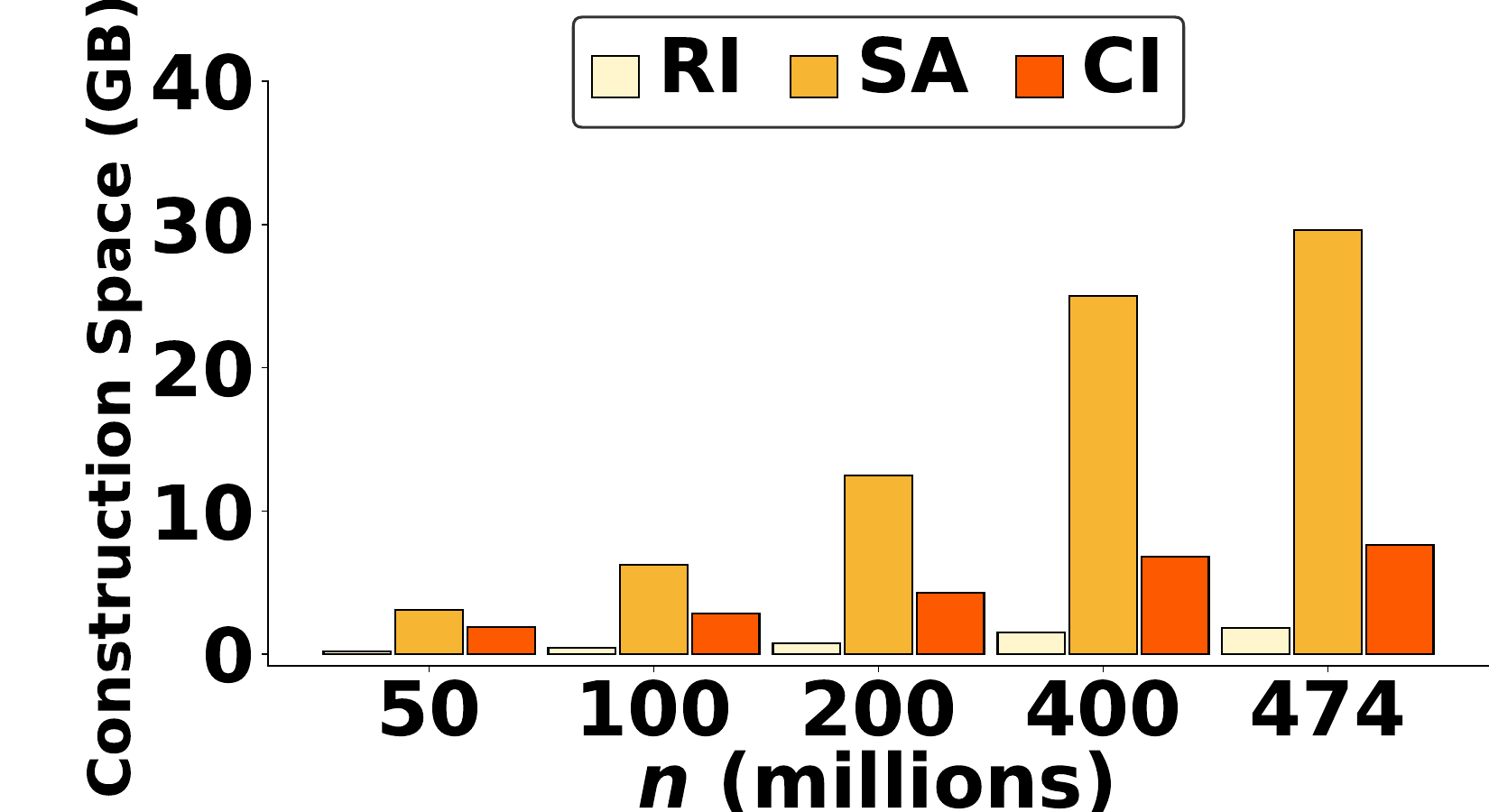}
        \caption{WIKI}
        \label{fig:construction_space_WIKI_n}
    \end{subfigure}%
 \hfill
    \begin{subfigure}{0.188\linewidth}
        \includegraphics[width=1.05\linewidth]{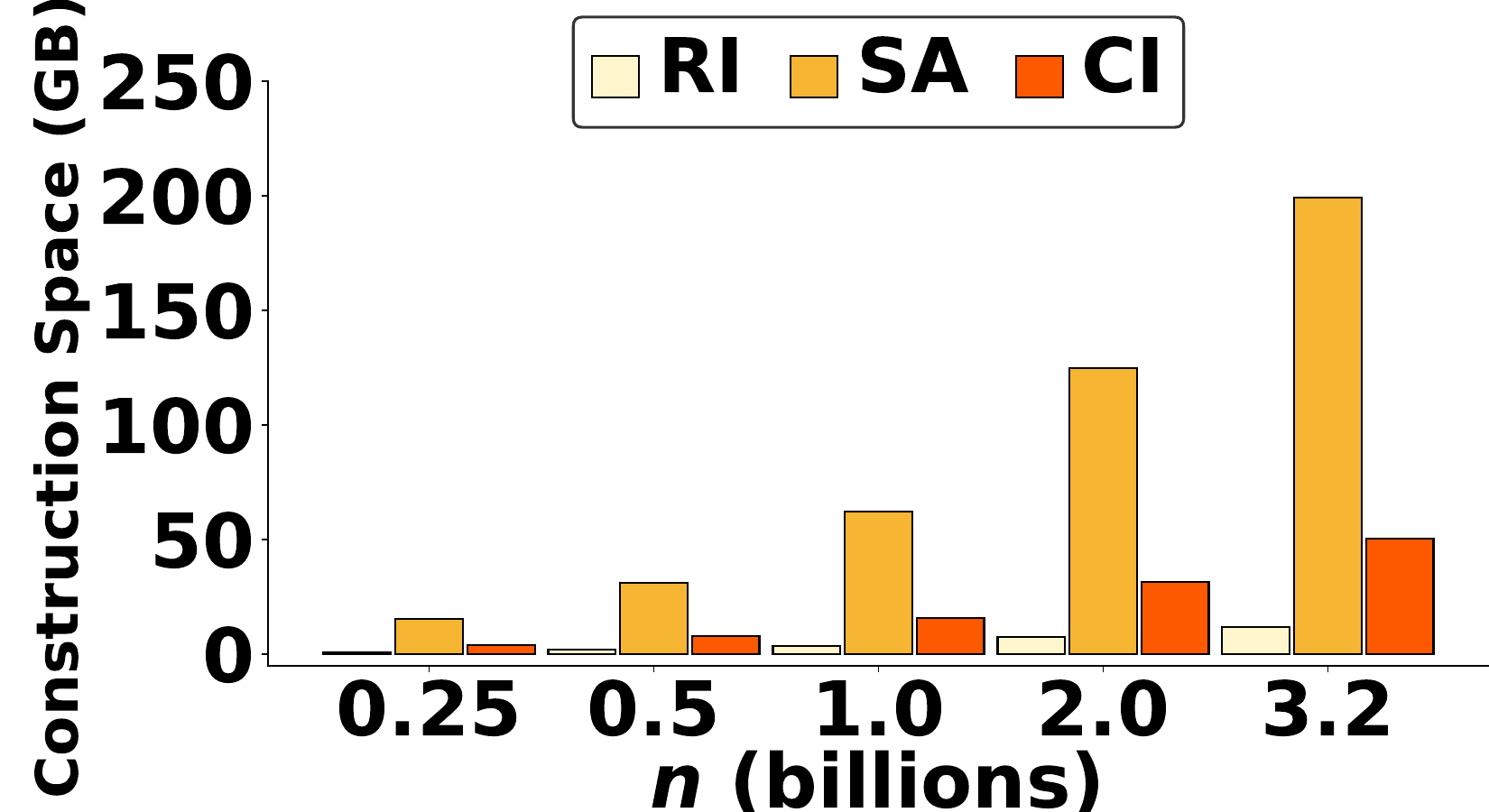}
        \caption{BST}
    \end{subfigure}%
\hfill
    \begin{subfigure}{0.188\linewidth}
        \includegraphics[width=1.05\linewidth]{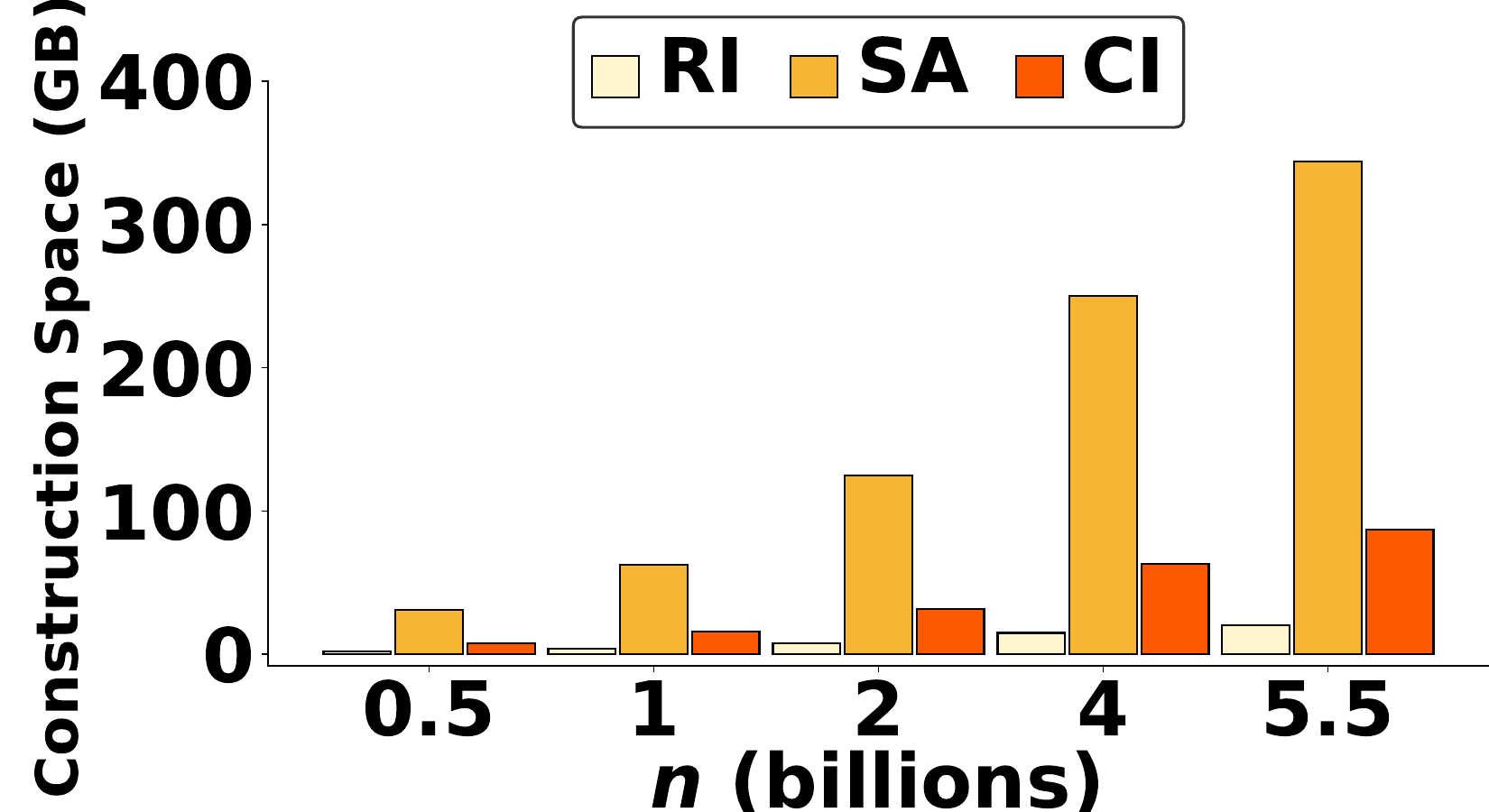}
        \caption{SDSL}
    \end{subfigure}%
\hfill
        \begin{subfigure}{0.188\linewidth}
        \includegraphics[width=1.05\linewidth]{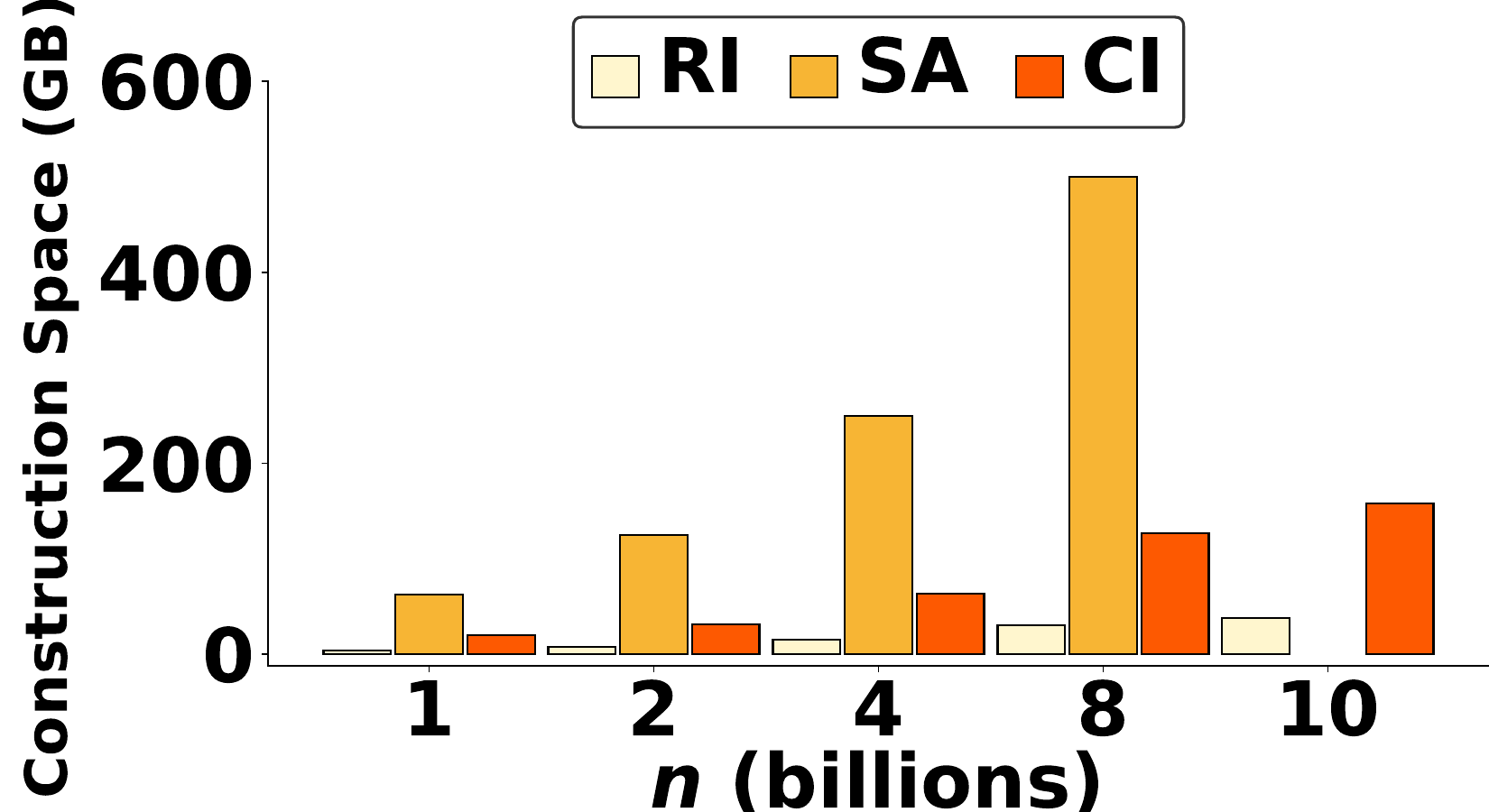}
        \caption{\sars}
    \end{subfigure}%
\hfill
    \begin{subfigure}{0.188\linewidth}
        \includegraphics[width=1.05\linewidth]{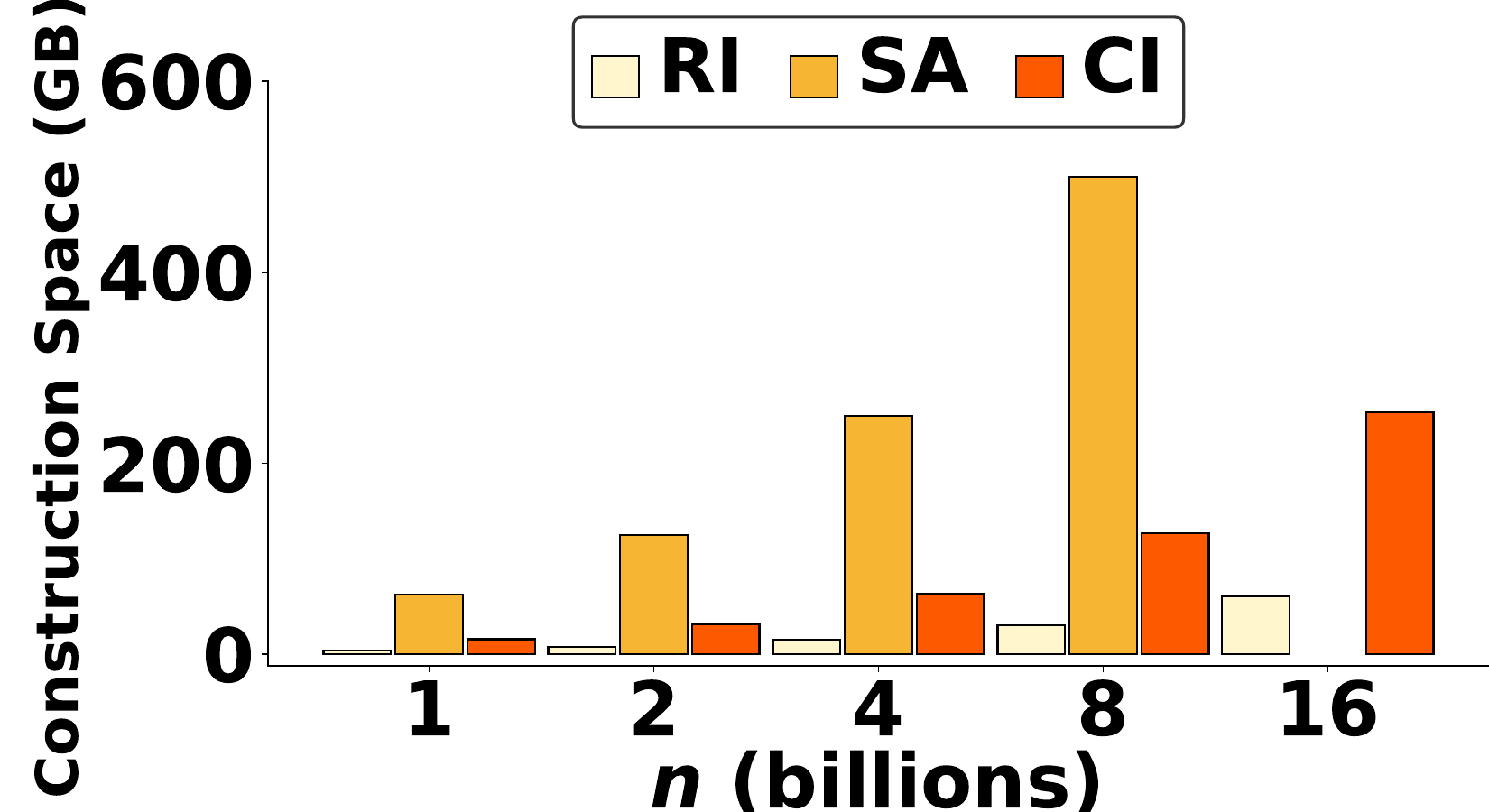}
        \caption{\chr}
        \label{fig:construction_space_CHR_n}
    \end{subfigure}%

    \begin{subfigure}{0.188\linewidth}
        \includegraphics[width=1.05\linewidth]{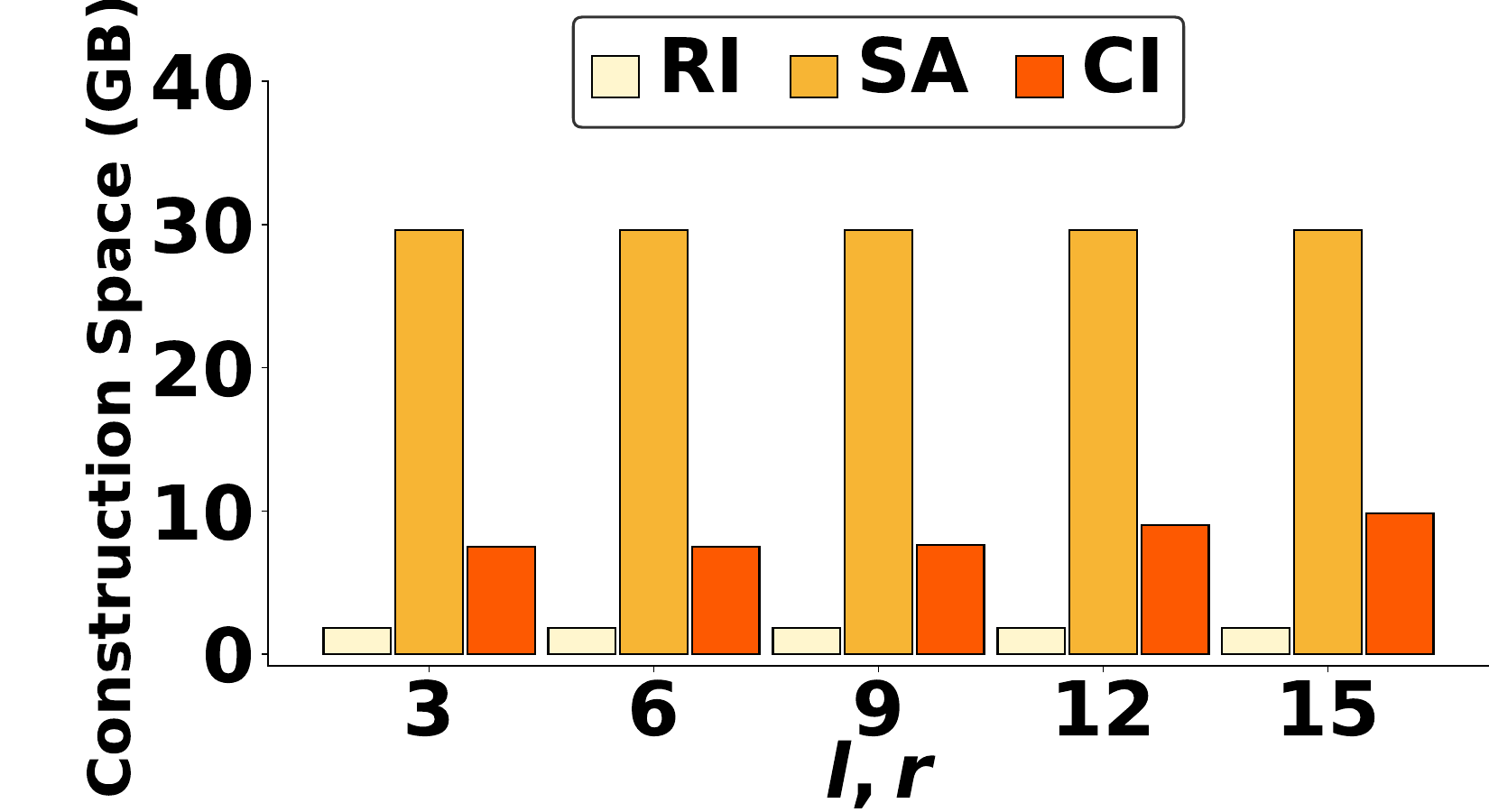}
        \caption{WIKI}
        \label{fig:construction_space_WIKI_xy}
    \end{subfigure}%
\hfill
    \begin{subfigure}{0.188\linewidth}
        \includegraphics[width=1.05\linewidth]{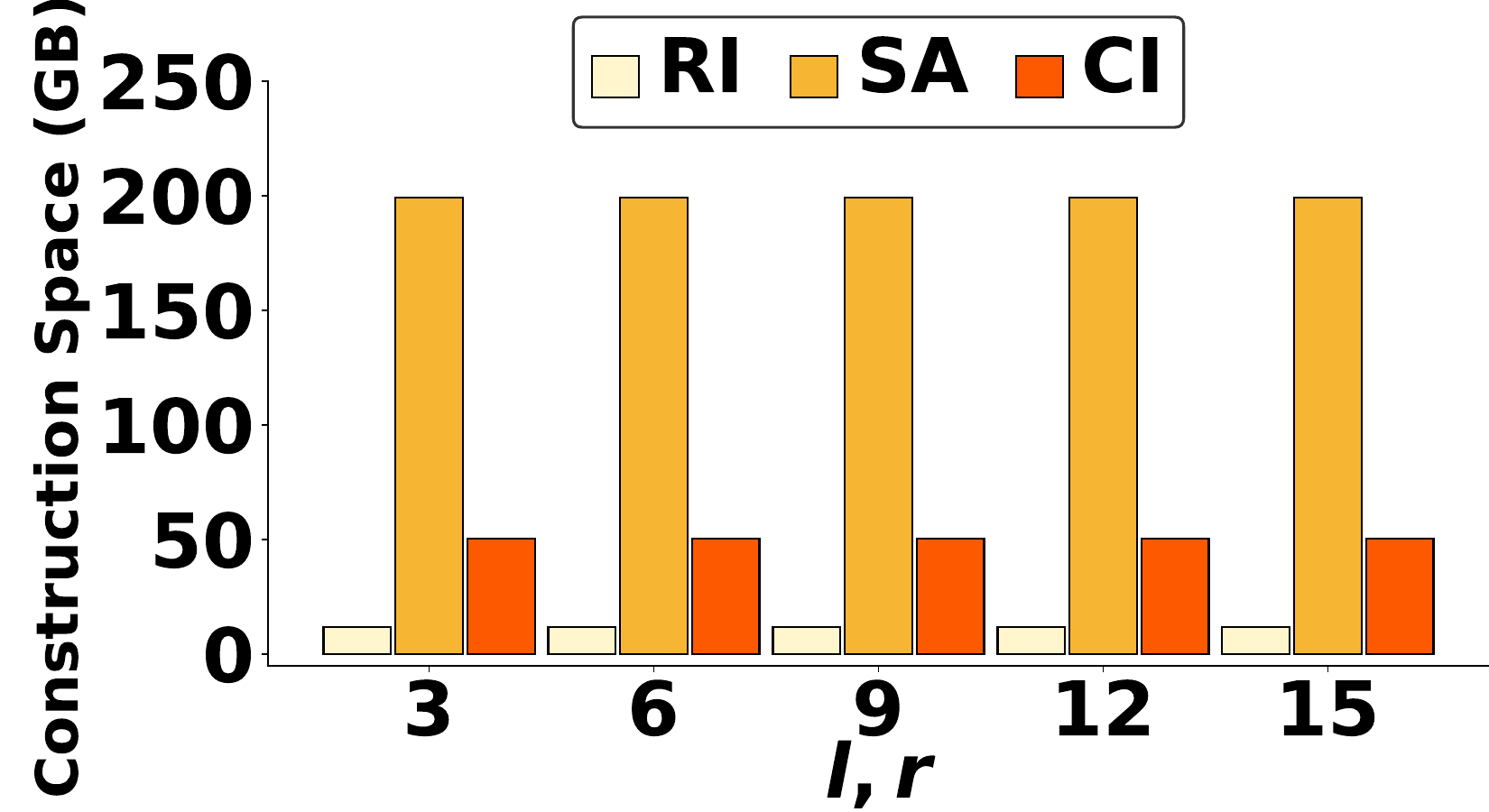}
        \caption{BST}
    \end{subfigure}%
\hfill
    \begin{subfigure}{0.188\linewidth}
        \includegraphics[width=1.05\linewidth]{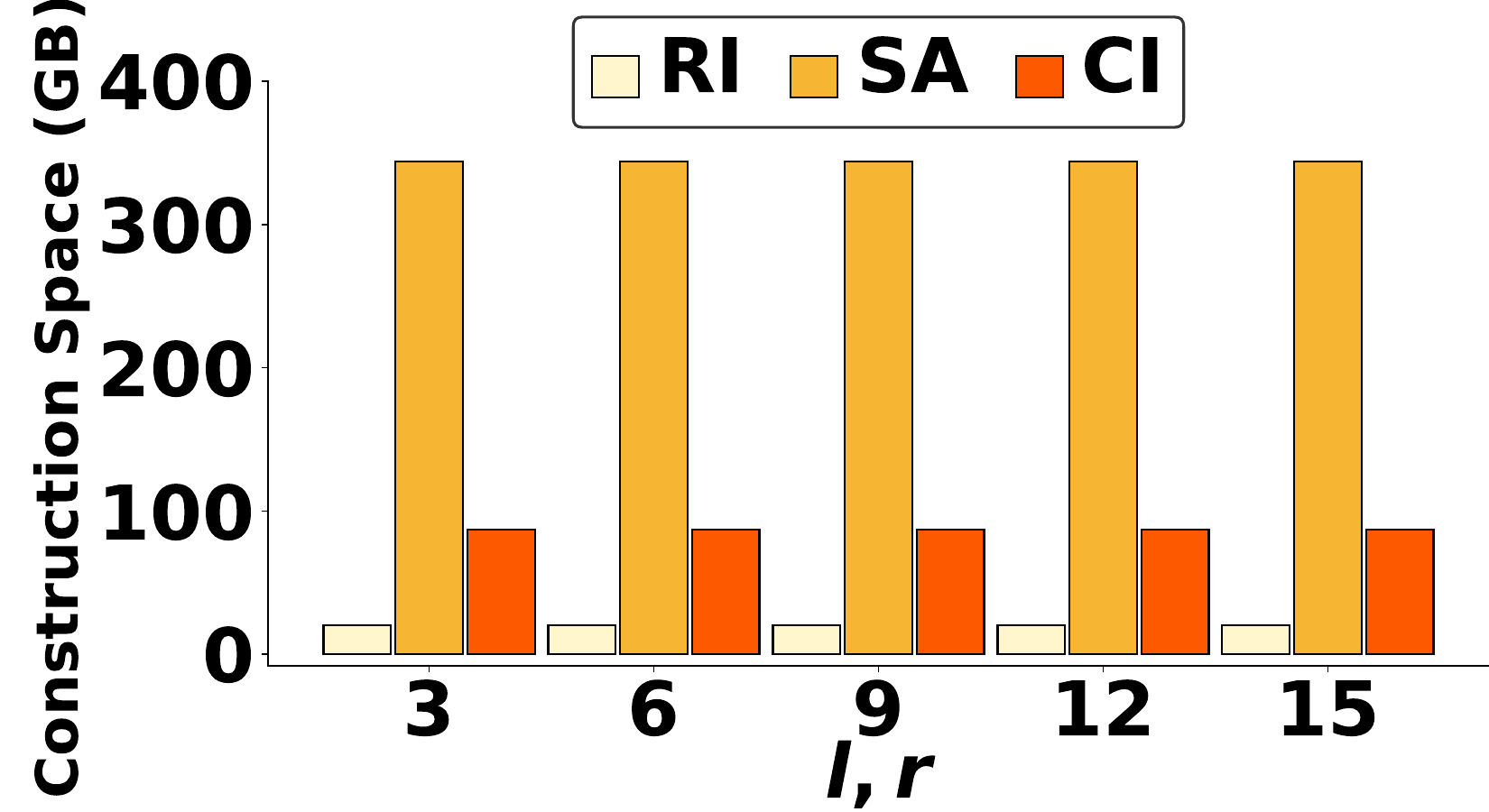}
        \caption{SDSL}
    \end{subfigure}%
\hfill
    \begin{subfigure}{0.188\linewidth}
        \includegraphics[width=1.05\linewidth]{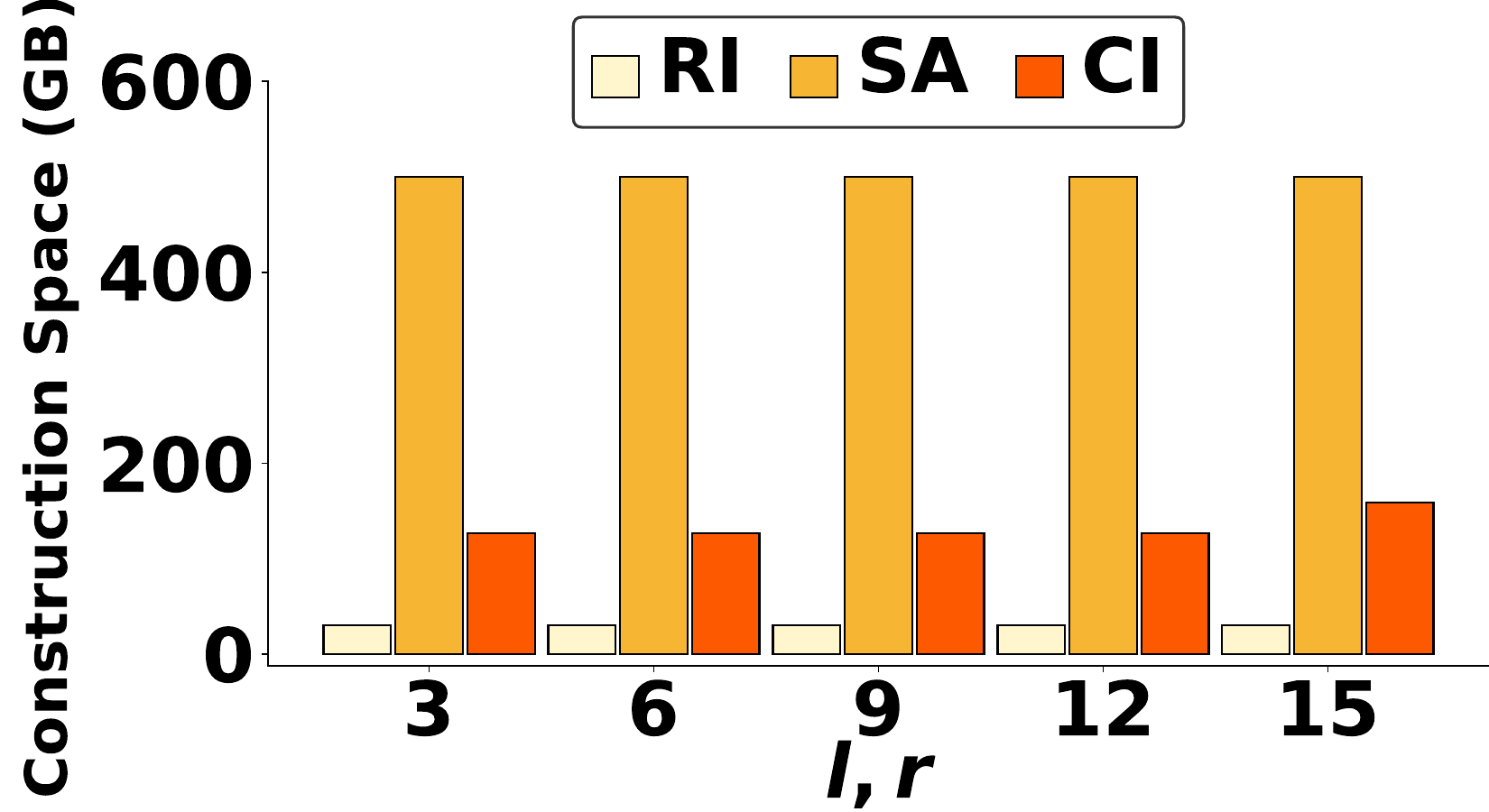}
        \caption{\sars}
    \end{subfigure}%
\hfill
    \begin{subfigure}{0.188\linewidth}
        \includegraphics[width=1.05\linewidth]{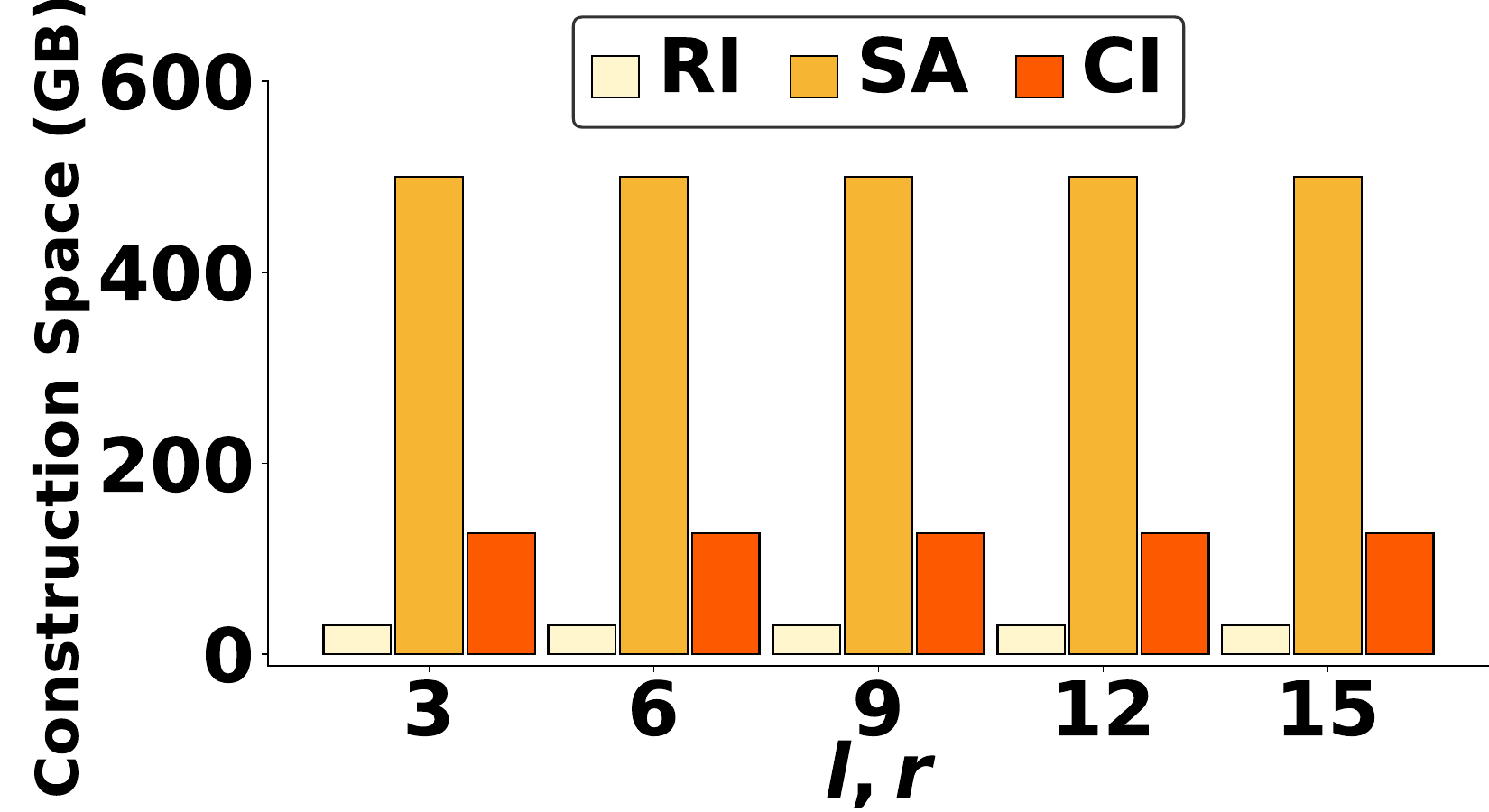}
        \caption{\chr}
        \label{fig:construction_space_CHR_xy}
    \end{subfigure}%

    \begin{subfigure}{0.188\linewidth}
        \includegraphics[width=1.05\linewidth]{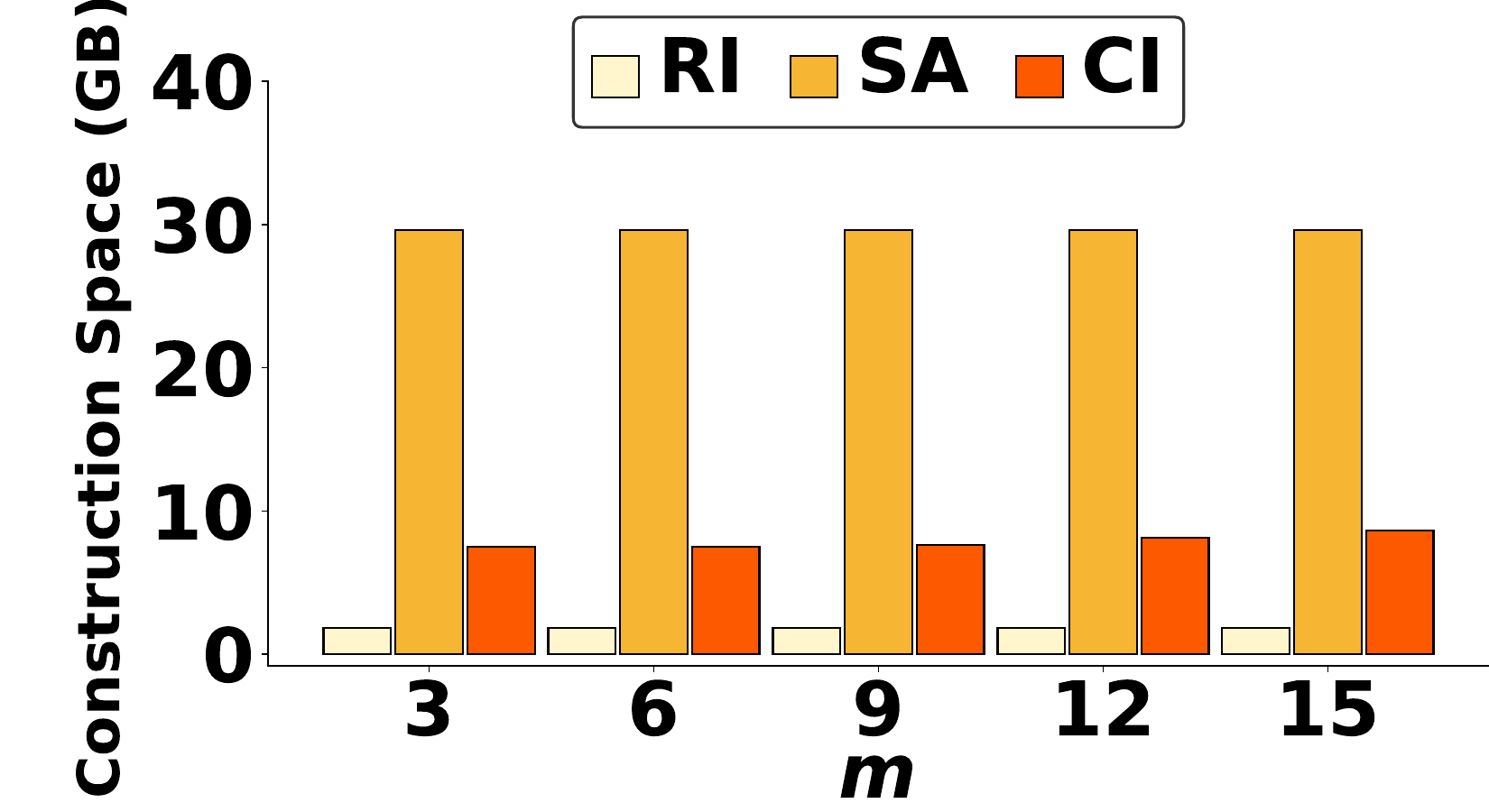}
        \caption{WIKI}
        \label{fig:construction_space_WIKI_m}
    \end{subfigure}%
\hfill
    \begin{subfigure}{0.188\linewidth}
        \includegraphics[width=1.05\linewidth]{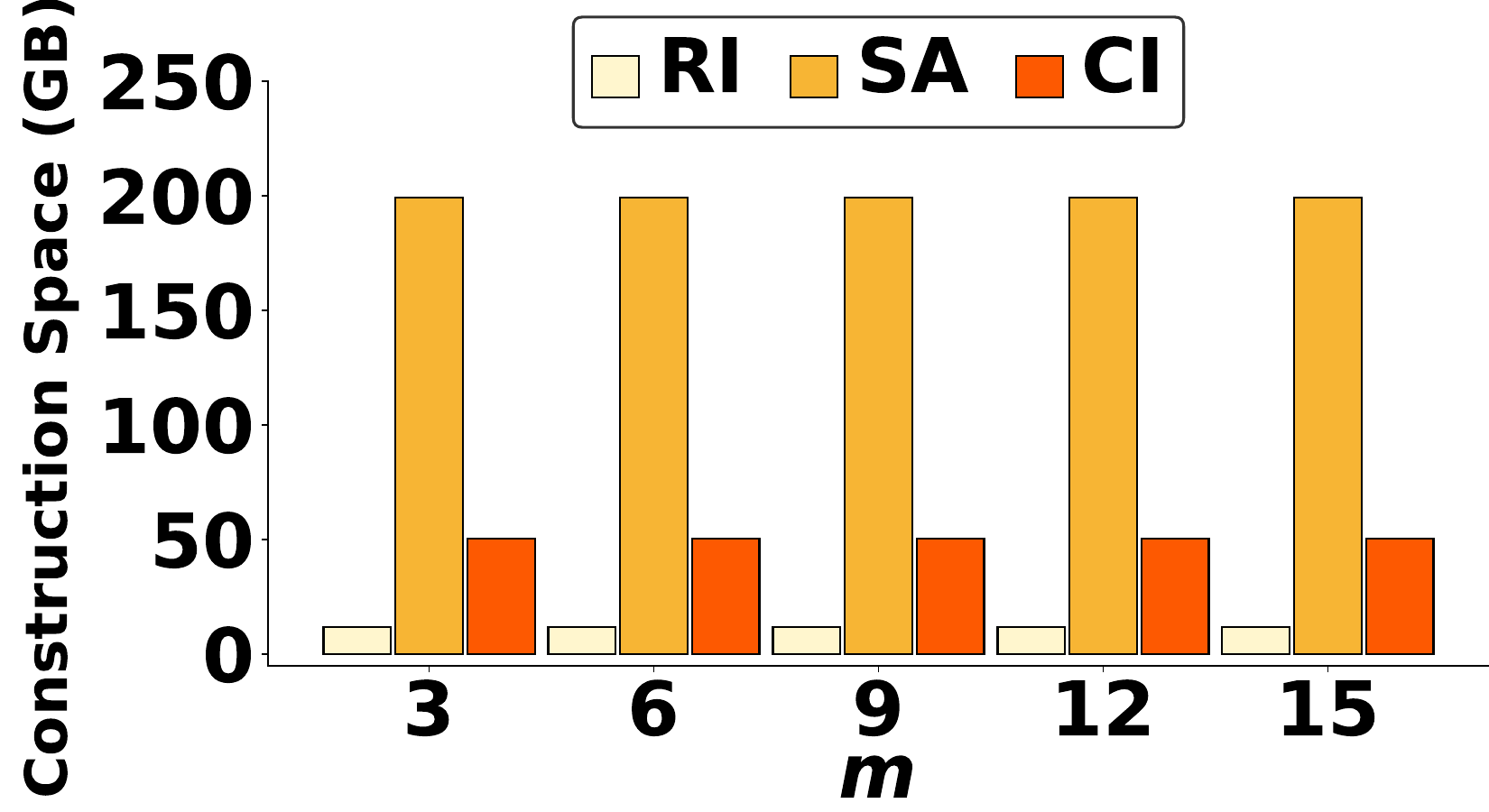}
        \caption{BST}
    \end{subfigure}%
\hfill
    \begin{subfigure}{0.188\linewidth}
        \includegraphics[width=1.05\linewidth]{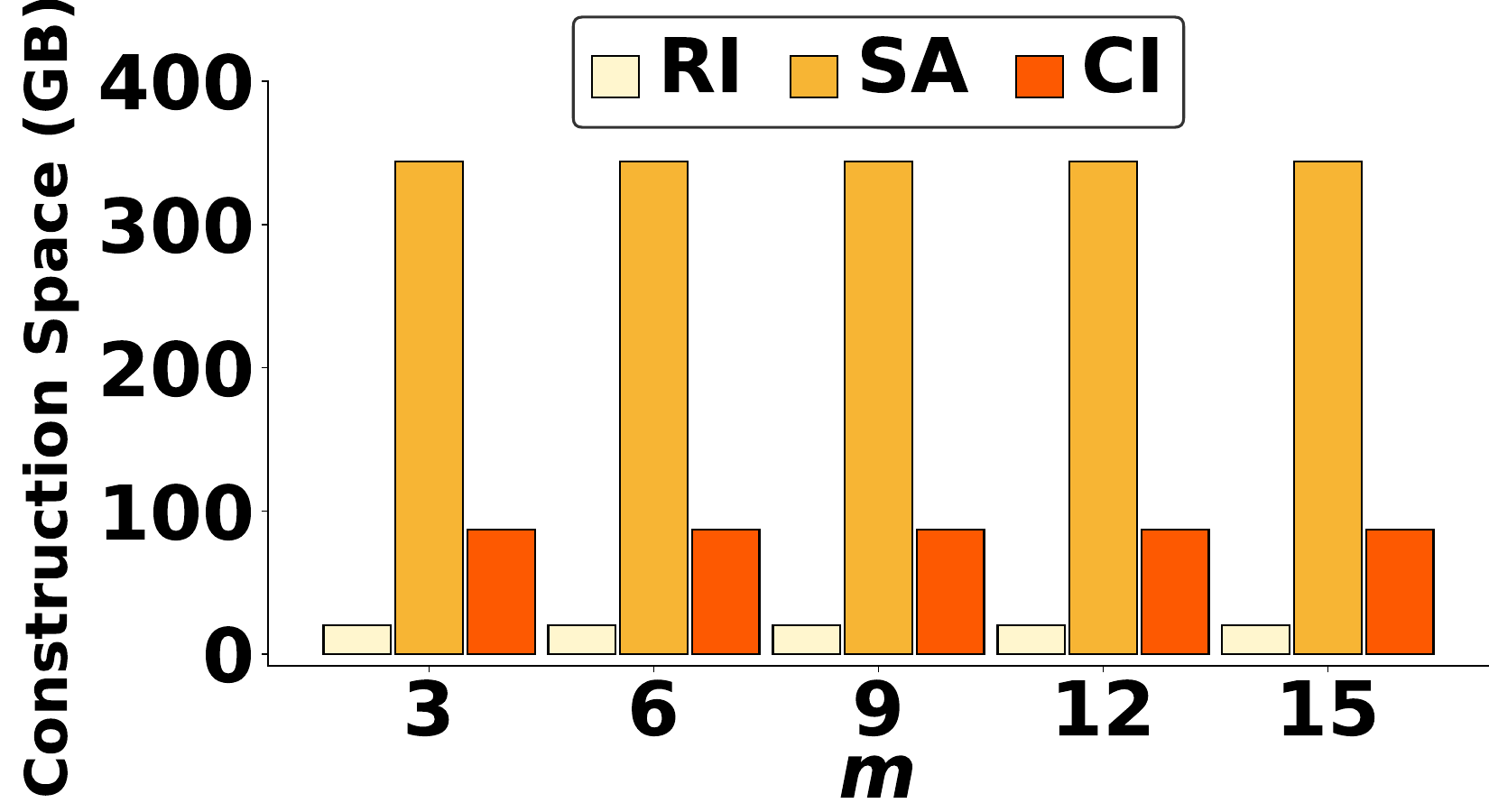}
        \caption{SDSL}
    \end{subfigure}%
\hfill
    \begin{subfigure}{0.188\linewidth}
        \includegraphics[width=1.05\linewidth]{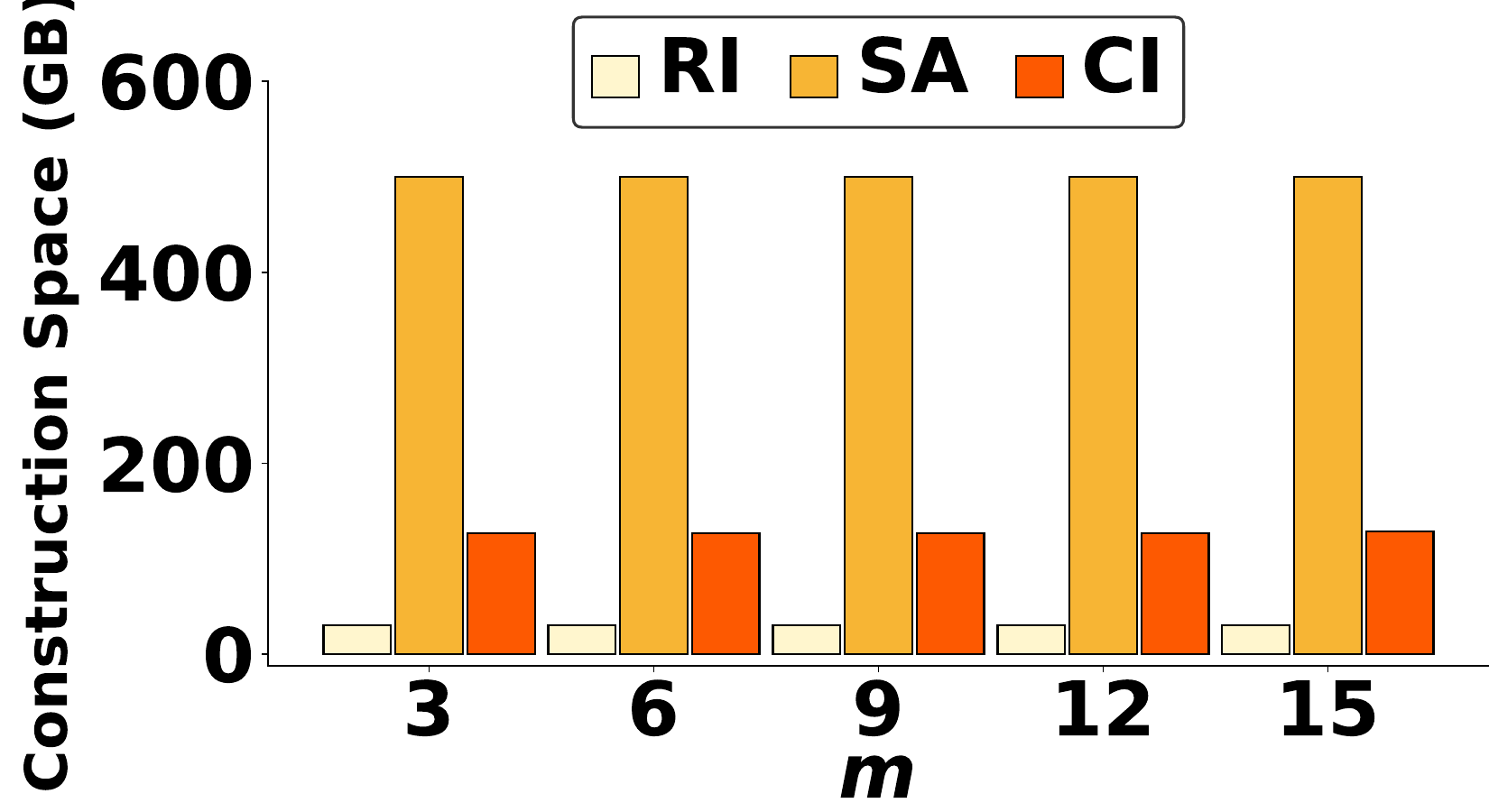}
        \caption{\sars}
    \end{subfigure}%
\hfill
        \begin{subfigure}{0.188\linewidth}
        \includegraphics[width=1.05\linewidth]{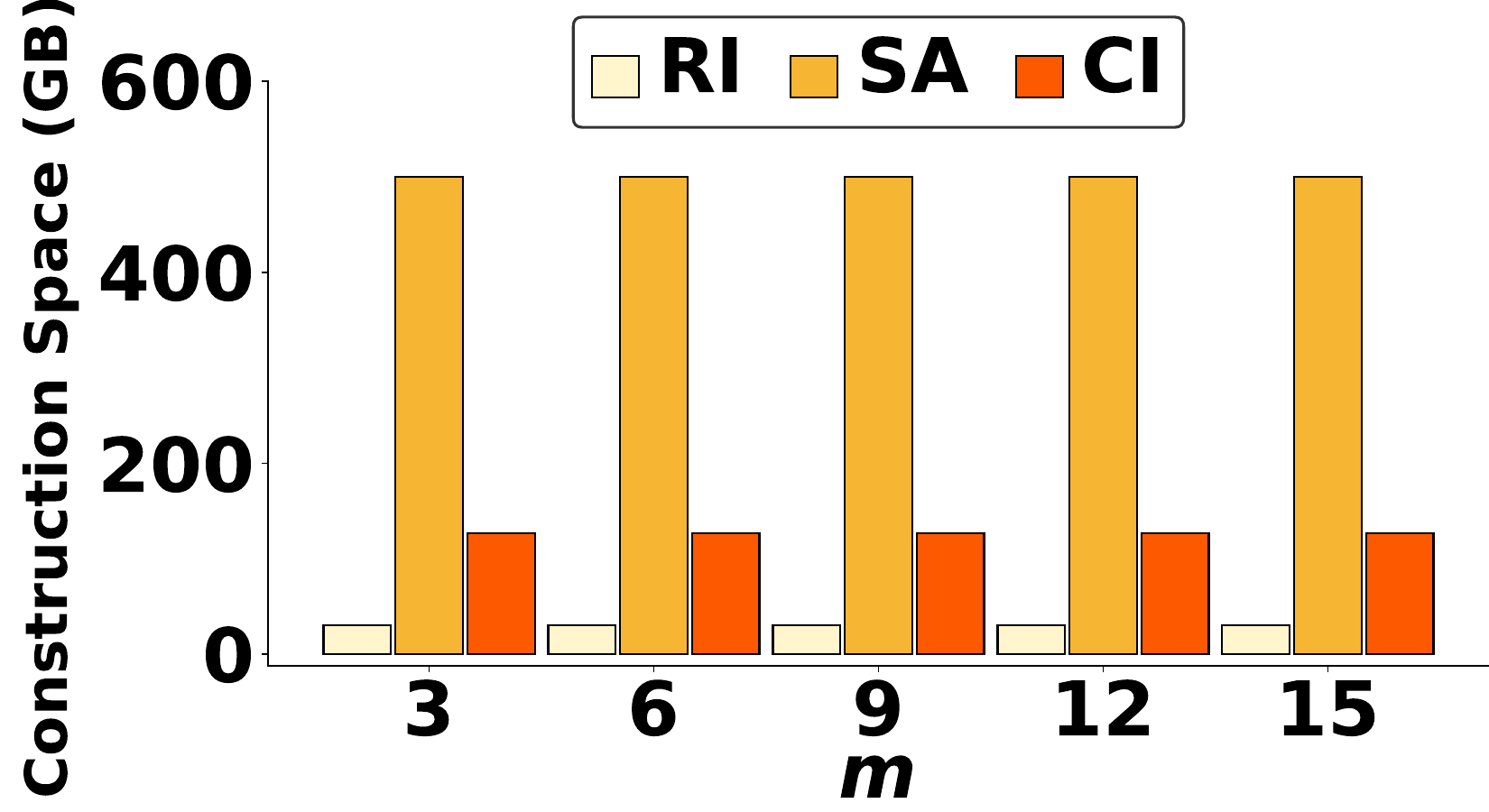}
        \caption{\chr}
    \label{fig:construction_space_CHR_m}
    \end{subfigure}%
    \caption{Construction space across all datasets: (a--e) vs.\ $n$, (f--j) vs.\ $l,r$, and (k--o) vs.\ $m$.}
    \label{fig:Constructionspace}
\end{figure*}

\subparagraph{Ablation Study.}~To show the benefit of the optimizations to our \CI index (see Section~\ref{sec:cpc:optimizations}), we compared it to \CI-, its  version with no optimizations. Note that all the optimizations have to be performed together. 
For completeness, we also include the results for \CPRIS and \RI. 

Fig.~\ref{fig:boost_ablation_n} shows the four measures of efficiency for all indexes for varying $n$. The performance of \RI, \CPRIS, and \CI in all measures was as in Figs.~\ref{fig:querytime_wiki_n} to \ref{fig:querytime_chr_n}, \ref{fig:construction_time_WIKI_n} to \ref{fig:construction_time_chr_n},  \ref{fig:index_space_WIKI_n} to \ref{fig:index_space_CHR_n},  and \ref{fig:construction_space_WIKI_n} to \ref{fig:construction_space_CHR_n}. \CI- was by far the least scalable index; it could only be applied to the prefix of BST with length $n=10^7$, as it needed more than 512GB of memory for prefixes with $n\geq 2\cdot 10^6=7$. Its query time was competitive with that of \RI and \CPRIS, as the number of pattern occurrences is relative small for the prefix used. However, the other measures were worse, which is in line with the complexity analysis (see Theorem~\ref{thm:CPC_index}). 

Fig.~\ref{fig:boost_ablation_z} shows the four measures of efficiency for all indexes for 
varying $z$. We varied $z$ by substituting letters selected uniformly at random from a prefix of BST with $n=10^7$. We also set $B=27$. 
The parameter $z$ appears only in the complexities of \CI (see Theorem~\ref{thm:BCPC_index}). Thus, only the performance of \CI was affected by $z$. Its query time  increased sublinearly with $z$ (see Fig.~\ref{all_boost_querytime_z}), as predicted by the query time complexity, while its construction time and index size increased superlinearly with $z$ (see Figs.~\ref{all_boost_construction_time_z} and ~\ref{all_boost_index_space_z}), as predicted by the corresponding  complexities. As shown in Fig.~\ref{all_boost_construction_space_z}, the construction space of \CI was not affected by $z$ for the same reason as in the experiment of  Fig.~\ref{fig:Constructionspace}. 

Overall, the optimizations in 
\CI are essential to achieving an excellent performance in terms of all measures.

\begin{figure}[t]
\centering
    \begin{subfigure}{0.24\linewidth}
        \includegraphics[width=1\linewidth]{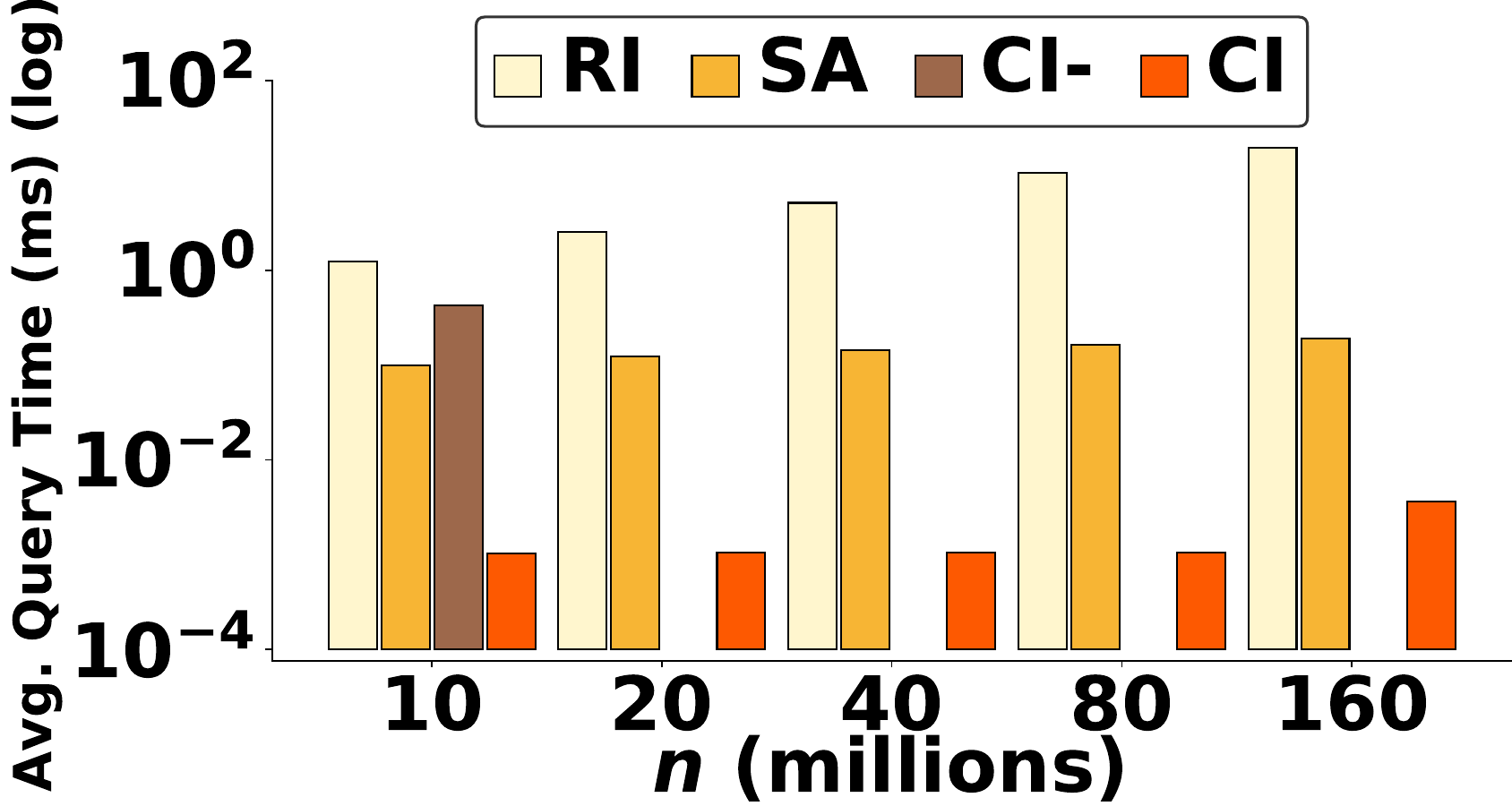}
        \caption{BST}
        \label{fig:all_boost_querytime_n} 
    \end{subfigure}%
    \begin{subfigure}{0.24\linewidth}
        \includegraphics[width=1\linewidth]{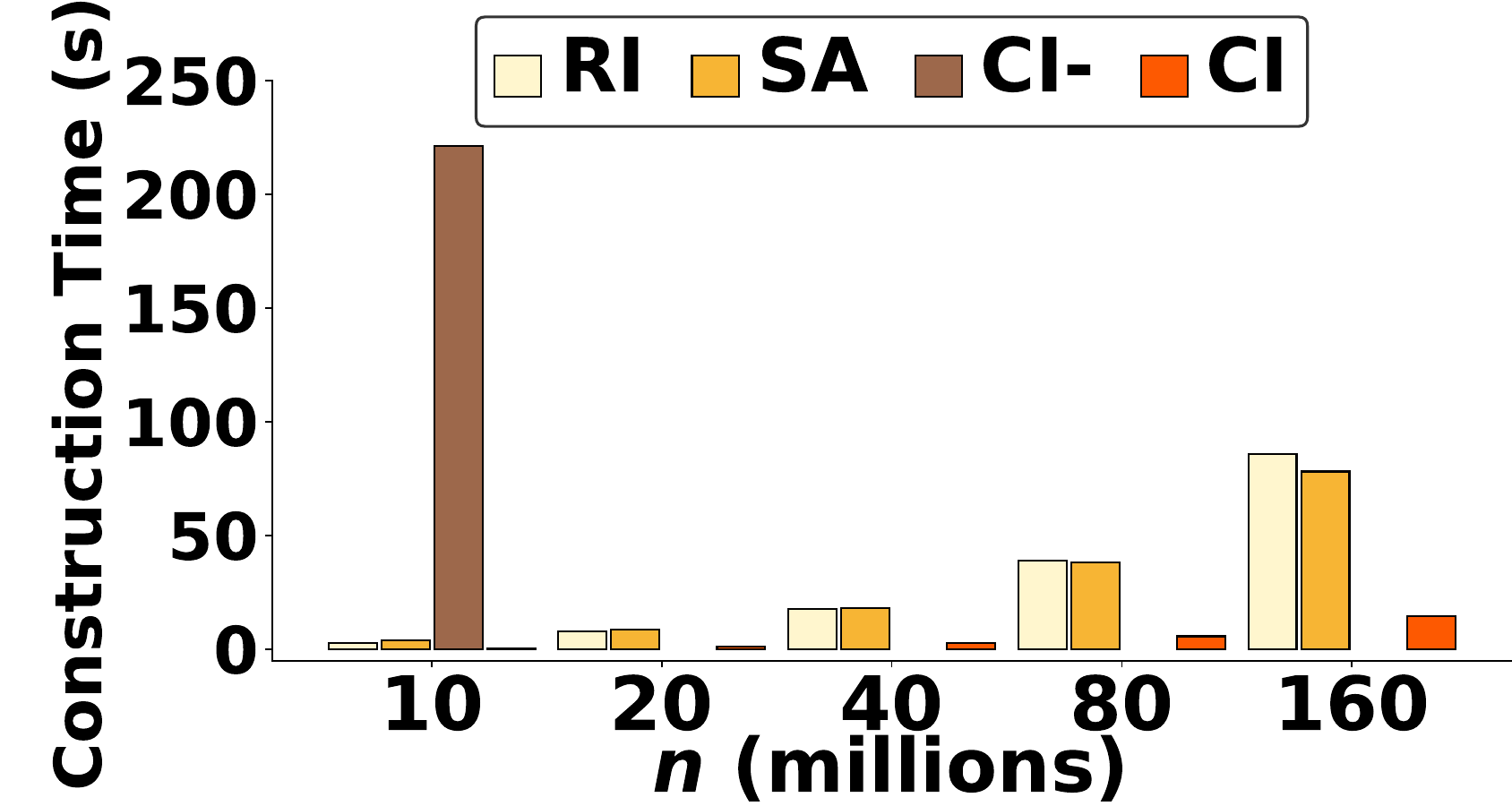}
        \caption{BST}
    \end{subfigure}
    \begin{subfigure}{0.24\linewidth}
        \includegraphics[width=1\linewidth]{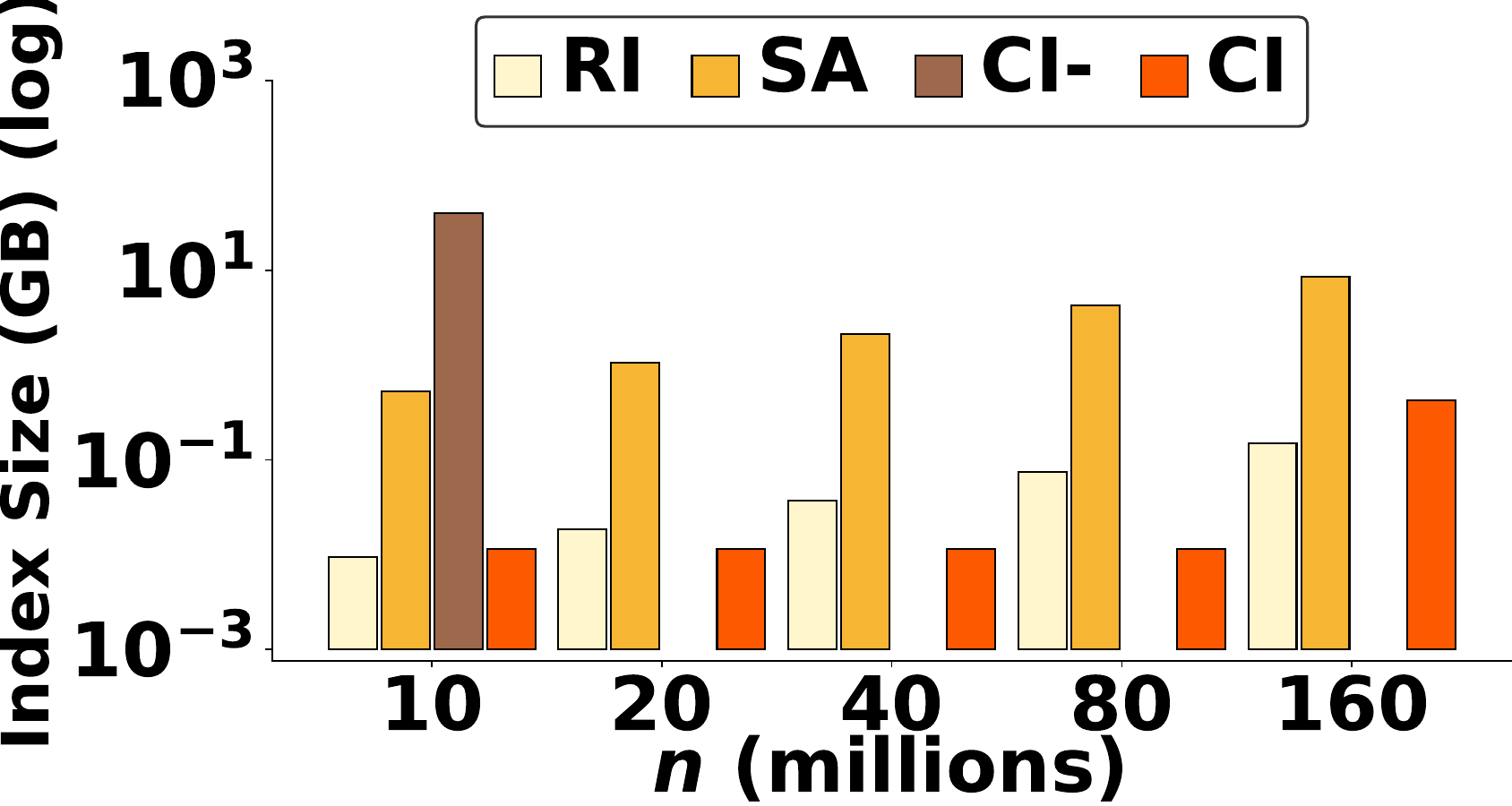}
        \caption{BST}
    \end{subfigure}%
    \begin{subfigure}{0.24\linewidth}
        \includegraphics[width=1\linewidth]{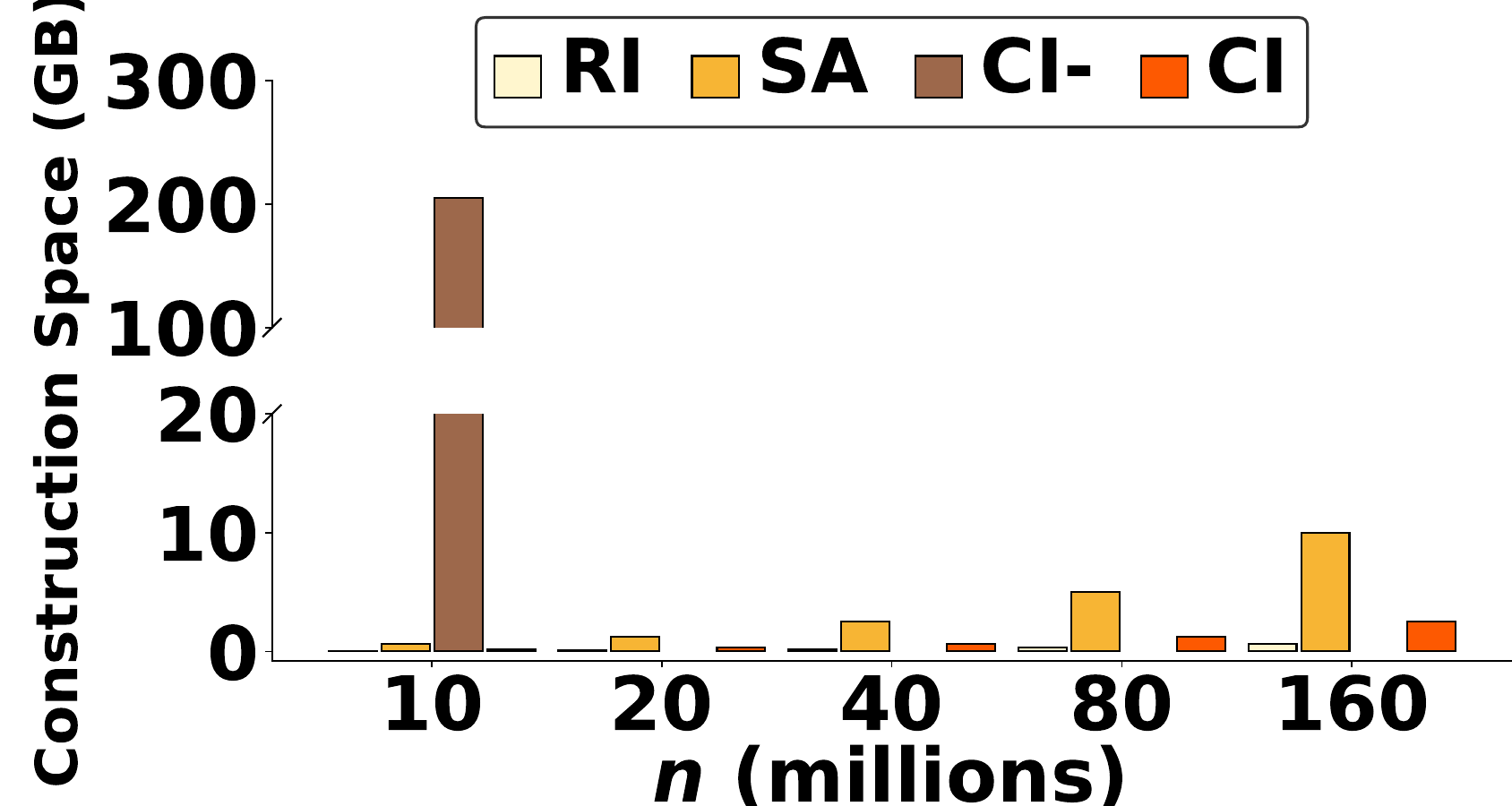}
        \caption{BST}
        \label{fig:all_boost_construction_space_n} 
    \end{subfigure}
    \caption{(a) average query time, (b) construction time, (c) index size, and (d) construction space for  %
    \RI, \CPRIS, CI-, and CI vs. $n$ with $B = l + m + r$ and $l = m = r = 9$.}
    \label{fig:boost_ablation_n}
\end{figure}

\begin{figure}[t]
\centering
    \begin{subfigure}{0.24\linewidth}
        \includegraphics[width=1\linewidth]{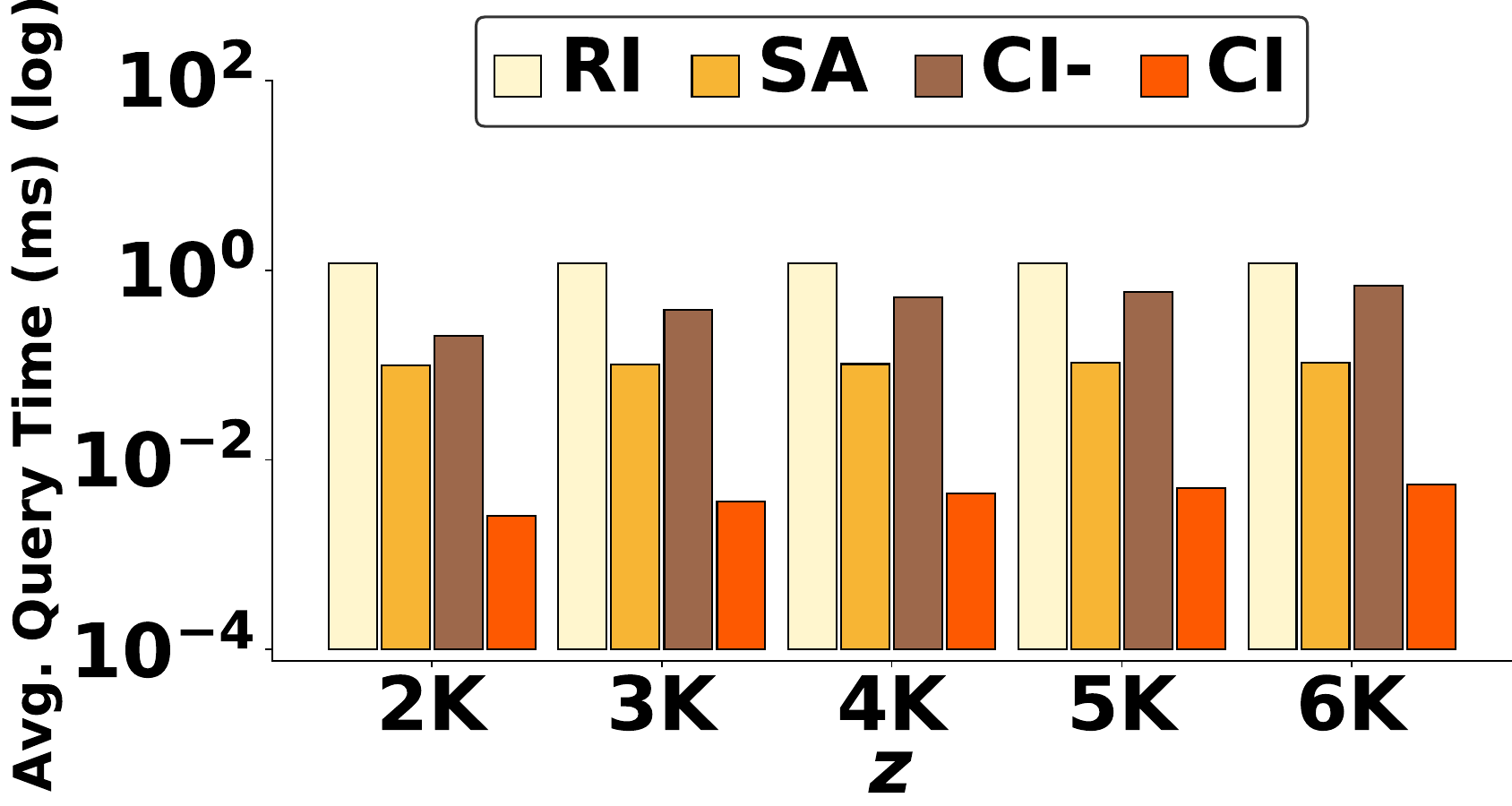}
        \caption{BST prefix}
        \label{all_boost_querytime_z}
    \end{subfigure}%
    \begin{subfigure}{0.24\linewidth}
        \includegraphics[width=1\linewidth]{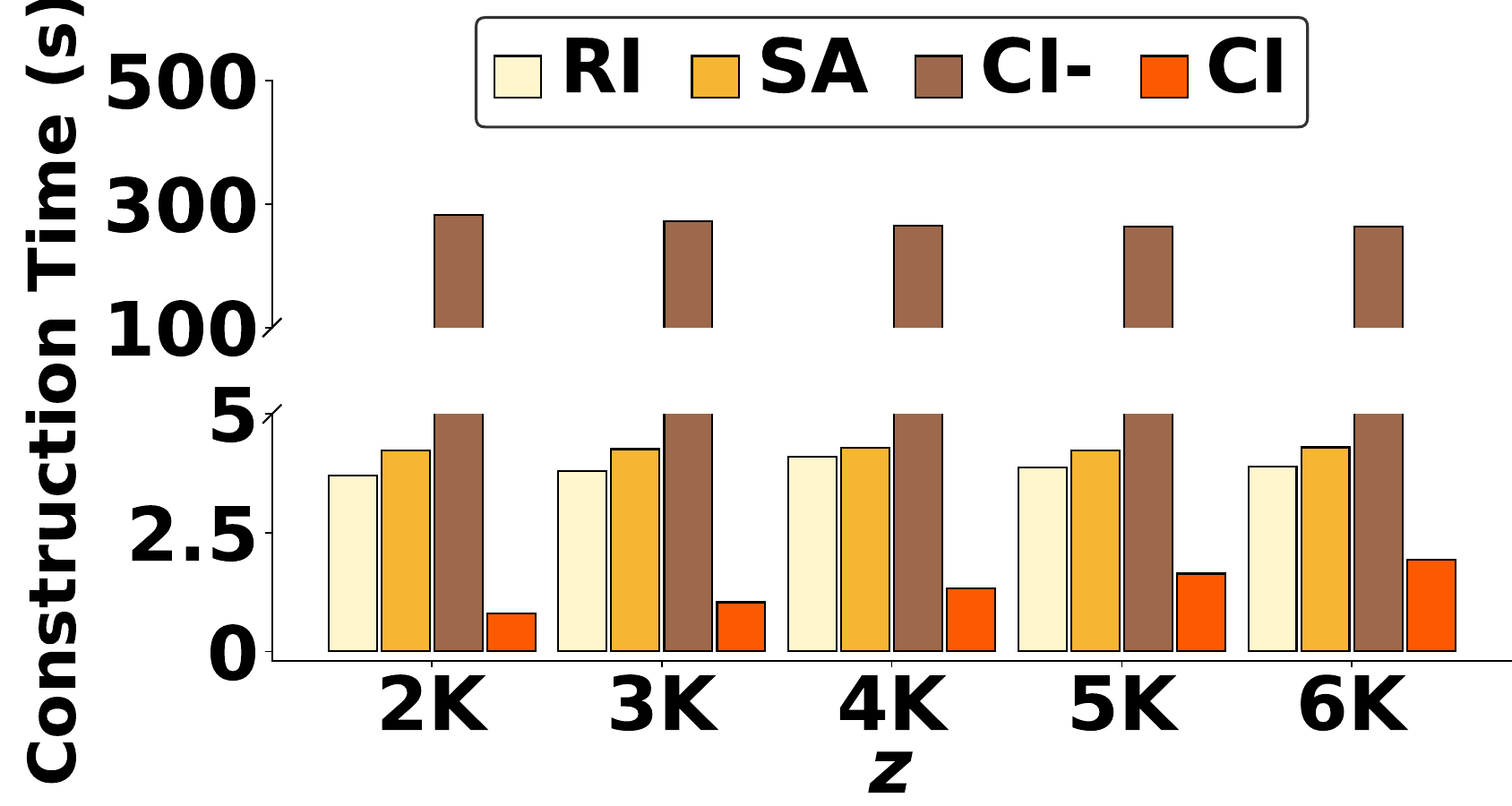}
        \caption{BST prefix}
        \label{all_boost_construction_time_z}
    \end{subfigure}
    \begin{subfigure}{0.24\linewidth}
        \includegraphics[width=1\linewidth]{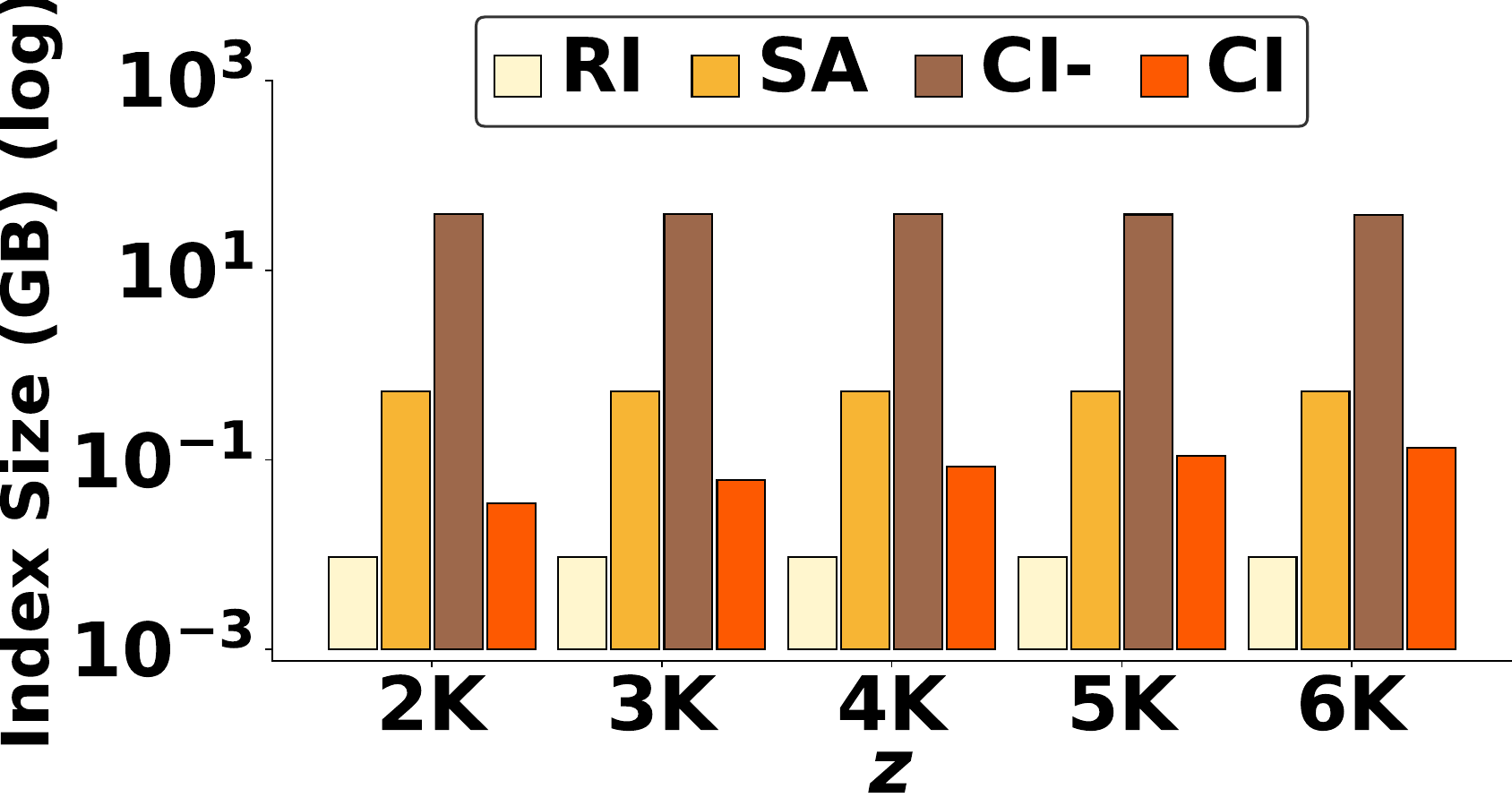}
        \caption{BST prefix}
        \label{all_boost_index_space_z}
    \end{subfigure}%
    \begin{subfigure}{0.24\linewidth}
        \includegraphics[width=1\linewidth]{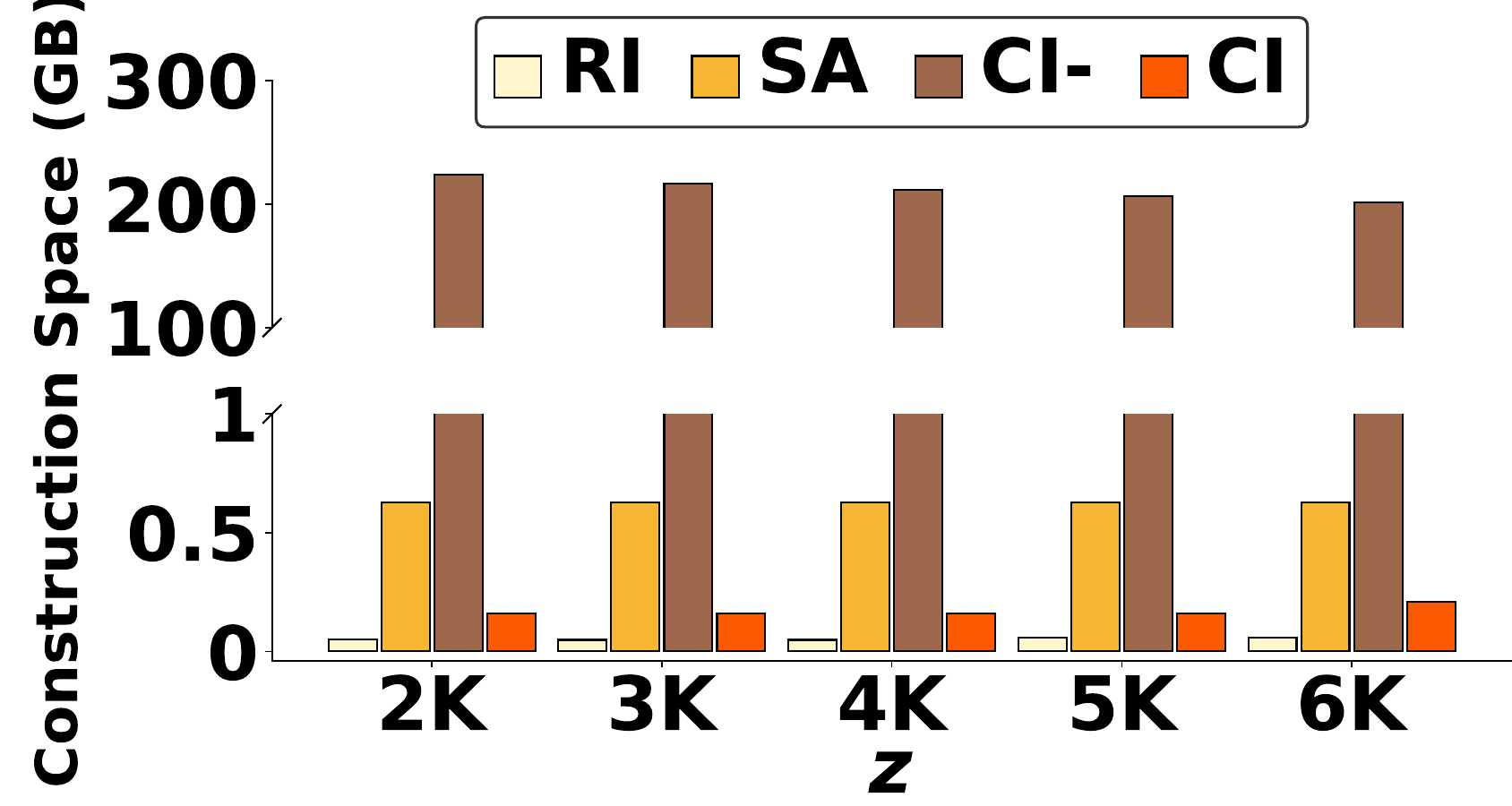}
        \caption{BST prefix}
        \label{all_boost_construction_space_z}
    \end{subfigure}
    \caption{(a) average query time, (b) construction time, (c) index size, and (d) construction space for  %
    \RI, \CPRIS, CI-, and CI vs. $z$ with $B = l + m + r$ and $l = m = r = 9$.}
    \label{fig:boost_ablation_z}
\end{figure}

\section{Conclusion}\label{sec:conclusion}

We introduced the \CPM problem and proposed a linear-work algorithm
with for it. Our algorithm has two implementations; one using exclusively internal memory and another that uses also external memory. 
We also introduced the \CPC problem, for which we proposed a near-linear space index that can answer queries in near-optimal time along with optimizations that significantly improve its efficiency in theory and in practice. 
We showed experimentally that the external memory version of our algorithm for \CPM can process very large datasets with a small amount of internal memory while its runtime is comparable to that of the internal memory version. Interestingly, we also showed that our optimized index for \CPC outperforms an approach based on the state of the art for the reporting version of \CPC in terms of all $4$ efficiency measures, often by more than an order of magnitude.

\section*{Acknowledgments}
 Grigorios Loukides would like to thank Eileen Kwan for discussions on the relevance of the proposed problems to bioinformatics. Solon P. Pissis  is supported by the PANGAIA and ALPACA projects that have received funding from the European Union’s Horizon 2020 research and innovation programme under the Marie Skłodowska-Curie grant agreements No 872539 and 956229, respectively.

\bibliographystyle{alphaurl}
\bibliography{main}

\newcommand{\etalchar}[1]{$^{#1}$}
\begin{thebibliography}{GGK{\etalchar{+}}14}

\bibitem[ACGT23]{abedinDCC23}
Paniz Abedin, Oliver~A. Chubet, Daniel Gibney, and Sharma~V. Thankachan.
\newblock Contextual pattern matching in less space.
\newblock In Ali Bilgin, Michael~W. Marcellin, Joan Serra{-}Sagrist{\`{a}}, and
  James~A. Storer, editors, {\em Data Compression Conference, {DCC} 2023,
  Snowbird, UT, USA, March 21-24, 2023}, pages 160--167. {IEEE}, 2023.
\newblock URL: \url{https://doi.org/10.1109/DCC55655.2023.00024}, \href
  {http://dx.doi.org/10.1109/DCC55655.2023.00024}
  {\path{doi:10.1109/DCC55655.2023.00024}}.

\bibitem[ALP23]{pvldb23}
Lorraine A.~K. Ayad, Grigorios Loukides, and Solon~P. Pissis.
\newblock Text indexing for long patterns: Anchors are all you need.
\newblock {\em Proc. {VLDB} Endow.}, 16(9):2117--2131, 2023.
\newblock URL: \url{https://www.vldb.org/pvldb/vol16/p2117-loukides.pdf}, \href
  {http://dx.doi.org/10.14778/3598581.3598586}
  {\path{doi:10.14778/3598581.3598586}}.

\bibitem[AMJ16]{asselin2016anomaly}
Eric Asselin, Carlos~Aguilar Melchor, and Gentian Jakllari.
\newblock Anomaly detection for web server log reduction: {A} simple yet
  efficient crawling based approach.
\newblock In {\em 2016 {IEEE} Conference on Communications and Network
  Security, {CNS} 2016, Philadelphia, PA, USA, October 17-19, 2016}, pages
  586--590. {IEEE}, 2016.
\newblock URL: \url{https://doi.org/10.1109/CNS.2016.7860553}, \href
  {http://dx.doi.org/10.1109/CNS.2016.7860553}
  {\path{doi:10.1109/CNS.2016.7860553}}.

\bibitem[AU07]{takeaki}
Hiroki Arimura and Takeaki Uno.
\newblock An efficient polynomial space and polynomial delay algorithm for
  enumeration of maximal motifs in a sequence.
\newblock {\em J. Comb. Optim.}, 13(3):243--262, 2007.
\newblock URL: \url{https://doi.org/10.1007/s10878-006-9029-1}, \href
  {http://dx.doi.org/10.1007/s10878-006-9029-1}
  {\path{doi:10.1007/s10878-006-9029-1}}.

\bibitem[BEGV18]{sharma9}
Philip Bille, Mikko~Berggren Ettienne, Inge~Li G{\o}rtz, and Hjalte~Wedel
  Vildh{\o}j.
\newblock Time-space trade-offs for lempel-ziv compressed indexing.
\newblock {\em Theor. Comput. Sci.}, 713:66--77, 2018.
\newblock URL: \url{https://doi.org/10.1016/j.tcs.2017.12.021}, \href
  {http://dx.doi.org/10.1016/J.TCS.2017.12.021}
  {\path{doi:10.1016/J.TCS.2017.12.021}}.

\bibitem[Ben75]{bentley1975multidimensional}
Jon~Louis Bentley.
\newblock Multidimensional binary search trees used for associative searching.
\newblock {\em Commun. {ACM}}, 18(9):509--517, 1975.
\newblock URL: \url{https://doi.org/10.1145/361002.361007}, \href
  {http://dx.doi.org/10.1145/361002.361007} {\path{doi:10.1145/361002.361007}}.

\bibitem[BFO16]{DBLP:journals/jea/Bingmann0O16}
Timo Bingmann, Johannes Fischer, and Vitaly Osipov.
\newblock Inducing suffix and {LCP} arrays in external memory.
\newblock {\em {ACM} J. Exp. Algorithmics}, 21(1):2.3:1--2.3:27, 2016.
\newblock URL: \url{https://doi.org/10.1145/2975593}, \href
  {http://dx.doi.org/10.1145/2975593} {\path{doi:10.1145/2975593}}.

\bibitem[BKSS90]{beckmann1990r}
Norbert Beckmann, Hans{-}Peter Kriegel, Ralf Schneider, and Bernhard Seeger.
\newblock The {R*-Tree}: An efficient and robust access method for points and
  rectangles.
\newblock In Hector Garcia{-}Molina and H.~V. Jagadish, editors, {\em
  Proceedings of the 1990 {ACM} {SIGMOD} International Conference on Management
  of Data, Atlantic City, NJ, USA, May 23-25, 1990}, pages 322--331. {ACM}
  Press, 1990.
\newblock URL: \url{https://doi.org/10.1145/93597.98741}, \href
  {http://dx.doi.org/10.1145/93597.98741} {\path{doi:10.1145/93597.98741}}.

\bibitem[BNL20]{DBLP:journals/pvldb/Boncz0L20}
Peter~A. Boncz, Thomas Neumann, and Viktor Leis.
\newblock {FSST:} fast random access string compression.
\newblock {\em Proc. {VLDB} Endow.}, 13(11):2649--2661, 2020.
\newblock URL: \url{http://www.vldb.org/pvldb/vol13/p2649-boncz.pdf}.

\bibitem[boo]{boost}
Boost github repository.
\newblock \url{https://github.com/boostorg/boost}.

\bibitem[BS78]{bentley1978decomposable}
Jon~Louis Bentley and James~B Saxe.
\newblock {\em Decomposable searching problems.}
\newblock Carnegie Mellon University, 1978.

\bibitem[CCJ{\etalchar{+}}16]{ndsi16}
Byungkwon Choi, Jongwook Chae, Muhammad Jamshed, KyoungSoo Park, and Dongsu
  Han.
\newblock {DFC:} accelerating string pattern matching for network applications.
\newblock In Katerina~J. Argyraki and Rebecca Isaacs, editors, {\em 13th
  {USENIX} Symposium on Networked Systems Design and Implementation, {NSDI}
  2016, Santa Clara, CA, USA, March 16-18, 2016}, pages 551--565. {USENIX}
  Association, 2016.
\newblock URL:
  \url{https://www.usenix.org/conference/nsdi16/technical-sessions/presentation/choi}.

\bibitem[CEK{\etalchar{+}}21]{talgChristiansen}
Anders~Roy Christiansen, Mikko~Berggren Ettienne, Tomasz Kociumaka, Gonzalo
  Navarro, and Nicola Prezza.
\newblock Optimal-time dictionary-compressed indexes.
\newblock {\em {ACM} Trans. Algorithms}, 17(1):8:1--8:39, 2021.
\newblock URL: \url{https://doi.org/10.1145/3426473}, \href
  {http://dx.doi.org/10.1145/3426473} {\path{doi:10.1145/3426473}}.

\bibitem[CHL07]{DBLP:books/daglib/0020103}
Maxime Crochemore, Christophe Hancart, and Thierry Lecroq.
\newblock {\em Algorithms on strings}.
\newblock Cambridge University Press, 2007.

\bibitem[chr21]{chr}
Pangeneome dataset, 2021.
\newblock URL: \url{https://github.com/koeppl/phoni}.

\bibitem[CKL15]{DBLP:journals/algorithmica/0001KL15}
Richard Cole, Tsvi Kopelowitz, and Moshe Lewenstein.
\newblock Suffix trays and suffix trists: Structures for faster text indexing.
\newblock {\em Algorithmica}, 72(2):450--466, 2015.
\newblock URL: \url{https://doi.org/10.1007/s00453-013-9860-6}, \href
  {http://dx.doi.org/10.1007/s00453-013-9860-6}
  {\path{doi:10.1007/s00453-013-9860-6}}.

\bibitem[CN12]{sh}
Francisco Claude and Gonzalo Navarro.
\newblock Improved grammar-based compressed indexes.
\newblock In Liliana Calder{\'{o}}n{-}Benavides, Cristina~N.
  Gonz{\'{a}}lez{-}Caro, Edgar Ch{\'{a}}vez, and Nivio Ziviani, editors, {\em
  String Processing and Information Retrieval - 19th International Symposium,
  {SPIRE} 2012, Cartagena de Indias, Colombia, October 21-25, 2012.
  Proceedings}, volume 7608 of {\em Lecture Notes in Computer Science}, pages
  180--192. Springer, 2012.
\newblock URL: \url{https://doi.org/10.1007/978-3-642-34109-0\_19}, \href
  {http://dx.doi.org/10.1007/978-3-642-34109-0\_19}
  {\path{doi:10.1007/978-3-642-34109-0\_19}}.

\bibitem[CPL{\etalchar{+}}13]{nucleicacidres}
Heejoon Chae, Jinwoo Park, Seong-Whan Lee, Kenneth~P. Nephew, and Sun Kim.
\newblock Comparative analysis using k-mer and k-flank patterns provides
  evidence for {CpG} island sequence evolution in mammalian genomes.
\newblock {\em Nucleic Acids Research}, 41(9):4783--4791, 03 2013.
\newblock URL: \url{https://doi.org/10.1093/nar/gkt144}, \href
  {http://arxiv.org/abs/https://academic.oup.com/nar/article-pdf/41/9/4783/7162900/gkt144.pdf}
  {\path{arXiv:https://academic.oup.com/nar/article-pdf/41/9/4783/7162900/gkt144.pdf}},
  \href {http://dx.doi.org/10.1093/nar/gkt144} {\path{doi:10.1093/nar/gkt144}}.

\bibitem[CYL15]{locis}
Meng Chen, Xiaohui Yu, and Yang Liu.
\newblock Mining moving patterns for predicting next location.
\newblock {\em Inf. Syst.}, 54:156--168, 2015.
\newblock URL: \url{https://doi.org/10.1016/j.is.2015.07.001}, \href
  {http://dx.doi.org/10.1016/J.IS.2015.07.001}
  {\path{doi:10.1016/J.IS.2015.07.001}}.

\bibitem[DEMT15]{DBLP:journals/tcs/DurocherEMT15}
Stephane Durocher, Hicham El{-}Zein, J.~Ian Munro, and Sharma~V. Thankachan.
\newblock Low space data structures for geometric range mode query.
\newblock {\em Theor. Comput. Sci.}, 581:97--101, 2015.
\newblock URL: \url{https://doi.org/10.1016/j.tcs.2015.03.011}, \href
  {http://dx.doi.org/10.1016/J.TCS.2015.03.011}
  {\path{doi:10.1016/J.TCS.2015.03.011}}.

\bibitem[Deu96]{DBLP:journals/rfc/rfc1951}
Peter Deutsch.
\newblock {DEFLATE} compressed data format specification version 1.3.
\newblock {\em {RFC}}, 1951:1--17, 1996.
\newblock URL: \url{https://doi.org/10.17487/RFC1951}, \href
  {http://dx.doi.org/10.17487/RFC1951} {\path{doi:10.17487/RFC1951}}.

\bibitem[Far97]{DBLP:conf/focs/Farach97}
Martin Farach.
\newblock Optimal suffix tree construction with large alphabets.
\newblock In {\em 38th Annual Symposium on Foundations of Computer Science,
  {FOCS} '97, Miami Beach, Florida, USA, October 19-22, 1997}, pages 137--143.
  {IEEE} Computer Society, 1997.
\newblock URL: \url{https://doi.org/10.1109/SFCS.1997.646102}, \href
  {http://dx.doi.org/10.1109/SFCS.1997.646102}
  {\path{doi:10.1109/SFCS.1997.646102}}.

\bibitem[FHK05a]{fisher}
Johannes Fischer, Volker Heun, and Stefan Kramer.
\newblock Fast frequent string mining using suffix arrays.
\newblock In {\em Proceedings of the 5th {IEEE} International Conference on
  Data Mining ({ICDM} 2005)}, pages 609--612. {IEEE} Computer Society, 2005.
\newblock \href {http://dx.doi.org/10.1109/ICDM.2005.62}
  {\path{doi:10.1109/ICDM.2005.62}}.

\bibitem[FHK05b]{icdm2005}
Johannes Fischer, Volker Heun, and Stefan Kramer.
\newblock Fast frequent string mining using suffix arrays.
\newblock In {\em Proceedings of the 5th {IEEE} International Conference on
  Data Mining {(ICDM} 2005), 27-30 November 2005, Houston, Texas, {USA}}, pages
  609--612. {IEEE} Computer Society, 2005.
\newblock URL: \url{https://doi.org/10.1109/ICDM.2005.62}, \href
  {http://dx.doi.org/10.1109/ICDM.2005.62} {\path{doi:10.1109/ICDM.2005.62}}.

\bibitem[FMV08]{icdm2008}
Johannes Fischer, Veli M{\"{a}}kinen, and Niko V{\"{a}}lim{\"{a}}ki.
\newblock Space efficient string mining under frequency constraints.
\newblock In {\em Proceedings of the 8th {IEEE} International Conference on
  Data Mining {(ICDM} 2008), December 15-19, 2008, Pisa, Italy}, pages
  193--202. {IEEE} Computer Society, 2008.
\newblock URL: \url{https://doi.org/10.1109/ICDM.2008.32}, \href
  {http://dx.doi.org/10.1109/ICDM.2008.32} {\path{doi:10.1109/ICDM.2008.32}}.

\bibitem[Gen08]{GenomesProject}
1000 genomes project, 2008.
\newblock URL: \url{https://www.internationalgenome.org/}.

\bibitem[Get]{Getgit}
Get git script.
\newblock \url{https://github.com/nicolaprezza/get-git-revisions}.

\bibitem[GGK{\etalchar{+}}14]{DBLP:conf/latin/GagieGKNP14}
Travis Gagie, Pawel Gawrychowski, Juha K{\"{a}}rkk{\"{a}}inen, Yakov Nekrich,
  and Simon~J. Puglisi.
\newblock {LZ77}-based self-indexing with faster pattern matching.
\newblock In Alberto Pardo and Alfredo Viola, editors, {\em {LATIN} 2014:
  Theoretical Informatics - 11th Latin American Symposium, Montevideo, Uruguay,
  March 31 - April 4, 2014. Proceedings}, volume 8392 of {\em Lecture Notes in
  Computer Science}, pages 731--742. Springer, 2014.
\newblock URL: \url{https://doi.org/10.1007/978-3-642-54423-1\_63}, \href
  {http://dx.doi.org/10.1007/978-3-642-54423-1\_63}
  {\path{doi:10.1007/978-3-642-54423-1\_63}}.

\bibitem[GKN17]{simonjea}
Simon Gog, Roberto Konow, and Gonzalo Navarro.
\newblock Practical compact indexes for top-k document retrieval.
\newblock {\em ACM J. Exp. Algorithmics}, 22, March 2017.
\newblock URL: \url{https://doi.org/10.1145/3043958}, \href
  {http://dx.doi.org/10.1145/3043958} {\path{doi:10.1145/3043958}}.

\bibitem[GMM16]{NGS}
Sara Goodwin, John~D. McPherson, and W.~Richard McCombie.
\newblock Coming of age: ten years of next-generation sequencing technologies.
\newblock {\em Nature Reviews Genetics}, 17:333--351, 2016.
\newblock URL: \url{https://doi.org/10.1038/nrg.2016.49}, \href
  {http://dx.doi.org/10.1038/nrg.2016.49} {\path{doi:10.1038/nrg.2016.49}}.

\bibitem[GNP20]{rindex}
Travis Gagie, Gonzalo Navarro, and Nicola Prezza.
\newblock Fully functional suffix trees and optimal text searching in
  {BWT}-runs bounded space.
\newblock {\em J. {ACM}}, 67(1):2:1--2:54, 2020.
\newblock URL: \url{https://doi.org/10.1145/3375890}, \href
  {http://dx.doi.org/10.1145/3375890} {\path{doi:10.1145/3375890}}.

\bibitem[HHC{\etalchar{+}}22]{he2021survey}
Shilin He, Pinjia He, Zhuangbin Chen, Tianyi Yang, Yuxin Su, and Michael~R.
  Lyu.
\newblock A survey on automated log analysis for reliability engineering.
\newblock {\em {ACM} Comput. Surv.}, 54(6):130:1--130:37, 2022.
\newblock URL: \url{https://doi.org/10.1145/3460345}, \href
  {http://dx.doi.org/10.1145/3460345} {\path{doi:10.1145/3460345}}.

\bibitem[HT84]{DBLP:journals/siamcomp/HarelT84}
Dov Harel and Robert~Endre Tarjan.
\newblock Fast algorithms for finding nearest common ancestors.
\newblock {\em {SIAM} J. Comput.}, 13(2):338--355, 1984.
\newblock URL: \url{https://doi.org/10.1137/0213024}, \href
  {http://dx.doi.org/10.1137/0213024} {\path{doi:10.1137/0213024}}.

\bibitem[JM09]{DBLP:books/lib/JurafskyM09}
Dan Jurafsky and James~H. Martin.
\newblock {\em Speech and language processing: an introduction to natural
  language processing, computational linguistics, and speech recognition, 2nd
  Edition}.
\newblock Prentice Hall series in artificial intelligence. Prentice Hall,
  Pearson Education International, 2009.
\newblock URL: \url{https://www.worldcat.org/oclc/315913020}.

\bibitem[JMS04]{DBLP:conf/isaac/JaJaMS04}
Joseph~F. J{\'{a}}J{\'{a}}, Christian~Worm Mortensen, and Qingmin Shi.
\newblock Space-efficient and fast algorithms for multidimensional dominance
  reporting and counting.
\newblock In Rudolf Fleischer and Gerhard Trippen, editors, {\em Algorithms and
  Computation, 15th International Symposium, {ISAAC} 2004, Hong Kong, China,
  December 20-22, 2004, Proceedings}, volume 3341 of {\em Lecture Notes in
  Computer Science}, pages 558--568. Springer, 2004.
\newblock URL: \url{https://doi.org/10.1007/978-3-540-30551-4\_49}, \href
  {http://dx.doi.org/10.1007/978-3-540-30551-4\_49}
  {\path{doi:10.1007/978-3-540-30551-4\_49}}.

\bibitem[KH23]{kostic2023mapping}
Daniel Kosti{\'c} and Willem Halffman.
\newblock Mapping explanatory language in neuroscience.
\newblock {\em Synthese}, 202(4):112, 2023.
\newblock \href {http://dx.doi.org/10.1007/s11229-023-04329-6}
  {\path{doi:10.1007/s11229-023-04329-6}}.

\bibitem[KLA{\etalchar{+}}01]{DBLP:conf/cpm/KasaiLAAP01}
Toru Kasai, Gunho Lee, Hiroki Arimura, Setsuo Arikawa, and Kunsoo Park.
\newblock Linear-time longest-common-prefix computation in suffix arrays and
  its applications.
\newblock In Amihood Amir and Gad~M. Landau, editors, {\em Combinatorial
  Pattern Matching, 12th Annual Symposium, {CPM} 2001 Jerusalem, Israel, July
  1-4, 2001 Proceedings}, volume 2089 of {\em Lecture Notes in Computer
  Science}, pages 181--192. Springer, 2001.
\newblock URL: \url{https://doi.org/10.1007/3-540-48194-X\_17}, \href
  {http://dx.doi.org/10.1007/3-540-48194-X\_17}
  {\path{doi:10.1007/3-540-48194-X\_17}}.

\bibitem[KN11]{DBLP:conf/cpm/KreftN11}
Sebastian Kreft and Gonzalo Navarro.
\newblock Self-indexing based on {LZ77}.
\newblock In Raffaele Giancarlo and Giovanni Manzini, editors, {\em
  Combinatorial Pattern Matching - 22nd Annual Symposium, {CPM} 2011, Palermo,
  Italy, June 27-29, 2011. Proceedings}, volume 6661 of {\em Lecture Notes in
  Computer Science}, pages 41--54. Springer, 2011.
\newblock URL: \url{https://doi.org/10.1007/978-3-642-21458-5\_6}, \href
  {http://dx.doi.org/10.1007/978-3-642-21458-5\_6}
  {\path{doi:10.1007/978-3-642-21458-5\_6}}.

\bibitem[KS03]{DBLP:conf/icalp/KarkkainenS03}
Juha K{\"{a}}rkk{\"{a}}inen and Peter Sanders.
\newblock Simple linear work suffix array construction.
\newblock In Jos C.~M. Baeten, Jan~Karel Lenstra, Joachim Parrow, and
  Gerhard~J. Woeginger, editors, {\em Automata, Languages and Programming, 30th
  International Colloquium, {ICALP} 2003, Eindhoven, The Netherlands, June 30 -
  July 4, 2003. Proceedings}, volume 2719 of {\em Lecture Notes in Computer
  Science}, pages 943--955. Springer, 2003.
\newblock URL: \url{https://doi.org/10.1007/3-540-45061-0\_73}, \href
  {http://dx.doi.org/10.1007/3-540-45061-0\_73}
  {\path{doi:10.1007/3-540-45061-0\_73}}.

\bibitem[KSB06]{DBLP:journals/jacm/KarkkainenSB06}
Juha K{\"{a}}rkk{\"{a}}inen, Peter Sanders, and Stefan Burkhardt.
\newblock Linear work suffix array construction.
\newblock {\em J. {ACM}}, 53(6):918--936, 2006.
\newblock URL: \url{https://doi.org/10.1145/1217856.1217858}, \href
  {http://dx.doi.org/10.1145/1217856.1217858}
  {\path{doi:10.1145/1217856.1217858}}.

\bibitem[KU96]{karkkainen1996lempel}
Juha K{\"a}rkk{\"a}inen and Esko Ukkonen.
\newblock {L}empel-{Z}iv parsing and sublinear-size index structures for string
  matching.
\newblock In {\em Proc. 3rd South American Workshop on String Processing
  (WSP)}, pages 141--155, 1996.

\bibitem[LCR24]{li2024predicting}
Jinsen Li, Tsu-Pei Chiu, and Remo Rohs.
\newblock Predicting {DNA} structure using a deep learning method.
\newblock {\em Nature Communications}, 15(1):1243, 2024.
\newblock \href {http://dx.doi.org/10.1038/s41467-024-45191-5}
  {\path{doi:10.1038/s41467-024-45191-5}}.

\bibitem[LP21]{DBLP:conf/esa/LoukidesP21}
Grigorios Loukides and Solon~P. Pissis.
\newblock Bidirectional string anchors: {A} new string sampling mechanism.
\newblock In Petra Mutzel, Rasmus Pagh, and Grzegorz Herman, editors, {\em 29th
  Annual European Symposium on Algorithms, {ESA} 2021, September 6-8, 2021,
  Lisbon, Portugal (Virtual Conference)}, volume 204 of {\em LIPIcs}, pages
  64:1--64:21. Schloss Dagstuhl - Leibniz-Zentrum f{\"{u}}r Informatik, 2021.
\newblock URL: \url{https://doi.org/10.4230/LIPIcs.ESA.2021.64}, \href
  {http://dx.doi.org/10.4230/LIPICS.ESA.2021.64}
  {\path{doi:10.4230/LIPICS.ESA.2021.64}}.

\bibitem[Lue78]{DBLP:conf/focs/Lueker78}
George~S. Lueker.
\newblock A data structure for orthogonal range queries.
\newblock In {\em 19th Annual Symposium on Foundations of Computer Science, Ann
  Arbor, Michigan, USA, 16-18 October 1978}, pages 28--34. {IEEE} Computer
  Society, 1978.
\newblock URL: \url{https://doi.org/10.1109/SFCS.1978.1}, \href
  {http://dx.doi.org/10.1109/SFCS.1978.1} {\path{doi:10.1109/SFCS.1978.1}}.

\bibitem[MLC{\etalchar{+}}21]{matlock2021flanker}
William Matlock, Samuel Lipworth, Bede Constantinides, Timothy~EA Peto, A~Sarah
  Walker, Derrick Crook, Susan Hopkins, Liam~P Shaw, and Nicole Stoesser.
\newblock Flanker: a tool for comparative genomics of gene flanking regions.
\newblock {\em Microbial Genomics}, 7(9):000634, 2021.
\newblock \href {http://dx.doi.org/10.1099/mgen.0.000634}
  {\path{doi:10.1099/mgen.0.000634}}.

\bibitem[MLRD24]{plosgen}
Matthew~P. Moore, Mirjam Laager, Paolo Ribeca, and Xavier Didelot.
\newblock {KmerAperture}: Retaining k-mer synteny for alignment-free extraction
  of core and accessory differences between bacterial genomes.
\newblock {\em PLOS Genetics}, 20(4):1--17, 04 2024.
\newblock URL: \url{https://doi.org/10.1371/journal.pgen.1011184}, \href
  {http://dx.doi.org/10.1371/journal.pgen.1011184}
  {\path{doi:10.1371/journal.pgen.1011184}}.

\bibitem[MM93]{DBLP:journals/siamcomp/ManberM93}
Udi Manber and Eugene~W. Myers.
\newblock Suffix arrays: {A} new method for on-line string searches.
\newblock {\em {SIAM} J. Comput.}, 22(5):935--948, 1993.
\newblock URL: \url{https://doi.org/10.1137/0222058}, \href
  {http://dx.doi.org/10.1137/0222058} {\path{doi:10.1137/0222058}}.

\bibitem[MNT15]{DBLP:conf/cccg/MunroNT15}
J.~Ian Munro, Yakov Nekrich, and Sharma~V. Thankachan.
\newblock Range counting with distinct constraints.
\newblock In {\em Proceedings of the 27th Canadian Conference on Computational
  Geometry, {CCCG} 2015, Kingston, Ontario, Canada, August 10-12, 2015}.
  Queen's University, Ontario, Canada, 2015.
\newblock URL:
  \url{http://research.cs.queensu.ca/cccg2015/CCCG15-papers/44.pdf}.

\bibitem[Nav20]{navarro2020contextual}
Gonzalo Navarro.
\newblock Contextual pattern matching.
\newblock In Christina Boucher and Sharma~V. Thankachan, editors, {\em String
  Processing and Information Retrieval - 27th International Symposium, {SPIRE}
  2020, Orlando, FL, USA, October 13-15, 2020, Proceedings}, volume 12303 of
  {\em Lecture Notes in Computer Science}, pages 3--10. Springer, 2020.
\newblock URL: \url{https://doi.org/10.1007/978-3-030-59212-7\_1}, \href
  {http://dx.doi.org/10.1007/978-3-030-59212-7\_1}
  {\path{doi:10.1007/978-3-030-59212-7\_1}}.

\bibitem[Nav22]{DBLP:journals/csur/Navarro21a}
Gonzalo Navarro.
\newblock Indexing highly repetitive string collections, part {I:}
  repetitiveness measures.
\newblock {\em {ACM} Comput. Surv.}, 54(2):29:1--29:31, 2022.
\newblock URL: \url{https://doi.org/10.1145/3434399}, \href
  {http://dx.doi.org/10.1145/3434399} {\path{doi:10.1145/3434399}}.

\bibitem[NCB24]{NCBI}
{National Center for Biotechnology Information}, 2024.
\newblock URL: \url{https://www.ncbi.nlm.nih.gov/datasets/genome/}.

\bibitem[NTM20]{aaai20}
Atsuyoshi Nakamura, Ichigaku Takigawa, and Hiroshi Mamitsuka.
\newblock Efficiently enumerating substrings with statistically significant
  frequencies of locally optimal occurrences in gigantic string.
\newblock In {\em The Thirty-Fourth {AAAI} Conference on Artificial
  Intelligence, {AAAI} 2020, The Thirty-Second Innovative Applications of
  Artificial Intelligence Conference, {IAAI} 2020, The Tenth {AAAI} Symposium
  on Educational Advances in Artificial Intelligence, {EAAI} 2020, New York,
  NY, USA, February 7-12, 2020}, pages 5240--5247. {AAAI} Press, 2020.
\newblock URL: \url{https://doi.org/10.1609/aaai.v34i04.5969}, \href
  {http://dx.doi.org/10.1609/AAAI.V34I04.5969}
  {\path{doi:10.1609/AAAI.V34I04.5969}}.

\bibitem[piz]{pizzachili}
{Pizza \& Chili} corpus.
\newblock \url{https://pizzachili.dcc.uchile.cl/repcorpus.html}.

\bibitem[sds]{sdsl}
{SDSL} {Github} repository.
\newblock \url{https://github.com/simongog/sdsl-lite}.

\bibitem[SPH16]{maxcfp}
Jingbo Shang, Jian Peng, and Jiawei Han.
\newblock {MACFP:} maximal approximate consecutive frequent pattern mining
  under edit distance.
\newblock In Sanjay~Chawla Venkatasubramanian and Wagner~Meira Jr., editors,
  {\em Proceedings of the 2016 {SIAM} International Conference on Data Mining,
  Miami, Florida, USA, May 5-7, 2016}, pages 558--566. {SIAM}, 2016.
\newblock URL: \url{https://doi.org/10.1137/1.9781611974348.63}, \href
  {http://dx.doi.org/10.1137/1.9781611974348.63}
  {\path{doi:10.1137/1.9781611974348.63}}.

\bibitem[SS82]{DBLP:journals/jacm/StorerS82}
James~A. Storer and Thomas~G. Szymanski.
\newblock Data compression via textual substitution.
\newblock {\em J. {ACM}}, 29(4):928--951, 1982.
\newblock URL: \url{https://doi.org/10.1145/322344.322346}, \href
  {http://dx.doi.org/10.1145/322344.322346} {\path{doi:10.1145/322344.322346}}.

\bibitem[ST83]{DBLP:journals/jcss/SleatorT83}
Daniel~Dominic Sleator and Robert~Endre Tarjan.
\newblock A data structure for dynamic trees.
\newblock {\em J. Comput. Syst. Sci.}, 26(3):362--391, 1983.
\newblock URL: \url{https://doi.org/10.1016/0022-0000(83)90006-5}, \href
  {http://dx.doi.org/10.1016/0022-0000(83)90006-5}
  {\path{doi:10.1016/0022-0000(83)90006-5}}.

\bibitem[Vit06]{DBLP:journals/fttcs/Vitter06}
Jeffrey~Scott Vitter.
\newblock Algorithms and data structures for external memory.
\newblock {\em Found. Trends Theor. Comput. Sci.}, 2(4):305--474, 2006.
\newblock URL: \url{https://doi.org/10.1561/0400000014}, \href
  {http://dx.doi.org/10.1561/0400000014} {\path{doi:10.1561/0400000014}}.

\bibitem[voy24]{voyant}
Voyant tools documentation, 2024.
\newblock URL: \url{https://voyant-tools.org/docs/\#!/guide/contexts}.

\bibitem[Wei73]{DBLP:conf/focs/Weiner73}
Peter Weiner.
\newblock Linear pattern matching algorithms.
\newblock In {\em 14th Annual Symposium on Switching and Automata Theory, Iowa
  City, Iowa, USA, October 15-17, 1973}, pages 1--11, 1973.
\newblock URL: \url{https://doi.org/10.1109/SWAT.1973.13}, \href
  {http://dx.doi.org/10.1109/SWAT.1973.13} {\path{doi:10.1109/SWAT.1973.13}}.

\bibitem[YLWT11]{gis}
Josh~Jia{-}Ching Ying, Wang{-}Chien Lee, Tz{-}Chiao Weng, and Vincent~S. Tseng.
\newblock Semantic trajectory mining for location prediction.
\newblock In Isabel~F. Cruz, Divyakant Agrawal, Christian~S. Jensen, Eyal Ofek,
  and Egemen Tanin, editors, {\em 19th {ACM} {SIGSPATIAL} International
  Symposium on Advances in Geographic Information Systems, {ACM-GIS} 2011,
  November 1-4, 2011, Chicago, IL, USA, Proceedings}, pages 34--43. {ACM},
  2011.
\newblock URL: \url{https://doi.org/10.1145/2093973.2093980}, \href
  {http://dx.doi.org/10.1145/2093973.2093980}
  {\path{doi:10.1145/2093973.2093980}}.

\bibitem[Zak19]{zaker2019online}
Farzin Zaker.
\newblock Online shopping store-web server logs.
\newblock {\em Harvard Dataverse}, 2019.

\bibitem[ZL77]{DBLP:journals/tit/ZivL77}
Jacob Ziv and Abraham Lempel.
\newblock A universal algorithm for sequential data compression.
\newblock {\em {IEEE} Trans. Inf. Theory}, 23(3):337--343, 1977.
\newblock URL: \url{https://doi.org/10.1109/TIT.1977.1055714}, \href
  {http://dx.doi.org/10.1109/TIT.1977.1055714}
  {\path{doi:10.1109/TIT.1977.1055714}}.

\bibitem[ZYHY07]{icdm2007}
Feida Zhu, Xifeng Yan, Jiawei Han, and Philip~S. Yu.
\newblock Efficient discovery of frequent approximate sequential patterns.
\newblock In {\em Proceedings of the 7th {IEEE} International Conference on
  Data Mining {(ICDM} 2007), October 28-31, 2007, Omaha, Nebraska, {USA}},
  pages 751--756. {IEEE} Computer Society, 2007.
\newblock URL: \url{https://doi.org/10.1109/ICDM.2007.75}, \href
  {http://dx.doi.org/10.1109/ICDM.2007.75} {\path{doi:10.1109/ICDM.2007.75}}.

\end{thebibliography}

\end{document}